\documentclass[reqno]{article}

\usepackage{amsthm}
\usepackage{amsmath}
\usepackage{amssymb}
\usepackage{graphicx}
\usepackage{color}
\usepackage{bbm}
\usepackage[top=2.5cm, bottom=2.5cm, left=2.5cm, right=2.5cm]{geometry}
\usepackage{enumerate}
\usepackage{faktor}
\usepackage{nicefrac}
\usepackage[nottoc]{tocbibind}
\usepackage{slashed}
\usepackage{authblk}
\usepackage[normalem]{ulem}
\usepackage{bbm}
\usepackage{mathabx}
\usepackage{comment}
\usepackage{float}

\usepackage[scr=euler]{mathalfa}

\def\mathclap#1{\text{\hbox to 0pt{\hss$\mathsurround=0pt#1$\hss}}}

\newcommand{\Z}{\mathbb{Z}}
\newcommand{\N}{\mathbb{N}}
\newcommand{\R}{\mathbb{R}}

\newcommand{\C}{\mathbb{C}}

\newcommand{\CH}{\mathcal{CH}^+}
\newcommand{\Hp}{\mathcal{H}^+}
\newcommand{\rred}{r_\mathrm{red}}

\newcommand{\T}{\mathcal{T}}
\newcommand{\M}{\underline{\overline{\mathcal{M}}}}
\newcommand{\Mun}{\underline{\mathcal{M}}}
\newcommand{\vol}{\mathrm{vol}}

\newcommand{\Si}{\mathscr{S}}
\newcommand{\Hpl}{\mathcal{H}^+_\ell}
\newcommand{\Hpr}{\mathcal{H}^+_r}

\newcommand{\Rea}{\mathfrak{Re}}

\newcommand{\Sp}{\mathbb{S}}
\newcommand{\Li}{\mathcal{L}}
\newcommand{\SWs}{\mathscr{I}_{[s]}^\infty(\Sp^2)} 
\newcommand{\vols}{\mathrm{vol}_{\mathbb{S}^2}}
\newcommand{\hp}{\hat{\psi}}
\newcommand{\eai}{\underset{a.i.}{=}}
\newcommand{\swl}{\mathring{\slashed{\Delta}}_{[s]}}
\newcommand{\zt}{\widetilde{Z}}
\newcommand{\ered}{r_{\mathrm{ered}}}
\newcommand{\Sml}{S^{[s]}_{ml}}
\newcommand{\mms}{\max(|m|, |s|)}
\newcommand{\Smlp}{S^{[s]}_{ml'}}

\newcommand{\rd}{\partial}
\newcommand{\da}{\dot{\alpha}}

\begin{document}

\numberwithin{equation}{section}
\newtheorem{theorem}[equation]{Theorem}
\newtheorem{remark}[equation]{Remark}
\newtheorem{assumption}[equation]{Assumption}
\newtheorem{claim}[equation]{Claim}
\newtheorem{lemma}[equation]{Lemma}
\newtheorem{definition}[equation]{Definition}
\newtheorem{corollary}[equation]{Corollary}
\newtheorem{proposition}[equation]{Proposition}
\newtheorem*{theorem*}{Theorem}
\newtheorem*{conjecture}{Strong cosmic censorship conjecture}

\setcounter{tocdepth}{3}

\title{Instability of the Kerr Cauchy horizon under linearised gravitational perturbations}
\author{Jan Sbierski\thanks{School of Mathematics, 
University of Edinburgh,
James Clerk Maxwell Building,
Peter Guthrie Tait Road, 
Edinburgh, 
EH9 3FD,
United Kingdom}}
\date{\today}

\maketitle

\begin{abstract}
This paper establishes a mathematical proof of the blue-shift instability at the sub-extremal Kerr Cauchy horizon for the linearised vacuum Einstein equations. More precisely, we exhibit conditions on the $s=+2$ Teukolsky field, consisting of suitable integrated upper and lower bounds on the decay along the event horizon, that ensure that the Teukolsky field, with respect to a frame that is regular at the Cauchy horizon, becomes singular. The conditions are in particular satisfied by solutions of the Teukolsky equation arising from generic and compactly supported initial data by the recent work \cite{MaZha21} of Ma and Zhang for slowly rotating Kerr. 
\end{abstract}

\tableofcontents

\section{Introduction}

The sub-extremal Kerr solution of the vacuum Einstein equations 
$$Ric(g) = 0$$
models a stationary and rotating black hole, devoid of any gravitational radiation. While we expect that the exterior is stable if small gravitational radiation is taken into account\footnote{See \cite{KlaiSzef17}, \cite{DafHolRodTay21}, \cite{KlaiSze21} for recent results on the black hole stability problem.}, heuristics going back to Penrose \cite{Pen68} indicate that the interior is subject to a blue-shift instability: gravitational radiation entering the black hole builds up at the Cauchy horizon $\CH$ and leads to the formation of a singularity. 
\begin{figure}[h]
  \centering
 \def\svgwidth{8cm}
    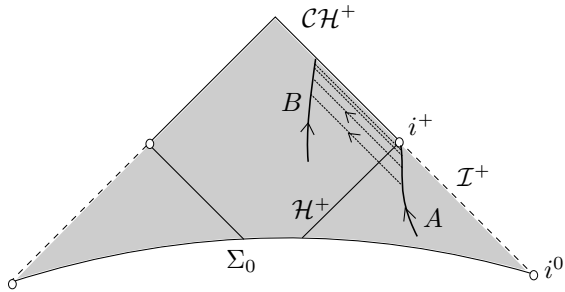
      \caption{The blue-shift effect. For observer $A$ an infinite time passes, while observer $B$ reaches the Cauchy horizon in finite time; signals sent by $A$ are received by $B$ shifted to the blue.} \label{FigHeu}
\end{figure}
Although the full resolution of this conjecture is still open, a large body of research concerning simplified models has since lent support to the validity of this scenario. The first class of simplified models we would like to mention here concerns non-linear spherically symmetric perturbations of the sub-extremal Reissner-Nordstr\"om black hole, which also possesses a Cauchy horizon in its interior that is subject to a blue-shift instability. The works by Hiscock \cite{His81}, Poisson-Israel \cite{PoiIs89}, \cite{PoiIs90}, and Ori \cite{Ori91} investigate and prove this blue-shift instability for the spherically symmetric Einstein-Maxwell-null dust system and the works of Dafermos \cite{Daf03}, \cite{Daf05a} and of Luk-Oh \cite{LukOh19I}, \cite{LukOh19II}\footnote{See also \cite{LukOh15} for the linearised case.}  do so for the spherically symmetric Einstein-Maxwell-scalar field system. The second class of simplified models  are linear models on a Kerr background -- and in particular the linear scalar wave equation which serves as a ``poor man's linearisation" of the vacuum Einstein equations. The study was initiated by McNamara \cite{McNam78}, who indeed also considers gravitational perturbations. Results of a similar nature for the scalar wave equation were proven by Dafermos-Shlapentokh--Rothman \cite{DafShla17} and in \cite{Sbie13b}. These results all have in common that they only ensure the abstract existence of solutions that become singular at the Cauchy horizon, but they do not provide explicit criteria that ensure that a particular solution becomes singular. This gap was filled for the scalar wave equation in collaboration with Luk in \cite{LukSbi15}, which shows that under the assumption of suitable upper and lower bounds on the decay along the event horizon, the energy of the scalar field becomes unbounded at the Cauchy horizon. (The wave itself remains bounded \cite{Hintz17}, \cite{Fra20}.) It was later shown by Hintz \cite{Hintz20} and Angelopoulos-Aretakis-Gajic \cite{AAG21} that the assumed bounds on the event horizon are generically satisfied. 

The present work makes the step from the scalar wave equation to linearised gravitational perturbations in the form of the Teukolsky field \cite{Teuk73}. Analogously to \cite{LukSbi15} we exhibit conditions on the Teukolsky field along the event horizon, consisting of integrated upper and lower bounds on the decay, which ensure the blow-up of the Teukolsky field at the Cauchy horizon. More precisely, we show
\begin{theorem}\label{ThmInt} 
Assume $\psi$ satisfies the Teukolsky equation with $s=+2$ and, along the event horizon $\Hp$,
\begin{itemize}
\item assume that there exists $p \in \N$ s.t.\ $\int_{\Hp \cap \{v_+ \geq 1\}} v_+^{2p} |\psi|^2 \, \vols dv_+ = \infty\;.$ 
Let $p_0$ be the smallest such integer and assume $p_0 \geq 2$,
\item $\int_{\Hp \cap \{v_+ \geq 1\}} v_+^{2p_0} |\psi_{S(m_0l_0)}|^2 \, dv_+ = \infty$ for some $m_0 \in \mathbb{Z}$, $l_0 \geq \max\{2, |m_0|\}$, where $\psi_{S(ml)}$ denotes the projection of $\psi$ on the $(m,l)$ spin $2$-weighted \emph{spherical} harmonic,
\item $\int_{\Hp \cap \{v_+ \geq 1\}} v_+^{2p_0} | \rd_{v_+} \psi|^2 \, \vols dv_+ < \infty$,
\item $\int_{\Hp \cap \{v_+ \geq 1\}} v_+^{q_r} |\partial^k \psi|^2 \, \vols dv_+ < \infty$ for some $2 < q_r < 2p_0$ with $q_r \in \R$ and for all $k = 0, 1, \ldots, 7$.\footnote{See assumption \eqref{AssumpDecayRHp} on page \pageref{AssumpDecayRHp} for the precise statement.} 
\end{itemize}
It then follows that 
\begin{equation}
\label{EqConCCCC}
\int_{\Sigma \cap \{v_+ \geq 1\}} v_+^{2p_0} | \psi|^2 \, \vols dv_+ = \infty\;,
\end{equation} where $\Sigma$ is a hypersurface transversal to $\CH$ as in Figure \ref{FigInt}.
\end{theorem}
\begin{figure}[h]
  \centering
 \def\svgwidth{2.7cm}
    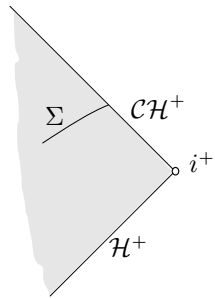
      \caption{The statement of Theorem \ref{ThmInt}. } \label{FigInt}
\end{figure}
Here, $v_+ = t + r^*$ and $\rd_{v_+}$ is the Killing vector field which is a time-translation at spatial infinity, see also Section \ref{SecManifoldMetric}. We also refer the reader to Theorem \ref{Thm2} in Section \ref{SecAssumptions} for the precise statement of Theorem \ref{ThmInt}.

We would like to bring to the reader's attention that the coordinate $v_+$ is not regular at the Cauchy horizon. There, $V^+_{r_-} = - e^{\kappa_- v_+}$ is a regular boundary defining function with $\{V^+_{r_-} = 0\}$ being the Cauchy horizon. The constant $\kappa_- <0$ is the surface gravity of $\CH$. Moreover, the regular Teukolsky field $\hat{\psi}$ at $\CH$, i.e., the linearisation of the Teukolsky $s=+2$ curvature component with respect to a regular frame at $\CH$, is given by $e^{-2 \kappa_- v_+} \psi = \frac{1}{(V^+_{r_-})^2} \psi$, modulo a regular factor which remains bounded away from zero (and infinity) at $\CH$. We thus obtain that the conclusion \eqref{EqConCCCC} of Theorem \ref{ThmInt} with respect to regular quantities at $\CH$ reads
\begin{equation}\label{EqRegCH}
\int_{\Sigma \cap \{v_+ \geq 1\}} \big[\log(-V^+_{r_-})\big]^{2p_0} (-V^+_{r_-})^3 |\hat{\psi}|^2 \, \vols dV^+_{r_-} = \infty \;,
\end{equation}
which makes manifest the blow-up of the  Teukolsky field with respect to a regular frame at the Cauchy horizon.

Moreover, we note that in the slowly rotating case  the assumptions made in Theorem \ref{ThmInt} were recently shown to be satisfied generically (\cite{MaZha21} and \cite{DafHolRod17}, \cite{Ma20}) for solutions arising from compactly supported initial data on a global Cauchy hypersurface $\Sigma_0$ as in Figure \ref{FigHeu} with $p_0 = 7$, $l_0 =2$ and $m_0 \in \{-2,-1,1,2\}$. The parameter $q_r$ can be chosen to be anything strictly less than $13$. See also Remark \ref{RemMainThm} for further discussion.

Let us also remark that we expect Theorem \ref{ThmInt} to be an important ingredient in the analysis of the blue-shift instability at the Cauchy horizon for the full non-linear vacuum Einstein equations.

\subsection{The case of the full non-linear Einstein equations}

Standard energy estimates entail that  solutions of \emph{linear} equations arising from regular initial Cauchy data can at most become singular at the (\emph{null}) boundary of the black hole interior, i.e., at the Cauchy horizon of Kerr -- but not earlier inside the black hole. For the vacuum Einstein equations, however, which are non-linear, it is a priori conceivable that the non-linearities amplify the blow-up and lead to the formation of a  singularity in the black hole interior which is everywhere \emph{spacelike}. Whether this happens or not has been contentious for a long time.

For the spherically symmetric Einstein-Maxwell-scalar field system numerical evidence was presented in \cite{BraSmith95} which indicated that the non-linearities do not amplify the blow-up in the sense that one always has a piece of a null singularity emanating from timelike infinity in the Penrose diagram\footnote{Which can later on collapse to a spacelike singularity, see also \cite{Ori91} and the recent \cite{VdM19}.}. This scenario in spherical symmetry was later rigorously confirmed in the works \cite{Daf03}, \cite{Daf05a},  \cite{LukOh19I}, \cite{LukOh19II}. Indeed, if one only considers sufficiently small perturbations of two-ended sub-extremal Reissner-Nordstr\"om initial data, then the singularity only occurs along the bifurcate Cauchy horizon, i.e., there is no piece of the singularity which is spacelike, see \cite{Daf14}.

Concerning the vacuum Einstein equations Dafermos and Luk established the following seminal result:
\begin{theorem}[Dafermos-Luk, \cite{DafLuk17}] \label{ThmDafLuk}
Consider a suitable spacelike hypersurface $\Sigma$ in the interior of a sub-extremal Kerr black hole, see Figure \ref{FigDafLuk1}, and consider small perturbations of the induced initial data which decay towards $i^+$ with a rate that is in particular compatible with what is expected to arise dynamically  from small perturbations of exact sub-extremal Kerr initial data on a global Cauchy hypersurface $\Sigma_0$ as in Figure \ref{FigHeu}. Then the maximal globally hyperbolic development of the perturbed initial data contains a region which is $C^0$-close to, and the Penrose diagram of which is given by, the darker shaded region of the unperturbed sub-extremal Kerr spacetime as in Figure \ref{FigDafLuk1}.
\end{theorem}
This result in particular entails that also for the vacuum Einstein equations, and under the assumptions of their theorem, the non-linearities do not amplify the blow-up to create a spacelike singularity emanating from timelike infinity in the Penrose diagram (cf.\ in Figure \ref{FigDafLuk2}). The result is only compatible with a null singularity emanating from timelike infinity  (i.e.\ the Cauchy horizon becoming singular) as in the spherically symmetric case. But whether  the Cauchy horizon is indeed generically singular is not established in \cite{DafLuk17}. The result obtained in this paper is a first step in this direction.
\begin{figure}[h]
\centering
\begin{minipage}{.5\textwidth}
  \centering
 \def\svgwidth{3.7cm}
    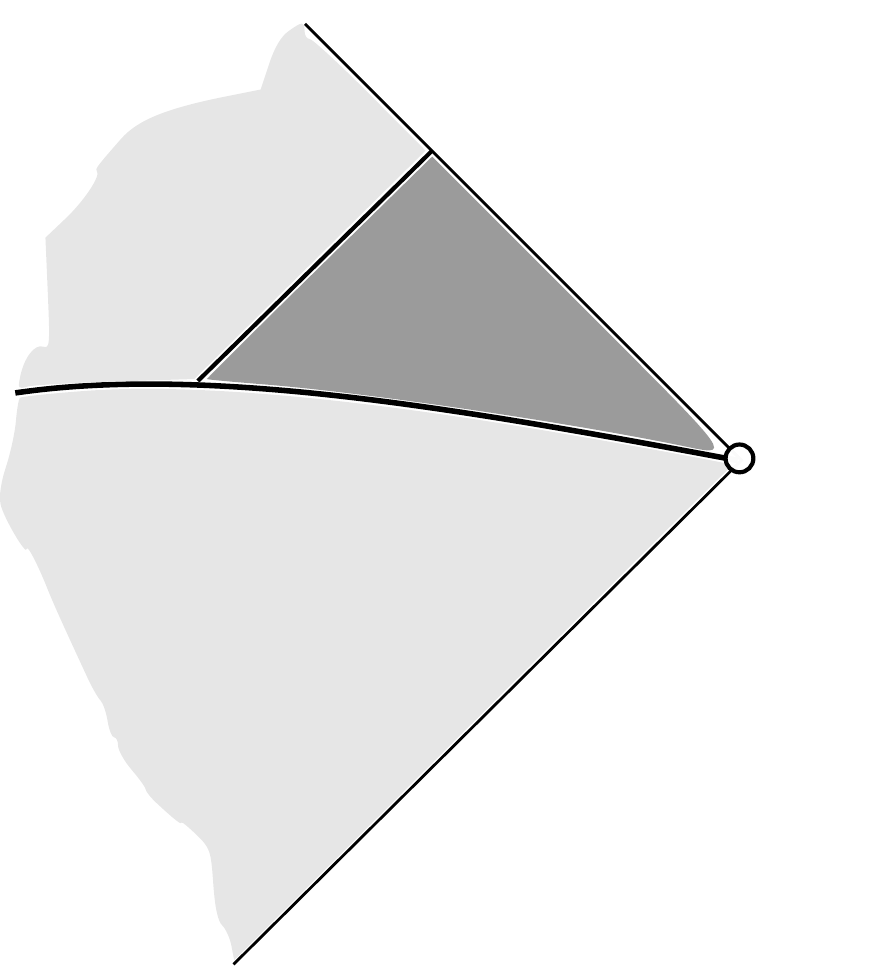
      \caption{The statement of Theorem \ref{ThmDafLuk}} \label{FigDafLuk1}
\end{minipage}%
\begin{minipage}{.5\textwidth}
  \centering
  \def\svgwidth{3.7cm}
    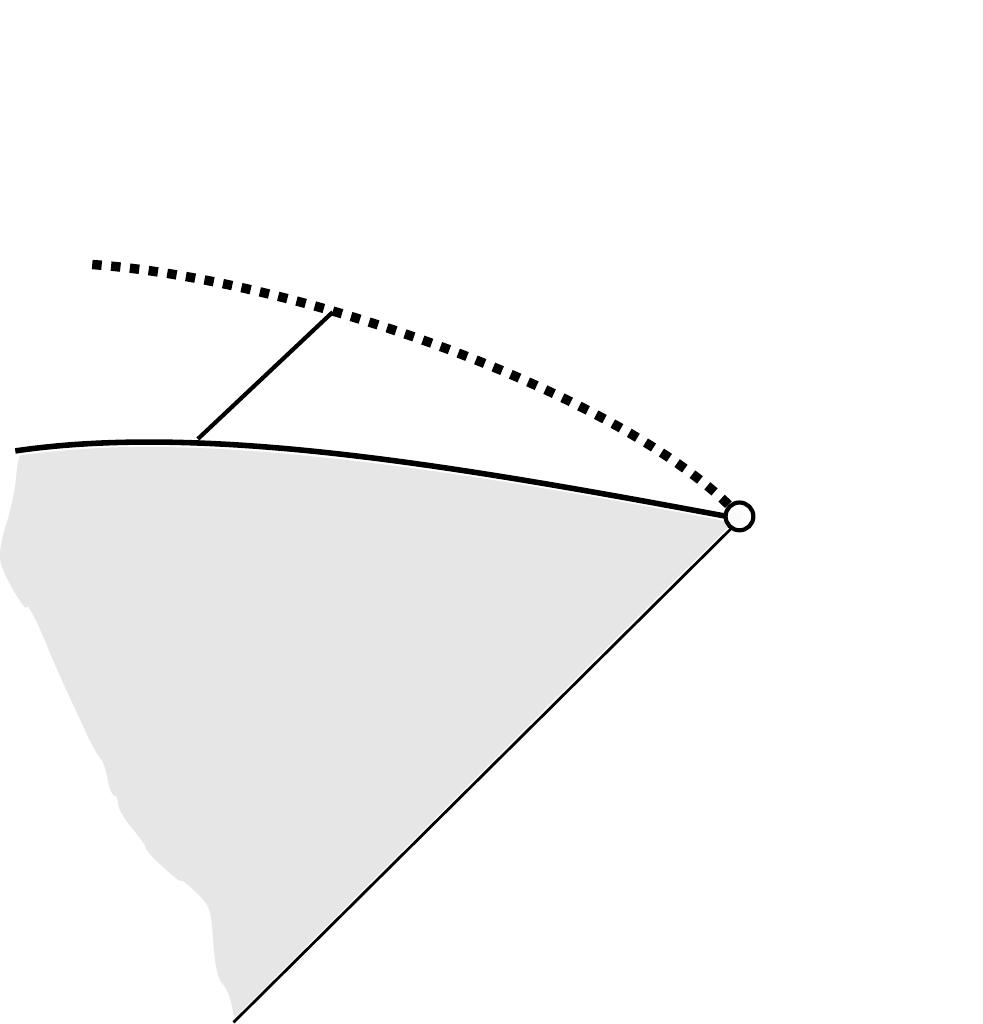
      \caption{Spacelike singularity emanating from $i^+$ is ruled out. Picture \emph{cannot} occur.} \label{FigDafLuk2}
\end{minipage}
\end{figure}

Note that Theorem \ref{ThmDafLuk} also shows that the metric remains continuous up to and including the Cauchy horizon. Thus, if a singularity forms, it is not at the level of the metric itself, as is the case for example for the Schwarzschild singularity (see \cite{Sbie15}, \cite{Sbie18}), but we expect that it is the connection which will generically become singular. This expectation is mainly based on the spherically symmetric models discussed earlier for which one also obtains that the metric extends continuously to the Cauchy horizon but the connection becomes unbounded \cite{Daf03}, \cite{Daf05a},  \cite{LukOh19I}, \cite{LukOh19II}, \cite{Sbie20}. Such singularities have been termed `weak  null singularities'. 
The \emph{construction} of weak null singularities in vacuum spacetimes without any symmetry was achieved in \cite{Luk18}, where it was also shown that they \emph{propagate} (for some finite time). We expect that such weak null singularities as given in \cite{Luk18} do generically \emph{form}  at the Cauchy horizon of perturbed Kerr. 

\subsection{Relation to the strong cosmic censorship conjecture}

Going back to the result of this paper in the form of \eqref{EqRegCH}, and \emph{if one trusts the naive expectation} that there is a linearised Christoffel symbol which is better than $\hat{\psi}$ by a power of $V^+_{r_-}$, i.e., of order $V^+_{r_-} \hat{\psi}$, then \eqref{EqRegCH}  shows that this linearised Christoffel symbol is not in $L^2_{\mathrm{loc}}$ at the Cauchy horizon with respect to the differentiable structure of the background. This makes contact with the modern formulation of the strong cosmic censorship conjecture:
\begin{conjecture}
The maximal globally hyperbolic development arising from generic asymptotically flat initial data for the vacuum Einstein equations is inextendible as a Lorentzian manifold with a continuous metric and locally square integrable Christoffel symbols.
\end{conjecture}
The strong cosmic censorship conjecture was originally conceived by Penrose \cite{Pen79}, the formulation given here in terms of the initial value problem and the conjectured breakdown of the regularity goes back to Christodoulou \cite{Chris09} and Chrusciel \cite{Chrus91}. The inextendibility as a Lorentzian manifold with $g \in C^0$ and $\partial g \in L^2_{\mathrm{loc}}$ in particular rules out the extension of the maximal globally hyperbolic development as a weak solution\footnote{See for example \cite{HawkEllis} or the introduction of \cite{Sbie20}.}. We note that for exact sub-extremal Kerr initial data the maximal globally hyperbolic development is given in Figure \ref{FigHeu} and \emph{is in fact extendible in various ways} across the Cauchy horizon even as a smooth solution: determinism is violated.
However, as we discussed earlier, for generic small perturbations of exact sub-extremal Kerr initial data we expect the blue-shift instability to turn the Cauchy horizon into a weak null singularity and in this way preventing non-unique extensions as weak solutions. Determinism would thus be restored generically. 

The result obtained in this paper can be thought of as a first step towards establishing the generic divergence of curvature at the Cauchy horizon of non-linearly perturbed sub-extremal Kerr -- and thus the generic inextendibility as a Lorentzian manifold with $g \in C^2$. And with the earlier naive expectation that there is a (linearised) Christoffel symbol of order $V^+_{r_-} \hat{\psi}$ it is also a first step towards showing that the metric cannot be extended with $g \in C^0$ and $\partial g \in L^2_{\mathrm{loc}}$ \emph{in a particular natural-looking coordinate system}. However, the result does not contribute to developing methods which show that \emph{no matter what coordinate system is chosen for the extension}, the metric cannot be extended in $g \in C^0$ and $\partial g \in L^2_{\mathrm{loc}}$. This is an open problem. For recent progress in this direction we refer the reader to \cite{Sbie20}.



\subsection{Related results and directions concerning the interior of black holes}

The studies mentioned earlier on perturbations of sub-extremal Reissner Nordstr\"om under the spherically symmetric Einstein-Maxwell-scalar field system were extended in \cite{VdM18}  to the spherically symmetric Einstein-Maxwell-massive and charged scalar field system. This matter model in particular allows for asymptotically flat one-ended spherically symmetric black hole solutions which possess a Cauchy horizon and is thus a good model to understand the contraction and breakdown of weak null singularities in the interior of black holes \cite{VdM19}.

For the behaviour of  linear waves and of axisymmetric and polarized perturbations in the interior of \emph{non-rotating} (Schwarzschild) black holes see  \cite{FouSbi20}, \cite{AleFourno20}.

Another interesting direction of research concerns the interior of extremal black holes where the blue-shift instability at the Cauchy horizon is much weaker than in the sub-extremal case. For results concerning linear waves see \cite{Gajic17}, \cite{Gajic18} and for non-linear results in spherical symmetry see \cite{GaLuk19}.

Finally, for the investigation of the blue-shift instability in the presence of a cosmological constant $\Lambda$ we refer the reader to \cite{Daf14}, \cite{HinVas17}, \cite{CaJoHi18}, \cite{DiEpReSa18}, \cite{DafShla18}  for $\Lambda >0$ and to \cite{Kehle21}, \cite{Kehle21a} for $\Lambda <0$ as well as to the references given in those papers.

\subsection{Outline of proof}

A good, simple, and instructive model problem for gravitational perturbations in the interior of a subextremal rotating Kerr black hole  is the spherically symmetric scalar wave equation in the interior of a subextremal charged Reissner-Nordstr\"om black hole. The blue-shift instability in this scenario is well-established and various results along with various methods of proof have been developed: the methods in  \cite{McNam78}, \cite{DafShla17} are based on the scattering map from characteristic initial data on the right even horizon $\Hp_r$ (past null infinity $\mathcal{I}^-$) to the trace of the wave on the left Cauchy horizon $\CH_l$, making crucial use of the time-translation invariance of this map. See Figure \ref{FigRN} below for the notation. The $C^1$-instability results in \cite{ChandHart82}, \cite{KehleShla19} are also obtained via scattering theory together with meromorphic continuation. One can also use the geometric optics (Gaussian beam) approximation together with an application of the closed graph theorem, see \cite{Sbie13b} and the introduction of \cite{LukSbi15}, to capture a formulation of the blue-shift instability. In \cite{LukOh15} a neat argument by contradiction is given, using that one can solve the linear wave equation in spherical symmetry sideways. A proof in physical space using energy estimates and at the heart of which is the conservation law associated to the spacelike Killing vector field $\rd_t$ is presented in \cite{LukSbi15}. And finally, in \cite{LukOhShlaForth}, Luk, Oh, and Shlapentokh--Rothman give another scattering theoretic proof of the blue-shift instability at the Cauchy horizon. It is this last method of proof which is being taken up in this paper and being implemented for the Teukolsky equation on Kerr. In the following we shall first outline the argument from \cite{LukOhShlaForth} in spherical symmetry and then discuss the main differences to the proof in this paper.

\subsubsection{Spherically symmetric scalar waves on Reissner-Nordstr\"om} \label{SecRN}

The \emph{interior of a charged subextremal Reissner-Nordstr\"om black hole} is the Lorentzian manifold\footnote{The definitions of symbols made here are only valid in this section. In the rest of the paper we will use $\mathcal{M}, g, r_+, r_-,$ etc.\ to refer to objects and quantities on Kerr.} $(\mathcal{M},g)$, where $\mathcal{M} := \R \times (r_-, r_+) \times \Sp^2$ with standard $(t,r, \theta, \varphi)$-coordinates and $r_\pm := M \pm \sqrt{M^2 - e^2}$, where $0 < |e| < M$ are real parameters modelling the charge and the mass of the black hole, respectively. The Lorentzian metric $g$ is given by $$g := - \frac{\Delta}{r^2} \,dt^2 + \frac{r^2}{\Delta} \,dr^2 + r^2 \,(d \theta^2 + \sin^2 \theta \, d \varphi^2)\;,$$
with $\Delta := r^2 - 2Mr + e^2$. 
The \emph{spherically symmetric scalar wave equation} $\Box_g \phi = 0$, where $\phi : \mathcal{M} \to \C$ is only a function of $t$ and $r$, takes in the above coordinates the form
\begin{equation}\label{EqWaveEqRN}
0 = \Box_g \phi = - \frac{r^2}{\Delta} \rd_t^2 \phi + \frac{1}{r^2} \rd_r (\Delta \rd_r \phi) \;.
\end{equation}
Let $r^*(r)$ be a function with $\frac{dr^*}{dr} = \frac{r^2}{\Delta}$ and then introduce the null coordinates $v := r^* + t$ and $u:= r^* - t$. 
We define $\kappa_\pm := \frac{r_\pm - r_\mp}{2r_\pm^2}$ and use those to introduce the Kruskal-like null coordinates\footnote{See also \cite{HawkEllis} for a more detailed discussion of the Reissner-Nordstr\"om spacetime.} $V_{r_+} := e^{\kappa_+ v}$ and $U_{r_+} := e^{\kappa_+ u}$ in which the Lorentzian manifold $(\mathcal{M},g)$ extends analytically to $r=r_+$ ($r$ as a function of $V_{r_+}, U_{r_+}$) and similarly $V_{r_-} := -e^{\kappa_- v}$ and $U_{r_-} := -e^{\kappa_- u}$ in which the Lorentzian manifold extends analytically to $r=r_-$. The boundary null hypersurface $\{V_{r_+} = 0\} =: \Hp_l$, at which we have $r = r_+$, is called the \emph{left event horizon}, the boundary null hypersurface $\{U_{r_+} = 0\} =: \Hp_r$, at which we also have $r=r_+$, the \emph{right event horizon}, and the boundary sphere $\{V_{r_+} = U_{r_+} = 0\} =:\Sp^2_b$ is the \emph{bottom bifurcation sphere}. Moreover, we call the boundary null hypersurface $\{U_{r_-} = 0 \} =: \CH_l$ the \emph{left Cauchy horizon}, the boundary null hypersurface $\{V_{r_-} = 0 \} =: \CH_r$ the \emph{right Cauchy horizon}, and the boundary sphere $\{V_{r_-} = U_{r_-} = 0\} =: \Sp_t^2$ the \emph{top bifurcation sphere}. A Penrose diagram of $(\mathcal{M},g)$ with the boundaries attached is given in Figure \ref{FigRN} below.
\begin{figure}[h]
 \centering
 \def\svgwidth{6cm}
    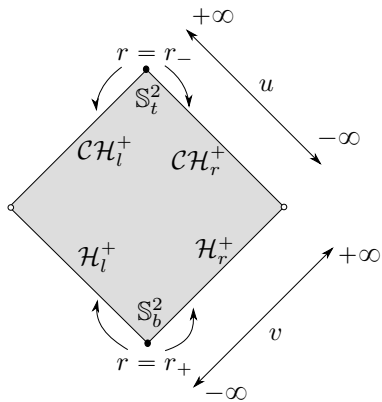
      \caption{The interior of subextremal Reissner-Nordstr\"om} \label{FigRN}
\end{figure}
In \cite{LukOhShlaForth} the following theorem is shown
\begin{theorem}[Corollary 4.2 in \cite{LukOhShlaForth}] \label{ThmRN1}
Consider the region $\big(\mathcal{M} \cup \Hp_r \big) \cap \{v \geq v_0\} \cap \{u \leq u_1\}$ for some $v_0, u_1 \geq 1$ and let $\phi$ be a smooth solution of the spherically symmetric wave equation \eqref{EqWaveEqRN} in this region, which, moreover, satisfies
\begin{equation}\label{EqRNAssump1}
\lim\limits_{v \to \infty} \phi|_{\Hp_r}(v) = 0 \qquad \textnormal{ and } \qquad \int\limits_{v_0}^\infty v^2 |\rd_v \phi|_{\Hp_r}|^2 \,  dv < \infty
\end{equation}
and there exists $\N \ni p_0 \geq 2$ such that
\begin{equation}\label{EqRNAssump2}
\int\limits_{v_0}^\infty v^{2p_0} |\rd_v \phi|_{\Hp_r}|^2 \, dv = \infty
\end{equation}
holds. We further assume that $p_0$ is the smallest such integer with this property. And finally we assume 
\begin{equation}\label{EqRNAssump3}
\int\limits_{v_0}^\infty v^{2p_0}|\rd_v^2 \phi|_{\Hp_r}|^2 \,dv < \infty \;.
\end{equation}
Then for any $u_2 \leq u_1$ we have
\begin{equation}\label{EqRNConcThm1}
\int\limits_{v_0}^\infty v^{2p_0} \big|\rd_v \phi\big|^2(u_2,v) \, dv = \infty \;.
\end{equation}
\end{theorem}
This is the local statement that is the analogue of Theorem \ref{Thm1} (or \ref{ThmInt}) for Teukolsky. It is inferred from the following global statement, which is the analogue of Theorem \ref{Thm2} for Teukolsky, by an extension procedure of the solution.
\begin{theorem}[Theorem 4.1 in \cite{LukOhShlaForth}] \label{ThmRN2}
Let $\phi$ be a smooth solution of the spherically symmetric wave equation \eqref{EqWaveEqRN} on $\mathcal{M} \cup \Hp_r \cup \Hp_l$. Suppose that in addition to \eqref{EqRNAssump1}, \eqref{EqRNAssump2}, \eqref{EqRNAssump3} (for some $v_0 \geq 1$) we also have that there exists $v_* \in \R$ such that $\phi|_{\Hp_r}(v) = 0$ for $v \leq v_*$ and there exists $u_* \in \R$ such that $\phi|_{\Hp_l}(u) = 0$ for $u \geq u_*$.\footnote{The  important properties here are that $\phi$ vanishes at the bottom bifurcation sphere $\Sp^2_b$ and decays sufficiently fast along $\Hp_l$. The first one is not strictly necessary, but simplifies the proof.} Then \eqref{EqRNConcThm1} holds for any $u_2 \in \R$ (and any $v_0 \in \R$).
\end{theorem}
Before we discuss the structure of the proof, let us recall the \emph{formal} separation of the spherically symmetric wave equation \eqref{EqWaveEqRN}. By taking the Fourier transform 
\begin{equation} \label{EqRNFT}
\widecheck{\phi}(r; \omega) := \frac{1}{\sqrt{2 \pi}} \int\limits_\R \phi(t,r) e^{i \omega t} \, dt 
\end{equation} 
of $\phi$ in $t$ one obtains that formally  $\phi$ satisfies \eqref{EqWaveEqRN} if, and only if, $\widecheck{\phi}(r; \omega)$ satisfies\footnote{Equation \eqref{EqRNRadialODE} should be compared with \eqref{EqTeukStandard}. For the Kerr case we will do the separation in the analogue of $(v,r)$-coordinates on Reissner-Nordstr\"om, which gives the radial ODE \eqref{EqThmSepVar3} which has solutions with slightly different asymptotics. But this is not essential.}
\begin{equation}\label{EqRNRadialODE}
0 = \frac{r^4 \omega^2}{\Delta^2} \widecheck{\phi}(r; \omega) + \frac{\rd_r \Delta}{\Delta} \rd_r \widecheck{\phi}(r; \omega) + \rd^2_r \widecheck{\phi}(r;\omega) \;.
\end{equation}
This ODE has two regular singular points at $r= r_+$ and $r=r_-$; for all $\omega \neq 0$ we can find a fundamental system of solutions with asymptotics\footnote{We introduce the following notation: for $f,g : \R \supseteq I \to \C$ the notation $f \sim g$ for $I \ni x \to x_0 \in \R $ stands for $\lim_{x \to x_0} \frac{f(x)}{g(x)} = 1$. When obvious which limit point is considered, we may just write $f \sim g$.}
\begin{equation}\label{EqRNAsyA}
A_{\Hp_r}(r; \omega) \sim e^{-i \omega r^*} \qquad \textnormal{ and } \qquad A_{\Hp_l}(r; \omega) \sim e^{i \omega r^*}
\end{equation}
for $r \to r_+$ and another fundamental system of solutions with asymptotics
\begin{equation} \label{EqRNAsyB}
B_{\CH_l}(r; \omega) \sim e^{-i \omega r^*} \qquad \textnormal{ and } \qquad B_{\CH_r}(r; \omega) \sim e^{i \omega r^*}
\end{equation}
for $r \to r_-$, where $r^*(r)$ is as defined earlier. Since any three solutions have to be linearly dependent, we can write 
\begin{equation*}
A_{\Hp_r}( r; \omega) = \mathfrak{T}_{\Hp_r}(\omega) B_{\CH_l}(r; \omega) + \mathfrak{R}_{\Hp_r}(\omega) B_{\CH_r}(r; \omega)
\end{equation*}
 and 
\begin{equation}\label{EqRNRem}
A_{\Hp_l}( r; \omega) = \mathfrak{T}_{\Hp_l}(\omega) B_{\CH_r}(r; \omega) + \mathfrak{R}_{\Hp_l}(\omega) B_{\CH_l}(r; \omega) \;,
\end{equation}
where $\mathfrak{T}_{\Hp_r}(\omega)$, $\mathfrak{R}_{\Hp_r}(\omega)$ and $\mathfrak{T}_{\Hp_l}(\omega)$, $\mathfrak{R}_{\Hp_l}(\omega)$ are the \emph{transmission and reflection coefficients} of the right event horizon and left event horizon, respectively. A priori they are only defined for $\omega \in \R \setminus \{0\}$, but it can be shown that they extend analytically to all of $\R$. A key ingredient needed for the proof of Theorem \ref{ThmRN2} is that $\mathfrak{T}_{\Hp_r}(0) \neq 0$, which can be shown using the $\rd_t$-conservation law (see for example \cite{LukOhShlaForth}, \cite{KehleShla19}) or by direct computation using special functions (see for example \cite{GurselEtAl79}, \cite{KehleShla19}).

For $\omega \neq 0$ we can thus expand any solution of \eqref{EqRNRadialODE} as $\widecheck{\phi}(r; \omega) = a_{\Hp_r}(\omega) A_{\Hp_r}(r; \omega) + a_{\Hp_l}(\omega) A_{\Hp_l}(r; \omega)$ with $a_{\Hp_r}, a_{\Hp_l} : \R \setminus \{0\} \to \C$  and thus, at least formally,
\begin{equation}\label{EqRNFormalSol}
\phi(t,r) = \frac{1}{\sqrt{2\pi}} \int\limits_{\R} \big(a_{\Hp_r}(\omega) A_{\Hp_r}(r; \omega) + a_{\Hp_l}(\omega) A_{\Hp_l}(r; \omega)\big) e^{-i \omega t} \, d \omega
\end{equation}
is a solution of \eqref{EqWaveEqRN}.

We now discuss the reduction of Theorem \ref{ThmRN1} to Theorem \ref{ThmRN2}. Let $\phi : (\mathcal{M}  \cup \Hp_r )\cap \{v \geq v_0\} \cap \{u \leq u_1\} \to \C$ be as in Theorem \ref{ThmRN1}. One extends the induced initial data on $\Hp_r \cap \{v \geq v_0\}$ smoothly to all of $\Hp_r$ in such a way that $\phi|_{\Hp_r}(v) = 0$ for $v \leq v_0 -1$. Using that we are in spherical symmetry, we can now solve the wave equation sideways to extend $\phi$ to the region $(\mathcal{M} \cup \Hp_r \cup \Hp_l) \cap \{u \leq u_1\}$. Again we extend the induced initial data on $\Hp_l \cap \{u \leq u_1\}$ to all of $\Hp_l$ such that $\phi|_{\Hp_l}(u) = 0$ for $u \geq u_1 +1$ and solve the wave equation forwards to get a global solution in $\mathcal{M} \cup \Hp_r \cup \Hp_l$ which satisfies the assumptions in Theorem \ref{ThmRN2}. This is the reduction of Theorem \ref{ThmRN1} to Theorem \ref{ThmRN2} by extension of $\phi$.

We now turn towards the sketch of a proof of Theorem \ref{ThmRN2}. 
One first shows that the solution $\phi$ is indeed (i.e., not just formally) given by \eqref{EqRNFormalSol} with 
\begin{equation}\label{EqRNInitialData}
a_{\Hp_r}(\omega) = \frac{1}{\sqrt{2\pi}} \int\limits_{\R} \phi|_{\Hp_r}(v) e^{- i \omega v} \, dv \qquad \textnormal{ and } \qquad a_{\Hp_l}(\omega) = \frac{1}{\sqrt{2 \pi}} \int\limits_{\R}\phi|_{\Hp_l}(u) e^{i \omega u} \, du
\end{equation} 
being the (inverse) Fourier transforms of the characteristic  initial data. This can be established in (at least) two ways: one way is to start from the expression \eqref{EqRNFormalSol} with the coefficients $a_{\Hp_r}, a_{\Hp_l}$ given by \eqref{EqRNInitialData} and to show by direct computation that it solves the wave equation \eqref{EqWaveEqRN} and attains the prescribed initial data when $r \to r_+$ and $u$ or $v$, respectively, are fixed. By the uniqueness of the characteristic initial value problem we thus obtain that \eqref{EqRNFormalSol} with \eqref{EqRNInitialData} is indeed the wanted solution. Another possibility, which will be implemented in this paper for Teukolsky on Kerr, is to first prove via energy estimates decay of $\phi(t,r)$ in $t$ for all $r \in (r_-,r_+)$ which one uses to justify that the Fourier transform \eqref{EqRNFT} is well-defined for all $r \in (r_-,r_+)$ and that it satisfies \eqref{EqRNRadialODE}. One then infers that $\phi$ must be given by \eqref{EqRNFormalSol} with \emph{some} $a_{\Hp_r}$, $a_{\Hp_l}$, which one then determines by passing the expression \eqref{EqRNFormalSol} to the limit $r \to r_+$ for either fixed $u$ or fixed $v$. Since, as will become clear below, we only use the  frequency regime around $\omega = 0$ of the wave to prove the blow-up, this second approach, in contrast to the first one, allows us to completely ignore the behaviour of the other frequency regimes in the separated picture. Let us also remark that since $\phi$ vanishes at the bifurcation sphere, we do have exponential decay  of $\phi$ in $v,u$ along $\Hp_r, \Hp_l$, respectively, when approaching the bifurcation sphere and thus $a_{\Hp_r}$ and $a_{\Hp_l}$ are in particular in $L^2_\omega(\R)$. If $\phi$ did not vanish at the bifurcation sphere, the coefficients $a_{\Hp_r}$, $a_{\Hp_l}$ would have additional poles at zero frequency which encode the constant at the bifurcation sphere.

We now investigate the regularity of the coefficient functions \eqref{EqRNInitialData} around $\omega = 0$. 
First note that by \eqref{EqRNAssump1} and a Hardy inequality\footnote{Recall that $\phi|_{\Hp_r}$ has exponential decay towards  the bifurcation sphere.} we have $\widecheck{\phi|_{\Hp_r}} \in L^2(\R)$. Furthermore \eqref{EqRNAssump2} and \eqref{EqRNAssump3} imply
\begin{align}
&\rd_\omega^p(\omega \widecheck{\phi|_{\Hp_r}}) \in L^2_\omega(\R) \qquad \textnormal{ for all } \N \ni p < p_0 \label{EqRN1} \\
&\rd_\omega^{p_0}(\omega \widecheck{\phi|_{\Hp_r}}) \notin L^2_\omega(\R) \label{EqRN2} \\
&\rd_\omega^{p_0}(\omega^2 \widecheck{\phi|_{\Hp_r}}) \in L^2_\omega(\R) \;. \label{EqRN3}
\end{align}
It follows from \eqref{EqRN1} and \eqref{EqRN3} that $\omega \cdot \rd_\omega^{p_0} (\omega \widecheck{\phi|_{\Hp_r}}) \in L^2_{\omega}(\R)$. Together with \eqref{EqRN2} this now implies
$\rd_\omega^{p_0}(\omega \widecheck{\phi|_{\Hp_r}}) \notin L^2_\omega\big((-\varepsilon, \varepsilon) \big)$
for any $\varepsilon >0$. By \eqref{EqRNInitialData} we thus obtain  for any $\varepsilon >0$
\begin{equation}
\label{EqRNID}
\rd_\omega^{p_0}(\omega a_{\Hp_r}) \notin L^2_\omega\big((-\varepsilon, \varepsilon) \big)\qquad \textnormal{ and } \qquad \rd_\omega^p (\omega a_{\Hp_r}) \in L^2_\omega\big((-1, 1)\big)  \quad \textnormal{ for all } \N \ni p < p_0 \;.
\end{equation}
Furthermore we straightforwardly obtain 
\begin{equation}
\label{EqRNID2}
\rd_\omega^p(\omega a_{\Hp_l}) \in L^2_\omega\big((-1, 1) \big) \qquad \textnormal{ for all } \N \ni p \leq p_0 \;.
\end{equation}

We now move on to the analysis of the wave near the Cauchy horizon at $r=r_-$. 
Using for example energy estimates one shows that the wave $\phi$ extends (even continuously) to the Cauchy horizon $\CH_l$\footnote{It also extends continuously to $\CH_r$.}  and satisfies
\begin{equation}
\label{EqRNCauchy}
\int\limits_{\R} \chi(v) \big|\rd_v \phi|_{\CH_l}\big|^2(v) \, dv < \infty \;,
\end{equation}
where $\chi(v) : \R \to (0, \infty)$ is a positive function with $\chi(v) \simeq |v|^{2p_0}$ for\footnote{For $v \to - \infty$ one can replace $|v|^{2p_0}$ by  $|v|^q$ for any $q \in \N$. For the definition of the notation $\simeq$ we refer the reader to the very beginning of Section \ref{SecInteriorKerr}.} $v \to -\infty$ and $\chi(v) \simeq |v|^{2(p_0 - 1)}$ for $v \to + \infty$. In particular one can take the Fourier transform of $\rd_v \phi|_{\CH_l}$ in $L^2$. Using the language of the transmission and reflection coefficients introduced earlier we can rewrite \eqref{EqRNFormalSol} as
\begin{equation}\label{EqRNSepCH}
\begin{split}
\phi(t,r) = \frac{1}{\sqrt{2\pi}} \int\limits_\R \Big( &\big[ \mathfrak{T}_{\Hp_r}(\omega) a_{\Hp_r}(\omega) + \mathfrak{R}_{\Hp_l}(\omega) a_{\Hp_l}(\omega) \big] B_{\CH_l}(r; \omega) \\
&+ \big[ \mathfrak{T}_{\Hp_l}(\omega) a _{\Hp_l}(\omega) + \mathfrak{R}_{\Hp_r}(\omega) a_{\Hp_r}(\omega) \big] B_{\CH_r}(r ; \omega) \Big) e^{-i \omega t} \, d \omega \;.
\end{split}
\end{equation}
Noticing that the Killing vector field $\rd_t$ equals $\rd_v$ on $\CH_l$ and using the asymptotics \eqref{EqRNAsyB} of $B_{\CH_l}(r; \omega)$ and $B_{\CH_r}(r; \omega)$, we can pass \eqref{EqRNSepCH} to the limit $r \to r_-$ for fixed $v$ to obtain
\begin{equation*}
\rd_v \phi|_{\CH_l} = \frac{1}{\sqrt{2\pi}} \int\limits_{\R} \underbrace{(-i) \big[ \mathfrak{T}_{\Hp_r}(\omega) \cdot \big(\omega a_{\Hp_r}(\omega)\big) + \mathfrak{R}_{\Hp_l}(\omega) \cdot \big( \omega a_{\Hp_l}(\omega) \big) \big]}_{= \widecheck{(\rd_v \phi|_{\CH_l})}} e^{- i \omega v} \, d \omega \;,
\end{equation*}
i.e., a Fourier representation of $\rd_v \phi_{\CH_l}$ in terms of the Fourier representations of the characteristic initial data and the transmission and reflection coefficients.\footnote{Let us remark that a fully fledged scattering theory for the wave equation in the interior of a Reissner-Nordstr\"om black hole has been presented in \cite{KehleShla19}.} We can now investigate the decay of $\rd_v \phi|_{\CH_l}$ in $v$ by considering the regularity of $\widecheck{(\rd_v \phi|_{\CH_l})}  $ at $\omega = 0$: 
\begin{equation}\label{EqRNFinal}
\begin{split}
i \rd_\omega^{p_0} ( \widecheck{\rd_v \phi|_{\CH_l}}) &= \mathfrak{T}_{\Hp_r}(\omega) \cdot \rd_\omega^{p_0} \big(\omega a_{\Hp_r}(\omega)\big) + \sum_{p=1}^{p_0} \binom{p_0}{p} \rd_\omega^p \mathfrak{T}_{\Hp_r}(\omega) \cdot \rd_\omega^{p_0 - p} \big(\omega a_{\Hp_r}(\omega)\big) \\
&\qquad \qquad +\sum_{p=0}^{p_0} \binom{p_0}{p} \rd_\omega^p \mathfrak{R}_{\Hp_r}(\omega) \cdot \rd_\omega^{p_0 - p} \big(\omega a_{\Hp_l}(\omega) \big) \;.
\end{split}
\end{equation}
The last two terms (the two sums) on the right hand side are in $L^2_\omega\big((-1,1)\big)$ by the analyticity of the transmission and reflection coefficients $\mathfrak{T}_{\Hp_r}(\omega), \mathfrak{R}_{\Hp_r}(\omega)$ and by \eqref{EqRNID2} and the second property in \eqref{EqRNID}. It now follows from the first property in \eqref{EqRNID} together with $\mathfrak{T}_{\Hp_r}(0) \neq 0$ (cf.\ remark below \eqref{EqRNRem}), applied to the first term on the right hand side of \eqref{EqRNFinal} that $\rd_\omega^{p_0} ( \widecheck{\rd_v \phi|_{\CH_l}}) \notin L^2_\omega \big((-\varepsilon, \varepsilon)\big)$ for any $\varepsilon >0$. Plancherel now implies
$$ \int\limits_{\R} v^{2 p_0} \big|\rd_v \phi|_{\CH_l}(v)\big|^2 \, dv = \infty \;.$$
This, however, does not tell us yet whether the slow decay of $\rd_v \phi|_{\CH_l}$ in $v$ is for $v \to + \infty$ or for $v \to - \infty$. However, with \eqref{EqRNCauchy} we can finally infer
$$\int\limits_1^\infty v^{2 p_0} \big|\rd_v \phi|_{\CH_l}(v)\big|^2 \, dv = \infty \;.$$
The statement \eqref{EqRNConcThm1} of Theorem \ref{ThmRN2} then follows by propagating the singularity backwards along $\CH_r$, using energy estimates. This is a standard propagation of regularity result. We have now concluded the sketch of a proof of Theorem \ref{ThmRN2} and will discuss next how this method of proof changes for the Teukolsky field on Kerr.

\subsubsection{Comparison to Teukolsky on Kerr}

We will mainly use the $(v_+,r, \theta, \varphi_+)$-coordinate system on Kerr, which can be thought of as the analogue of the $(v,r, \theta, \varphi)$-coordinate system on Reissner-Nordstr\"om. However, $v_+$ is not a null coordinate any more, but its level sets are timelike. The Teukolsky equation takes the form\footnote{We refer the reader to Section \ref{SecInteriorKerr} for the Kerr-related terminology. Here $\Delta = r^2 - 2Mr + a^2$ and $\mathring{\slashed{\Delta}}_{[s]}$ is the spin $s$-weighted spherical Laplacian, see Section \ref{SecSpinWeightedLaplacian}.} 
\begin{equation}\label{EqTeukIntro}
\begin{split}
0=\mathcal{T}_{[s]} \psi := & a^2 \sin^2 \theta \,\partial_{v_+}^2\psi + 2a \,\partial_{v_+}\partial_{\varphi_+} \psi + 2(r^2 + a^2)\, \partial_{v_+}\partial_r \psi 
+2 a\, \partial_{\varphi_+}\partial_r \psi + \Delta \,\partial_r^2 \psi \\ &+  2\Big( r(1-2s) - isa\cos \theta\Big)\, \partial_{v_+} \psi +\dashuline{2(r-M)(1-s) \,\partial_r \psi} + \mathring{\slashed{\Delta}}_{[s]} \psi - 2s \psi 
 \;,
\end{split}
\end{equation}
where the Teukolsky field $\psi$ is with respect to an algebraically special frame which is regular at the right event horizon $\Hp_r$, cf.\ Sections \ref{SecPrincipalNull} and \ref{SecTeukAndSpin}.  For $s=+2$, the case we are concerned with, the frame component entering the Teukolsky field degenerates near $\Hp_l$ and thus a regular Teukolsky field $\psi$ vanishes on the left event horizon including at the bifurcation sphere.

Let us begin by discussing the differences between the energy estimates for Teukolsky and the linear wave equation. As is well-known, the spacetime geometry near the event horizons is such that localised energy of linear waves decays exponentially. This is usually referred to as the `red-shift effect'; it helps the analyst to close energy estimates. The name of course derives from a shift in frequency, which is also present at the event horizons. The shift in frequency and the decay of energy are not one and the same thing -- indeed, they decouple for the Teukolsky equation. We give a detailed discussion in Remark \ref{RemRedShift}. For the energy estimates it is of course the decay of localised energy which is most relevant -- let us refer to this effect as the `red-shift effect for energy' in order to keep in touch with standard terminology. For the Teukolsky field $\psi$ (and for $s = +2$) we now have an \emph{effective} \textbf{blue}-shift for the energy at the right event horizon $\Hp_r$. This can be seen from the dashed term in \eqref{EqTeukIntro}. It is effective in the sense that it turns into a red-shift for the energy after two commutations with $\rd_r$. It is thus at this level that we close the energy estimate for the Teukolsky field near $\Hp_r$. The Teukolsky equation for $\hat{\psi} := \Delta^{-s} \psi$, which is the Teukolsky field with respect to a frame that is regular at the left event horizon\footnote{Recall that $\psi$ degenerates (vanishes) at $\Hp_l$.} $\Hp_l$, does still have a red-shift for energy near $\Hp_l$; so there, the energy estimates can be closed at the level of $\hat{\psi}$ as for the wave equation.

On the other hand, the blue-shift for energy for the wave equation near the Cauchy horizon turns into an effective \textbf{red}-shift for energy for the Teukolsky field $\psi$ near the left Cauchy horizon $\CH_l$. This makes the energy estimates for \eqref{EqTeukIntro} near $r=r_-$ in a sense even easier than for the wave equation (disregarding the more technical nature of implementing the energy estimates for Teukolsky). It is again `effective' in the sense that after two commutations with $\rd_r$ we have again a blue-shift for energy.

We now discuss the formal separation. Denoting with $Y^{[s]}_{ml}(\theta, \varphi ; \omega) = S^{[s]}_{ml}( \cos \theta ; \omega) e^{im \varphi} $ the spin $s$-weighted spheroidal harmonics (see Section \ref{SecSpin2Harmonics}; we have $\N \ni l \geq \max\{|s|, |m|\}$, $m \in \Z$), the Teukolsky transform of $\psi$ is given by
\begin{equation}\label{EqIntroTeukTrafo}
\widecheck{\psi}_{ml}(r; \omega) = \frac{1}{\sqrt{2\pi}} \int_\R \int_{\Sp^2} \psi(v_+, r ,\theta, \varphi_+) e^{i \omega v_+} \overline{Y^{[s]}_{ml}(\theta, \varphi_+; \omega) }\, dv_+ \vols \;.
\end{equation}
Formally, $\psi$ satisfies the Teukolsky equation \eqref{EqTeukIntro} if, and only if, $\widecheck{\psi}_{ml}(r; \omega)$ satisfies\footnote{Here, $\lambda_{ml}^{[s]}(\omega)$ denotes the eigenvalue associated to the eigenfunction $Y^{[s]}_{ml}(\cdot; \omega)$ of the spin $2$-weighted spheroidal Laplacian, see Section \ref{SecSpin2Harmonics}.}
\begin{equation}\label{EqIntroODE}
\begin{split}
\Delta \frac{d^2}{dr^2} \widecheck{\psi}_{ml}(r; \omega) &+ 2 \Big( - (r^2 + a^2) i \omega + i a m + (r-M) (1-s)\Big) \frac{d}{dr} \widecheck{\psi}_{ml}(r; \omega) \\
&+ \Big( \lambda_{ml}^{[s]}(\omega) - (a \omega)^2 + 2 \omega m a - 2i \omega r (1-2s) - 2s\Big) \widecheck{\psi}_{ml}(r; \omega) = 0  \;.
\end{split}
\end{equation}
Like \eqref{EqRNRadialODE}, the radial ODE \eqref{EqIntroODE} has two regular singular points at $r=r_-$ and $r=r_+$. Let $\omega_\pm = \frac{a}{2Mr_\pm}$ and fix $s=+2$. For $\omega \neq \omega_+m$ we can find a fundamental system of solutions with asymptotics
\begin{equation*}
A_{\Hp_r, ml} (r; \omega) \sim 1 \qquad \textnormal{ and } \qquad A_{\Hp_l, ml}(r; \omega) \sim \Big(\frac{r_+ - r}{r_+ - r_-} \Big)^{2 + \frac{4iMr_+}{r_+ -r_-}(\omega - \omega_+m)}
\end{equation*}
for $r \to r_+$ and for $\omega \neq \omega_- m$ another fundamental system of solutions with asymptotics
\begin{equation*}
B_{\CH_l, ml} (r; \omega) \sim 1 \qquad \textnormal{ and } \qquad B_{\CH_r, ml}(r; \omega) \sim \Big(\frac{r - r_-}{r_+ - r_-} \Big)^{2 -\frac{4iMr_-}{r_+ -r_-}(\omega - \omega_-m)} 
\end{equation*}
for $r \to r_-$. The fact that $A_{\Hp_r, ml}$ and $B_{\CH_l, ml}$ do not have oscillating phases as for the wave equation in \eqref{EqRNAsyA} and \eqref{EqRNAsyB} is due to our choice of $(v_+,r, \theta, \varphi_+)$-coordinates. If we had used Boyer-Lindquist coordinates $(t,r,\theta, \varphi)$ for the separation, both branches would be oscillatory. Note, however, the difference in the $r$-weights between the two branches, which is related to $\psi$ being regular at $\Hp_r$ and degenerate at $\Hp_l$, and similarly for the Cauchy horizons. Another important difference is that while the branches $A_{\Hp_l, ml}$ and $B_{\CH_r, ml}$ extend analytically to $\omega = \omega_+m$ and $\omega = \omega_- m$, respectively, the branches $A_{\Hp_r, ml}$ and $B_{\CH_l, ml}$ become singular at $\omega = \omega_+m$ and $\omega = \omega_- m$, respectively. This should be contrasted with both branches $A_{\Hp_r}$ and $A_{\Hp_l}$ in \eqref{EqRNAsyA} for the wave equation having a regular (and indeed identical) limit $\omega \to 0$ (similarly for the other two branches in \eqref{EqRNAsyB}). This difference impacts a priori on relating the coefficients in the separated picture to the Teukolsky transform of the characteristic initial data (more about this later) and also on the regularity of the transmission and reflection coefficients: as before we can write
\begin{equation*}
\begin{split}
A^{[s]}_{\Hp_r, ml}(r; \omega) &= \mathfrak{T}^{[s]}_{\Hp_r, ml}(\omega) \cdot B^{[s]}_{\CH_l, ml}(r; \omega) + \mathfrak{R}^{[s]}_{\Hp_r, ml} (\omega) \cdot B^{[s]}_{\CH_r, ml} (r; \omega) \\
A^{[s]}_{\Hp_l, ml}(r; \omega) &= \mathfrak{T}^{[s]}_{\Hp_l, ml}(\omega) \cdot B^{[s]}_{\CH_r, ml}(r; \omega) + \mathfrak{R}^{[s]}_{\Hp_l, ml} (\omega) \cdot B^{[s]}_{\CH_l, ml} (r; \omega) \;,
\end{split}
\end{equation*}
where the transmission and reflection coefficients are a priori only defined and analytic on $\R \setminus \{\omega_+m, \omega_-m\}$. Recall that the structure of the blow-up argument only requires information on the frequency regime near $\omega = 0$. So for $m \neq 0$ we know that the transmission and reflection coefficients are analytic in a neighbourhood of $\omega = 0$. Moreover, for non-vanishing $m$ we show by direct computation that $\mathfrak{T}_{\Hp_r, ml}(0) \neq 0$, where we use that for $\omega = 0$ the radial ODE \eqref{EqIntroODE} turns into a hypergeometric equation. 
For $m=0$, however, the potentially problematic frequency at $\omega = 0$ cannot be avoided. We show that $\mathfrak{T}_{\Hp_l, 0l}, \mathfrak{R}_{\Hp_l,0l}, \mathfrak{T}_{\Hp_r,0l}$ all extend analytically to $\omega = 0$, but for the reflection coefficient of the right event horizon we only show that $\omega \cdot \mathfrak{R}_{\Hp_r, 0l}$ extends analytically to $\omega = 0$.\footnote{The analyticity of $\mathfrak{T}_{\Hp_l, 0l}$ at $\omega = 0$ is of no relevance to this paper and has not been explicitly stated, but is also proven as a side-result in the proof of Proposition \ref{PropTRM0}. And while we do not show that $\mathfrak{R}_{\Hp_r, 0l}$ has indeed a pole at $\omega = 0$, this is what we would expect -- and it can be decided by a longer and direct computation. However, this is of no relevance to this paper.} 

It can also be shown by direct computation for $m = 0$ that $\mathfrak{T}_{\Hp_r, 0l}(0) \neq 0$. However, this is more complicated than in the case $m \neq 0$, because it cannot be inferred alone from the $\omega \to 0$ limit of \eqref{EqIntroODE}, which is a hypergeometric equation, but we also need to get information on the $\omega$-derivatives of solutions to \eqref{EqIntroODE} at $\omega = 0$. We take this as an opportunity to implement and demonstrate a second approach to showing the non-vanishing of the transmission coefficients at $\omega = 0$, namely by making use of the Teukolsky-Starobinsky conservation law, which can be thought of as the equivalent to using the conservation law associated to the Killing vector field $\rd_t$ in the case of spherically symmetric waves on Reissner-Nordstr\"om mentioned in Section \ref{SecRN}. It is for this implementation where we need that $\omega \mathfrak{R}_{\Hp_r, 0l}$ extends continuously to $\omega = 0$.  Let us mention that we also show how the Teukolsky-Starobinsky conservation law can be used to obtain $\mathfrak{T}_{\Hp_r, ml}(0) \neq 0$ for $m \neq 0$, but in this case, which gives the leading blow-up at the Cauchy horizon, the direct computation is much easier. 

Finally, we also mention at this point that for a reason to be explained below we also need in the case $m=0$ the vanishing of $\mathfrak{R}_{\Hp_l, 0l}(0)$ in order to implement the blow-up argument. Again, this is shown by direct computation.

For $\omega \neq \omega_+m$ we can expand any solution of \eqref{EqIntroODE} as $$\widecheck{\psi}_{ml}(r; \omega) = a_{\Hp_r, ml}(\omega) A_{\Hp_r, ml}(r; \omega) + a_{\Hp_l, ml}(\omega) A_{\Hp_l, ml}(r; \omega)$$ with $a_{\Hp_r, ml}, a_{\Hp_l, ml} : \R \setminus \{ \omega_+m\} \to \C$ and thus, at least formally, we obtain that
\begin{equation} \label{EqIntroSepTeuk}
\psi(v_+,y, \theta, \varphi_+) = \frac{1}{\sqrt{2\pi}} \int_\R \sum_{m,l} \big[ a_{\Hp_r, ml}(\omega) A_{\Hp_r, ml}(y; \omega) + a_{\Hp_l, ml}(\omega) A_{\Hp_l, ml}(y; \omega) \big]  Y^{[s]}_{ml}(\theta, \varphi_+; \omega) e^{-i\omega v_+} \, d\omega 
\end{equation}
is a solution to \eqref{EqTeukIntro}.

In a similar way to how the local Theorem \ref{ThmRN1} for the spherically symmetric wave equation on Reissner-Nordstr\"om is reduced to the global Theorem \ref{ThmRN2}, we also reduce the local Theorem \ref{ThmInt} (or Theorem \ref{Thm1}) for the Teukolsky field of Kerr to a global theorem, see Theorem \ref{Thm2}. In spherical symmetry we extended the local solution to a global one by first solving sideways and in this way  ensuring that the extended solution vanishes at the bottom bifurcation sphere. For the Teukolsky equation we can no longer solve sideways, but, by solving two initial value problems, we can still extend the local solution to a global one which is compactly supported on $\Hp_l \cup \Sp^2_b$. This is done in Theorem \ref{ThmExt} in Section \ref{SecExtension}, see also Figure \ref{FigExt}. However, we can no longer ensure that the regular Teukolsky field vanishes at the bottom bifurcation sphere, which entails that we have to deal with what is the analogue  of the poles  in the Fourier expansion coefficients $a_{\Hp_r}$ and $a_{\Hp_l}$ in the spherically symmetric case, cf.\ discussion above \eqref{EqRN1}.\footnote{\label{FootnoteBSP} To be slightly more precise here, recall that $\psi$ vanishes automatically at the bottom bifurcation sphere because of the degeneration of the frame chosen. The Teukolsky analogue of the vanishing of the scalar field at $\Sp^2_b$, which avoids poles in the Fourier expansion coefficients, is the vanishing of $\rd_r^2 \psi$, which is non-degenerate at $\Sp^2_b$ due to the blow-up of $\rd_r$ in $(v_+,r, \theta, \varphi_+)$-coordinates at $\Sp^2_b$. 

Being confronted with a non-trivial field at the bottom bifurcation sphere one might still entertain the following approach, which can easily be implemented for the wave equation: by decomposing the initial data, we write the solution $\psi$ obtained by the above extension procedure as  a superposition of a solution $\psi_1$, the initial data of which is supported on $\Hp_l \cup \Hp_r$ only in a compact neighbourhood of $\Sp^2_b$, and another solution $\psi_2$ that vanishes on $\Hp_l$ including at $\Sp^2_b$ (and agrees with $\psi$ on $\Hp_r$ for late affine times). One can now run the desired argument for $\psi_2$ to obtain the singularity at the Cauchy horizon and then use   standard energy estimates  for $\psi_1$ to see that $\psi_1$ is much more regular at the Cauchy horizon -- and can essentially be neglected. However, one runs into difficulty when trying to implement this strategy for Teukolsky due to the effective blue-shift effect on $\Hp_r$ mentioned earlier. The reader can see directly from \eqref{EqTeukIntro} that the transversal derivative $\rd_r \psi_1$ of the solution $\psi_1$, whose trace on $\Hp_r$ vanishes for late affine times,  will in general grow exponentially along $\Hp_r$ -- thus prohibiting the stability estimates. (For the wave equation, due to the red-shift effect, the transversal derivative decays exponentially.) For this reason our proof of the blow-up of the Teukolsky field at the Cauchy horizon is more global in nature than for the wave equation.}  

We now discuss the implementation of the proof of  the global Theorem \ref{ThmRN2} to Teukolsky on Kerr (i.e., the proof of  Theorem \ref{Thm1} in Section \ref{SecAssumptions}). 
In Sections \ref{SecEERed} and \ref{SecEENoShift} we prove the energy estimates needed  to establish the representation \eqref{EqIntroSepTeuk}. The coefficients $a_{\Hp_r, ml}$ and $a_{\Hp_l, ml}$ are being determined in Section \ref{SecDetA}.  Keeping the $v_+$ coordinate fixed one can pass to the limit $r \to r_+$ in an analogous manner as for the wave equation to establish that 
\begin{equation*}
a_{\Hp_l, ml}(\omega) = \widecheck{\psi|_{\Hp_r}}_{ml}(\omega) \;,
\end{equation*}
where $\widecheck{(\cdot)}_{ml}$ is the Teukolsky transform \eqref{EqIntroTeukTrafo}. Note that because of the exponential decay in $v_+$ of $\psi|_{\Hp_r}$ towards $\Sp_b^2$ we have that $\psi|_{\Hp_r}$ is in particular in $L^2_{v_+}L^2(\Sp^2)$, so no poles are present. Because $\psi$ vanishes on $\Hp_l$, we first go over to the quantity $\rd_r^2 \psi$, which is regular at $\Hp_l$ due to the blow-up of $\rd_r$ near $\Hp_l$ (see also Footnote \ref{FootnoteBSP}). We thus take two $r$-derivatives of \eqref{EqIntroSepTeuk} and then pass to the limit $r \to r_+$ with fixed $v_-$.\footnote{See Section \ref{SecManifoldMetric} for the definition of $v_-$. It can be thought of as the analogue of $u$ in Reissner-Nordstr\"om.} However, it is clear from the preceding discussion that one cannot hope to establish an $L^2$-limit, since $\rd_r^2 \psi$ does not vanish at $\Sp_b^2$. We take a limit in the sense of distributions to recover that $a_{\Hp_l, ml}$ is related to the Teukolsky transform of $\rd_r^2 \psi|_{\Hp_l}$ (modulo a delta distribution). The support of the Teukolsky field at $\Sp_b^2$ implies that $a_{\Hp_l, ml}$ has a pole at $\omega = \omega_+m$. 
For $m \neq 0$ we can ignore this pole, since it is disjoint from a neighbourhood of $\omega = 0$ which is important for the argument. But for $m = 0$ the pole potentially interferes with our argument which is based on exploiting the limited regularity of the Fourier coefficients at $\omega = 0$. It is for this reason that $\mathfrak{R}_{\Hp_l, 0l}(0) = 0$ is needed later, which cancels the pole.

Recall how we inferred for the spherically symmetric wave the limited regularity \eqref{EqRNID} of $a_{\Hp_r}$ at $\omega = 0$ from the decay assumptions of $\psi$ along $\Hp_r$. In spherical symmetry we only had one mode -- the spherically symmetric one -- for Teukolsky we want to work with the $m_0l_0$-mode for which we assume slow decay in Theorem \ref{ThmInt} (or Theorem \ref{Thm1}). Note, however, that in the assumptions the $m_0l_0$-mode is with respect to spin $2$-weighted \emph{spherical} harmonics and not spin $2$-weighted \emph{spheroidal} harmonics. 

So we would like to obtain for any $\varepsilon >0$
\begin{equation} \label{EqIntroKerrSlowD}
\rd_\omega^{p_0} a_{\Hp_r, m_0l_0} \notin L^2_\omega \big((-\varepsilon, \varepsilon)\big) \qquad \textnormal{ and } \qquad \rd_\omega^p a_{\Hp_r, m_0l_0} \in L^2_\omega \big(-1,1\big) \quad \textnormal{ for all } \N \ni p < p_0 \;.
\end{equation}
Note that the derivation of \eqref{EqRNID}  used at its heart that $v$-weights translate in the Fourier picture as $\omega$-derivatives. Since the spin weighted spheroidal harmonics in the Teukolsky transform \eqref{EqIntroTeukTrafo} depend on $\omega$, this correspondence does not hold true any more for Kerr.\footnote{This is not an issue arising from considering Teukolsky versus the wave equation, but already appears when considering the wave equation on Kerr.} Exploiting, however, that for $\omega = 0$ the spin weighted \emph{spheroidal} harmonics agree with the spin weighted \emph{spherical} harmonics, we can still obtain \eqref{EqIntroKerrSlowD}, see Proposition \ref{PropFourierAssump}.

In an analogous way as for the spherically symmetric wave equation (see \eqref{EqRNSepCH}) we can now express \eqref{EqIntroSepTeuk} in terms of the fundamental solutions $B_{\CH_l, ml}$, $B_{\CH_r, ml}$ normalised at the Cauchy horizons and the transmission and reflection coefficients and prove energy estimates which allow us to pass to the limit $r \to r_-$ with fixed $v_+$ to obtain that 
\begin{equation*}
\psi|_{\CH_l}(v_+, \theta, \varphi_+) = \frac{1}{\sqrt{2 \pi}} \int_{\R} \sum_{m,l} \underbrace{\big[\mathfrak{R}_{\Hp_l, ml}(\omega) a_{\Hp_l, ml}(\omega) + \mathfrak{T}_{\Hp_r, ml}(\omega) a_{\Hp_r, ml}(\omega) \big]}_{= \widecheck{(\psi|_{\CH_l})}_{ml}(\omega)} Y^{[s]}_{ml}(\theta, \varphi_+; \omega) e^{-i \omega v_+} \, d \omega \;.
\end{equation*}
As before, and using $\mathfrak{R}_{\Hp_l, 0l}(0) = 0$ in the case $m_0 = 0$, we deduce that $\rd_\omega^{p_0} \widecheck{(\psi|_{\CH_l})}_{ml}(\omega) \notin L^2_\omega \big((-\varepsilon, \varepsilon) \big)$ for any $\varepsilon >0$. When converting this into the statement
\begin{equation} \label{EqIntorNearlyDone}
\int_{\R} \int_{\Sp^2} |v_+^{p_0} \psi|_{\CH_l}(v_+,  \theta, \varphi_+)|^2 \, \vols dv_+ = \infty 
\end{equation}
we again have to address the complication that $v$-weights do not exactly correspond to $\omega$-derivatives. This is done by proving bounds on $\rd_\omega^k Y^{[s]}_{ml}(\theta, \varphi_+; \omega)$, see Propositions \ref{PropEstimatesDerivativesEigenfunctions} and \ref{PropCharSlowDecay}. As for the spherically symmetric model we prove energy estimates in Section \ref{SecEECauchy} to  show that the infinitude of the integral in \eqref{EqIntorNearlyDone} is due to the behaviour of $\psi$ for large positive $v_+$ and also that we can propagate the singularity backwards. This concludes the outline of the proof.

\subsection{Outline of paper}

In Section \ref{SecInteriorKerr} we begin by introducing the interior of the Kerr black hole, then recall briefly the derivation of the Teukolsky equation, and we define spin weighted functions on the sphere as well as on spacetime. Moreover, we show that the Teukolsky field has the regularity of such a spin weighted function on spacetime and we record the form of the Teukolsky equation in various coordinate systems for later reference. Section \ref{SecAssumptions} formulates the main theorems of this paper and their assumptions. The proof of the main theorems begins in Section \ref{SecEE} where we establish the energy estimates required and record some corollaries which are needed later for the limits $r \to r_\pm$, the separation of the solution, the extension to the Cauchy horizon, and the backwards propagation of the singularity. In Section \ref{SecTeukSep} we recall the spin weighted spheroidal harmonics, establish a couple of results which are needed for the translation of $v_+$-weights to $\omega$-derivatives, and then use the energy estimates to give the separation of the Teukolsky field. We continue in Section \ref{SecAnaHeun} with the analysis of the radial Teukolsky ODE, introduce the fundamental systems of solutions we work with, and prove the required properties of the transmission and reflection coefficients. Section \ref{SecDetA} is concerned with the passing to the limit $r \to r_+$ and the determination of the Fourier coefficients in terms of the characteristic initial data. And finally in Section \ref{SecPfMainThm} we conclude the proofs of the main theorems. Appendix \ref{CoordTransformation} records the form of the Teukolsky equation in coordinates which are regular near the bottom bifurcation sphere and discusses the initial value problem for Teukolsky, which is needed for the extension procedure which reduces the local Theorem \ref{Thm1} to the global Theorem \ref{Thm2}.  The  Appendices \ref{AppendixCommTeukHat}, \ref{AppendixTeukStar}, and \ref{AppendixBL} collect commutator expressions required for the energy estimates in Section \ref{SecEE}.

\subsection*{Acknowledgements}

I would like to acknowledge that this project started out as a collaboration with Jonathan Luk which had to be discontinued because of time constraints. Several ideas in this paper have been developed together with him and I am also grateful to Jonathan for bringing the Teukolsky-Starobinsky conservation law to my attention. Moreover, I would like to thank Yakov Shlapentokh-Rothman for a helpful discussion. A large body of this work was completed while I was supported through the Sylvester Research Fellowship at the University of Oxford. I now acknowledge support through the Royal Society University Research Fellowship URF\textbackslash R1\textbackslash 211216. There are no competing interests.

\section{The interior of sub-extremal Kerr and gravitational perturbations} \label{SecInteriorKerr}

This section presents the set-up of this paper. We first introduce the geometry of the interior of a sub-extremal Kerr black hole and then recall the derivation of the  Teukolsky equation along with the notion of spin weighted functions. We also show that the geometrically arising Teukolsky field  is indeed such a spin weighted function.

We also introduce the following \emph{notation}: for a function $f$ and a non-negative function $g$ the notation $f \lesssim g$  means that there exists a constant $C>0$ such that $|f(x)| \leq C \cdot g(x)$ holds for all points $x$ for which both functions are defined. If we say `$f \lesssim g$ on $A$', where $A$ is a subset of the domains of definition of $f$ and $g$, then this means that there exists a constant $C>0$ such that $|f| \leq C \cdot g$ holds on $A$. Similarly, `$f \lesssim g$ for $x \to x_0$' means that there exists a neighbourhood $A$ of $x_0$ such that $f \lesssim g$ on $A$ . Here, $x_0$ may also be $\infty$. The notations `$f \lesssim g$ as $x \to x_0$' and `$f = \mathcal{O}(g)$ as $x \to x_0$' have the same meaning.\footnote{The reason we use both notations is that we find it convenient to use the $\mathcal{O}$ notation within equations: an equation of the form $f = t \cdot \mathcal{O}(g) + h$ has to be read as `$f = t \cdot u + h$ with $u = \mathcal{O}(g)$'. The limit associated with the $\mathcal{O}$ notation is often understood from the context and not mentioned explicitly.} Finally, if both $f$ and $g$ are non-negative, then the notation $f \simeq g$ stands for `$f \lesssim g$ and $g \lesssim f$', i.e., there exists a constant $C >0$ such that $\frac{1}{C} f \leq g \leq C \cdot g$. Again, we may specify a region or a limit in which $f \simeq g$ is supposed to hold.

\subsection{The manifold and metric of the interior of sub-extremal Kerr} \label{SecManifoldMetric}

We consider the standard $(t,r,\theta, \varphi)$ coordinates on the smooth manifold $\mathcal{M} = \R \times (r_-,r_+) \times \mathbb{S}^2$, where $r_- = M - \sqrt{M^2 - a^2}$, $r_+ =M + \sqrt{M^2 - a^2}$, and $0 < |a| < M $  are constants which later represent the angular momentum per unit mass and the mass of the black hole, respectively. A Lorentzian metric $g$ on $\mathcal{M}$ is defined by
\begin{equation}\label{Kerr.metric.BL}
g = g_{tt} \, dt^2 + g_{t\varphi}\,(dt \otimes d\varphi + d\varphi \otimes dt) + \frac{\rho^2}{\Delta} \, dr^2 + \rho^2 \, d\theta^2 + g_{\varphi \varphi} \, d\varphi^2 \;,
\end{equation}
where
\begin{equation*}
\begin{aligned}
&\rho^2 = r^2 + a^2 \cos^2\theta\;,  \qquad \qquad &&g_{tt} = -1 + \frac{2Mr}{\rho^2}\;{,} \\
&\Delta = r^2 -2Mr + a^2\;,  &&g_{t\varphi} = -\frac{2Mra\sin^2\theta}{\rho^2}\;{,} \\
& &&g_{\varphi \varphi} = \big[ r^2 + a^2 +\frac{2Mra^2 \sin^2\theta}{\rho^2}\big] \sin^2\theta \;.
\end{aligned}
\end{equation*}
Note that $r_- <  r_+$  are the roots of $\Delta$. We also compute $\det g = -\rho^4 \sin^2 \theta$ for later convenience. We fix a time orientation on the Lorentzian manifold $(\mathcal{M},g)$ by stipulating that $-\partial_r$ is future directed. The time oriented Lorentzian manifold $(\mathcal{M},g)$ is called the \emph{interior of a sub-extremal Kerr black hole} and the coordinates $(t,r,\theta, \varphi)$ are called \emph{Boyer--Lindquist coordinates}.  Moreover, let us fix an orientation by stipulating that the Lorentzian volume form $\vol = \rho^2 \sin \theta \,dt \wedge dr \wedge d\theta \wedge {d}\varphi$ is positive. A longer computation yields that $(\mathcal{M},g)$ is a solution to the vacuum Einstein equations $\mathrm{Ric}(g)=0$.

For later reference we note that the inverse metric $g^{-1}$ in the \emph{Boyer--Lindquist coordinates} $(t, \varphi, r, \theta)$ is given by
\begin{equation}
\label{gInverse}
g^{-1} = \begin{pmatrix}
-\frac{g_{\varphi \varphi}}{\Delta \sin^2 \theta} & \frac{g_{t \varphi}}{\Delta \sin^2 \theta} & 0 & 0 \\
\frac{g_{t \varphi}}{ \Delta \sin^2 \theta} & -\frac{g_{tt}}{\Delta \sin^2 \theta} & 0 & 0 \\
0 & 0 & \frac{\Delta}{\rho^2} &0 \\
0 & 0 & 0& \frac{1}{\rho^2}
\end{pmatrix} \;.
\end{equation}

In the following we will attach boundaries to $\mathcal{M}$. 
Let $r^*(r)$ be a function on $(r_-,r_+)$ satisfying $\frac{dr^*}{dr} = \frac{r^2 + a^2}{\Delta}$ and $\overline{r}(r)$ a function on $(r_-,r_+)$ satisfying $\frac{d\overline{r}}{dr} = \frac{a}{\Delta}$. We now define the following functions on $\mathcal M$:
\begin{align*}
v_+ := t + r^* \quad &, \qquad \varphi_+ := \varphi + \overline{r}\quad  \mathrm{mod}\; 2\pi\;,\\
v_- := r^* - t \quad &, \qquad \varphi_- := \varphi - \overline{r} \quad \mathrm{mod}\; 2\pi \;.
\end{align*}
It is easy to check that $(v_+, \varphi_+, r, \theta)$ and $(v_-, \varphi_-, r, \theta)$ are coordinate systems for $\mathcal M$. The metric $g$ in these coordinates takes the following form:
\begin{equation*}
\begin{split}
g&= g_{tt} \, dv_+^2 + g_{t\varphi} \, \big( dv_+ \otimes d\varphi_+ + d\varphi_+ \otimes dv_+\big) + g_{\varphi \varphi} \, d\varphi^2_+  +\big(dv_+ \otimes dr + dr \otimes dv_+\big) \\[2pt] &\qquad - a\sin^2\theta \, \big( dr \otimes d\varphi_+ + d\varphi_+ \otimes dr\big) + \rho^2 \, d\theta^2 \\[7pt]
&= g_{tt} \, dv_-^2 - g_{t\varphi} \, \big( dv_- \otimes d\varphi_- + d\varphi_- \otimes dv_-\big) + g_{\varphi \varphi} \, d\varphi^2_-  + \big(dv_- \otimes dr + dr \otimes dv_-\big)\\[2pt]
&\qquad + a\sin^2\theta \, \big( dr \otimes d\varphi_- + d\varphi_- \otimes dr\big) + \rho^2 \, d\theta^2 \;.
\end{split}
\end{equation*}
A simple computation shows that those expressions define non-degenerate (and analytic) Lorentzian metrics for all positive values of $r$.
We now set
\begin{equation*}
\kappa_\pm = \frac{r_\pm - r_\mp}{2(r_\pm^2 + a^2)}
\end{equation*} 
and define the Kruskal-like coordinate functions
\begin{align*}
V_{r_+}^+ &:= e^{\kappa_+ v_+} \\
V_{r_+}^- &:= e^{\kappa_+ v_-} \\
\Phi_{r_+} &:= \varphi - \frac{at}{r_+^2 + a^2} \;.
\end{align*}
The Kruskal-like coordinates $(V_{r_+}^+, V_{r_+}^-, \theta, \Phi_{r_+})$ map $\mathcal{M}$ onto $(0, \infty) \times (0, \infty) \times \mathbb{S}^2$.
It can be shown (see \cite{ONeillKerr}, Chapter 3.5) that the Kerr metric \eqref{Kerr.metric.BL} extends, under this mapping, regularly to the manifold $[0, \infty) \times [0,\infty) \times \mathbb{S}^2$. We call the null hypersurface $\{0\} \times [0, \infty) \times \mathbb{S}^2 =: \Hpl$ the \emph{(left) event horizon} and the null hypersurface $[0, \infty) \times \{0\} \times \mathbb{S}^2 =: \Hpr$ the \emph{(right) event horizon}. The sphere $\{0\} \times \{0\} \times \mathbb{S}^2 = \Hpr \cap \Hpl =:\mathbb{S}^2_b$ is called the \emph{(bottom) bifurcation sphere}.

In order to extend $\mathcal{M}$ to $r = r_-$, we define another set of Kruskal-like coordinate functions by
\begin{align*}
V_{r_-}^+ &:= -e^{\kappa_- v_+} \\
V_{r_-}^- &:= -e^{\kappa_- v_-} \\
\Phi_{r_-} &:= \varphi - \frac{at}{r_-^2 + a^2} \;.
\end{align*}
The Kruskal-like coordinates $(V_{r_-}^+, V_{r_-}^-, \theta, \Phi_{r_-})$ map $\mathcal{M}$ onto $(-\infty, 0) \times (-\infty, 0) \times \mathbb{S}^2$ and in the same way it can be shown that the Kerr metric \eqref{Kerr.metric.BL} extends in these coordinates regularly to $(-\infty, 0] \times (-\infty, 0] \times \Sp^2$. We call the null hypersurface $\{0\} \times (-\infty, 0] \times \Sp^2 =: \CH_r$ the \emph{(right) Cauchy horizon} and the null hypersurface $(-\infty,0] \times \{0\} \times \Sp^2 =: \CH_l$ the \emph{(left) Cauchy horizon}. The sphere $\{0\} \times \{0\} \times \Sp^2 = \CH_r \cap \CH_l =: \Sp_t$ is called the \emph{(top) bifurcation sphere}.

Using the two Kruskal-like coordinate systems we define the manifold with corners $\M := \mathcal{M} \cup \Hp_l \cup \Hp_r \cup \CH_l \cup \CH_r$, which is depicted in a Penrose-style diagram\footnote{To be more precise, depicted is a slice of constant $0<\theta < \pi$ and each point represents an $\Sp^1$.} in Figure \ref{FigInt1}. Figure \ref{FigInt2} shows the behaviour and range of the functions $t, r, v_-,$ and $v_+$. We also define the manifolds with corners $\underline{\mathcal{M}} := \mathcal{M} \cup \Hp_l \cup \Hp_r$ and $\overline{\mathcal{M}} := \mathcal{M} \cup \CH_l \cup \CH_r$.

\begin{figure}[h]
\centering
\begin{minipage}{.5\textwidth}
  \centering
 \def\svgwidth{6cm}
    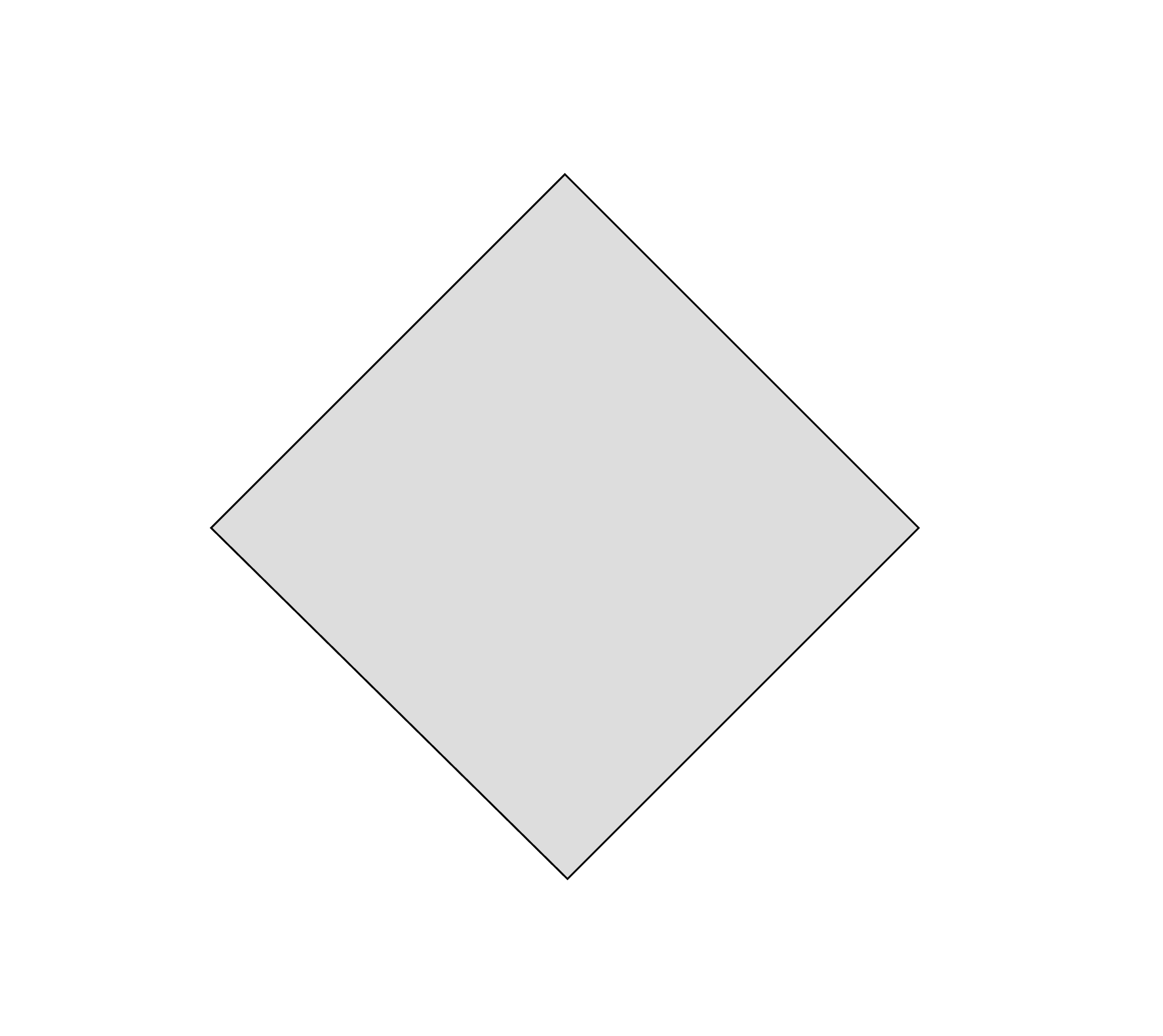
      \caption{The interior of sub-extremal Kerr} \label{FigInt1}
\end{minipage}%
\begin{minipage}{.5\textwidth}
  \centering
  \def\svgwidth{6cm}
    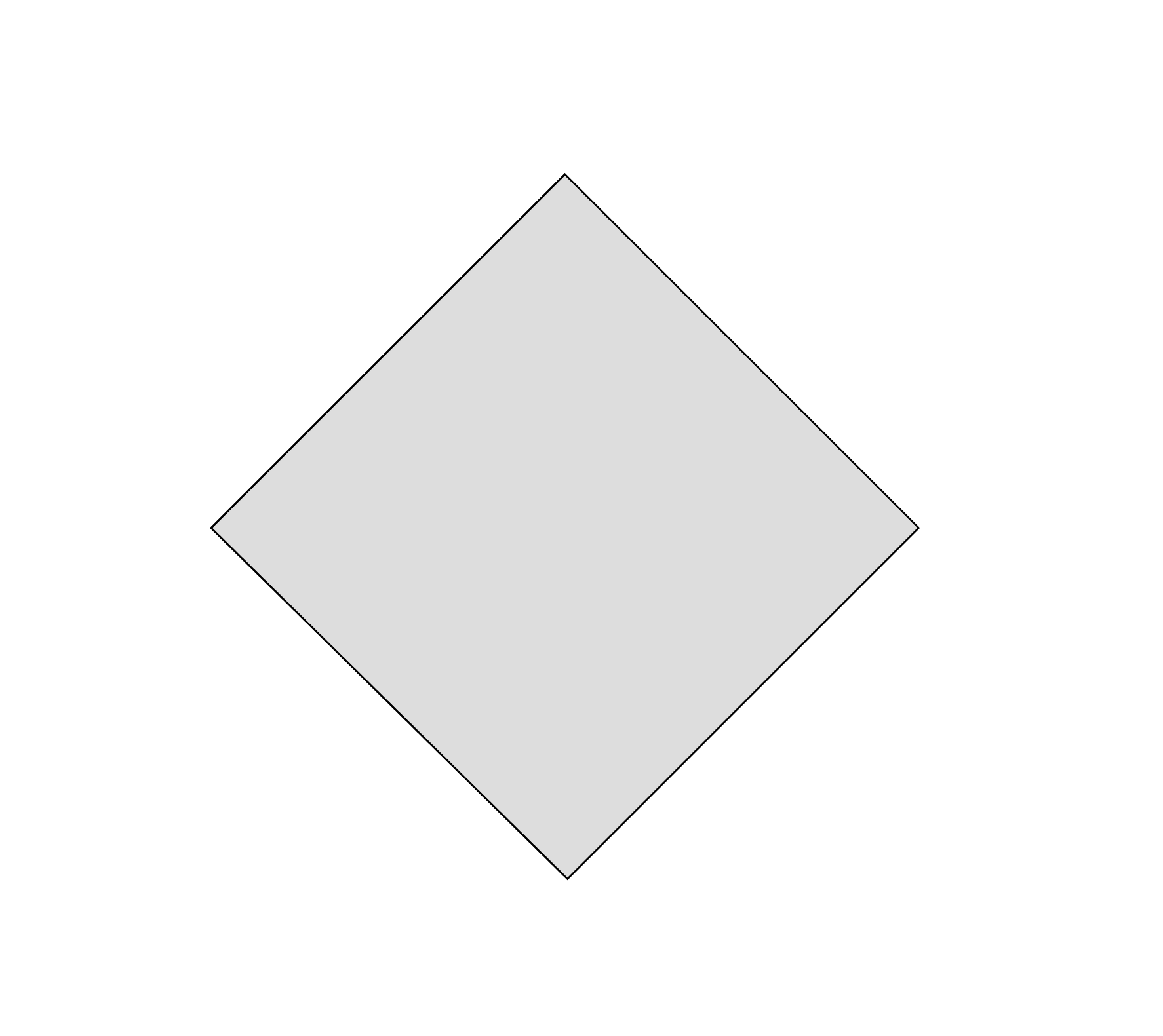
      \caption{The coordinate functions $t,r,v_-$, and $v_+$} \label{FigInt2}
\end{minipage}
\end{figure}
We also note that the coordinates $\{v_+, \varphi_+, r, \theta\}$ cover $\mathcal{M} \cup (\Hp_r \setminus \Sp^2_b) \cup (\CH_l \setminus \Sp^2_t)$. For later reference we express the Boyer-Lindquist coordinate vector fields (on the left) in terms of the $\{v_+, \varphi_+, r, \theta\}$ coordinate vector fields (on the right):
\begin{equation}
\label{CoordTrafoStar}
\begin{aligned}
\frac{\partial}{\partial r} \Big|_t &= \frac{r^2 + a^2}{\Delta} \partial_{v_+} + \frac{a}{\Delta} \partial_{\varphi_+} + \partial_r \qquad \quad &&\frac{\partial}{\partial t} = \partial_{v_+} \\
\frac{\partial}{\partial \varphi} &= \partial_{\varphi_+} &&\frac{\partial}{\partial \theta} = \partial_\theta \;.
\end{aligned}
\end{equation}
We also note that the volume form in $\{v_+, \varphi_+, r, \theta\}$-coordinates is given by $\vol = \rho^2 \sin \theta \, dv_+ \wedge dr \wedge d\theta \wedge d\varphi_+$.

Similarly we express the Boyer-Lindquist coordinate vector fields (on the left) in terms of the $\{v_-, \varphi_-, r, \theta\}$ coordinate vector fields (on the right):
\begin{equation}
\label{CoordTrafoStarMinus}
\begin{aligned}
\frac{\partial}{\partial r} \Big|_t &= \frac{r^2 + a^2}{\Delta} \partial_{v_-} - \frac{a}{\Delta} \partial_{\varphi_-} + \partial_r \qquad \quad &&\frac{\partial}{\partial t} = -\partial_{v_-} \\
\frac{\partial}{\partial \varphi} &= \partial_{\varphi_-} &&\frac{\partial}{\partial \theta} = \partial_\theta \;.
\end{aligned}
\end{equation}
We also note that the volume form in $\{v_-, \varphi_-, r, \theta\}$-coordinates is given by $\vol = \rho^2 \sin \theta \, dv_- \wedge dr \wedge d\theta \wedge d\varphi_-$.

Note that $<dv_+, dv_+> = <dv_-, dv_-> = \frac{a^2 \sin^2 \theta}{\rho^2}$, thus showing that for $a >0$ the level sets of $v_+$ and $v_-$ are timelike hypersurfaces away from the axis. 

We now define the functions $f^+ := v_+ -r + r_+$ and $f^- := v_- - r + r_-$. An easy computation gives 
\begin{equation}
\label{NormF}
<df^+, df^+> = <df^-, df^-> =  \frac{a^2 \sin^2 \theta}{\rho^2} + \frac{\Delta}{\rho^2} - \frac{2(r^2 + a^2)}{\rho^2}
\end{equation}
which shows that the level sets of $f^+$ and $f^-$ are spacelike hypersurfaces, cf.\ Figure \ref{FigInt2}. Moreover, it is immediate that the level sets of $r$ are spacelike hypersurfaces.

\subsubsection{Relation of $\Phi_{r_+}$ and $\varphi_+$  on $\Hp_r$ -- and similarly for $\Phi_{r_+}, \varphi_-$ on $\Hp_l$} \label{SecRelTwoPhis}

We define $\omega_\pm := \frac{a}{r_\pm^2 + a^2}$ and set $$\phi_\pm(r) := \omega_\pm r^* - \overline{r} = \frac{a}{r_\pm^2 + a^2} r^* - \overline{r} \;.$$  This defines  smooth functions for $r \in (r_-, r_+)$. Moreover, $\phi_+$ extends smoothly to $r_+$ and $\phi_-$ extends smoothly to $r_-$: for $\phi_+$ this follows from $$\frac{d}{dr} \phi_+(r) = \frac{a}{r_+^2 + a^2} \frac{dr^*}{dr} - \frac{d\overline{r}}{dr} = \frac{a}{r_+^2 + a^2} \frac{r^2 + a^2}{ \Delta} - \frac{a}{\Delta} = \frac{a}{ \Delta}\Big( \frac{r^2 + a^2}{r_+^2 + a^2} - 1\Big) \;,$$
where the right hand side clearly extends smoothly to $r=r_+$. We denote $\lim_{r \to r_+} \phi_+(r) =: \phi_+(r_+)$. Similarly for $\phi_-$ and we denote $\lim_{r \to r_-} \phi_-(r) =: \phi_-(r_-)$.

We will also need to relate the angular functions $\Phi_{r_+}$ and $\varphi_+$ ($\varphi_-$) on the right (left) event horizon, where they are both defined. 
For $r \in (r_-, r_+)$ we have
\begin{equation*}
\begin{split}
\Phi_{r_+} &= \varphi - \frac{at}{r_+^2 + a^2} \\
&= \varphi_+ - \frac{a}{r_+^2 + a^2} v_+ + \phi_+(r) \\
&= \varphi_- + \frac{a}{r_+^2 + a^2} v_- - \phi_+(r) \;.
\end{split}
\end{equation*}
On $\Hp_r$ we thus have
\begin{equation*}
\Phi_{r_+} = \varphi_+ - \omega_+ v_+ + \phi_+(r_+) \;,
\end{equation*}
while on $\Hp_l$ we have
\begin{equation*}
\Phi_{r_+} = \varphi_- + \omega_+ v_- - \phi_+(r_+) \;.
\end{equation*}

\subsubsection{Estimates for $r^*$ near $r = r_\pm$} \label{SecDelta}

We write $\frac{dr^*}{dr} = \frac{r^2 + a^2}{(r - r_+)(r-r_-)} = \frac{r_+^2 + a^2}{(r-r_+)(r_+ - r_-)} + f_+(r)$ with a function $f_+ : (r_-, r_+) \to \R$ that extends regularly to $r_+$. Integration gives
\begin{equation} \label{EqRStarDelta}
r^*(r) = \frac{1}{2\kappa_+}\log (r_+-r) + F_+(r)
\end{equation}
with a function $F_+ : (r_-,r_+) \to \R$ that extends regularly to $r_+$.
Recalling $r^* = \frac{1}{2}(v_+ + v_-)$ we obtain
\begin{equation} \label{r-r_+}
V^+_{r_+} V^-_{r_+} = e^{\kappa_+(v_+ + v_-)} = (r_+ - r)e^{2\kappa_+F_+(r)} = \mathcal{O}(r_+-r) \;.
\end{equation}

Similarly, we obtain
\begin{equation}\label{EqRStarCauchy}
r^*(r) = \frac{1}{2\kappa_-} \log (r - r_-) + F_-(r)
\end{equation}
with a function $F_- : (r_-, r_+) \to \R$ that extends regularly to $r_-$.

\subsection{The principal null frame} \label{SecPrincipalNull}
For convenience we introduce the abbreviations $\mathscr{S} = \sin \theta$ and $\mathscr{C}=\cos \theta$. Moreover, using the Boyer--Lindquist coordinates, we define
\begin{equation*}
V= (r^2 + a^2) \partial_t + a \partial_\varphi  \qquad \textnormal{ and } \qquad W= \partial_\varphi + a \Si^2 \partial_t \;.
\end{equation*}
A \emph{principal null frame} is then given by
\begin{equation*}
\begin{aligned}
&e_1 := \frac{1}{\rho} \partial_\theta \;,  && \hat{e}_3 := \frac{\Delta}{\rho^2} \partial_r - \frac{1}{\rho^2} V \;,\\
&e_2 := \frac{W}{|W|} = \frac{1}{\rho \Si}(\partial_\varphi + a \Si^2 \partial_t)\;, \qquad \qquad
&&\hat{e}_4 := -\partial_r - \frac{1}{\Delta}V \;.
\end{aligned}
\end{equation*}
The vector fields $\hat{e}_3$ and $\hat{e}_4$ are null and future directed and satisfy $\langle \hat{e}_3,\hat{e}_4\rangle  = -2$. Let us denote the distribution spanned by $\hat{e}_3$ and $\hat{e}_4$ by $\Pi$ and the distribution orthogonal to $\Pi$  by $\Pi^\perp$. The vector fields $e_1$ and $e_2$ are not defined on the axis, but where defined they form an orthonormal basis for $\Pi^\perp$.  

Note that in $(v_-,r,\theta,\varphi_-)$-coordinates we have\footnote{In the following $\big{|}_\pm$ indicates a partial derivative in the $(v_\pm, r, \theta, \varphi_\pm)$ coordinate system.} 
\begin{equation*}
\hat{e}_3 = \frac{\Delta}{\rho^2} \frac{\partial}{\partial r} \Big|_- - \frac{2}{\rho^2} V \qquad \textnormal{ and } \qquad \hat{e}_4 = - \frac{\partial}{\partial r}\Big|_-\;,
\end{equation*}
while in $(v_+, r , \theta, \varphi_+)$-coordinates we have
\begin{equation*}
\hat{e}_3 = \frac{\Delta}{\rho^2} \frac{\partial}{\partial r}\Big|_+ \qquad \textnormal{ and } \qquad \hat{e}_4 = - \frac{\partial}{\partial r} \Big|_+ - \frac{2}{\Delta} V \;. 
\end{equation*}
Hence, the null vectors $\hat{e}_3$ and $\hat{e}_4$ are regular at the left event horizon $\Hp_l$ and at the right Cauchy horizon $\CH_r$, but not at the right event horizon $\Hp_r$ and at the left Cauchy horizon $\CH_l$. There, the vector fields 
\begin{equation*}
e_3 := -\frac{1}{\Delta}\, \hat{e}_3 \qquad \textnormal{ and } \qquad  e_4 := - \Delta \, \hat{e}_4
\end{equation*}
are regular.



\subsection{The Teukolsky equation and spin-weighted functions}
\label{SecTeukAndSpin}

\subsubsection{Gravitational perturbations in the Newman-Penrose formalism}\label{SecGravPerturb}

In the following we recall the basic steps in the derivation of the Teukolsky equation for gravitational perturbations of Kerr, see \cite{Teuk73}. We start by clarifying that our convention for the Riemann curvature tensor is
\begin{equation*}
\begin{split}
R^{\mu}_{\;\,\nu \rho \sigma} &= dx^\mu\big(R(\partial_\rho, \partial_\sigma)\partial_\nu\big) = dx^\mu\Big(\nabla_{\partial_\rho} \nabla_{\partial_\sigma} \partial_\nu - \nabla_{\partial_\sigma} \nabla_{\partial_\rho} \partial_\nu\Big)\\
&= \partial_\rho \Gamma^\mu_{\nu \sigma} - \partial_\sigma \Gamma^\mu_{\nu \rho} + \Gamma^\mu_{\kappa \rho}\Gamma^\kappa_{\nu \sigma} - \Gamma^\mu_{\kappa \sigma} \Gamma^\kappa_{\nu \rho}
\end{split}
\end{equation*}
where $x^\mu$ denotes a local coordinate system. 

We now make contact with and follow \cite{Teuk73} by setting
\begin{equation}\label{NPFrame}
l =- \hat{e}_4\, ,\qquad \quad n = -\frac{1}{2} \hat{e}_3\, , \qquad \quad {m_a} = \frac{1}{\sqrt{2}}\cdot \frac{\rho}{r + i a \cos \theta}(e_1 + i \cdot e_2) \;.
\end{equation}
With respect to this complex principal null frame we have\footnote{See \cite{Teuk73} or \cite{ChandBlackHole}, taking into account that they consider Lorentzian metrics of signature $(+, -, -, -)$, i.e., $R(-g)_{\mu \nu \rho \sigma} = -R(g)_{\mu \nu \rho \sigma}$.}
\begin{equation}\label{EqCurvatureAlgSpFrame}
\begin{aligned}
\Psi_0 &= R(l,{m_a},l,{m_a}) = 0 \\
\Psi_1 &= R(l,n,l,{m_a}) = 0 \\
\Psi_2 &= R(l,{m_a},\overline{{m_a}},n) = - \frac{M}{(r - ia \cos \theta)^3} \\
\Psi_3 &= R(l,n,\overline{{m_a}}, n) = 0 \\
\Psi_4 &= R(n, \overline{{m_a}}, n, \overline{{m_a}}) = 0 \;.
\end{aligned}
\end{equation}
Let now $g(s)$, $s \in [0, \varepsilon)$, be a smooth family of Lorentzian metrics defined on $\mathcal{M} \cup \Hp_l \cup \Hp_r$ satisfying the vacuum Einstein equations $\mathrm{Ric}\big(g(s)\big) = 0$ and such that $g(0)$ is the metric \eqref{Kerr.metric.BL} of sub-extremal Kerr. Moreover, let $l(s)$, $n(s)$, ${m_a}(s)$, $\overline{{m_a}}(s)$ be a complex frame field (not necessarily null) such that for $s=0$ they agree with \eqref{NPFrame} and define $\Psi_i(s)$ in analogy with \eqref{EqCurvatureAlgSpFrame} for all $s$. It now follows from \eqref{EqCurvatureAlgSpFrame} that
\begin{equation*}
\dot{\Psi}_0(0) =  \frac{d}{ds}\Big|_{s=0} \Big(R\big(g(s)\big)\big(l(s), {m_a}(s), l(s), {m_a}(s)\big) \Big) \overset{!}{=} \Big(\frac{d}{ds}\Big|_{s=0}  R\big(g(s)\big)\Big)\big(l(0), {m_a}(0), l(0), {m_a}(0)\big) \;,
\end{equation*}
i.e., $\dot{\Psi}_0(0)$ is in fact independent of the continuation of the complex principal null frame \eqref{NPFrame} for $s>0$. Moreover, because $\Psi_0$ is a vanishing scalar, $\dot{\Psi}_0(0)$ is also gauge invariant. The same observations hold for $\dot{\Psi}_4(0)$. In \cite{Teuk73}, Teukolsky derived the following equation, now called the \emph{Teukolsky equation},
\begin{equation} \label{TeukolskyEquation1}
\begin{split}
-\Big[\frac{(r^2 + a^2)^2}{\Delta} &- a^2 \sin^2 \theta\Big] \partial_t^2\hat{\psi}_s - \frac{4Mar}{ \Delta} \partial_t \partial_\varphi \hat{\psi}_s - \Big[\frac{a^2}{\Delta} - \frac{1}{\sin^2 \theta}\Big] \partial^2_\varphi \hat{\psi}_s \\ 
&+ \Delta^{-s} \partial_r (\Delta^{s+1} \partial_r \hat{\psi}_s) + \frac{1}{ \sin \theta}\partial_\theta(\sin \theta \partial_\theta \hat{\psi}_s) + 2s\Big[\frac{a(r-M)}{\Delta} + \frac{i \cos\theta}{\sin^2\theta}\Big] \partial_\varphi \hat{\psi}_s \\ 
&+ 2s\Big[\frac{M(r^2 -a^2)}{\Delta} - r - ia\cos \theta\Big] \partial_t \hat{\psi}_s -\Big[\frac{s^2 \cos^2 \theta}{\sin^2\theta} -s\Big] \hat{\psi}_s = 0 \;,
\end{split}
\end{equation}
which is satisfied for $s = +2$ by $\hat{\psi}_2 = \dot{\Psi}_0(0)$ and for $s=-2$ by $\hat{\psi}_{-2} = (r - ia \cos \theta)^4 \cdot \dot{\Psi}_4(0)$.


\subsubsection{The Teukolsky equation for a regular field near $\Hp_r$} \label{SecAnotherFormTeuk}

We recall that $l = -\hat{e}_4$ blows up at the right event horizon $\Hp_r$ and that $\Delta l = e_4$ is regular at $\Hp_r$. Hence, the curvature component $\alpha = R(e_4, {m_a} , e_4,{m_a}) = \Delta^2 \Psi_0$ is regular at $\Hp_r$ (and vanishes at $\Hp_l$) and for its linearisation $\dot{\alpha}$ we obtain $\dot{\alpha} = \Delta^2 \dot{\Psi}_0$. This motivates to set $\psi_s := \Delta^s \hat{\psi}_s$.
It now follows that if $\hat{\psi}_s$ satisfies \eqref{TeukolskyEquation1}, then $\psi_s$ satisfies
\begin{equation} \label{TeukolskyEquation2}
\begin{split}
\mathcal{T}_{[s]} \psi_s := &-\Big[\frac{(r^2 + a^2)^2}{\Delta} - a^2 \sin^2 \theta\Big] \partial_t^2\psi_s - \frac{4Mar}{ \Delta} \partial_t \partial_\varphi \psi_s - \Big[\frac{a^2}{\Delta} - \frac{1}{\sin^2 \theta}\Big] \partial^2_\varphi \psi_s \\ 
&+ \Delta^{-s} \partial_r (\Delta^{s+1} \partial_r \psi_s) + \frac{1}{ \sin \theta}\partial_\theta(\sin \theta \partial_\theta \psi_s) + 2s\Big[\frac{a(r-M)}{\Delta} + \frac{i \cos\theta}{\sin^2\theta}\Big] \partial_\varphi \psi_s \\ 
&+ 2s\Big[\frac{M(r^2 -a^2)}{\Delta} - r - ia\cos \theta\Big] \partial_t \psi_s -\Big[\frac{s^2 \cos^2 \theta}{\sin^2\theta} +s\Big] \psi_s -4s(r-M) \partial_r \psi_s = 0 \;,
\end{split}
\end{equation}
where we have used
\begin{equation*}
\Delta^s \cdot  \Delta^{-s} \partial_r (\Delta^{s+1} \partial_r \hat{\psi}_s)  =\Delta^{-s} \partial_r (\Delta^{s+1} \partial_r \psi_s) -4s(r-M) \partial_r \psi_s  - 2s \psi_s \;.
\end{equation*}
\emph{In particular, the quantity we are most interested in, $\da =\psi_2$, satisfies \eqref{TeukolskyEquation2} for $s=+2$.}

Using the definition of the wave operator
\begin{equation*}
\Box_g \psi  = \frac{1}{\sqrt{-\det g}} \partial_\mu (g^{\mu \nu} \sqrt{- \det g} \, \partial_\nu \psi) \;,
\end{equation*}
we can rewrite \eqref{TeukolskyEquation2} as
\begin{equation} \label{TeukolskyWave}
\begin{split}
\frac{1}{\rho^2}\mathcal{T}_{[s]} \psi_s = \Box_g \psi_s &- \frac{2s}{\rho^2}(r-M) \partial_r \psi_s + \frac{2s}{\rho^2}\Big(\frac{a(r-M)}{\Delta} + i \frac{\cos \theta}{\sin^2 \theta}\Big) \partial_\varphi \psi_s \\
&+ \frac{2s}{\rho^2}\Big( \frac{M(r^2 - a^2)}{\Delta} - r -ia \cos \theta\Big) \partial_t \psi_s - \frac{1}{\rho^2}(s + s^2 \frac{\cos^2 \theta}{\sin^2 \theta}) \psi_s = 0 \;.
\end{split}
\end{equation}

\subsubsection{Spin $s$-weighted functions on $\Sp^2$} \label{SecSpinWeightedFunctions}

In the following we will exhibit the appropriate function space on which the Teukolsky equation \eqref{TeukolskyEquation2} is defined -- and in particular which function spaces $\da$ and $\dot{\Psi}_0(0)$ belong to (it is immediate from their definition that they are not regular at $\theta = 0, \pi$). We begin by discussing spin $s$-weighted functions on the $2$-sphere which arise by expressing tensors on $\Sp^2$ with respect to a (necessarily) non-global frame field. We consider the standard $(\theta, \varphi)$ coordinate system on $\Sp^2$ in which the round metric takes the form $g_{\Sp^2} = d\theta^2 + \sin^2 \theta \, d\varphi^2$ and choose as an orthonormal frame field $E_1 = \partial_\theta$ and $E_2 =\frac{1}{\sin \theta}\partial_\varphi$, which are defined away from the north pole at $\{\theta = 0\}$ and the south pole at $\{\theta = \pi\}$. We combine this frame field into a single complex vector
\begin{equation*}
m = \frac{1}{\sqrt{2}}(\partial_\theta + \frac{i}{\sin \theta} \partial_\varphi) \;.
\end{equation*}
Consider now the space $\Gamma^\infty(S^2T^*\Sp^2)$ of all smooth symmetric  $2$-covariant tensor fields on $\Sp^2$ and define a map 
\begin{equation}\label{EqMapIota}
\iota_m : \Gamma^\infty(S^2T^*\Sp^2) \to C^\infty(\Sp^2 \setminus \{\theta = 0, \pi\}) \cap L^\infty(\Sp^2)
\end{equation}
by $\iota_m(\alpha) = \alpha(m, m)$. 
\begin{definition}
The space of smooth spin $2$-weighted functions on $\Sp^2$ is defined as 
\begin{equation*}
\mathscr{I}_{[2]}^\infty(\Sp^2) := \iota_m\big(\Gamma^\infty(S^2T^*\Sp^2)\big) \subseteq C^\infty(\Sp^2 \setminus \{\theta = 0, \pi\}) \cap L^\infty(\Sp^2)
\end{equation*}
\end{definition}
\begin{remark} \label{RemSymTfTen}
For  $\alpha \in S^2T^*\Sp^2$ we compute
\begin{equation*}
\alpha(m,m) = \frac{1}{2}\big(\alpha_{\theta \theta} + \frac{i}{\sin \theta} \alpha_{\theta \varphi} + \frac{i}{\sin \theta} \alpha_{\varphi \theta} - \frac{1}{\sin^2 \theta} \alpha_{\varphi \varphi}\big) = \frac{1}{2}\big(\alpha_{\theta \theta} - \frac{1}{\sin^2 \theta} \alpha_{\varphi \varphi}\big) + \frac{i}{\sin \theta} \alpha_{\theta \varphi} \;
\end{equation*}
where we have used the symmetry of $\alpha$. It follows that $g_{\Sp^2}(m,m) =0$ and thus the kernel of $\iota_m$ contains $g_{\Sp^2} \cdot C^\infty(\Sp^2)$. We now show that the kernel of $\iota_m$ equals $g_{\Sp^2} \cdot C^\infty(\Sp^2)$. We note that $\nicefrac{\Gamma^\infty(S^2T^*\Sp^2) }{g_{\Sp^2} \cdot C^\infty(\Sp^2)} \simeq \Gamma^\infty(S^2_{\mathrm{tf}}T^*\Sp^2) $, where $\Gamma^\infty(S^2_{\mathrm{tf}}T^*\Sp^2) $ denotes the space of all smooth symmetric and trace-free $2$-covariant tensor fields on $\Sp^2$.
For $\alpha \in  \Gamma^\infty(S^2_{\mathrm{tf}}T^*\Sp^2)$ we have $\alpha_{\theta \theta} + \frac{1}{\sin^2 \theta} \alpha_{\varphi \varphi} = 0$ and thus
$$\alpha(m,m) =  \alpha_{\theta \theta} + \frac{i}{\sin \theta} \alpha_{\theta \varphi} \;,$$
which shows that $\alpha(m,m) $ characterises $\alpha$ uniquely. This shows that $\iota_m\big|_{\Gamma^\infty(S^2_{\mathrm{tf}}T^*\Sp^2)} : \Gamma^\infty(S^2_{\mathrm{tf}}T^*\Sp^2) \to \mathscr{I}_{[2]}^\infty(\Sp^2)$ is an isomorphism.\footnote{Indeed, one could have defined the space $\mathscr{I}_{[2]}^\infty(\Sp^2)$ as the image of $\Gamma^\infty(S^2_{\mathrm{tf}}T^*\Sp^2)$ under $\iota_m$. However, for the proof of Proposition \ref{PropAlphaDotSpinW} we will need that $\alpha(m,m) \in \mathscr{I}_{[2]}^\infty(\Sp^2)$ even if $\alpha$ is not trace-free.}
\end{remark}

We also remark that the space of smooth spin $-2$-weighted functions is defined as the image of $\Gamma^\infty(S^2T^*\Sp^2)$ under $\iota_{\overline{m}}$ in $C^\infty(\Sp^2 \setminus \{\theta = 0, \pi\})$, where $\overline{m}$ is the complex conjugate of $m$. Although not needed in this paper, we also briefly remark that smooth spin ($\pm 1)$-weighted functions are defined as the images (under $\iota_m$ and $\iota_{\overline{m}}$) of all smooth one-forms on $\Sp^2$. We also remark that it follows directly from the definition that the spaces of smooth spin weighted functions are invariant under multiplication by smooth functions on $\Sp^2$.


We now give an intrinsic characterisation of the spin $s$-weighted functions on $\Sp^2$. We define
\begin{equation}\label{ZTilde}
\begin{aligned}
\tilde{Z}_1 &= - \sin \varphi \, \partial_\theta + \cos \varphi(-is \frac{1}{\sin \theta} - \frac{\cos \theta}{\sin \theta} \partial_\varphi) \\
\tilde{Z}_2 &= -\cos \varphi \, \partial_\theta - \sin \varphi(-is \frac{1}{\sin \theta} - \frac{\cos \theta}{\sin \theta} \partial_\varphi) \\
\tilde{Z}_3 &= \partial_\varphi \;.
\end{aligned}
\end{equation}
These first order differential operators satisfy $[\tilde{Z}_1 , \tilde{Z}_2] = \tilde{Z}_3$, $[\tilde{Z}_2, \tilde{Z}_3] = \tilde{Z}_1$, and $[\tilde{Z}_3, \tilde{Z}_1] = \tilde{Z}_2$.

\begin{proposition} \label{PropFirstCharSpinWeighted}
$f \in C^\infty(\Sp^2\setminus \{\theta = 0, \pi\})$ lies in $\mathscr{I}^\infty_{[s]}(\Sp^2)$ if, and only if, $e^{is\varphi}(\tilde{Z}_1)^{k_1}(\tilde{Z}_2)^{k_2}(\tilde{Z}_3)^{k_3} f$ extends continuously to the north pole $\theta = 0$ and $e^{-is\varphi}(\tilde{Z}_1)^{k_1}(\tilde{Z}_2)^{k_2}(\tilde{Z}_3)^{k_3} f$ extends continuously to the south pole $\theta = \pi$ for all $0 \leq k_1 + k_2 + k_3 < \infty$, $k_i \in \N_0$.
\end{proposition}

Before we give the proof we recall that the vector fields
\begin{equation}\label{DefEqVectorFieldZ}
\begin{split}
Z_1 &= - \sin \varphi \, \partial_\theta - \cos \varphi \frac{\cos \theta}{\sin \theta} \, \partial_\varphi \\
Z_2 &= - \cos \varphi \, \partial_\theta + \sin \varphi \frac{\cos \theta}{\sin \theta} \, \partial_\varphi \\
Z_3 &= \partial_\varphi
\end{split}
\end{equation}
are smooth on $\Sp^2$, span $T\Sp^2$ at each point of $\Sp^2$, and satisfy $[Z_1, Z_2] = Z_3$, $[Z_2, Z_3] = Z_1$, and $[Z_3, Z_1] = Z_2$.

\begin{proof}
We observe that
\begin{equation*}
\begin{split}
\sqrt{2} e^{i \varphi} m &= (\cos \varphi \,\partial_\theta - \frac{\sin \varphi}{\sin \theta} \partial_\varphi) + i ( \sin \varphi\, \partial_\theta + \frac{\cos \varphi}{\sin \theta} \partial_\varphi) \\
&= (-Z_2 + \frac{\sin \varphi}{\sin \theta}[\cos \theta - 1] Z_3) + i( - Z_1 - \frac{\cos \varphi}{\sin \theta}[\cos \theta - 1] Z_3)
\end{split}
\end{equation*}
is continuous at the north pole $\theta = 0$ and, similarly,
\begin{equation*}
\begin{split}
\sqrt{2} e^{-i \varphi} m &= (\cos \varphi \, \partial_\theta + \frac{ \sin \varphi}{\sin \theta} \, \partial_\varphi) + i ( \frac{ \cos \varphi}{\sin \theta} \partial_\varphi - \sin \varphi\, \partial_\theta) \\
&= (-Z_2 + \frac{\sin \varphi}{\sin \theta} [\cos \theta + 1] Z_3) + i(Z_1 + \frac{\cos \varphi}{\sin \theta}[\cos \theta + 1] Z_3)
\end{split}
\end{equation*}
is continuous at the south pole $\theta = \pi$. Moreover, we compute
\begin{equation*}
\begin{aligned}
\Li_{Z_1} m &= i \frac{\cos \varphi}{\sin\theta} \cdot m \\
\Li_{Z_2} m &= -i \frac{\sin \varphi}{\sin \theta} \cdot m \\
\Li_{Z_3} m &= 0 \;.
\end{aligned}
\end{equation*}
For $s = +2$ and $\alpha \in \Gamma^\infty(S^2T^*\Sp^2)$ we now compute
\begin{equation}\label{LieZTilde}
(\Li_{Z_i}\alpha)(m,m) = \Li_{Z_i}\big(\alpha(m,m)\big) - 2\alpha(\Li_{Z_i}m,m) = \tilde{Z}_i\big(\alpha(m,m)\big) \;.
\end{equation}
Iteratively, we obtain
\begin{equation} \label{EqContDepSpin}
\big((\Li_{Z_1})^{k_1}(\Li_{Z_2})^{k_2}(\Li_{Z_3})^{k_3} \alpha \big)(m,m) = (\tilde{Z}_1)^{k_1}(\tilde{Z}_2)^{k_2}(\tilde{Z}_3)^{k_3} \big(\alpha(m,m)\big)
\end{equation}
for $0 \leq k_1 + k_2 +k_3 < \infty$.

Given now $f = \alpha(m,m) \in \mathscr{I}^\infty_{[s]}(\Sp^2)$, it follows from \eqref{EqContDepSpin} together with the above observations that
\begin{equation*}
\begin{split}
e^{i2\varphi}(\tilde{Z}_1)^{k_1}(\tilde{Z}_2)^{k_2}(\tilde{Z}_3)^{k_3} f &= e^{2i \varphi} \big((\Li_{Z_1})^{k_1}(\Li_{Z_2})^{k_2}(\Li_{Z_3})^{k_3} \alpha \big)(m,m) \\
&= \big((\Li_{Z_1})^{k_1}(\Li_{Z_2})^{k_2}(\Li_{Z_3})^{k_3} \alpha \big)(e^{i\varphi} m,e^{i \varphi} m)
\end{split}
\end{equation*}
extends continuously to the north pole. The analogous computation shows the claim for the south pole. 

Vice versa, let $f \in C^\infty(\Sp^2\setminus\{\theta = 0, \pi\})$ satisfy the continuity properties stated in the proposition. By Remark \ref{RemSymTfTen} $\alpha(m,m) := f$ defines a  smooth symmetric and trace-free two-covariant tensor field (over $\R$) on $\Sp^2\setminus\{\theta = 0, \pi\})$. It now follows as before from \eqref{EqContDepSpin} that this tensor field extends smoothly to the north and south pole. 

The statement of the proposition for $s = -2$ (as well as for $s = \pm 1$) follows analogously.
\end{proof}


Now we introduce spin weighted Sobolev spaces. Some properties of those will later be needed for the energy estimates and Sobolev embeddings of spin weighted functions.
\begin{definition} 
The \emph{spin $s$-weighted Sobolev space} $H^m_{[s]}(\Sp^2)$ is defined by
\begin{equation*}
H^m_{[s]}(\Sp^2) := \{f \in L^2(\Sp^2) \; | \; (\tilde{Z}_1)^{k_1}(\tilde{Z}_2)^{k_2}(\tilde{Z}_3)^{k_3} f \in L^2(\Sp^2) \textnormal{ for all } 0 \leq k_1 + k_2 + k_3 \leq m\;, k_i \in \N_0\} \;.
\end{equation*}
\end{definition}

We denote with $\Sp^2_+$ the (closed) northern hemisphere of $\Sp^2$ and with $\Sp^2_-$ the (closed) southern hemisphere.

\begin{lemma} \label{LemSobolev}
If $f \in H^j_{[s]}(\Sp^2)$, then $e^{is\varphi} f \in H^j(\Sp^2_+)$ and $e^{-is\varphi}f \in H^j(\Sp^2_-)$ for $j = 1,2$.
\end{lemma}

\begin{proof}
We compute
\begin{equation} \label{EqCommuteExp}
\begin{aligned}
Z_1(e^{\pm is\varphi}f) &= e^{\pm is\varphi}\big(\tilde{Z}_1 f + \underbrace{is \frac{\cos \varphi}{\sin \theta}( 1 \mp \cos \theta)}_{=:a_\mp} f\big) \\
Z_2(e^{\pm is \varphi}f) &= e^{\pm is \varphi} \big(\tilde{Z}_2 f - \underbrace{is \frac{\sin \varphi}{\sin \theta}(1 \mp \cos \theta)}_{=:b_\mp}f\big) \\
Z_3(e^{\pm is \varphi}f) &= e^{\pm is \varphi}\big(\tilde{Z}_3 f \pm is f\big) \;.
\end{aligned}
\end{equation}
Let us now restrict to the upper sign and to the northern hemisphere. It then follows that $a_-(\theta,\varphi) = is \cos \varphi \cdot \mathcal{O}(\theta) \in C^{0,1}(\Sp^2_+)$, and similarly $b_- \in C^{0,1}(\Sp^2_+)$. Thus all the terms in \eqref{EqCommuteExp} are in $L^2(\Sp^2_+)$. Moreover, it now follows easily that $Z_i\big(Z_j( e^{is\varphi} f)\big) \in L^2(\Sp^2_+)$ for all $i,j \in \{1,2,3\}$. For example we have
\begin{equation*}
\begin{split}
Z_2\big(Z_1(e^{is \varphi} f)\big) &=Z_2(e^{is\varphi}\tilde{Z}_1 f) + Z_2(a_- \cdot e^{is \varphi} f) \\
&= e^{is \varphi}\Big(\tilde{Z}_2(\tilde{Z}_1 f) + a_- \tilde{Z}_1 f\Big) + (Z_2 a_-) \cdot e^{is\varphi} f + a_- Z_2(e^{is\varphi}f) \;. 
\end{split}
\end{equation*}
Similarly for the lower sign and the southern hemisphere.
\end{proof}

\begin{proposition} \label{PropSmoothSW}
We have $\mathscr{I}^\infty_{[s]}(\Sp^2) = \bigcap_{0 \leq m < \infty} H^m_{[s]}(\Sp^2)$
\end{proposition}

\begin{proof}
The inclusion ``$\subseteq$'' follows directly from Proposition \ref{PropFirstCharSpinWeighted}. For the reverse inclusion let $f \in \bigcap_{0 \leq m < \infty} H^m_{[s]}(\Sp^2)$ and note that for $0 \leq k_1 + k_2 + k_3$ we have $(\tilde{Z}_1)^{k_1}(\tilde{Z}_2)^{k_2}(\tilde{Z}_3)^{k_3} f \in H^2_{[s]}(\Sp^2)$.  It now follows from Lemma \ref{LemSobolev} together with the standard Sobolev embedding that $e^{is\varphi} (\tilde{Z}_1)^{k_1}(\tilde{Z}_2)^{k_2}(\tilde{Z}_3)^{k_3} f $ is continuous at the north pole $\theta = 0$ while $e^{-is\varphi} (\tilde{Z}_1)^{k_1}(\tilde{Z}_2)^{k_2}(\tilde{Z}_3)^{k_3} f $ is continuous at the south pole $\theta = \pi$. The conclusion now follows again from Proposition \ref{PropFirstCharSpinWeighted}.
\end{proof}

Let us denote the standard volume form on $\Sp^2$ by $\vols = \sin \theta d\theta \wedge d\varphi$. We now derive an integration by parts formula for spin weighted functions.

\begin{proposition} \label{PropIntPartsSphere}
For $f,h \in \mathscr{I}^\infty_{[s]}(\Sp^2)$ and $i \in \{1,2,3\}$ we have
\begin{equation*}
\int_{\Sp^2} \tilde{Z}_i f \cdot \overline{h} \, \vols = -\int_{\Sp^2} f \cdot \overline{\tilde{Z}_i h} \, \vols \;.
\end{equation*}
\end{proposition}

\begin{proof}
We give the proof for $s = +2$, but the other cases are analogous.
We begin by noticing that
\begin{equation}\label{EqMMbarMetric}
m \otimes \overline{m} = E_1 \otimes E_1 + E_2 \otimes E_2 - i (E_1 \otimes E_2 - E_2 \otimes E_1) = g_{\Sp^2} - i \varepsilon \;,
\end{equation}
where $\varepsilon = \vols^\sharp \in \Gamma^\infty(\Lambda^2T^*\Sp^2)$ is the raised volume form. Note that $m \otimes \overline{m}$ is a \emph{smooth} tensor on $\Sp^2$. In particular, since the vector fields $Z_i$ are Killing vector fields, we obtain
\begin{equation}\label{KillingM}
\mathcal{L}_{Z_i} (m \otimes \overline{m}) = 0 \;.
\end{equation}
Let now $\alpha, \beta \in \Gamma^\infty(S^2T^*\Sp^2)$ with $\alpha(m, m) = f$ and $\beta(m,m) = h$. Using \eqref{LieZTilde}, \eqref{KillingM}, and the smoothness of $m \otimes \overline{m}$ we compute
\begin{equation*}
\begin{aligned}
\int_{\Sp^2} \tilde{Z}_i f \cdot \overline{h} \, \vols &= \int_{\Sp^2} \tilde{Z}_i\big(\alpha(m,m)\big) \overline{\beta(m,m)} \, \vols \\
&= \int_{\Sp^2} (\mathcal{L}_{Z_i} \alpha)(m,m) \cdot \beta(\overline{m}, \overline{m}) \, \vols \\
&= - \int_{\Sp^2} \alpha(m,m) \cdot (\mathcal{L}_{Z_i}\beta)(\overline{m}, \overline{m}) \, \vols \\
&= - \int_{\Sp^2} f \overline{\tilde{Z}_i h} \, \vols \;.
\end{aligned}
\end{equation*}
\end{proof}

\begin{remark}
Note that the smoothness of $m \otimes \overline{m}$ implies that if $f,h \in \mathscr{I}^\infty_{[s]}(\Sp^2)$, then $f \overline{h} \in C^\infty(\Sp^2)$. The above can now also be derived from observing $Z_i(f \overline{h}) = (Z_i f) \overline{h} + f Z_i \overline{h} =(\tilde{Z}_i f)\overline{h} + f \overline{\tilde{Z}_i h}$.
\end{remark}

\subsubsection{The spin $s$-weighted Laplacian} \label{SecSpinWeightedLaplacian}

The spin $s$-weighted Laplacian $\mathring{\slashed\Delta}^{[s]}$ on $\Sp^2$ is defined for $f \in \SWs$ in standard $(\theta, \varphi)$ coordinates by\footnote{This differs from the spin $s$-weighted Laplacian in \cite{DafHolRod17} by an overall minus sign.}
\begin{equation} \label{EqSpinWeightedLap}
\mathring{\slashed\Delta}_{[s]} f = \frac{1}{\sin\theta} \partial_\theta\big(\sin \theta \, \partial_\theta f\big) + \frac{1}{\sin^2 \theta} \partial_\varphi^2f + 2si\frac{\cos \theta}{\sin^2 \theta}\partial_\varphi f -\big( s^2\frac{\cos^2\theta}{\sin^2\theta} -s\big)f \;.
\end{equation}
We note that 
\begin{equation}\label{EqSWLV}
(\mathring{\slashed\Delta}_{[s]} - s -s^2)f = \big((\tilde{Z}_1)^2 + (\tilde{Z}_2)^2 + (\tilde{Z}_3)^2\big)f\;,
\end{equation}
such that it follows easily from Proposition \ref{PropFirstCharSpinWeighted} that \eqref{EqSpinWeightedLap} is a smooth operator on $\SWs$.

It follows directly from Proposition \ref{PropIntPartsSphere} and \eqref{EqSWLV} that for $f \in \SWs$ we have
\begin{equation}\label{EqSWLIntParts}
\begin{split}
-\int_{\Sp^2} \mathring{\slashed\Delta}_{[s]} f  \cdot \overline{f} \; \vols &= - \int_{\Sp^2} \Big( \sum_{i =1}^3 \zt_i^2 f + (s + s^2) f\Big) \cdot \overline{f} \; \vols \\ 
&= \int_{\Sp^2} \sum_{i=1}^3 |\zt_i f|^2 \; \vols - \int_{\Sp^2} (s + s^2) |f|^2 \; \vols \;.
\end{split}
\end{equation}

Note that for $s=0$ the right hand side of \eqref{EqSWLIntParts} is equal to $\int_{\Sp^2} \big( |\rd_\theta f|^2 + \frac{1}{\sin^2 \theta} |\rd_\varphi f|^2\big) \; \vols$, which gives non-degenerate control of the $\rd_\varphi$ derivative towards the north and south pole of $\Sp^2$. For $s \neq 0$, however, $\frac{1}{\sin \theta} \rd_\varphi f$ has in general a pole in $\theta$ at $\theta = 0, \pi$ and thus, in particular, is not square integrable on $\Sp^2$. The next lemma gives the appropriate generalisation, which is needed in Sections \ref{SecEENoShift} and \ref{SecEECauchy}.

\begin{lemma}\label{LemStrongerBoundPhiDer}
For $f \in \SWs$ we have 
\begin{equation}\label{LemLInfty}
\begin{aligned}
&\rd_\theta f \in L^\infty(\Sp^2)\\
&\frac{1}{\sin \theta} (is \cos \theta  + \rd_\varphi) f \in L^\infty(\Sp^2)
\end{aligned}
\end{equation} 
and the following holds:
\begin{equation}\label{LemEqStrongerBoundPhiDer}
\sum_{i=1}^3 |\zt_i f|^2 = |\rd_\theta f|^2 + \frac{1}{\sin^2 \theta} |is \cos\theta \cdot f + \rd_\varphi f|^2 + s^2 |f|^2 \;.
\end{equation}
\end{lemma}

\begin{proof}
In order to prove \eqref{LemLInfty} we note that by Proposition \ref{PropFirstCharSpinWeighted}  we have
\begin{align}
-\sin \varphi \cdot \zt_1 f - \cos \varphi \cdot \zt_2 f &= \rd_\theta f \in L^\infty(\Sp^2) \notag \\ 
 \sin \varphi \cdot \zt_2 f - \cos \varphi \cdot \zt_1 f &= \frac{1}{\sin \theta} (is + \cos \theta \cdot \rd_\varphi)f \in L^\infty(\Sp^2) \;. \label{EqSpinorialDer1}
 \end{align}
Multiplying \eqref{EqSpinorialDer1} by $\cos \theta$ and adding $\frac{1}{\sin \theta} \sin^2 \theta \cdot \rd_{\varphi} f$, which is clearly also bounded on $\Sp^2$, we obtain the second claim in \eqref{LemLInfty}.
The proof of \eqref{LemEqStrongerBoundPhiDer} is a direct computation:
\begin{equation*}
\begin{split}
\sum_{i=1}^3 |\zt_i f|^2 &= |\rd_\theta f|^2  + \frac{1}{\sin^2 \theta} |is f + \cos \theta \cdot \rd_\varphi f|^2 + |\rd_\varphi f|^2 \\
&= |\rd_\theta f|^2 + \frac{1}{\sin^2 \theta} | is \cos \theta \cdot f + \cos^2 \rd_\varphi f|^2 + |is f + \cos \theta \cdot \rd_\varphi f|^2 \\ &\qquad \qquad + \frac{1}{\sin^2 \theta}|\sin^2\theta \cdot \rd_\varphi f|^2 + \frac{\cos^2 \theta}{\sin^2 \theta} | \sin \theta \cdot \rd_\varphi f|^2 \\
&=|\rd_\theta f|^2 +  \frac{1}{\sin^2 \theta} | is \cos \theta \cdot f + \rd_\varphi f|^2 -\frac{1}{\sin^2 \theta} 2 \Rea\big( (is \cos \theta \cdot f + \cos^2 \theta \cdot \rd_\varphi f)\overline{(\sin^2 \theta \cdot \rd_\varphi f)}\big) \\
&\qquad \qquad + |is f + \cos \theta \cdot \rd_\varphi f|^2 + \cos^2 \theta | \rd_\varphi f|^2 \\ 
&= |\rd_\theta f|^2 + \frac{1}{\sin^2 \theta} |is \cos\theta \cdot  f + \rd_\varphi f|^2 + s^2 |f|^2 \;.
\end{split}
\end{equation*}
\end{proof}

\begin{lemma} \label{LemSecondAngularDer}
For $f \in \SWs$ we have
\begin{equation*}
\swl f \overline{\swl f} \eai \sum_{i,j =1}^3 |\zt_j \zt_i f|^2  - 2(s+s^2) \sum_{i = 1}^3 |\zt_i f|^2 + (s+s^2)^2 |f|^2 \;,
\end{equation*}
where $\eai$ denotes equality after integration over the sphere.
\end{lemma}

\begin{proof}
Using Proposition \ref{PropIntPartsSphere} we compute
\begin{equation*}
\swl f \overline{\swl f} = \Big( \sum_{i =1}^3 \zt_i^2 f + (s+s^2) f\Big) \Big(\sum_{j=1}^3 \overline{\zt_j^2 f +(s + s^2) f}\Big) \eai \sum_{i,j = 1}^3 \zt_i^2 f \overline{\zt_j^2 f} +(s+s^2)^2 |f|^2 - 2(s + s^2) \sum_{i = 1}^3 |\zt_i f|^2 \;.
\end{equation*}
Moreover, using the commutation relations $[\zt_i, \zt_j] = \varepsilon_{ijk} \zt_k$, we further compute
\begin{equation*}
\begin{split}
\sum_{i,j = 1}^3 \zt_i^2 f \overline{\zt_j^2 f} &\eai - \sum_{i,j = 1}^3 \zt_i f \overline{\zt_i \zt_j^2 f} \\
&= - \sum_{i,j = 1}^3 \zt_i f \overline{ \zt_j \zt_i \zt_j f} - \sum_{i,j,k = 1}^3 \varepsilon_{ijk} \zt_i f \overline{\zt_k \zt_j f} \\
&\eai \sum_{i,j = 1}^3 \zt_j \zt_i f \overline{\zt_i \zt_j f} - \sum_{i,j,k =1}^3 \varepsilon_{ijk} \zt_i f \overline{\zt_k \zt_j f} \\
&= \sum_{i,j = 1}^3 |\zt_j \zt_i f|^2 + \sum_{i,j,k = 1}^3 \varepsilon_{ijk} \zt_j \zt_i f \overline{\zt_k f} - \sum_{i,j,k =1}^3 \varepsilon_{ijk} \zt_i f \overline{\zt_k \zt_j f}  \\
&\eai \sum_{i,j = 1}^3 |\zt_j \zt_i f|^2 - \underbrace{\sum_{i,j,k = 1}^3 \varepsilon_{ijk} \big( \zt_i f ( \overline{ \zt_j \zt_k f + \zt_k \zt_j f})\big)}_{=0} \;.
\end{split}
\end{equation*}
\end{proof}

\subsubsection{The Teukolsky equation in Boyer-Lindquist coordinates}

Using \eqref{EqSpinWeightedLap} and $\Delta^{-s} \partial_r (\Delta^{s+1} \partial_r \psi_s) = 2(r-M)(s+1)\partial_r\psi_s + \Delta \partial_r^2 \psi_s$  we can rewrite the Teukolsky equation \eqref{TeukolskyEquation2} as
\begin{equation} \label{TeukolskyEquation3}
\begin{split}
\mathcal{T}_{[s]} \psi_s := &-\Big[\frac{(r^2 + a^2)^2}{\Delta} - a^2 \sin^2 \theta\Big] \partial_t^2\psi_s - \frac{4Mar}{ \Delta} \partial_t \partial_\varphi \psi_s - \frac{a^2}{ \Delta}\partial^2_\varphi \psi_s \\ 
&+ \Delta \partial_r^2 \psi_s +2(r-M)(1-s) \partial_r \psi_s  + 2s\frac{a(r-M)}{\Delta} \partial_\varphi \psi_s \\ 
&+ 2s\Big[\frac{M(r^2 -a^2)}{\Delta} - r - ia\cos \theta\Big] \partial_t \psi_s 
+\mathring{\slashed\Delta}_{[s]} \psi_s - 2s\psi_s = 0\;.
\end{split}
\end{equation}

\subsubsection{The Teukolsky equation in $\{v_+, \varphi_+, r,\theta\}$ coordinates}

Using \eqref{CoordTrafoStar} we rewrite \eqref{TeukolskyEquation3} in terms of $\{v_+, \varphi_+, r,\theta\}$ coordinates, which are regular at the right event horizon $\Hp_r$, to  obtain
\begin{equation}
\label{TeukolskyStar}
\begin{split}
\mathcal{T}_{[s]} \psi_s := & a^2 \sin^2 \theta \,\partial_{v_+}^2\psi_s + 2a \,\partial_{v_+}\partial_{\varphi_+} \psi_s + 2(r^2 + a^2)\, \partial_{v_+}\partial_r \psi_s \\
&+2 a\, \partial_{\varphi_+}\partial_r \psi_s + \Delta \,\partial_r^2 \psi_s + 2\Big( r(1-2s) - isa\cos \theta\Big)\, \partial_{v_+} \psi_s \\
&+2(r-M)(1-s) \,\partial_r \psi_s + \mathring{\slashed{\Delta}}_{[s]} \psi_s - 2s \psi_s 
= 0 \;.
\end{split}
\end{equation}

\subsubsection{The Teukolsky equation in $\{v_-, \varphi_-, r,\theta\}$ coordinates}

We express the Teukolsky equation \eqref{TeukolskyEquation1} for $\hat{\psi}_s$ (which is regular at $\Hp_l$) in terms of $\{v_-,\varphi_-,r,\theta\}$ coordinates (which are also regular at $\Hp_l$), using \eqref{CoordTrafoStarMinus}, to obtain
\begin{equation}\label{TeukolskyEquationHat}
\begin{split}
 \hat{\mathcal{T}}_{[s]} \hat{\psi}_s := & a^2 \sin^2 \theta \,\partial_{v_-}^2\hat{\psi}_s - 2a \,\partial_{v_-}\partial_{\varphi_-} \hat{\psi}_s + 2(r^2 + a^2)\, \partial_{v_-}\partial_r \hat{\psi}_s \\
&-2 a\, \partial_{\varphi_-}\partial_r \hat{\psi}_s + \Delta \,\partial_r^2 \hat{\psi}_s + 2\Big( r(1+2s) + isa\cos \theta\Big)\, \partial_{v_-} \hat{\psi}_s \\
&+2(r-M)(1+s) \,\partial_r \hat{\psi}_s + \mathring{\slashed{\Delta}}_{[s]} \hat{\psi}_s 
= 0 \;.
\end{split}
\end{equation}

\subsubsection{Spin weighted functions on spacetime}
\label{SubSecSpinWFctSpacetime}

We consider $\mathcal{M}$ and observe that the vector field $m = \frac{1}{\sqrt{2}}(\partial_\theta + \frac{i}{\sin \theta} \partial_\varphi)$, given in Boyer Lindquist coordinates, extends smoothly to $\M \setminus\{\theta = 0, \pi\}$ by virtue of $\partial_\varphi = \partial_{\Phi_{r_+}} = \partial_{\Phi_{r_-}}$. We consider the space $\Gamma^\infty(S^2(T^*\underline{\overline{\mathcal{M}}}))$ of all smooth and symmetric sections of $T^*\M \otimes T^*\M$ and the map $\iota_m$ which acts on an element $\alpha$ of $\Gamma^\infty(S^2(T^*\underline{\overline{\mathcal{M}}}))$ by $\iota_m \alpha = \alpha(m,m)$.

\begin{definition}
The space $\mathscr{I}^\infty_{[2]}(\M)$ of smooth spin $2$-weighted functions on $\M$ is defined as the image of $\Gamma^\infty(S^2(T^*\underline{\overline{\mathcal{M}}}))$ under $\iota_m$, i.e.
\begin{equation*}
\mathscr{I}^\infty_{[2]}(\M) := \iota_m\Big(\Gamma^\infty(S^2(T^*\underline{\overline{\mathcal{M}}}))\Big) \subseteq C^\infty(\M \setminus \{\theta = 0,\pi\}, \C) \;.
\end{equation*}
\end{definition}

\begin{remark} \label{RemDefSpinSpacetime}
\begin{enumerate}
\item As before, the space of smooth spin $-2$-weighted functions is defined as the image of $\Gamma^\infty(S^2(T^*\underline{\overline{\mathcal{M}}}))$ under $\iota_{\overline{m}}$ and the spin $\pm 1$-weighted functions are defined as the images of the space of smooth one-forms on $\M$.
\item It follows from the definition of the spin weighted spaces that they are invariant under multiplication by elements in $C^\infty(\M, \C)$. To see  this we note that multiplication by $i$ of a smooth spin $1$-weighted function corresponds to a concatenation of the one-covector field by a rotation of $\frac{\pi}{2}$ (with respect to the oriented frame field $\{\partial_\theta, \frac{1}{\sin\theta}\partial_\varphi\}$) while for smooth spin $2$-weighted functions it corresponds to a concatenation of the symmetric two covector field with a rotation of $\frac{\pi}{4}$.
\end{enumerate}
\end{remark}
Let us define the distribution $D \subset T\M$ which is annihilated by $\{dV^+_{r_+}, dV^-_{r_+}, dV^+_{r_-}, dV^-_{r_-}\}$ (where defined). Its integral manifolds in the interior of $\M$ are exactly the Boyer-Lindquist spheres of constant $t$ and $r$. We note that $m$ lies in the complexification of $D$.  Moreover, we denote the dual bundle of $D$ by $D^*$.
\begin{remark}\label{RemSpinWSpacetimeS2Tensor}
\begin{enumerate}
\item Given a subset $A \subseteq \M$ with the property that the integral manifolds of $D$ restricted to $A$ are complete spheres, we define the spin weighted spaces  $\mathscr{I}^\infty_{[s]}(A)$ analogously. For example we will choose $A = \Mun$ later.
\item  We define an auxiliary round metric $\slashed{g}_{\Sp^2}$ on the integral manifolds of $D$ by the symmetric part of $m \otimes \overline{m}$, cf.\ \eqref{EqMMbarMetric}. The kernel of the map $\iota_m : \Gamma^\infty(S^2(T^*\underline{\overline{\mathcal{M}}})) \to C^\infty(\M \setminus \{\theta = 0,\pi\}, \C)$ is the span of all those symmetric two-tensor fields that, when restricted to $D$, vanish or are proportional to $\slashed{g}_{\Sp^2}$. Thus, the space $\mathscr{I}^\infty_{[2]}(\M)$ of smooth spin $2$-weighted functions on $\M$ is isomorphic to the space  $\Gamma^\infty\big(S^2(D^* \to \M)\big)$, the space of all smooth, symmetric, and trace-free (with respect to $\slashed{g}_{\Sp^2}$)  sections of $D^* \otimes D^* \to \M$.
\item Given the above, a convenient realisation of the space $\mathscr{I}^\infty_{[2]}(\Mun)$ is as all those elements in $\Gamma^\infty(S^2(T^*\underline{\mathcal{M}}))$ that
\begin{itemize}
\item  vanish if $\rd_{V^+_{r_+}}$ is inserted in one of the slots
\item vanish if $\rd_{V^-_{r_+}}$ is inserted in one of the slots
\item are trace-free with respect to $\slashed{g}_{\Sp^2}$.
\end{itemize}
We will call such an element a \emph{symmetric and trace-free $\Sp^2$ $2$-covariant tensor field}. On this subset of $\Gamma^\infty(S^2(T^*\underline{\mathcal{M}}))$, $\iota_m$ is an isomorphism.
\end{enumerate}
\end{remark}

As before we can characterise the spin weighted functions on $\M$ among the elements of $C^\infty(\M\setminus\{\theta = 0, \pi\},\C)$. We do this in regions on which we have global coordinate charts. For example on $\Mun$ we introduce the first order differential operators $\tilde{Z}_{i,r_+}$, $i = 1,2,3$, which are defined as in \eqref{ZTilde} but with respect to the $\{V^+_{r_+}, V^-_{r_+}, \theta, \Phi_{r_+}\}$ coordinate system, i.e., we replace $\varphi$ in \eqref{ZTilde} by $\Phi_{r_+}$.  We obtain
\begin{proposition} \label{PropCharacterisationSpinSpacetime}
$f \in C^\infty(\Mun \setminus \{\theta = 0,\pi\}),\C)$ lies in $\mathscr{I}^\infty_{[s]}(\Mun)$ if, and only if, 
\begin{equation*}
e^{is\Phi_{r_+}} (\partial_{V^+_{r_+}})^{l_1}(\partial_{V^-_{r_+}})^{l_2} (\tilde{Z}_{1,r_+})^{k_1} (\tilde{Z}_{2,r_+})^{k_2} (\tilde{Z}_{3,r_+})^{k_3} f
\end{equation*} 
extends continuously to $\Mun \setminus \{\theta = \pi\}$ and 
\begin{equation*}
e^{-is\Phi_{r_+}} (\partial_{V^+_{r_+}})^{l_1}(\partial_{V^-_{r_+}})^{l_2} (\tilde{Z}_{1,r_+})^{k_1} (\tilde{Z}_{2,r_+})^{k_2} (\tilde{Z}_{3,r_+})^{k_3} f
\end{equation*} 
extends continuously to $\Mun \setminus \{\theta =0\})$ for all $l_1, l_2, k_1, k_2, k_3 \in \N_0$.
\end{proposition}

\begin{proof}
This is the same as the proof of Proposition \ref{PropFirstCharSpinWeighted}, noticing that we have $\partial_\varphi = \partial_{\Phi_{r_+}}$, $\mathcal{L}_{\partial_{V^+_{r_+}}} m = 0$,  $\mathcal{L}_{\partial_{V^-_{r_+}}} m = 0$, and also the last point in Remark \ref{RemSpinWSpacetimeS2Tensor}.
\end{proof}

Similarly we choose $\{V^+_{r_-}, V^-_{r_-}, \theta, \Phi_{r_-}\}$ coordinates on $\overline{\mathcal{M}}$ and define the operators $\tilde{Z}_{i,r_-}$, $i=1,2,3$, by replacing $\varphi$ in \eqref{ZTilde} by $\Phi_{r_-}$. We obtain an analogous characterisation of elements in $\mathscr{I}^\infty_{[s]}(\overline{\mathcal{M}})$. Taken together, this gives a characterisation of elements in $\mathscr{I}^\infty_{[s]}(\M)$ among those of $C^\infty(\M \setminus \{\theta = 0,\pi\}),\C)$.

We will also need to define the operators $\tilde{Z}_{i,+}$, $i=1,2,3$, with respect to the $\{v_+, \varphi_+, r, \theta\}$ coordinate system, i.e., we replace $\varphi$ in \eqref{ZTilde} by $\varphi_+$. Similarly we define the operators $\tilde{Z}_{i,-}$, $i=1,2,3$, with respect to the $\{v_-, \varphi_-, r, \theta\}$ coordinate system. We obtain analogous characterisations to Proposition \ref{PropCharacterisationSpinSpacetime} in the regions covered by each of these coordinate systems.

It now follows from \eqref{EqSWLV} (which obviously holds for any of the sets of $\tilde{Z}$ defined), the second part of Remark \ref{RemDefSpinSpacetime}, and Proposition \ref{PropCharacterisationSpinSpacetime} that the Teukolsky operator $\mathcal{T}_{[s]}$, defined in \eqref{TeukolskyStar}, is a smooth operator on $\mathscr{I}^\infty_{[s]}\big(\mathcal{M} \cup (\Hp_r \setminus \Sp^2_b) \cup (\CH_l \setminus \Sp^2_t)\big)$.  Similalry, the Teukolsky operator $\hat{\mathcal{T}}_{[s]}$, defined in \eqref{TeukolskyEquationHat}, is a smooth operator on $\mathscr{I}^\infty_{[s]}\big(\mathcal{M} \cup (\Hp_l \setminus \Sp^2_b) \cup (\CH_r \setminus \Sp^2_t)\big)$.

\subsubsection{The spin weighted Carter operator}

\begin{definition} \label{DefCarterOp}
We define the spin $s$-weighted Carter operator $\mathcal{Q}_{[s]}$ by
\begin{equation*}
\begin{split}
\mathcal{Q}_{[s]} :&= a^2 \sin^2 \theta \, \rd_{v_+}^2 - 2isa\cos \theta \, \rd_{v_+} + \swl \\
&= a^2 \sin^2 \theta\,  \rd_{v_-}^2 + 2isa\cos \theta \, \rd_{v_-} + \swl
\end{split}
\end{equation*}
\end{definition}
Note that it follows directly from \eqref{TeukolskyStar} and \eqref{TeukolskyEquationHat} that the spin $s$-weighted Carter operator commutes with the Teukolsky operator, i.e.\ we have $[\mathcal{T}_{[s]}, \mathcal{Q}_{[s]}] = 0 = [\hat{\mathcal{T}}_{[s]}, \mathcal{Q}_{[s]}]$.

\subsubsection{The regularity of $\dot{\alpha}$ as defined in Section \ref{SecGravPerturb} and \ref{SecAnotherFormTeuk}}

Our following global theorem will concern spin $2$-weighted functions, satisfying the Teukolsky equation, that satisfy the following smoothness properties.

\begin{assumption} \label{AssumptionReg}
 \begin{itemize}
 \item $\psi \in  \mathscr{I}^\infty_{[2]}(\Mun)$.
 \item $\frac{1}{(V^+_{r_+})^2}\psi \in \mathscr{I}^\infty_{[2]}(\Mun)$. Note that this implies, using \eqref{r-r_+}, that $\hat{\psi}:=\frac{1}{\Delta^2} \psi \in \mathscr{I}^\infty_{[2]}\big(\mathcal{M} \cup (\Hp_l \setminus \Sp^2_b)\big)$.
 \item $\psi$ satisfies $\mathcal{T}_{[2]}\psi = 0$ in $\mathcal{M} \cup (\Hp_r \setminus \Sp^2_b)$. Note that this implies that $\hat{\psi}$ satisfies $\hat{\T}_{[2]} \hat{\psi} = 0$ in $\mathcal{M} \cup (\Hp_l \setminus \Sp^2_b)$.
 \end{itemize}
\end{assumption}

We have dropped here the subscript $2$ from $\psi$ to shorten notation. No confusion can arise, since the remainder of the paper is only concerned with spin $2$-weighted functions.

We investigate what the above regularity assumptions imply for $\rd_r \psi$ and $\rd_r^2 \psi$, where the partial derivative is with respect to the $(v_+, r ,\theta, \varphi_+)$ coordinate system. By \eqref{KruskalCoordTrafo} and \eqref{r-r_+} we have $\rd_r = \frac{1}{V^+_{r_+}} f_1(r) \rd_{V^-_{r_+}} + f_2(r) \rd_{\Phi_{r_+}}$, where $f_1(r)$ and $f_2(r)$ are functions which extend smoothly to $r = r_+$. It now follows from Assumption \ref{AssumptionReg} that $\psi$ decays at least like $(V^+_{r_+})^2$ for $V^+_{r_+} \to 0$, $\rd_r\psi$ at least like $V^+_{r_+}$, and $\rd_r^2 \psi$ does in general not decay but is a regular smooth spin $2$-weighted function on $\Mun$.

We now show that $\dot{\alpha}$, as defined in Section \ref{SecGravPerturb} and \ref{SecAnotherFormTeuk}, satisfies the above smoothness assumptions \ref{AssumptionReg}.\footnote{Recall that $m_a$ differs from $m$ by a term proportional to $\partial_t$ -- thus the claim that $\dot{\alpha}$ is spin $2$-weighted is not trivial.}

\begin{proposition}\label{PropAlphaDotSpinW}
The quantity $\psi :=\dot{\alpha}$ from Section \ref{SecAnotherFormTeuk} satisfies the Assumptions \ref{AssumptionReg}. 
\end{proposition}

\begin{proof}
Recall that $e_4$ is a smooth vector field on $\Mun$, vanishing at $\Hp_l$. Also recall that
\begin{equation*}
\dot{\alpha} = \dot{R}(e_4(0), m_a(0), e_4(0), m_a(0)) =: \beta(m_a(0),m_a(0)) \;,
\end{equation*}
where we have defined $\beta$, a smooth and symmetric tensor field on $\Mun$. By \eqref{NPFrame} we have
\begin{equation*}
\beta(m_a(0),m_a(0)) = \frac{1}{(r + ia \cos \theta)^2} \beta(m + \frac{ia \sin\theta}{\sqrt{2} } \partial_t, m + \frac{ia \sin\theta}{\sqrt{2}} \partial_t) \;.
\end{equation*}
By the second point in Remark \ref{RemDefSpinSpacetime} it suffices to show that $\beta(m + \frac{ia \sin\theta}{\sqrt{2} } \partial_t, m + \frac{ia \sin\theta}{\sqrt{2}} \partial_t) \in \mathscr{I}^\infty_{[2]}(\Mun)$. By definition of the spin weighted spaces we have $\beta(m,m) \in \mathscr{I}^\infty_{[2]}(\Mun)$ and thus it remains to establish that $\sin \theta \cdot \beta(m,\partial_t) \in  \mathscr{I}^\infty_{[2]}(\Mun)$ and $\sin^2 \theta \in \mathscr{I}^\infty_{[2]}(\Mun)$. Let us define the smooth one-form $\gamma := \sin \theta \, d \theta$ on $\Mun$. Then $\gamma \otimes \gamma$ is a symmetric two-covector field with $(\gamma \otimes \gamma)(m, m) = \frac{1}{2} \sin^2 \theta$, which lies in $\mathscr{I}^\infty_{[2]}(\Mun)$. Similarly, defining the symmetric two-covector field $\gamma \otimes \beta(\cdot, \partial_t) + \beta(\cdot, \partial_t) \otimes \gamma$ shows that $\sin \theta \cdot \beta(m,\partial_t) \in  \mathscr{I}^\infty_{[2]}(\Mun)$. This shows the first point. The second point follows analogously recalling from Section \ref{CoordTransformation} that $\frac{1}{V^+_{r_+}} e_4$ is a smooth vector field on $\Mun$. The last two points were established in the previous sections.
\end{proof}

\section{Assumptions on the event horizon and the main theorem} \label{SecAssumptions}
In addition to the smoothness assumptions in Assumption \ref{AssumptionReg} we make the following assumptions on $\psi$ along the event horizons: 

Along the right event horizon $\Hp_r$:
Assume that there exists a $p \in \N$ such that $\int_{\Hp_r \cap \{v_+ \geq 1\}} v_+^{2p} |\psi|^2 \, \vols\,dv_+ = + \infty$. Let $p_0\in \mathbb N$ be the smallest integer such that this holds, i.e., we have
\begin{equation} \label{AssumpP0}
\int_{\Hp_r \cap \{v_+ \geq 1\}} v_+^{2p_0} |\psi|^2 \,\vols\,dv_{+} = +\infty\;.
\end{equation}
Assume that $p_0 \geq 2$. Moreover, we assume that there is $m_0 \in \Z$ and $ \N \ni l_0 \geq \max\{2, |m_0|\}$ such that
\begin{equation}\label{AssumpP0SphHarm}
\int_{v_+ \geq 1} v_+^{2p_0} | (\psi|_{\Hp_r})_{S(m_0l_0)}|^2 \, dv_+ = + \infty \;,
\end{equation}
where $(\psi|_{\Hp_r})_{S(ml)}(v_+) =\int_{\Sp^2} \psi|_{\Hp_r} (v_+, \theta, \varphi_+) \overline{Y^{[+2]}_{ml}(\theta, \varphi_+;0)} \, \vols$ denotes the projection of $\psi|_{\Hp_r}$ onto the spin $2$-weighted \emph{spherical} harmonic $Y^{[+2]}_{ml}(\theta, \varphi_+;0) $, cf.\ Section \ref{SecSpin2Harmonics}.
 We also assume that
\begin{equation}\label{AssumpP0Better}
\int_{\Hp_r \cap \{v_+ \geq 1\}} v_+^{2p_0} |\rd_{v_+} \psi|^2 \,\vols\,dv_{+} < +\infty
\end{equation}
holds and that for some $2 <q_r <2p_0$, $q_r \in \R$, we have\footnote{We have made no attempt in this paper to keep the number of derivatives required as low as possible,  one can certainly improve on it.  It is also likely that one can improve on the requirement $2 < q_r$ and thus also on the lower bound on $2 \leq p_0$. The bound $2 <q_r$ is used in Theorem \ref{ThmSeparationVariables} (via Corollary \ref{CorollaryNoShift} -- for which we also use all the derivatives assumed) to derive the radial ODE \eqref{EqThmSepVar3}. }
\begin{equation}\label{AssumpDecayRHp}
\sum_{0\leq i_1+i_2+i_3+j \leq 1} \int_{\Hp_r \cap \{v_+ \geq 1\}} v_+^{q_r} |\widetilde Z_{1,+}^{i_1}\widetilde Z_{2,+}^{i_2}\widetilde Z_{3,+}^{i_3}\rd_{v_+}^{j} f|^2 \,\vols\,dv_{+} < +\infty
\end{equation}
with $f \in \{\rd_{v_+}^a\rd_{\varphi_+}^b\rd_r^c\psi, \rd_{v_+} \rd_{v_+}^a\rd_{\varphi_+}^b\rd_r^c\psi, \mathcal{Q}_{[s]} \rd_{v_+}^a\rd_{\varphi_+}^b\rd_r^c \psi, \}$, $0 \leq a + b \leq 2$, $c = 0,1,2$.

Along the left event horizon $\Hp_l$: Assume that
\begin{equation}\label{AssumpLHp}
\hat{\psi} \textnormal{ is compactly supported on } \Hp_l \cup \Sp_b,
\end{equation}
i.e., there exists a $v_0 \in \R$ such that $\hat{\psi}$ vanishes in $\Hp_l \cap \{v_- \geq v_0\}$. However, all our results remain true if we replace \eqref{AssumpLHp} by the much weaker
\begin{equation}\label{AssumpLHpWeaker}
\sum_{0\leq i_1+i_2+i_3+j \leq 1 } \int_{\Hp_l \cap \{v_- \geq 1\}} v_-^{q_l} |\widetilde Z_{1,-}^{i_1}\widetilde Z_{2,-}^{i_2}\widetilde Z_{3,-}^{i_3}\rd_{v_-}^{j}f|^2 \,\vols\,dv_{-} < +\infty
\end{equation}
with  $f \in \{\rd_{v_-}^a \rd_{\varphi_-}^b \rd_r^c \hp, \rd_{v_-} \rd_{v_-}^a \rd_{\varphi_-}^b \rd_r^c \hp, \mathcal{Q}_{[s]} \rd_{v_-}^a \rd_{\varphi_-}^b \rd_r^c \hp \}$, $0 \leq a+b+c \leq 2$, $a,b,c \in \N_0$ and $\R \ni q_l \geq 2p_0$.\footnote{The asymmetry between the number of derivatives assumed on the left and right event horizons can be traced back to the necessity to close the energy estimate near $\Hp_r$ at the level of $(\rd_r|_+)^2 \psi$ while near $\Hp_l$ we close it at the level of $\hp$. The higher number of derivatives assumed on $\Hp_r$ allows us to ease the presentation of the proof of Proposition \ref{PropRightHP} in Step 6. However, one can certainly improve on that.} To see that \eqref{AssumpLHp} implies \eqref{AssumpLHpWeaker} for $f = \rd_r \hat{\psi}$, we notice that the Teukolsky equation \eqref{TeukolskyEquationHat} reduces in the region $\Hp_l \cap \{v_- \geq v_0\}$, where $\hp$ vanishes, to
$\big(2(r^2 + a^2) \rd_{v_-} - 2a \rd_{\varphi_-} \big) \rd_r \hp + 2(r-M)(1+s) \rd_r \hp =0$. This shows that $\rd_r \hp$  decays exponentially along $\Hp_l$ -- a manifestation of the red-shift effect. Further commutations with $\rd_r$ even improve the red-shift.

For the statements of the intermediate results in the main body of the paper we will often use the phrase `\emph{under the assumptions of Section \ref{SecAssumptions}}'. Let us make explicit that by this we mean the Assumption \ref{AssumptionReg} together with \eqref{AssumpP0}, \eqref{AssumpP0SphHarm}, \eqref{AssumpP0Better}, \eqref{AssumpDecayRHp}, and \eqref{AssumpLHpWeaker}. However, usually not all of these assumptions are required for the specific partial result proven. 

With the exception of Section \ref{SecM0}, where we briefly consider the case $s=-2$, this paper is only concerned with the case $s= +2$. However, we will not replace the $s$ in the Teukolsky equation by $2$ so that the reader can follow the importance of the value of $s$ for the validity of our estimates. With the exception of Section \ref{SecM0} the convention in this paper is that $s=+2$.

\begin{theorem} \label{Thm1}
Let $\psi$ satisfy the  Assumptions \ref{AssumptionReg},  \eqref{AssumpP0}, \eqref{AssumpP0SphHarm}, \eqref{AssumpP0Better}, \eqref{AssumpDecayRHp}, and  \eqref{AssumpLHpWeaker}.  Let $v_2 \in \R$ and consider the spacelike hypersurface $\Sigma := \{ f^- = v_2\}$ which is transversal to $\CH_r$. We then have 
\begin{equation}\label{EqThm1}
\int\limits_{\Sigma \cap \{v_+ \geq 1\}} v_+^{2p_0} |\psi|^2 \, \vols dv_+ = \infty \;.
\end{equation}
\end{theorem}

The above theorem is global in nature, it concerns solutions of the Teukolsky equation defined in all of the interior of asymptotically flat two-ended Kerr black holes. As stated, it is not a useful ingredient for treating realistic one-ended rotating black holes. In the following we give a version of Theorem \ref{Thm1} localised to a neighbourhood of timelike infinity.

\begin{theorem} \label{Thm2}
Consider a patch of $\Mun$ given by $\Mun \cap \{f_- \leq v_1\} \cap \{f_+ \geq v_0\}$ for some $v_1, v_2 \in \R$, see also Figure \ref{FigExt} on page \pageref{FigExt}. Let $\psi \in \mathscr{I}^\infty_{[2]}(\Mun\cap \{f_- \leq v_1\} \cap \{f_+ \geq v_0\} )$ satisfy the Teukolsky equation $\mathcal{T}_{[2]}\psi = 0$ and\footnote{with the integration $\Hp_r \cap \{v_+ \geq 1\}$ replaced by $\Hp_r \cap \{v_+ \geq v_0\}$} \eqref{AssumpP0}, \eqref{AssumpP0SphHarm}, \eqref{AssumpP0Better}, and \eqref{AssumpDecayRHp} . Let $v_2 \leq v_1$ and consider the spacelike hypersurface $\Sigma := \{ f^- = v_2\}$ which is transversal to $\CH_r$. We then have 
\begin{equation} \label{EqThmConc}
\int\limits_{\Sigma \cap \{v_+ \geq 1\}} v_+^{2p_0} |\psi|^2 \, \vols dv_+ = \infty \;.
\end{equation}
\end{theorem}

\begin{remark} \label{RemMainThm}
\begin{enumerate}
\item The proof of Theorems \ref{Thm1} and \ref{Thm2} contains a crucial Fourier-theoretic component. To obtain the instability at the Cauchy horizon we use that there is a smallest $p_0 \in \N$ and $m_0 \in \Z$ and $\N \ni l_0 \geq \max\{2, |m_0|\}$ such that 
\begin{equation}
\label{EqAlternativeAssump}
\rd_\omega^{p_0} (\widecheck{\psi|_{\Hp_r}})_{m_0l_0}(\omega) \notin L^2_\omega(-\varepsilon, \varepsilon) \; \textnormal{ for any } \;\varepsilon >0\;.
\end{equation}
Here $ (\widecheck{\psi|_{\Hp_r}})_{m_0l_0}$ results from $\psi|_{\Hp_r}$ by taking the Fourier transform in $v_+$ and subsequent projection on the $m_0,l_0$ spin $2$-weighted \emph{spheroidal} harmonic, cf.\ Section \ref{SecTeukExp}. The physical space assumptions \eqref{AssumpP0}, \eqref{AssumpP0SphHarm}, \eqref{AssumpP0Better} are only used to guarantee \eqref{EqAlternativeAssump}, see Proposition \ref{PropFourierAssump}. In particular the above Theorems remain true if \eqref{AssumpP0}, \eqref{AssumpP0SphHarm}, \eqref{AssumpP0Better}  are replaced by \eqref{EqAlternativeAssump}. Note that the assumption \eqref{AssumpDecayRHp} implies that $2p_0 $ must be greater than $q_r$.
\item Having dropped the subscript $s$ from $\psi$, we introduce the notation $$\psi_m(v_+,r,\theta, \varphi_+) := \int_{\Sp^1} \psi(v_+, r, \theta, \varphi_+') \cdot e^{-im\varphi_+'} \, d \varphi_+' \cdot e^{im \varphi_+}$$ for the projection on the $m$-th azimuthal mode, $m \in \Z$, and also $\psi_{\neq 0} := \psi - \psi_0$. Note that if $\psi$ solves the Teukolsky equation then so does $\psi_m$. We can thus apply the above theorems also to the projections $\psi_m$ individually to obtain statements which, through the ensuing  $m$-dependent parameter $p_0$, depend on $m$. 
\item It was shown recently in \cite{MaZha21} (see also \cite{DafHolRod17}, \cite{Ma20} and also \cite{BarOri99}) that for slowly rotating\footnote{One expects that these results remain true in the full sub-extremal range, see \cite{MaZha21}.} sub-extremal Kerr and for compactly supported initial data for the Teukolsky equation, posed on a spacelike hypersurface connecting the event horizon with spacelike infinity, one has 
\begin{equation*}
\big| \rd_{v_+}^j \psi_{m \neq 0}|_{\Hp_r} - \sum_{m = \pm 1, \pm 2}Q_{m,2} Y^{[+2]}_{m2}(\theta, \varphi_+; 0) v_+^{-7-j} \big| \lesssim v_+^{-7-j - \varepsilon} \;,
\end{equation*}
where $\varepsilon>0$ and $Q_{m,2} \in \C$ is  \emph{generically} non-vanishing, and
\begin{equation*}
\big| \rd_{v_+}^j \psi_0|_{\Hp_r} - Q_{0,2} Y^{[+2]}_{02}(\theta, \varphi_+; 0) v_+^{-8-j} \big| \lesssim v_+^{-8-j - \varepsilon} \;,
\end{equation*}
where again $Q_{0,2} \in \C$ is \emph{generically} non-vanishing.
For $v_+$ large enough we thus obtain for $m = \pm 1, \pm 2$ generically $|(\psi|_{\Hp_r})_{S(m2)}| \geq c v_+^{-7}$ with $c>0$ and for $m = 0$ generically $|(\psi|_{\Hp_r})_{S(02)}| \geq c v_+^{-8}$ with $c>0$. Hence for $m_0 = \pm 1, \pm 2$ the assumptions made in this section are generically satisfied with $l_0 = 2 $ and $p_0 = 7$ and for $m_0 = 0$ with $l_0 = 2$ and $p_0 = 8$. If we do not decompose into azimuthal modes the assumptions are generically satisfied with $p_0 = 7$, $l_0=2$ and $m_0 \in \{-2,-1,1,2\}$. The parameter $q_r$ can  be chosen to be anything strictly less than $13$. 

If we do not assume the initial data to be compactly supported, but still to be smooth with respect to the conformal compactification at future null infinity, we expect the generic decay rates to be slower by a power of $v_{+}^{-1}$, see also \cite{AAG21}. There is evidence that the assumption of smoothness at future null infinity is not satisfied in many physically interesting situations (see \cite{Chris02}, \cite{Kehrb21a}) and that this impacts the late time tails \cite{Kehrb21b}. There is also evidence that tails arising on dynamical black hole exteriors differ from those on stationary exteriors \cite{LukOhForth}.
\item We rewrite \eqref{EqThmConc} in terms of quantities that are regular at $\CH_r$: we first recall that $\hat{\psi} = \frac{1}{\Delta^2} \psi$ is the linearisation of the curvature component with respect to the algebraically special frame that is regular at $\CH_r$ and that we have $\Delta^2 \simeq e^{2 \kappa_-(v_+ + v_-)} \simeq e^{2 \kappa_- v_+}$ along $\Sigma$ for $v_+ \to \infty$, where we have used \eqref{EqRStarCauchy}. Moreover, we have $V^{+}_{r_-} = - e^{\kappa_- v_+}$ and thus $\log(-V^+_{r_-}) = \kappa_- v_+$ and $dv_+ = \frac{1}{\kappa_- V^+_{r_-}}dV^+_{r_-}$. We thus find
$v_+^{2p_0} |\psi|^2  \simeq \big[\log(-V^+_{r_-})\big]^{2p_0} (-V^+_{r_-})^4 |\hat{\psi}|^2  $
along $\Sigma$ for $v_+ \to \infty$ and hence \eqref{EqThmConc} is equivalent to 
\begin{equation}
\label{EqRegularFinal}
\int_{\Sigma \cap \{v_+ \geq 1\}} \big[\log(-V^+_{r_-})\big]^{2p_0} (-V^+_{r_-})^3 |\hat{\psi}|^2 \, \vols dV^+_{r_-} = \infty \;.
\end{equation}
\item As a side result we also prove
\begin{equation*}
\int\limits_{\Sigma \cap \{v_+ \geq 1\}} v_+^{q_r} |\psi|^2 \, \vols dv_+ < \infty \;,
\end{equation*}
see \eqref{EqAddOn}. Hence, the integral in \eqref{EqRegularFinal} with $2p_0$ replaced by $q_r$ is finite. Recall that we said that in particular for compactly supported initial data $q_r$ can be chosen to be anything strictly less that $2p_0 - 1$.
\end{enumerate}

\end{remark}

\section{Energy estimates for the Teukolsky equation: upper bounds} \label{SecEE}

In this section we prove stability estimates which are being used to justify Teukolsky's separation of variables, to pass to the limits $r \to r_\pm$, and to propagate the singularity backwards along $\CH_r$.


\subsection{Estimates near the event horizons} \label{SecEERed}

We begin with the semi-global estimates near the left event horizon, since they are the simplest and thus the structure is easier to understand here.

\begin{proposition} \label{PropLeftHp}
Under the assumptions of Section \ref{SecAssumptions} there exists an $\rred \in (r_-, r_+)$ and a $C>0$ such that
\begin{equation} \label{EqPropLeftH}
\begin{split}
\sup_{r' \in [r_\mathrm{red}, r_+]} &\sum_{0 \leq i_1 + i_2 + i_3 + j + k \leq 1} \int\limits_{\{r=r'\} \cap \{v_- \geq 1\}} v_-^{q_l} |\Delta|^k|\widetilde Z_{1,-}^{i_1}\widetilde Z_{2,-}^{i_2}\widetilde Z_{3,-}^{i_3}\rd_{v_-}^{j}\rd_r^{k}f|^2 \,\vols\,dv_{-}  \\
&+ \sum_{0 \leq i_1 + i_2 + i_3 + j + k \leq 1}\int\limits_{\{\rred \leq r \leq r_+\} \cap \{v_- \geq 1\}}     v_-^{q_l} |\widetilde Z_{1,-}^{i_1}\widetilde Z_{2,-}^{i_2}\widetilde Z_{3,-}^{i_3}\rd_{v_-}^{j}\rd_r^{k}f|^2 \,\vols\,dv_{-} dr
\leq C
\end{split}
\end{equation}
holds for $f \in \{\rd_{v_-}^a \rd_{\varphi_-}^b \rd_r^c \hp, \rd_{v_-} \rd_{v_-}^a \rd_{\varphi_-}^b \rd_r^c \hp, \mathcal{Q}_{[s]} \rd_{v_-}^a \rd_{\varphi_-}^b \rd_r^c \hp \}$, $0 \leq a+b+c \leq 2$, $a,b,c \in \N_0$.\footnote{Note that away from $r_\pm$ we have $\mathrm{span}\{\rd_r|_+, \rd_{v_+}, \rd_{\varphi_+}\} = \mathrm{span}\{\rd_r|_-, \rd_{v_-}, \rd_{\varphi_-}\} = \mathrm{span}\{\rd_r, \rd_{t}, \rd_{\varphi}\}$. Since we carry out the different energy estimates in different coordinates, it is convenient to always consider this combination of derivatives.}
\end{proposition}

Here, and in the following propositions and corollaries throughout Section \ref{SecEE}, the constant $C$ depends in particular on the initial data of the Teukolsky field on $\Hp_l$ and $\Hp_r$, on $q_l$ and $q_r$, on the black hole parameters, and, in general, on the region in which the estimate holds. The exact dependency and the optimal value of the constant is, however, of no interest to this paper. We only need the qualitative statement that the quantity in question is finite.

\begin{proof}
\underline{\bf{Step 1: The multiplier.}}
In the following we restrict to $v_- \geq 1$.
We start from the following multiplier identity, where $\lambda, \eta, \mu>0$ are constants to be chosen:
\begin{equation}\label{EqMultRedLeft}
\begin{split}
0&=\Rea\Big( \hat{\mathcal{T}}_{[s]} \hp \cdot v_-^{q_l} \big( -(1 + \lambda \Delta)\partial_r + (1+ \lambda \Delta)\rd_{v_-}\big) \overline{\hp}\Big) \\
&\qquad + \underbrace{\rd_r(v_-^{q_l}\mu e^{\eta r}|\hp|^2) - v_-^{q_l} \mu \eta e^{\eta r} |\hp|^2 - 2v_-^{q_l} \mu e^{\eta r} \Rea(\overline{\hp}\rd_r \hp)}_{=0}
\end{split}
\end{equation}
Here, for $\lambda>0$ suitably, the vector field  $ -(1 + \lambda \Delta)\partial_r + (1+ \lambda \Delta)\rd_{v_-}$ is a choice of the redshift vector field of Dafermos and Rodnianski, \cite{DafRod09a}, \cite{DafRod08}, and the underbraced term is added in order to control the zeroth order terms, as will become clear in the following. After integration over the spheres, and using the form \eqref{TeukolskyEquationHat}    of the Teukolsky equation, the right hand side of \eqref{EqMultRedLeft} is the sum of
\begin{enumerate}
\item the sum of all the  terms on the right hand sides of \ref{ApHatR} and \ref{ApHatV}
\item the real part of the terms
\begin{equation*}\begin{split}
&2(r(1+2s) + isa\cos \theta) \partial_{v_-} \hp v_-^{q_l} \big( -(1 + \lambda \Delta)\partial_r + (1+ \lambda \Delta)\rd_{v_-}\big) \overline{\hp} \\
&\qquad - \dashuline{v_-^{q_l}(1+ \lambda \Delta) 2(r-M)(1+s)|\rd_r \hp|^2} + v_-^{q_l}(1+ \lambda \Delta)2(r-M)(1+s)\rd_r \hp \overline{\rd_{v_-} \hp}
\end{split}
\end{equation*}
\item the underbraced term in \eqref{EqMultRedLeft}.
\end{enumerate}
As will become clear later, we can derive a boundedness statement if the bulk terms (those terms which are not total derivatives) are \emph{negative}. Recall that $\rd_r \Delta (r_+) = 2(r_+ -M) >0$.

\underline{\bf{Step 2: Estimating all bulk terms that are quadratic in derivatives.}}

The two dashed terms, which are the most important terms, combine to give a negative contribution for $r$ close enough (depending on $\lambda$) to $r_+$. Indeed, for $s=0$ this is the familiar red-shift for the wave equation and we see that for $s=+2$ we even get an improved red-shift for the energy.\footnote{Note that the structure for $s=+2$ is the following:  for $\hat{\psi}$ strong red-shift for the energy at $\Hp_l$, strong blue-shift at $\CH_r$; for $\psi$ blue-shift at $\Hp_r$, red-shift at $\CH_l$. This is the reason why the estimate for $\psi$ at $\Hp_r$ is slightly more complicated.} 

We now investigate all the $\uwave{\textnormal{terms}}$ with a wavy underline, which are all those that are leading order in $\lambda$. The first of those terms in \ref{ApHatR} is negative. The second of those terms in \ref{ApHatR}, which indeed appears again from the fourth equation in \ref{ApHatV}, can be controlled as follows:
\begin{equation}\label{EqEstLeadTermHpl}
|2av_-^{q_l} \lambda (\rd_r \Delta) \Rea(\rd_{\varphi_-} \hp \overline{\rd_{v_-} \hp})| \leq v_-^{q_l} \lambda (\rd_r \Delta)\Big( \frac{1}{2} \alpha |\widetilde{Z}_{3,-} \hp|^2 + \alpha^{-1} 2a^2 | \rd_{v_-} \hp |^2\Big)
\end{equation} 
Note that $|a| < M < r_+$ so that there exists $0< \alpha <1$ close to $1$ such that $r_+^2 + a^2 > 2 \alpha^{-1} a^2$. Hence, \eqref{EqEstLeadTermHpl} can be estimated uniformly in $\lambda$ by the last wavily underlined term in \ref{ApHatR} and the one in \ref{ApHatV}. In summary, all the wavily underlined terms and the dashed terms are estimated from above by
\begin{equation}\label{EqEstLambda}
-v_-^{q_l} f(r,\lambda)\big( \lambda (\sum_i | \tilde{Z}_{i,-} \hp|^2 + |\rd_{v_-} \hp|^2) + |\rd_r \hp|^2\big)\;,
\end{equation}
where $f(r_+, \lambda)>0$ is independent of $\lambda$. All the other bulk terms which are quadratic in derivatives of $\hp$ can now be controlled in absolute value by $-\frac{1}{2}\times\textnormal{ \eqref{EqEstLambda}}$ by choosing first $\lambda>0$ big enough and then restricting to  $r \in [r_\mathrm{red}, r_+]$, $v_- \geq v_0$  with $r_\mathrm{red} < r_+$ close enough to $r_+$ and $v_0 >1$ large enough.\footnote{We need to choose $v_0$ large enough to control the second term on the right hand side of the third multiplier expression computed in \ref{ApHatR}. This one is quadratic in $\rd_r \hat{\psi}$, has a positive sign, but a sub-leading $v_-$-weight.}

\underline{\bf{Step 3: Estimating boundary terms.}}

We gather all the total derivatives appearing on the right hand side of \eqref{EqMultRedLeft}. They are $\partial_{v_-} \big(A(v_-,r,\hp, \rd\hp)\big)$ and $\rd_r\big( B(v_-,r,\hp, \rd \hp)\big)$ with
\begin{equation*}
\begin{split}
A(v_-,r,\hp,\rd \hp) &= -a^2 \sin^2\theta v_-^{q_l} (1 + \lambda \Delta)\Rea(\rd_{v_-} \hp \overline{\rd_r\hp}) + 2av_-^{q_l}(1+ \lambda \Delta)\Rea(\rd_{\varphi_-}\hp \overline{\rd_r\hp}) - v_-^{q_l}(1+ \lambda \Delta)(r^2 + a^2) |\rd_r\hp|^2 \\
&+\frac{1}{2}v_-^{q_l}(1+ \lambda \Delta)a^2 \sin^2\theta |\rd_{v_-}\hp|^2 - \frac{1}{2} v_-^{q_l}(1+ \lambda \Delta) \Delta |\rd_r\hp|^2 + \frac{1}{2} v_-^{q_l}(1+ \lambda \Delta))(s + s^2) |\hp|^2 \\
&-\frac{1}{2} v_-^{q_l}(1+ \lambda \Delta) \sum_i | \tilde{Z}_{i,-} \hp|^2
\end{split}
\end{equation*}
and
\begin{equation}\label{EqBBdry}
\begin{split}
B(v_-,r,\hp,\partial \hp) &= \frac{1}{2} a^2 \sin^2 \theta v_-^{q_l} (1 + \lambda \Delta)|\rd_{v_-}\hp|^2 - 2av_-^{q_l} (1 + \lambda \Delta)\Rea(\rd_{\varphi_-}\hp \overline{\rd_{v_-}\hp}) - \frac{1}{2} v_-^{q_l} (1 + \lambda \Delta) \Delta |\rd_r \hp|^2 \\
&-\frac{1}{2}v_-^{q_l} (1 + \lambda \Delta)(s+s^2) |\hp|^2 + \frac{1}{2} v_-^{q_l} (1 + \lambda \Delta) \sum_i | \tilde{Z}_{i,-} \hp|^2 + v_-^{q_l} (1 + \lambda \Delta)(r^2 +a^2)|\rd_{v_-}\hp|^2 \\
&+v_-^{q_l} (1 + \lambda \Delta) \Delta \Rea(\rd_r \hp \overline{\rd_{v_-} \hp}) + v_-^{q_l} \mu e^{\eta r} |\hp|^2 \;.
\end{split}
\end{equation}
We begin by establishing coercivity of $B$ for $r$ close enough to $r_+$. The second term in \eqref{EqBBdry} can be absorbed by the fifth and sixth term as follows
\begin{equation*}
2|a|v_-^{q_l} (1 + \lambda \Delta)|\rd_{\varphi_-}\hp \overline{\rd_{v_-}\hp}| \leq v_-^{q_l} (1 + \lambda \Delta)(\frac{1}{2}\alpha |\tilde{Z}_{3,-} \hp|^2 + 2a^2 \alpha^{-1} |\rd_{v_-} \hp|^2)
\end{equation*}
where $0< \alpha <1$ and we argue as in \eqref{EqEstLeadTermHpl}. The seventh term in \eqref{EqBBdry} is estimated by the third and sixth term by
\begin{equation*}
v_-^{q_l} (1 + \lambda \Delta) |\Delta| |\rd_r \hp \overline{\rd_{v_-} \hp}| \leq v_-^{q_l} (1 + \lambda \Delta) |\Delta|(\frac{1}{2} \alpha |\rd_r\hp|^2 + \frac{1}{2} \alpha^{-1}|\rd_{v_-}\hp|^2)\;,
\end{equation*}
for $0 < \alpha < 1$, where we note that the additional $|\Delta|$ allows us to absorb the $|\rd_{v_-} \hp|^2$ term. Finally we choose $\mu>0$ as a function of $\eta>0$ (to be determined later) such that $\mu(\eta) e^{\eta r_+}  = 2(s+s^2)$. It thus follows that for $r_\mathrm{red}<r_+$ close enough to $r_+$ we have
\begin{equation}\label{EqLowBoundB}
B(v_-,r,\hp, \rd \hp) \gtrsim v_-^{q_l} (|\Delta||\rd_r \hp|^2 + |\rd_{v_-}\hp|^2 + \sum_i | \tilde{Z}_{i,-} \hp|^2  +|\hp|^2)
\end{equation}
for $r \in [r_\mathrm{red}, r_+]$.

Next we establish the coercivity of $B(v_-,r, \hp, \rd \hp) - A(v_-,r, \hp, \rd \hp)$. We first compute
\begin{equation*}
\begin{split}
B(v_-,r,\hp, \partial \hp) - A(v_-,r,\hp, \rd \hp) &=  -2av_-^{q_l}(1 + \lambda \Delta) \Rea(\rd_{\varphi_-} \hp \overline{\rd_{v_-} \hp}) - v_-^{q_l}(1 + \lambda \Delta)(s+s^2)|\hp|^2  \\
&+v_-^{q_l}(1 + \lambda \Delta) \sum_i|\tilde{Z}_{i,-} \hp|^2 + v_-^{q_l}(1 + \lambda \Delta)(r^2 + a^2) |\rd_{v_-} \hp|^2 \\
&+v_-^{q_l}(1 + \lambda \Delta) \Delta \Rea(\rd_r \hp \overline{\rd_{v_-} \hp}) + v^{-{q_l}}\mu e^{\eta r} |\hp|^2 \\
&+ \underline{a^2 \sin^2 \theta v_-^{q_l}(1 + \lambda \Delta) \Rea(\rd_{v_-} \hp \overline{\rd_r \hp})} - 2av_-^{q_l}(1 + \lambda \Delta) \Rea(\rd_{\varphi_-} \hp \overline{\rd_r \hp}) \\
&+ v_-^{q_l}(1 + \lambda \Delta)(r^2 + a^2) |\rd_r \hp|^2 \;.
\end{split}
\end{equation*}
In particular completing the square for the underlined term gives
\begin{equation}\label{EqLowBoundAB}
\begin{split}
B(v_-,r,\hp, \partial \hp) - A(v_-,r,\hp, \rd \hp) &=  -2av_-^{q_l}(1 + \lambda \Delta) \Rea(\rd_{\varphi_-} \hp \overline{(\rd_{v_-} +\rd_{r}) \hp}) - v_-^{q_l}(1 + \lambda \Delta)(s+s^2)|\hp|^2  \\
&+v_-^{q_l}(1 + \lambda \Delta) \sum_i|\tilde{Z}_{i,-} \hp|^2 +v_-^{q_l}(1 + \lambda \Delta) \Delta \Rea(\rd_r \hp \overline{\rd_{v_-} \hp})  \\
&+ v^{-{q_l}}\mu e^{\eta r} |\hp|^2   + \frac{1}{2} v_-^{q_l}(1 + \lambda \Delta) (r^2 + a^2 + \frac{1}{2}a^2 \sin^2 \theta) |(\rd_{v_-} + \rd_r) \hp|^2 \\
&+\frac{1}{2} v_-^{q_l}(1 + \lambda \Delta)(r^2 + \frac{1}{2}(a^2 +a^2 \cos^2 \theta))|(\rd_{v_-} - \rd_r)\hp|^2 \;.
\end{split}
\end{equation}
The first term is estimated in the same way as before now by the third and sixth term. Note that the fourth term vanishes at $r= r_+$. Thus, with our choice of $\eta$ from above we obtain
\begin{equation*}
B(v_-,r,\hp, \partial \hp) - A(v_-,r,\hp, \rd \hp)  \gtrsim v_-^{q_l} (|\rd_r \hp|^2 + |\rd_{v_-}\hp|^2 + \sum_i | \tilde{Z}_{i,-} \hp|^2  +|\hp|^2)
\end{equation*}
for $r \in [r_\mathrm{red}, r_+]$ with $r_\mathrm{red}$ close enough to $r_+$.

\underline{\bf{Step 4: Estimating the remaining bulk terms:}}

The last two terms in \eqref{EqMultRedLeft} are estimated by
\begin{equation}\label{EqEstAuxTerms}
-v_-^{q_l} \eta \mu e^{\eta r} |\hp| -2v_-^{q_l} \mu e^{\eta r} \Rea(\overline{\hp}\rd_r \hp) \leq -\frac{1}{2} v_-^{q_l} \eta \mu e^{\eta r} |\hp|^2 + 2v_-^{q_l} \eta^{-1} \mu e^{\eta r} |\rd_r \hp|^2 \;.
\end{equation}
We can now choose $\eta>0$ sufficiently large such that the last term can be controlled by $-\frac{1}{4}\times \textnormal{\eqref{EqEstLambda}}$ and such that the first term controls the zeroth order terms arising in the bulk from \eqref{ApHatR} and \eqref{ApHatV} (those have an overall `bad' positive sign and need to be controlled).

\underline{\bf{Step 5: Putting it all together:}}

We thus obtain after integration over the spheres
\begin{equation*}
\rd_{v_-} \big(A(v_-,r,\hp,\rd \hp)\big) + \rd_r\big(B(v_-,r,\hp, \rd \hp)\big) \underset{{\mathrm{a.i.}}}{\gtrsim} v_-^{q_l}\big(|\rd_r \hp|^2 + |\rd_{v_-}\hp|^2 + \sum_i|\tilde{Z}_{i,-} \hp|^2 + |\hp|^2\big)
\end{equation*}
for $v_- \geq v_0$ and $r_\mathrm{red} \leq r \leq r_+$. Let $r' \in [\rred, r_+)$. We integrate over the region $\{2v_0 \leq f^- \leq v_1 \} \cap \{r' \leq r \leq r_+\}$ with respect to $dv_-\wedge dr \wedge \vols = \frac{1}{\rho^2} \vol$ and use that on a level set of $f^-$ we have $dr = dv_-$ to obtain
\begin{equation*}
\begin{split}
\int\limits_{\mathclap{\substack{\{r=r'\} \\ \cap \{2v_0 \leq f^- \leq v_1\}}}} &B \, \vols dv_- 
+ \int\limits_{\mathclap{\substack{\{f^-=v_1\} \\ \cap \{r' \leq r \leq r_+\}}}} (B - A) \, \vols dv_- + c \int\limits_{\mathclap{\substack{\{2v_0 \leq f^- \leq v_1\} \\ \cap \{r' \leq r \leq r_+\}}}} v_-^{q_l}\big(|\rd_r \hp|^2 + |\rd_{v_-}\hp|^2 + \sum_i|\tilde{Z}_{i,-} \hp|^2 + |\hp|^2\big) \, \vols dv_-dr \\
&\leq \int\limits_{\mathclap{\Hp_l \cap \{2v_0 \leq f^- \leq v_1\}}} B \, \vols dv_-  + 
\int\limits_{\mathclap{\substack{\{f^-=2v_0\} \\ \cap \{r' \leq r \leq r_+\}}}} (B - A) \, \vols dv_-  \;,
\end{split}
\end{equation*}
where $c>0$. Using now the lower bounds \eqref{EqLowBoundB} and \eqref{EqLowBoundAB}, the trivial upper bounds on $A$ and $B$, the assumption \eqref{AssumpLHpWeaker} on $\hp$ as well as the regularity Assumption \ref{AssumptionReg} for the  boundary term on $\{f^- = 2v_0\}$, and letting $v_1 \to \infty$ we obtain
\begin{equation*}
\begin{split}
\sup_{r' \in [\rred, r_+]}&\int\limits_{\{r=r'\} \cap \{f^- \geq 2v_0\}} v_-^{q_l}\big(|\Delta||\rd_r \hp|^2 + |\rd_{v_-}\hp|^2 + \sum_i|\tilde{Z}_{i,-} \hp|^2 + |\hp|^2\big) \, \vols dv_- \\
&+ \int\limits_{\mathclap{\{2v_0 \leq f^- \}  \cap \{\rred \leq r \leq r_+\}}} v_-^{q_l}\big(|\rd_r \hp|^2 + |\rd_{v_-}\hp|^2 + \sum_i|\tilde{Z}_{i,-} \hp|^2 + |\hp|^2\big) \, \vols dv_-dr
\leq C \;,
\end{split}
\end{equation*}
 for some $C>0$, where, in a second step we have also taken the limit $r' \to \rred$ to obtain the bulk term.
Together with the regularity Assumption \ref{AssumptionReg} used for the remaining compact spacetime region this shows \eqref{EqPropLeftH} with $f = \hp$.

\underline{\bf{Step 6: Estimating higher derivatives:}}

Recall that $\partial_{v_-}$, $\rd_{\varphi_-}$, and $\mathcal{Q}_{[s]}$ commute with $\hat{\mathcal{T}}_{[s]}$. By our assumptions on the left event horizon \eqref{AssumpLHpWeaker} we can thus repeat the above argument now with $\hp$ replaced by $\mathcal{Q}_{[s]}^d \partial_{v_-}^a \rd_{\varphi_-}^b  \hp$  to obtain \eqref{EqPropLeftH} for $f = \mathcal{Q}_{[s]}^d \partial_{v_-}^a \rd_{\varphi_-}^b \hp$, for $0 \leq a + b \leq 2$, $d = 0,1$.

We now commute the Teukolsky equation with $\partial_r$:
\begin{equation}\label{EqDiffTeukMinus}
\begin{split}
0 = \rd_r  \hat{\mathcal{T}}_{[s]} \hp &= a^2 \sin^2 \theta \rd_{v_-}^2 \rd_r \hp - 2a \rd_{v_-} \rd_{\varphi_-} \rd_r \hp + 2(r^2 + a^2) \rd_{v_-}\rd_r^2 \hp \\
&-2a \rd_{\varphi_-} \rd_r^2 \hp + \Delta \rd_r^3 \hp + 2\big(r(3+2s) +isa \cos\theta\big) \rd_{v_-}\rd_r \hp +  \dashuline{2(r-M)(2+s) \rd_r^2 \hp }\\
&+ \swl \rd_r \hp + \underline{2(1+2s)\rd_{v_-} \hp} + 2(1+s) \rd_r \hp \;.
\end{split}
\end{equation}
Of course, the principal part is unchanged. Also note that the dashed red-shift term is even \emph{improved}. Thus, the same vector field multiplier (with $\hp$ replaced by $\rd_r \hp$) can be used to control all bulk terms quadratic in derivatives of $\rd_r \hp$. 
Also the same modification, i.e., the underbraced term in \eqref{EqMultRedLeft} with $\hp$ replaced by $\rd_r \hp$, can be used to generate an arbitrarily large bulk term quadratic in $\rd_r \hp$ of the `good' negative sign.
The boundary terms are exactly of the same form with $\hp$ replaced by $\rd_r \hp$. So the only qualitatively new term we need to estimate is the underlined term, which is neither $\rd_r \hp$ nor derivatives of it. This term can be estimated either by a modification similarly to the one we used above, now with $\hp$ replaced by $\partial_{v_-} \hp$, or, more straightforwardly, we can use  directly the bulk term in \eqref{EqPropLeftH}. Thus, we obtain, after possibly choosing $\rred$ closer to $r_+$
\begin{equation*}
\begin{split}
\sup_{r' \in [\rred, r_+]}&\int\limits_{\{r=r'\} \cap \{f^- \geq 2v_0\}} v_-^{q_l}\big(|\Delta||\rd^2_r \hp|^2 + |\rd_{v_-} \rd_r\hp|^2 + \sum_i|\tilde{Z}_{i,-} \rd_r\hp|^2 + |\rd_r\hp|^2\big) \, \vols dv_- \\
&+ \int\limits_{\mathclap{\{2v_0 \leq f^- \}  \cap \{\rred \leq r \leq r_+\}}} v_-^{q_l}\big(|\rd_r^2 \hp|^2 + |\rd_{v_-}\rd_r\hp|^2 + \sum_i|\tilde{Z}_{i,-} \rd_r\hp|^2 + |\rd_r\hp|^2\big) \, \vols dv_-dr
\leq C \;,
\end{split}
\end{equation*}
for some $C>0$, which is \eqref{EqPropLeftH} with $f = \rd_r \hp$. Again, we can in addition commute with the Killing vector fields $\rd_{v_-}$, $\rd_{\varphi_-}$, as well as with $\mathcal{Q}_{[s]}$.

Differentiating \eqref{EqDiffTeukMinus} once more in $r$ we obtain
\begin{equation*}
\begin{split}
0 = \rd_r^2  \hat{\mathcal{T}}_{[s]} \hp &= a^2 \sin^2 \theta \rd_{v_-}^2 \rd_r^2 \hp - 2a \rd_{v_-} \rd_{\varphi_-} \rd_r^2 \hp + 2(r^2 + a^2) \rd_{v_-}\rd_r^3 \hp \\
&-2a \rd_{\varphi_-} \rd_r^3 \hp + \Delta \rd_r^4 \hp + 2\big(r(5+2s) +isa \cos\theta\big) \rd_{v_-}\rd_r^2 \hp +  \dashuline{2(r-M)(3+s) \rd_r^3 \hp }\\
&+ \swl \rd_r^2 \hp + 8(1+s)\rd_{v_-} \rd_r\hp + 2(3+2s) \rd_r^2 \hp \;.
\end{split}
\end{equation*}
The dashed red-shift term is even further improved and no qualitatively new terms compared to \eqref{EqDiffTeukMinus} (with $\hp$ replaced by $\rd_r \hp$) have appeared.
This completes the proof.
\end{proof}

We now continue with the red shift estimate near the right event horizon. 

\begin{proposition}\label{PropRightHP}
Under the assumptions of Section \ref{SecAssumptions} there exists an $\rred \in (r_-, r_+)$ and a $C>0$ such that
\begin{equation}
\label{EqPropRightH}
\begin{split}
\sup_{r' \in [r_\mathrm{red}, r_+]} &\sum_{0 \leq i_1 + i_2 + i_3 + j + k \leq 1} \int\limits_{\{r=r'\} \cap \{v_+ \geq 1\}} v_+^{q_r}|\Delta|^k |\widetilde Z_{1,+}^{i_1}\widetilde Z_{2,+}^{i_2}\widetilde Z_{3,+}^{i_3}\rd_{v_+}^{j}\rd_r^{k}f|^2 \,\vols\,dv_{+}  \\
&+\sum_{0 \leq i_1 + i_2 + i_3 + j + k \leq 1} \int\limits_{\{\rred \leq r \leq r_+\} \cap \{v_+ \geq 1\}} v_+^{q_r}|\widetilde Z_{1,+}^{i_1}\widetilde Z_{2,+}^{i_2}\widetilde Z_{3,+}^{i_3}\rd_{v_+}^{j}\rd_r^{k}f|^2 \,\vols\,dv_{+} dr
\leq C
\end{split}
\end{equation}
holds for $f \in \{\rd_{v_+}^a\rd_{\varphi_+}^b\rd_r^c\psi, \rd_{v_+} \rd_{v_+}^a\rd_{\varphi_+}^b\rd_r^c\psi, \mathcal{Q}_{[s]} \rd_{v_+}^a\rd_{\varphi_+}^b\rd_r^c \psi, \}$, with $0 \leq a + b +c \leq 2$.
\end{proposition}

The symmetry between left and right event horizon for the wave equation is broken for the Teukolsky equation because of a choice of frame field. Indeed, near the right event horizon $\Hp_r$ we do have a \emph{blue-shift} for the energy of the Teukolsky field $\psi$. This is the reason why in the following proof we need to commute twice with $\rd_r$ in order to get a red-shift near $\Hp_r$. 

\begin{proof}
Many elements of the proof are the same as those of the proof of Proposition \ref{PropLeftHp}. For this reason we will be more concise here and highlight the essential differences. We begin by observing that the crucial seventh term of \eqref{TeukolskyStar} has a `bad' sign for $s=2$. Multiplying by $-v_+^{q_r} (1 + \lambda \Delta) \overline{\rd_r \psi}$ as we did before would give a bulk term in $|\rd_r \psi|^2$ of positive sign -- but we recall that for stability we needed the good negative sign. We differentiate \eqref{TeukolskyStar} in $r$ to obtain
\begin{equation*}
\begin{split}
0 = \partial_r \mathcal{T}_{[s]} \psi &= a^2 \sin^2 \theta \rd_{v_+}^2 \rd_r \psi + 2a \rd_{v_+} \rd_{\varphi_+} \rd_r \psi + 2(r^2 + a^2) \rd_{v_+} \rd_r^2 \psi + 2a\rd_{\varphi_+} \rd_r^2 \psi + \Delta \rd_r^3 \psi\\
&\quad  + \swl \rd_r \psi +2\big(r(3-2s) - isa \cos \theta\big) \rd_{v_+} \rd_r \psi + 2(r-M)(2-s) \rd_r^2 \psi \\
&\quad +2(1-2s) \rd_r \psi +2(1-2s)\rd_{v_+} \psi \;.
\end{split}
\end{equation*}
Differentiating once more we obtain
\begin{equation}\label{EqTwiceDiffTeukP}
\begin{split}
0 = \rd_r^2 \mathcal{T}_{[s]} \psi &= a^2 \sin^2 \theta \rd_{v_+}^2 \rd_r^2 \psi + 2a \rd_{v_+} \rd_{\varphi_+} \rd_r^2 \psi + 2(r^2 + a^2) \rd_{v_+} \rd_r^3 \psi + 2a\rd_{\varphi_+} \rd_r^3 \psi + \Delta \rd_r^4 \psi \\
&\quad + \swl \rd_r^2 \psi +2\big(r(5-2s) - isa \cos \theta\big) \rd_{v_+} \rd_r^2 \psi + 2(r-M)(3-s) \rd_r^3 \psi  \\
&\quad +6(1-s) \rd_r^2 \psi +8(1-s)\rd_{v_+}\rd_r \psi \;.
\end{split}
\end{equation}
\underline{\bf{Step 1: The multiplier:}} We restrict in the following to $v_+ \geq 1$. We consider the following multiplier identity, where $\lambda, \eta, \mu >0$ are constants to be chosen:
\begin{equation}\label{EqPfPropRHPMultiplier}
\begin{split}
0 &= \Rea\Big(\rd_r^2 \mathcal{T}_{[s]} \psi \cdot v_+^{q_r}\big(-(1 + \lambda \Delta)\rd_r +(1 + \lambda \Delta)\rd_{v_+}\big) \overline{\rd_r^2\psi}\Big) \\
&\qquad + \underbrace{\rd_r(v_+^{q_r}\mu e^{\eta r}|\rd_r^2 \psi|^2) - v_+^{q_r} \mu \eta e^{\eta r} |\rd_r^2 \psi|^2 - 2v_+^{q_r} \mu e^{\eta r} \Rea(\overline{\rd_r^2 \psi}\rd_r^3 \psi)}_{=0} \\
&\qquad + \underbrace{\rd_r(v_+^{q_r}\mu e^{\eta r}|\rd_{v_+}\rd_r \psi|^2) - v_+^{q_r} \mu \eta e^{\eta r} |\rd_{v_+}\rd_r \psi|^2 - 2v_+^{q_r} \mu e^{\eta r} \Rea(\overline{\rd_{v_+}\rd_r \psi}\rd_{v_+}\rd_r^2 \psi)}_{=0}
\end{split}
\end{equation}
After integration over the spheres, the right hand side of \eqref{EqPfPropRHPMultiplier} is the sum of
\begin{enumerate}
\item the sum of all the terms on the right hand sides of \ref{ApR} and \ref{ApV} with $\chi(v_+) = v_+^{q_r}$ and $\psi$ replaced by $\rd_r^2 \psi$
\item the real part of the terms
\begin{equation*}
\begin{split}
&2\big(r(5-2s) - isa \cos \theta\big) \rd_{v_+} \rd_r^2 \psi \cdot v_+^{q_r}\big(-(1 + \lambda \Delta)\rd_r +(1 + \lambda \Delta)\rd_{v_+}\big) \overline{\rd_r^2\psi} \\[4pt]
&\quad \dashuline{-v_+^{q_r}(1 + \lambda \Delta) \cdot 2(r-M)(3-s) |\rd_r^3 \psi|^2} + v_+^{q_r}(1 + \lambda \Delta)2(r-M)(3-s)\rd_r^3\psi \overline{\rd_{v_+} \rd_r^2 \psi} \\[4pt]
&\quad +\big[6(1-s)\rd_r^2 \psi + \uuline{8(1-s)\rd_{v_+} \rd_r \psi}\big] \cdot v_+^{q_r}\big(-(1 + \lambda \Delta)\rd_r +(1 + \lambda \Delta)\rd_{v_+}\big) \overline{\rd_r^2\psi}  
\end{split}
\end{equation*}
\item the underbraced terms in \eqref{EqPfPropRHPMultiplier}.
\end{enumerate}
The second underbraced term in \eqref{EqPfPropRHPMultiplier} has been added to control the double underlined term above. As before, we can derive a boundedness statement if the bulk terms are negative. We proceed as before:

\underline{\bf{Step 2: Estimating all bulk terms that are quadratic in derivatives of $\rd_r^2 \psi$.}}

As before the two dashed terms combine to give a negative definite contribution in $|\rd_r^3 \psi|^2$ for $r$ close enough to $r_+$. Next, we look at the wavily underlined terms which are all those that are leading order in $\lambda$. Again, the first of those terms in \ref{ApR} has a good negative sign, the second can be controlled by the third one in \ref{ApR} and by the one in \ref{ApV} (as in \eqref{EqEstLeadTermHpl}) for $r$ close enough to $r_+$ so that the wavily underlined terms and the dashed terms combined can be estimated from above by
\begin{equation}\label{EqPfRedShiftRight}
-v_+^{q_r} f(r, \lambda) \big(\lambda (\sum_i|\zt_{i,+}\rd_r^2\psi|^2 + |\rd_{v_+} \rd_r^2 \psi|^2) + |\rd_r^3 \psi|^2\big) \;,
\end{equation}
where $f(r_+, \lambda)>0$ is independent of $\lambda$. Choosing now $\lambda >0$ and $v_0\geq 1$ large enough and $\rred < r_+$ close enough to $r_+$, all other bulk terms that are quadratic in derivatives of $\rd_r^2 \psi$ can be controlled in absolute value by $-\frac{1}{2} \times$\eqref{EqPfRedShiftRight} in the region $\{ \rred \leq r \leq r_+\} \cap \{v_+ \geq v_0\}$.

\underline{\bf{Step 3: Estimating boundary terms.}}

We gather all the total derivatives $\rd_{v_+} \big(A(v_+,r,\rd_r^2\psi, \partial \rd_r^2 \psi)\big)$ and $\rd_r\big(B(v_+,r,\rd_{v_+}\rd_r \psi, \rd_r^2 \psi, \rd \rd_r^2 \psi)\big)$ appearing on the right hand side of \eqref{EqPfPropRHPMultiplier}, where we find
\begin{equation*}
\begin{split}
A(v_+,r,\rd_r^2 \psi, \rd \rd_r^2 \psi) &= v_+^{q_r} ( 1 + \lambda \Delta)\Big( -a^2 \sin^2 \theta  \Rea (\rd_{v_+} \rd_r^2 \psi \overline{ \rd_r^3 \psi}) - 2a \Rea(\rd_{\varphi_+}\rd_r^2 \psi \overline{\rd_r^3 \psi}) -(r^2 + a^2) |\rd_r^3 \psi|^2 \\ 
&+ \frac{1}{2}  a^2 \sin^2 \theta |\rd_{v_+} \rd_r^2 \psi|^2 -\frac{1}{2}\Delta |\rd_r^3 \psi|^2 + \frac{1}{2} (s + s^2)|\rd_r^2 \psi|^2 -\frac{1}{2} \sum_i |\zt_{i,+} \rd_r^2 \psi|^2 \Big)
\end{split}
\end{equation*}
and
\begin{equation*}
\begin{split}
B(v_+,r,\rd_{v_+}\rd_r \psi, \rd_r^2 \psi, \rd \rd_r^2 \psi) &= v_+^{q_r}(1 + \lambda \Delta) \Big( \frac{1}{2} a^2 \sin^2\theta |\rd_{v_+}\rd_r^2\psi|^2 + 2a\Rea(\rd_{\varphi_+} \rd_r^2 \psi \overline{\rd_{v_+} \rd_r^2 \psi}) - \frac{1}{2} \Delta |\rd_r^3 \psi|^2 \\
&-\frac{1}{2}(s + s^2)|\rd_r^2\psi|^2 + \frac{1}{2} \sum_i |\zt_{i,+}\rd_r^2 \psi|^2 + (r^2 + a^2)|\rd_{v_+} \rd_r^2 \psi|^2 \\
&+\Delta \Rea(\rd_r^3 \psi \overline{\rd_{v_+}\rd_r^2 \psi})\Big) +v_+^{q_r}\mu e^{\eta r} |\rd_r^2 \psi|^2 + v_+^{q_r} \mu e^{\eta r}|\rd_{v_+} \rd_r \psi|^2 \;.
\end{split}
\end{equation*}
The coercivity of $B$ and $B - A$ for $r$ close enough to $r_+$ is established in the same way as in the proof of Proposition \ref{PropLeftHp}, Step 3, by choosing $\mu(\eta)$ such that $\mu(\eta) e^{\eta r_+} = 2(s + s^2)$ to obtain
\begin{equation}\label{EqLowBoundAR}
B(v_+,r,\rd_{v_+} \rd_r \psi, \rd_r^2 \psi, \rd \rd_r^2 \psi) \gtrsim v_+^{q_r}(|\Delta| |\rd_r^3 \psi|^2 + |\rd_{v_+} \rd_r^2 \psi|^2 + \sum_i |\zt_{i,+}\rd_r^2\psi|^2 + |\rd_r^2 \psi|^2 + |\rd_{v_+} \rd_r \psi|^2)
\end{equation}
and 
\begin{equation}\label{EqLowBoundABR}
(B - A)(v_+,r,\rd_{v_+} \rd_r \psi, \rd_r^2 \psi, \rd \rd_r^2 \psi) \gtrsim v_+^{q_r} (|\rd_r^3 \psi|^2 + |\rd_{v_+} \rd_r^2 \psi|^2 + \sum_i |\zt_{i,+} \rd_r^2 \psi|^2  +|\rd_r^2\psi|^2 + |\rd_{v_+} \rd_r \psi|^2)
\end{equation}
for $r \in [\rred, r_+]$ with $\rred$ close enough to $r_+$.

\underline{\bf{Step 4: Estimating the remaining bulk terms.}}

The last two terms of each underbraced term in \eqref{EqPfPropRHPMultiplier} are estimated as in \eqref{EqEstAuxTerms} of Step 4 of the proof of Proposition \ref{PropLeftHp}, where we again choose $\eta>0$ so large that the resulting terms quadratic in derivatives of $\rd_r^2 \psi$ are absorbed by $\frac{1}{4}\times$ \eqref{EqPfRedShiftRight} and such that the terms $-\frac{1}{2} v_+^{q_r} \eta \mu e^{\eta r} (|\rd_r^2 \psi|^2 + |\rd_{v_+} \rd_r \psi|^2)$ control all the remaining bulk terms.

\underline{\bf{Step 5: Putting it all together.}}

We obtain from \eqref{EqPfPropRHPMultiplier} after integration over the spheres
\begin{equation*}
\begin{split}
\rd_{v_+} \big(A(v_+,r,\rd_r^2\psi,\rd \rd_r^2 \psi)\big) &+ \rd_r\big(B(v_+,r,\rd_{v_+}\rd_r \psi , \rd_r^2 \psi, \rd \rd_r^2 \psi)\big) \\
&\underset{{\mathrm{a.i.}}}{\gtrsim} v_+^{q_r}\big(|\rd_r^3 \psi|^2 + |\rd_{v_+}\rd_r^2\psi|^2 + \sum_i|\tilde{Z}_{i,+} \rd_r^2\psi|^2 + |\rd_r^2 \psi|^2 + |\rd_{v_+} \rd_r \psi|^2\big)
\end{split}
\end{equation*}
for $v_+ \geq v_0$ and $\rred \leq r \leq r_+$. Let now $r' \in [\rred , r_+)$. We integrate over the region $\{2v_0 \leq f^+ \leq v_1\} \cap \{r' \leq r \leq r_+\}$ with respect to $dv_+ \wedge dr \wedge \vols = \frac{1}{\rho^2} \vol$ and use that on a level set of $f^+$ we have $dr = dv_+$ to obtain
\begin{equation}\label{EqEnEstRight}
\begin{split}
\int\limits_{\mathclap{\substack{\{r=r'\} \\ \cap \{2v_0 \leq f^+ \leq v_1\}}}} B \, \vols dv_+ 
&+ \int\limits_{\mathclap{\substack{\{f^+=v_1\} \\ \cap \{r' \leq r \leq r_+\}}}} (B - A) \, \vols dv_+ \\
&\quad + c \int\limits_{\mathclap{\substack{\{2v_0 \leq f^+ \leq v_1\} \\ \cap \{r' \leq r \leq r_+\}}}} v_+^{q_r}\big(|\rd_r^3 \psi|^2 + |\rd_{v_+}\rd_r^2\psi|^2 + \sum_i|\tilde{Z}_{i,+} \rd_r^2\psi|^2 + |\rd_r^2 \psi|^2 + |\rd_{v_+} \rd_r \psi|^2\big)\, \vols dv_+dr \\
&\leq \int\limits_{\mathclap{\Hp_r \cap \{2v_0 \leq f^+ \leq v_1\}}} B \, \vols dv_+  + 
\int\limits_{\mathclap{\substack{\{f^+=2v_0\} \\ \cap \{r' \leq r \leq r_+\}}}} (B - A) \, \vols dv_+  \;,
\end{split}
\end{equation}
where $c>0$. Using the lower bounds \eqref{EqLowBoundAR} and \eqref{EqLowBoundABR}, the trivial upper bounds on $A$ and $B$, the Assumptions \ref{AssumptionReg} and \eqref{AssumpDecayRHp} on $\psi$ to control the boundary terms on the right hand side, we obtain from this
\begin{equation}\label{Eq3rDerPsi}
\begin{split}
&\sup_{r' \in [\rred, r_+]}\int\limits_{\{r=r'\} \cap \{f^+ \geq 2v_0\}} v_+^{q_r}\big(|\Delta||\rd_r^3 \psi|^2 + |\rd_{v_+}\rd_r^2\psi|^2 + \sum_i|\tilde{Z}_{i,+} \rd_r^2 \psi|^2 + |\rd_r^2 \psi|^2 + |\rd_{v_+}\rd_r \psi|^2\big) \, \vols dv_+ \\
&\qquad +\int\limits_{\{\rred \leq r \leq r_+\} \cap \{f_+ \geq 2v_0\}} v_+^{q_r}\big(|\rd_r^3 \psi|^2 + |\rd_{v_+}\rd_r^2\psi|^2 + \sum_i|\tilde{Z}_{i,+} \rd_r^2\psi|^2 + |\rd_r^2 \psi|^2 + |\rd_{v_+} \rd_r \psi|^2\big)\, \vols dv_+dr 
\leq C 
\end{split}
\end{equation}
for some $C>0$. Together with Assumption \ref{AssumptionReg} this in particular gives \eqref{EqPropRightH} with $f = \rd_r^2 \psi$.

\underline{\bf{Step 6: Estimating higher, lower, and other derivatives.}}

We can again just commute \eqref{EqTwiceDiffTeukP} with $\rd_{v_+}$, $\rd_{\varphi_+}$, and $\mathcal{Q}_{[s]}$ to obtain \eqref{EqPropRightH} also for $f  \in \{ \rd_{v_+}^a \rd_{\varphi_+}^b \rd_r^2\psi, $ $ \rd_{v_+} \rd_{v_+}^a \rd_{\varphi_+}^b\rd_r^2 \psi, \mathcal{Q}_{[s]} \rd_{v_+}^a \rd_{\varphi_+}^b \rd_r^2 \psi\}$ for $0 \leq a+b \leq 2$.\footnote{This gives control over some higher derivatives which are not stated in Proposition \ref{PropRightHP} and which are not needed. We are wasteful here with derivatives in order to streamline the presentation. Being a bit more careful one can safe a couple of derivatives here.}
The lower order terms are now estimated by integrating  in $r$ using the fundamental theorem of calculus together with Minkowski's inequality\footnote{Concretely, we use $$\Big(\int\limits_{\{r=r'\} \cap \\ \{v_+ \geq 1\}} h^2 \vols dv_+ \Big)^\frac{1}{2} \leq \Big(\int\limits_{\{r=r_+\} \cap \\ \{v_+ \geq 1\}} h^2 \vols dv_+ \Big)^\frac{1}{2} + \int_{r'}^{r_+} \Big(\int\limits_{\{r=\tilde{r}\} \cap \{v_+ \geq 1\}} (\rd_rh)^2 \vols dv_+ \Big)^\frac{1}{2} dr$$ for $h \in \{ \zt^{i_1}_{1,+} \zt^{i_2}_{2,+} \zt^{i_3}_{3,+} \rd_{v_+}^j\rd_r^k (\rd_{v_+}^d \mathcal{Q}_{[s]}^e \rd_{v_+}^a \rd_{\varphi_+}^b \rd_r \psi)\}$, $0 \leq i_1 + i_2 + i_3 + j +k \leq 1$, and $0 \leq a + b \leq 2$, $0 \leq d + e \leq 1$.} and using the assumptions \eqref{AssumpDecayRHp} on the right event horizon to obtain
\begin{equation*}
\sup_{r' \in [\rred, r_+]}\int\limits_{\{r=r'\} \cap \{v_+ \geq 1\}} v_+^{q_r}\big(|\rd_r^2 f|^2 + |\rd_{v_+}\rd_rf|^2 + \sum_i|\tilde{Z}_{i,+} \rd_r f|^2 + |\rd_r f|^2 \big) \, \vols dv_+ \leq C 
\end{equation*}
 for some $C>0$ and $f  \in \{ \rd_{v_+}^a \rd_{\varphi_+}^b \rd_r\psi, $ $ \rd_{v_+} \rd_{v_+}^a \rd_{\varphi_+}^b\rd_r \psi, \mathcal{Q}_{[s]} \rd_{v_+}^a \rd_{\varphi_+}^b \rd_r \psi\}$ for $0 \leq a+b \leq 2$.
Integrating once more  in this way concludes the proof of Proposition \ref{PropRightHP}.
\end{proof}

The following remark about the altered red-shift effect for the Teukolsky equation and Gaussian beams is not needed for the result of the paper but the reader might still find it instructive.
\begin{remark}\label{RemRedShift}
The red-shift effect along the event horizon for the scalar wave equation is by now a classic effect which has been used in various guises. To understand how it changes for the Teukolsky equation it is helpful to differentiate between the following three manifestations of the red-shift effect (one could easily consider more). We consider a family of observers with timelike velocity vector fields $N$ which are Lie-transported along the Hawking Killing vector field $T_{\Hp} = \rd_{v_+} + \frac{a}{r_+^2 + a^2} \rd_{\varphi_+}$ along the event horizon. 
\begin{enumerate}
\item The frequency, as measured by the family of observers, of a (Gaussian) beam propagating along the event horizon is shifted exponentially to the red. This could be seen as the original red-shift effect.
\item The energy of a (Gaussian) beam propagating along the event horizon decays exponentially, see \cite{Sbie13b}, \cite{Sbie14}. Note that this is a priori independent of the change of colour of the light, but it is the most relevant manifestation of the so-called red-shift effect on energy estimates.
\item Consider compactly supported initial data along the event horizon. Then the transversal derivative decays exponentially along the event horizon, see \eqref{TeukolskyStar}.
\end{enumerate}
For the Teukolsky equation, as we will see, it no longer makes sense to refer to those three effects collectively as the `red-shift effect'. Dividing \eqref{TeukolskyStar} by $\rho^2$ we obtain that the Teukolsky equation in $(v_+,r, \theta, \varphi_+)$ coordinates is of the form $\Box_g \psi + X\psi + f \psi =0$ with $$X= - \frac{1}{\rho^2}(4sr + 2isa \cos \theta) \rd_{v_+} - \frac{2s(r-M)}{\rho^2} \rd_r + \frac{2si}{\rho^2} \frac{\cos \theta}{\sin^2 \theta} \rd_{\varphi_+} \;, \quad f = - \frac{2s}{\rho^2} - \frac{1}{\rho^2}(s^2 \frac{\cos^2 \theta}{\sin^2 \theta} - s) \;.$$

Note that in the construction of Gaussian beams for wave equations with lower order terms, the lower order terms only impinge on the amplitude, but not on the phase function, see Appendix 3.D of \cite{Sbie14}. Thus the frequency/colour is still shifted to the red for all values of $s$.

We now consider the behaviour of the energy of Gaussian beams for which we refer the reader to Appendix 3.D of \cite{Sbie14}. It follows from $\nabla_{T_{\Hp}}T_{\Hp} = \kappa_+ T_{\Hp}$ that $e^{-\kappa_+ v_+} T_{\Hp}$ is a null geodesic velocity vector field along the event horizon. The $N$-energy of a Gaussian beam for the wave equation localised along one of the integral curves thus behaves like $e^{-\kappa_+v_+}$. Let us now choose either the integral curve at $\theta = 0$ or $\theta = \pi$ so that $g(X, e^{-\kappa_+ v_+} T_{\Hp_+}) = - e^{-\kappa_+ v_+} 2s \kappa_+$. With the terminology from \cite{Sbie14} we hence obtain the modulating factor $|m_X(v_+)|^2 = e^{2s \kappa_+ v_+}$    of the amplitude of the Gaussian beam for the Teukolsky equation compared to that for the wave equation. Hence, the $N$-energy of such a Gaussian beam for the Teukolsky equation behaves like $e^{(2s-1)\kappa_+ v_+}$.

In Appendix 3.E of \cite{Sbie14} it was obtained that an integrated local energy decay statement for the Teukolsky equation cannot hold in the exterior of a Kerr black hole without the `loss of a derivative' by considering Gaussian beams localised along trapped null geodesics away from the horizon. By considering a Gaussian beam along the event horizon as above it follows that not even a uniform energy boundedness statement for the Teukolsky equation for $s=+1,+2$ can hold without the `loss of a derivative'.

Finally, for compactly supported initial data along the right event horizon it directly follows from \eqref{TeukolskyStar} that for $s=+1$ the transversal derivative remains constant for large $v_+$ while for $s=+2$ it grows in general exponentially  (i.e., if it does not vanish). This shows very nicely how these three different effects decouple for the Teukolsky equation.
\end{remark}

\subsubsection{Corollaries}

Let $\chi : \R \to (0, \infty)$ be a fixed positive smooth function with $\chi(v_+) = v_+^{q_r}$ for $v_+ \geq 1$ and $\chi(v_+) = |v_+|^{q_l}$ for $v_+ \leq -1$. The next corollary will be our starting point for the estimates in the next section which are needed for the separation of the Teukolsky field. It combines the results of Proposition \ref{PropLeftHp} and \ref{PropRightHP}, but we can afford to discard uniformity up to the event horizons. 
\begin{corollary}\label{CorFinEnEstHorizons}
Under the assumptions of Section \ref{SecAssumptions} there exists an $\rred \in (r_-, r_+)$ such that for any $r_1 \in (r_\mathrm{red}, r_+)$ there exists a $C>0$ such that
\begin{equation}\label{CorConclusionForNoShift2}
\sup_{r' \in [r_\mathrm{red}, r_1]} \sum_{0 \leq i_1 + i_2 + i_3 + j + k \leq 1} \int_{\{r=r'\}} \chi(v_+) |\widetilde Z_{1,+}^{i_1}\widetilde Z_{2,+}^{i_2}\widetilde Z_{3,+}^{i_3}\rd_{v_+}^{j}\big(\rd_r|_+\big)^{k}f|^2 \,\vols\,dv_{+}  \leq C
\end{equation}
holds for $f \in \{ \rd_{v_+}^a \rd_{\varphi_+}^b(\rd_r|_+)^c \psi, \rd_{v_+}\rd_{v_+}^a \rd_{\varphi_+}^b(\rd_r|_+)^c \psi,  \mathcal{Q}_{[s]} \rd_{v_+}^a \rd_{\varphi_+}^b(\rd_r|_+)^c \psi\}$, $0 \leq a + b + c \leq 2$. 
\end{corollary}
Here we have employed the notation $\rd_r|_+$ to emphasise that this is a partial derivative in $r$ with respect to the $(v_+, r, \theta, \varphi_+)$-coordinate system. 

\begin{proof}
It follows from \eqref{CoordTrafoStar} and \eqref{CoordTrafoStarMinus} that we have $\rd_r|_+ = 2 \frac{r^2 + a^2}{\Delta} \rd_{v_-} - 2 \frac{a}{\Delta} \rd_{\varphi_-} + \rd_r|_-$. Recalling moreover that $\psi = \Delta^2 \hp$, we obtain
\begin{equation}\label{EqRPsiStillWeight}
\rd_r|_+ \psi = 2 \Delta (r^2 + a^2) \rd_{v_-} \hp - 2a \Delta \rd_{\varphi_-} \hp + \Delta^2 \rd_r|_- \hp + 2(\rd_r \Delta)\Delta \hp
\end{equation}
and
\begin{equation}\label{EqRPsiNotWeight}
\begin{split}
\big(\rd_r|_+\big)^2 \psi  &=4(r^2 + a^2)^2 \rd_{v_-}^2 \hp - 8(r^2 + a^2) a \rd_{v_-} \rd_{\varphi_-} \hp + 4 \Delta(r^2 + a^2)\rd_{v_-} \rd_r|_- \hp + 4(r^2 + a^2)(\rd_r \Delta) \rd_{v_-} \hp \\
&\qquad + 4a^2 \rd_{\varphi_-}^2 \hp - 4a \Delta \rd_{\varphi_-} \rd_r|_- \hp - 4a (\rd_r \Delta) \rd_{\varphi_-} \hp + 2\rd_r(\Delta(r^2 + a^2)) \rd_{v_-} \hp \\
&\qquad - 2a(\rd_r\Delta) \rd_{\varphi_-} \hp + \Delta^2 \big(\rd_r|_-\big)^2 \hp + 4\Delta(\rd_r \Delta)\rd_r|_- \hp + 2\rd_r((\rd_r \Delta) \Delta) \hp 
\end{split}
\end{equation}
and $(\rd_r|_+)^3 \psi$ is a linear combination of $\rd_{v_-}^a \rd_{\varphi_-}^b (\rd_r|_-)^c \hp$ with $0 \leq a + b + c \leq 3$, $a,b,c, \in \N_0$.
Moreover, using $\varphi_+ = \varphi_- + 2\overline{r}$ we directly compute
\begin{equation}
\label{EqZTildeTrafo}
\begin{aligned}
\zt_{1,+} & =  \cos (2\overline{r}) \cdot \zt_{1,-} + \sin(2\overline{r}) \cdot \zt_{2,-} \\
\zt_{2,+} & = \cos (2 \overline{r}) \cdot \zt_{2,-} - \sin(2\overline{r}) \cdot \zt_{1,-} \\
\zt_{3,+} & = \zt_{3,-} \;.
\end{aligned}
\end{equation}
We also observe $\rd_{v_+} = - \rd_{v_-}$. Moreover, it follows from $v_+ = - v_- + 2r^*$ that for $r' \in [r_\mathrm{red}, r_+)$ we have $|v_+| \leq C(r') |v_-|$ for $v_+ \leq -C(r')$ with the constant $C(r')$ blowing up for $r' \to r_+$. Now \eqref{CorConclusionForNoShift2} follows directly from the Propositions \ref{PropLeftHp} and \ref{PropRightHP} and the regularity Assumption \ref{AssumptionReg}.
\end{proof}

\begin{remark}
The constant on the right hand side of \eqref{CorConclusionForNoShift2} will in general blow up for $r_1 \to r_+$, because of the conversion of the $v_-$-weights from Proposition \ref{PropLeftHp} into $v_+$-weights. However, for $f = \psi, \rd_r \psi$, we do have exponential decay in $v_+$ for $v_+ \to -\infty$ approaching $\Hp_l$ by the regularity Assumption \ref{AssumptionReg}, which compensates for the blow up of the constant in the conversion and \eqref{CorConclusionForNoShift2} can actually be shown to hold uniformly up to $r_+$. Since $f = (\rd_r|_+)^2 \psi$ is in general regular and non-vanishing near the bottom bifurcation sphere $\Sp^2_b$ we do no longer have decay for $v_+ \to - \infty$ approaching $\Hp_l$ and so the constant blows up for $r_1 \to r_+$.
\end{remark}

The next corollary is needed in Section \ref{SecDetA} for passing to the limit $r \to r_+$ in the separated picture, in particular for Proposition \ref{PropAHP}   and Proposition \ref{PropLimitDis}.
\begin{corollary}\label{CorLimitHP}
Under the assumptions in Section \ref{SecAssumptions} we have for $r_+ > r \to r_+$
\begin{equation}\label{EqCorL2Limit}
\psi(v_+, r, \theta, \varphi_+) \to \psi(v_+, r_+, \theta, \varphi_+) \qquad \textnormal{ in } L^2_{v_+} L^2_{\Sp^2}
\end{equation}
and 
\begin{equation}\label{EqCorL2LocLimit}
\mathbbm{1}_{(v_0, \infty)}(v_-) \cdot (\rd_r|_+)^2 \psi (v_-, r, \theta, \varphi_-) \to \mathbbm{1}_{(v_0, \infty)}(v_-) \cdot (\rd_r|_+)^2 \psi (v_-, r_+, \theta, \varphi_-) \qquad \textnormal{ in } L^2_{v_-} L^2_{\Sp^2}
\end{equation}
for any $v_0 \in \R$.
\end{corollary}

\begin{proof}
We begin with proving \eqref{EqCorL2Limit}. The fundamental theorem of calculus gives $|\psi(v_+, r, \theta, \varphi_+) - \psi(v_+, r_+, \theta, \varphi_+)| \leq \int_{[r,r_+]} |\rd_r \psi(v_+, r', \theta, \varphi_+) |\, dr'$. Cauchy Schwarz yields $$|\psi(v_+, r, \theta, \varphi_+) - \psi(v_+, r_+, \theta, \varphi_+)|^2 \leq \int_{[r,r_+]} |\rd_r \psi(v_+, r', \theta, \varphi_+)|^2 \, dr' \cdot |r- r_+|$$ which thus gives
$$\int_{\R \times \Sp^2} |\psi(v_+, r, \theta, \varphi_+) - \psi(v_+, r_+, \theta, \varphi_+)|^2 \vols dv_+ \leq \int_{[r, r_+]} \int_{\R \times \Sp^2} |\rd_r \psi(v_+, r', \theta, \varphi_+)|^2 \, dr' \vols dv_+ \cdot |r - r_+| \;.$$
It follows from \eqref{EqRPsiStillWeight} together with \eqref{EqPropLeftH}, from the bulk term in \eqref{EqPropRightH}, as well as from the regularity Assumption \ref{AssumptionReg} that the spacetime integral is uniformly bounded. 
This shows \eqref{EqCorL2Limit}.

To prove \eqref{EqCorL2LocLimit} we compute in an analogous manner as before
\begin{equation}\label{EqIEstLoc}
\begin{split}
\int_{\R \times \Sp^2} \mathbbm{1}_{(v_0, \infty)}(v_-) &\Big|(\rd_r|_+)^2 \psi(v_-, r, \theta, \varphi_-) - (\rd_r|_+)^2\psi(v_-, r_+, \theta, \varphi_-)\Big|^2 \vols dv_- \\
&\leq \int_{[r, r_+]} \int_{\R \times \Sp^2}\mathbbm{1}_{(v_0, \infty)}(v_-) \Big|(\rd|_-)(\rd_r|_+)^2 \psi(v_-, r', \theta, \varphi_-)|^2 \, dr' \vols dv_- \cdot |r - r_+| \;.
\end{split}
\end{equation}
Differentiating \eqref{EqRPsiNotWeight} once in $(\rd_r|_-)$ we obtain that $(\rd_r|_-)(\rd_r|_+)^2 \psi$ is a linear combination (with uniformly bounded coefficients) of the terms $\rd_{v_-}^a \rd_{\varphi_-}^b (\rd_r|_-)^c \hp$ with $0 \leq a + b + c \leq 3$. For $v_0 \geq 1$ all those terms are controlled by the bulk term in \eqref{EqPropLeftH} -- and for $v_0 \leq 1$ we complement this bulk term by the regularity Assumption \ref{AssumptionReg}. Hence, the spacetime integral in \eqref{EqIEstLoc} is uniformly bounded. This shows \eqref{EqCorL2LocLimit}.
\end{proof}



\subsection{Estimates away from the event and Cauchy horizons} \label{SecEENoShift}

\begin{proposition} \label{PropEnergyEstNoShift}
Under the assumptions of Section \ref{SecAssumptions}, and with $\rred$ as in Corollary \ref{CorFinEnEstHorizons}, we have that for any $r_0 \in (r_-, \rred]$ there exists a constant $C>0$ (depending on $r_0$) such that
\begin{equation}\label{PropEqNoShift}
\sup_{r' \in [r_0, \rred]} \sum_{0 \leq i_1 + i_2 + i_3 + j + k \leq 1} \int_{\{r=r'\}} \chi(v_+)  |\widetilde Z_{1,+}^{i_1}\widetilde Z_{2,+}^{i_2}\widetilde Z_{3,+}^{i_3}\rd_{v_+}^{j}\rd_r^{k}f|^2 \,\vols\,dv_{+}  \leq C
\end{equation}
holds for $f \in \{\rd_r^c\psi, \rd_{v_+} \rd_r^c\psi, \mathcal{Q}_{[s]} \rd_r^c\psi\}$ for $c = 0,1,2$.
\end{proposition}

\begin{proof}
We use Boyer-Lindquist coordinates for the proof. Since the region under consideration in \eqref{PropEqNoShift} is bounded away from $r_-$ and $r_+$, we have that $\rd_r|_+$ is a bounded linear combination of $\rd_t, \rd_\varphi, \rd_r|_{\mathrm{BL}}$. Thus, it is straightforward to see that \eqref{PropEqNoShift} follows from
\begin{equation}\label{EqPropMainBoyer}
\sup_{r' \in [r_0, \rred]} \sum_{0 \leq i_1 + i_2 + i_3 + j + k \leq 1} \int_{\{r=r'\}} \chi(t)  |\widetilde Z_{1}^{i_1}\widetilde Z_{2}^{i_2}\widetilde Z_{3}^{i_3}\rd_{t}^{j}\rd_r^{k}f|^2 \,\vols\,dt  \leq C
\end{equation}
for $f \in \{\rd_t^a \rd_\varphi^b (\rd_r|_{\mathrm{BL}})^c\psi, \rd_t \rd_t^a \rd_\varphi^b (\rd_r|_{\mathrm{BL}})^c\psi,  \mathcal{Q}_{[s]} \rd_t^a \rd_\varphi^b (\rd_r|_{\mathrm{BL}})^c\psi\}$, with $0 \leq a + b + c \leq 2$.  In the following we will prove \eqref{EqPropMainBoyer}.

\underline{\bf{Step 1: The multiplier.}} We start out from the following multiplier identity, where $\lambda, \mu, \eta >0$ are constants to be chosen and $\chi$ is as above:
\begin{equation}\label{EqPropMultiplierNoShift}
0 = \Rea( \mathcal{T}_{[s]} \psi (-\chi(t) e^{\lambda r} \overline{\rd_r \psi})) +  \underbrace{\rd_r(\chi(t) \mu e^{\eta r}|\psi|^2) - \chi(t)\mu \eta e^{\eta r} |\psi|^2 - 2\chi(t) \mu e^{\eta r} \Rea(\overline{\psi}\rd_r \psi)}_{=0} \;.
\end{equation}
After integration over the spheres, and using the form \eqref{TeukolskyEquation3} of the Teukolsky equation, the right hand side of \eqref{EqPropMultiplierNoShift} equals the sum of
\begin{enumerate}
\item the sum of all the terms on the right hand side of \ref{ApBoyerR}
\item the terms
\begin{equation*}
\begin{split}
&-\chi(t) e^{\lambda r} 2(r-M)(1-s)|\rd_r\psi|^2 - \chi(t) e^{\lambda r} 2s\frac{a(r-M)}{\Delta} \Rea(\rd_\varphi \psi \overline{\rd_r \psi}) \\
&-\chi(t)e^{\lambda r} 2s \Big[\frac{M(r^2 - a^2)}{\Delta} - r -ia \cos \theta\Big]\Rea(\rd_t \psi \overline{\rd_r \psi}) + \chi(t) e^{\lambda r}2s \Rea(\psi \overline{\rd_r \psi})
\end{split}
\end{equation*}
\item the underbraced terms in \eqref{EqPropMultiplierNoShift}.
\end{enumerate}
As before, it will turn out that we can derive a boundedness statement if all the bulk terms are negative.

\underline{\bf{Step 2: Estimating all bulk terms that are quadratic in derivatives of $\psi$.}}
We collect the leading order terms in $\lambda$ from \ref{ApBoyerR}, which are the wavily underlined terms:
\begin{equation}\label{EqPropNoShiftLeadLambda}
\frac{1}{2} \chi(t) e^{\lambda r} \lambda \Big[ \Big( \frac{(r^2 + a^2)^2}{\Delta} - a^2 \sin^2 \theta\Big) |\rd_t \psi|^2 + \frac{4Mar}{\Delta} \Rea(\rd_\varphi \psi \overline{\rd_t \psi}) + \frac{a^2}{\Delta} |\rd_\varphi \psi|^2 + \Delta |\rd_r \psi|^2 - \sum_i|\zt_i \psi|^2\Big]
\end{equation}
It turns out that in order to control the non-definite second term it is actually not sufficient just to use the $|\rd_\varphi \psi|^2$ control of the last term. Instead, we need to use the strengthened control provided by Lemma \ref{LemStrongerBoundPhiDer} together with the $a^2 \sin^2 \theta$ contribution of the first term in \eqref{EqPropNoShiftLeadLambda}. Thus, using Lemma \ref{LemStrongerBoundPhiDer}, we rewrite \eqref{EqPropNoShiftLeadLambda} as
\begin{equation}\label{EqRewriteNoShiftLambda}
\begin{split}
\frac{1}{2} \chi(t) e^{\lambda r} \lambda \Big[ \Big( &\frac{(r^2 + a^2)^2}{\Delta} - a^2 \sin^2 \theta\Big) |\rd_t \psi|^2 + \frac{4Mar}{\Delta} \Rea([is \cos \theta \cdot \psi + \rd_\varphi \psi] \overline{\rd_t \psi}) + \frac{a^2}{\Delta} |is\cos \theta \cdot \psi + \rd_\varphi \psi|^2 \\
 &+ \Delta|\rd_r \psi|^2 - |\rd_\theta \psi|^2 - \frac{1}{\sin^2 \theta} | is \cos \theta \cdot \psi + \rd_\varphi \psi|^2  - s^2|\psi|^2\Big] \\
&- \underbrace{\frac{1}{2} \chi(t) e^{\lambda r} \lambda \Big[ \frac{4Mar}{\Delta} \Rea(is \cos \theta \cdot \psi \overline{\rd_t \psi}) +\frac{a^2}{\Delta} s^2 \cos^2 \theta | \psi|^2 + 2\frac{a^2}{\Delta} \Rea(is \cos \theta \cdot \psi \overline{\rd_\varphi \psi})\Big]}
\end{split}
\end{equation}
The underbraced terms will be treated as error terms. We show now that for $r \in [r_0, \rred]$ with $r_0 > r_-$ the remaining terms (modulo zeroth order terms) are uniformly bounded from above by
\begin{equation}\label{EqPropNoShiftUniformControlLambda}
-\frac{1}{2} \chi(t) e^{\lambda r} \lambda \cdot c \big( |\rd_t \psi|^2 + |\rd_r \psi|^2 + \sum_i|\zt_i \psi|^2\big)
\end{equation}
for some $c>0$ depending on $r_- <r_0< \rred < r_+$. For this it is clearly sufficient to show that the non-underbraced terms in \eqref{EqRewriteNoShiftLambda} are uniformly negative definite in $\rd_t \psi$ and $\frac{1}{\sin \theta} (is \cos \theta \cdot \psi + \rd_\varphi \psi)$. A straightforward computation gives
\begin{equation}\label{EqDetSmallMinor}
\det \begin{pmatrix}
\frac{(r^2 + a^2)^2 }{\Delta} - a^2 \sin^2 \theta & \frac{2Mar}{\Delta} \sin \theta \\ 
\frac{2Mar}{\Delta} \sin \theta & \frac{a^2}{\Delta}\sin^2 \theta -1 
\end{pmatrix}
= - \frac{1}{\Delta} \big(r^2 + a^2\cos^2 \theta\big)^2 >0 \;,
\end{equation}
which shows the claim.

We can now choose $\lambda >0$ large enough such that all bulk terms that are quadratic in derivatives of $\psi$ can be controlled in absolute value by $-\frac{1}{4}\times$\eqref{EqPropNoShiftUniformControlLambda}. The underbraced bulk terms in \eqref{EqRewriteNoShiftLambda} of the form $\Rea(\psi \overline{\rd \psi})$, which are also \emph{leading order} in $\lambda$, can be estimated by $|\Rea(\psi \overline{\rd \psi})| \leq \frac{1}{2} \varepsilon |\rd \psi|^2 + \frac{1}{2} \varepsilon^{-1} |\psi|^2$, where we choose $\varepsilon>0$ so small that the arising first term can be bounded by $-\frac{1}{4}\times$\eqref{EqPropNoShiftUniformControlLambda}. It thus only remains to estimate the zeroth order bulk terms, which will be done in Step 4.

\underline{\bf{Step 3: Estimating boundary terms.}}
We collect all the total derivatives appearing on the right hand side of \eqref{EqPropMultiplierNoShift}. They are of the form $\rd_t (A)$ and $\rd_r(B)$, where
\begin{equation*}
A = \chi(t) e^{\lambda r} \Big[ \Big( \frac{(r^2 + a^2)^2}{\Delta} - a^2 \sin^2 \theta\Big) \Rea(\rd_t \psi \overline{\rd_r \psi}) + \frac{2Mar}{\Delta} \Rea(\rd_\varphi \psi \overline{\rd_r \psi})\Big] 
\end{equation*}
and
\begin{equation}\label{EqBNoShift}
\begin{split}
B =\frac{1}{2} \chi(t) e^{\lambda r}\Big[&-\Big(\frac{(r^2 + a^2)^2}{\Delta} - a^2 \sin^2 \theta\Big)|\rd_t \psi|^2 - \frac{4Mar}{\Delta} \Rea(\rd_\varphi \psi \overline{\rd_t \psi}) - \frac{a^2}{\Delta} |\rd_\varphi \psi|^2 \\
&- \Delta | \rd_r \psi|^2 - (s + s^2) |\psi|^2 + \sum_i |\zt_i \psi|^2 \Big] + \chi(t) \mu e^{\eta r} |\psi|^2 \;.
\end{split}
\end{equation}
The coercivity of $B$,
\begin{equation}\label{EqCoerBNoShift}
B \gtrsim \chi(t) \big( |\rd_r \psi|^2 + |\rd_t \psi|^2 + \sum_i |\zt_i \psi|^2 + |\psi|^2\big) \;,
\end{equation}
in the region $r_0 \leq r \leq \rred$ is established using the same computation as in Step 2: First, we use Lemma \ref{LemStrongerBoundPhiDer} and moreover replace every $\rd_\varphi \psi$ in \eqref{EqBNoShift} by $\frac{1}{\sin \theta}(is \cos \theta \cdot \psi + \rd_\varphi \psi)$, thus obtaining again error terms. The lower bound of $B$ in $\chi(t)\big(|\rd_r \psi|^2 + |\rd_t \psi|^2 + \sum_i |\zt_i \psi|^2\big)$ then follows again from \eqref{EqDetSmallMinor} \emph{at the expense of a large zeroth order error term}. We now choose $\mu$ as a function of $\eta$ such that the last term in \eqref{EqBNoShift}, $\chi(t)\mu e^{\eta r}|\psi|^2$, is large enough in the region $r_0 \leq r \leq \rred$ to dominate this error term. This yields \eqref{EqCoerBNoShift}.

Next we establish the coercivity of $B \pm \frac{\Delta}{2Mr} A$,
\begin{equation}\label{EqCoerABNoShift}
B \pm \frac{\Delta}{2Mr} A \gtrsim \chi(t) \big( |\rd_r \psi|^2 + |\rd_t \psi|^2 + \sum_i |\zt_i \psi|^2 + |\psi|^2\big) \;,
\end{equation}
in the region $r_0 \leq r \leq \rred$. It follows again from Lemma \ref{LemStrongerBoundPhiDer} that we have
\begin{equation}\label{EqReplaceNoShift}
\begin{split}
B \pm \frac{\Delta}{2Mr} A &= \chi(t) e^{\lambda r} \Big[ - \frac{1}{2} \Big( \frac{(r^2 + a^2)^2}{\Delta} - a^2 \sin^2 \theta \Big) |\rd_t \psi|^2 - \frac{2Mar}{\Delta} \sin \theta \Rea\big(\frac{1}{\sin \theta} (is \cos \theta \cdot \psi + \rd_\varphi \psi) \overline{\rd_t \psi}\big) \\[3pt]
&\qquad \qquad \quad  - \frac{1}{2} \frac{a^2}{\Delta} \sin^2\theta \Big|\frac{1}{\sin \theta}(is \cos \theta \cdot \psi + \rd_\varphi \psi)\Big|^2  -\frac{1}{2}\Delta |\rd_r \psi|^2 - \frac{1}{2}(s + s^2) |\psi|^2 \\[3pt]
&\qquad \qquad \quad + \frac{1}{2}\big(|\rd_\theta \psi|^2 + \frac{1}{\sin^2 \theta}|is \cos \theta \cdot \psi + \rd_\varphi \psi|^2 + s^2 |\psi|^2\big)\Big] + \chi(t) \mu e^{\eta r} |\psi|^2 \\[3pt]
&\pm \chi(t) e^{\lambda r} \Big[\Big(\frac{(r^2+ a^2)^2}{2Mr} - \frac{\Delta a^2 \sin^2\theta}{2Mr}\Big) \Rea(\rd_t \psi \overline{\rd_r \psi}) + a \sin \theta \Rea\big( \frac{1}{\sin \theta}(is \cos \theta \cdot \psi + \rd_\varphi \psi)\overline{\rd_r \psi}\big) \Big] \\[3pt]
&-\underbrace{\chi(t) e^{\lambda r}\Big[ - \frac{2Mar}{\Delta} \Rea(is \cos \theta \cdot \psi \overline{\rd_t \psi}) - \frac{1}{2} \frac{a^2}{\Delta} s^2 \cos^2 \theta |\psi|^2 - \frac{a^2}{\Delta} \sin \theta \Rea(is \cos \theta \cdot \psi \overline{\rd_\varphi \psi})\Big]} \\[3pt]
&\mp \underbrace{\chi(t) e^{\lambda r}a \Rea(is \cos \theta \cdot \psi \overline{\rd_r \psi})}\;,
\end{split}
\end{equation}
where the underbraced terms are considered as error terms. Again, we consider the part of the above expression that is quadratic in $\{\rd_t \psi, \frac{1}{\sin \theta}(is \cos \theta \cdot \psi + \rd_\varphi \psi), \rd_r \psi\}$ as a quadratic form. Its associated matrix is
\begin{equation*}
M_{\pm} := \begin{pmatrix}
 - \frac{1}{2} \Big( \frac{(r^2 + a^2)^2}{\Delta} - a^2 \sin^2 \theta \Big) & -\frac{Mar}{\Delta} \sin \theta &  \pm \frac{1}{2} \Big(\frac{(r^2+ a^2)^2}{2Mr} - \frac{\Delta a^2 \sin^2\theta}{2Mr}\Big) \\
 -\frac{Mar}{\Delta} \sin \theta & - \frac{1}{2} \big(\frac{a^2}{\Delta} \sin^2 \theta -1\big) & \pm \frac{1}{2} a \sin \theta \\
 \pm \frac{1}{2} \Big(\frac{(r^2+ a^2)^2}{2Mr} - \frac{\Delta a^2 \sin^2\theta}{2Mr}\Big) & \pm \frac{1}{2} a \sin \theta & - \frac{1}{2} \Delta \;,
\end{pmatrix}
\end{equation*}
which  we claim is positive definite in the region $r_0 \leq r \leq \rred$: Obviously, the first main minor is positive, the  second main minor was computed in \eqref{EqDetSmallMinor} to be positive, and a computation gives 
\begin{equation*}
\det M_{\pm} = \frac{-\Delta}{32M^2 r^2} (r^2 + a^2 \cos^2 \theta)^2(r^2 + a^2 \cos^2\theta + 2Mr) >0  \qquad \textnormal{ for } r_0 \leq r \leq \rred\;,
\end{equation*}
from which the claim follows. Now, if necessary, choosing $\mu(\eta)$ even larger, we can control all the error terms in \eqref{EqReplaceNoShift} to obtain \eqref{EqCoerABNoShift}.

\underline{\bf{Step 4: Estimating the remaining bulk terms.}}
As familiar from the proof of Propositions \ref{PropLeftHp} and \ref{PropRightHP} we estimate the last two of the underbraced terms in \eqref{EqPropMultiplierNoShift} by
\begin{equation*}
-\chi(t) \mu \eta e^{\eta r} |\psi|^2 - 2 \chi(t) \mu e^{\eta r} \Rea(\overline{\psi} \rd_r \psi) \leq - \frac{1}{2} \chi(t) \mu \eta e^{\eta r} | \psi|^2 + 2 \chi(t) \eta^{-1} \mu e^{\eta r} |\rd_r \psi|^2 \;.
\end{equation*}
We now choose $\eta>0$ sufficiently large so that the last term can be controlled by $-\frac{1}{4} \times$ \eqref{EqPropNoShiftUniformControlLambda} and such that the first term controls all the zeroth order terms in the bulk (including those generated at the end of Step 2).

\underline{\bf{Step 5: Putting it all together.}}
After integration over the spheres we thus obtain from \eqref{EqPropMultiplierNoShift}
\begin{equation}\label{EqPutAllTogetherNoShift}
\rd_t(A) + \rd_r(B) \underset{\mathrm{a.i.}}{\gtrsim} \chi(t)\big(|\rd_r \psi|^2 + |\rd_t \psi|^2 + \sum_i|\zt_i \psi|^2 + |\psi|^2 \big) 
\end{equation}
for $r_0 \leq r \leq \rred$.
Let $r' \in [r_0, \rred)$. Integrating \eqref{EqPutAllTogetherNoShift} over $\{r' \leq r \leq \rred\} \cap \{f^- \leq t_0\} \cap \{f^+ \leq t_0\}$ with respect to $\frac{1}{\rho^2} \vol = dt \wedge dr \wedge \vols$, where $t_0 \gg 1$, and using that on the level sets of $f^-$ as well as on those of $f^+$ we have $\big|\frac{dr}{dt}\big| = \frac{|\Delta|}{2Mr}$, we obtain
\begin{equation*}
\begin{split}
\int\limits_{\mathclap{\substack{\{r=r'\}  \cap \{ f^- \leq t_0\} \\ \cap \{f^+ \leq t_0\} }}} B \, \vols dt &+ \int \limits_{\mathclap{\substack{\{f^- = t_0\} \\ \cap \{r' \leq r \leq \rred\}}}}\big(B + \frac{|\Delta|}{2Mr} A\big) \, \vols dt + \int\limits_{\mathclap{\substack{\{f^+ = t_0\} \\ \cap \{r' \leq r \leq \rred\}}}} \big( B - \frac{|\Delta|}{2Mr} A \big) \, \vols dt \\
&+ c \int\limits_{\mathclap{\substack{\{r' \leq r \leq \rred\} \cap \{f^- \leq t_0\} \\ \cap \{f^+ \leq t_0\}}}}\chi(t)\big(|\rd_r \psi|^2 + |\rd_t \psi|^2 + \sum_i |\zt_i \psi|^2 + |\psi|^2\big) \, \vols dt \,dr \\
&\leq \int\limits_{\mathclap{\substack{\{r=\rred\} \cap \{f^- \leq t_0\} \\ \cap \{f^+ \leq t_0\}}}} B \, \vols dt \;,
\end{split}
\end{equation*}
where $c>0$ is a constant depending on $r_0$. Using \eqref{EqCoerABNoShift} to infer the positivity of the second and third term, \eqref{EqCoerBNoShift}, and letting $t_0 \to \infty$, we obtain
\begin{equation}\label{EqFinEnEstNoShift}
\begin{split}
\int\limits_{\mathclap{\substack{\{r=r'\}  }}}  \chi(t)\big(|\rd_r \psi|^2 &+ |\rd_t \psi|^2 + \sum_i |\zt_i \psi|^2 + |\psi|^2\big) \, \vols dt + \int\limits_{\mathclap{\substack{\{r' \leq r \leq \rred\} }}}\chi(t)\big(|\rd_r \psi|^2 + |\rd_t \psi|^2 + \sum_i |\zt_i \psi|^2 + |\psi|^2\big) \, \vols dt \,dr \\
&\leq C \int\limits_{\mathclap{\substack{\{r=\rred\} }}} B(\psi) \, \vols dt \;,
\end{split}
\end{equation}
where $C>0$ is a constant depending on $r_0$. This, together with the trivial upper bounds on $B$ and Corollary \ref{CorFinEnEstHorizons} (note that $\rd_r|_{\mathrm{BL}}$ is a bounded linear combination of $\rd_{v_+}, \rd_{\varphi_+}, \rd_r|_+$), gives \eqref{EqPropMainBoyer} for $f = \psi$.

\underline{\bf{Step 6: Estimating higher derivatives.}} 
Because $\rd_t$, $\rd_\varphi$, $\mathcal{Q}_{[s]}$ commute with $\mathcal{T}_{[s]}$, it follows directly that \eqref{EqFinEnEstNoShift} also holds with $\psi$ replaced by $\rd_t^a \rd_\varphi^b \psi$, $\rd_t \rd_t^a \rd_\varphi^b \psi$, and $\mathcal{Q}_{[s]}\rd_t^a \rd_t^b \psi$ with $0 \leq a + b \leq 2$. The conclusion of Corollary \ref{CorFinEnEstHorizons} implies that the boundary term at $\{r = \rred\}$ is bounded, thus giving \eqref{EqPropMainBoyer} with $c =0$.

Moreover, we observe that
\begin{equation}\label{EqCommuBoyer}
\begin{split}
[\rd_r, \mathcal{T}_{[s]}] \psi &= \rd_r \mathcal{T}_{[s]} \psi - \mathcal{T}_{[s]} \rd_r \psi \\
&=- \rd_r\Big(\frac{(r^2 + a^2)^2}{\Delta}\Big) \rd_t^2 \psi - \rd_r\Big(\frac{4Mar}{\Delta}\Big) \rd_t \rd_\varphi \psi - \rd_r\Big(\frac{a^2}{\Delta}\Big) \rd_\varphi^2 \psi + \dotuline{(\rd_r \Delta) \rd_r^2 \psi} \\
&+ 2(1-s)\rd_r \psi + \rd_r\Big(2s \frac{a(r-M)}{\Delta}\Big) \rd_\varphi \psi + \rd_r(2s \Big(\frac{M(r^2 - a^2)}{\Delta} - r\Big))\rd_t \psi \;.
\end{split}
\end{equation}
Thus, all the additional bulk terms in the energy estimate
\begin{equation*}
\begin{split}
0 = \Rea( \mathcal{T}_{[s]} \rd_r\psi (-\chi(t) e^{\lambda r} \overline{\rd_r^2 \psi})) &+ \Rea([\rd_r, \mathcal{T}_{[s]}] \psi (-\chi(t) e^{\lambda r} \overline{\rd_r^2 \psi})) \\
&+  \underbrace{\rd_r(\chi(t) \mu e^{\eta r}|\rd_r\psi|^2) - \chi(t)\mu \eta e^{\eta r} |\rd_r\psi|^2 - 2\chi(t) \mu e^{\eta r} \Rea(\overline{\rd_r\psi}\rd_r^2 \psi)}_{=0} \;,
\end{split}
\end{equation*}
after Cauchy Schwarz, have either already been controlled by the integrated \eqref{EqPropMainBoyer} with $c=0$ (or the bulk term in \eqref{EqFinEnEstNoShift} with $\psi$  replaced by $\rd_t^a \rd_\varphi^b\psi$) or are at the level of energy for $\rd_r \psi$ (i.e., the dotted term in \eqref{EqCommuBoyer}). We thus also obtain \eqref{EqFinEnEstNoShift} with $\psi$ replaced by $\rd_r \psi$. Since $\rd_r|_{\mathrm{BL}}$ is a bounded linear combination of $\rd_{v_+}, \rd_{\varphi_+}, \rd_r|_+$, the boundary term at $\{r = \rred\}$ is bounded by Corollary \ref{CorFinEnEstHorizons}. We can again commute with $\rd_t, \rd_\varphi, \mathcal{Q}_{[s]}$ to obtain \eqref{EqPropMainBoyer} for $c = 1$.

Finally, commuting \eqref{TeukolskyEquation3} once more with $\rd_r$ we find that $[\rd_r^2, \mathcal{T}_{[s]}]$ is a bounded linear combination of the terms $\rd_r \rd_t^2 \psi, \rd_t^2 \psi, \rd_r\rd_t \rd_\varphi\psi, \rd_t \rd_\varphi \psi, \rd_r \rd_\varphi^2 \psi, \rd_\varphi^2 \psi, \dotuline{\rd_r^3 \psi}, \rd_r^2 \psi, \rd_\varphi \psi, \rd_t \psi, \rd_r \rd_t \psi$. We can repeat the same energy estimate, but now for $\rd_r^2 \psi$. The dotted term is again at the level of the energy for $\rd_r^2\psi$ and all the other terms have already been controlled. Commutation with $\rd_t$ and $\mathcal{Q}_{[s]}$ then concludes the proof.
\end{proof}

\subsubsection{Corollaries}

The following corollary is needed for Teukolsky's separation of variables in Theorem \ref{ThmSeparationVariables}.

\begin{corollary}\label{CorollaryNoShift}
Under the assumptions of Section \ref{SecAssumptions} and for $r_- < r_0 < r_1 < r_+$ there exists a constant $C>0$ (depending on $r_0, r_1$) such that for $f \in \{  \psi, \rd_r \psi, \rd_r^2 \psi\}$
\begin{equation}\label{EqCorH2NoShift}
\sup_{r' \in [r_0,r_1]} \sum_{0 \leq i_1 + i_2 + i_3 + j \leq 2} \int\limits_{\{r = r'\}} \chi(v_+) |\widetilde Z_{1,+}^{i_1}\widetilde Z_{2,+}^{i_2}\widetilde Z_{3,+}^{i_3}\rd_{v_+}^{j}f|^2 \,\vols\,dv_{+}  \leq C
\end{equation}
holds  and
\begin{equation}\label{EqCorSobolevNoShift}
|f(v_+, r,\theta, \varphi_+)| \leq \frac{C}{\sqrt{\chi(v_+)}}
\end{equation}
holds for $r \in (r_0, r_1)$ and all $(\theta, \varphi_+) \in \Sp^2 \setminus \{ \theta = 0, \pi\}$.
\end{corollary}

\begin{proof}
We begin by proving the bound \eqref{EqCorH2NoShift}; first for $f = \psi$. Then all the terms with $0 \leq i_1 + i_2 + i_3 + j \leq 1$ are controlled by \eqref{PropEqNoShift} and \eqref{CorConclusionForNoShift2} with $f = \psi$. Using $f = \rd_{v_+} \psi$ in \eqref{PropEqNoShift}  and \eqref{CorConclusionForNoShift2} extends control to all terms except those with $i_1 + i_2 + i_3 = 2$. We now use $f = \mathcal{Q}_{[s]} \psi$ in \eqref{PropEqNoShift} and \eqref{CorConclusionForNoShift2} and use just the $L^2$-control. By the Definition \ref{DefCarterOp} of the Carter operator and by the fact that we have already controlled $\rd_{v_+}\psi$ and $\rd_{v_+}^2 \psi$ in $L^2$, this gives us $L^2$-control of $\swl \psi$. Lemma \ref{LemSecondAngularDer}, together with the $L^2$-control of the first angular derivatives already obtained, now controls the remaining terms with $i_1 + i_2 + i_3 = 2$. The cases of $f = \rd_r \psi, \rd_r^2 \psi$ can be treated analogously using that \eqref{PropEqNoShift} and \eqref{CorConclusionForNoShift2} hold for $f \in \{\rd_r^k\psi, \rd_{v_+} \rd_r^k\psi, \mathcal{Q}_{[s]} \rd_r^k\psi\}$ for $k = 0,1,2$.

To prove \eqref{EqCorSobolevNoShift} we observe that for $1 \leq v_0 < \infty$ \eqref{EqCorH2NoShift} implies $$\sup_{r' \in [r_0,r_1]} \sum_{0 \leq i_1 + i_2 + i_3 + j \leq 2} \int_{v_0}^\infty \int_{\Sp^2} |\widetilde Z_{1,+}^{i_1}\widetilde Z_{2,+}^{i_2}\widetilde Z_{3,+}^{i_3}\rd_{v_+}^{j}f(v_+, r', \theta,\varphi_+)|^2 \,\vols\,dv_{+}  \leq \frac{C}{v_0^{q_r}} \;.$$ By Lemma \ref{LemSobolev} we thus have $$\sup_{r' \in [r_0,r_1]} \sum_{0 \leq i_1 + i_2 + i_3 + j \leq 2} \int_{v_0}^\infty \int_{\Sp^2_+} | Z_{1,+}^{i_1} Z_{2,+}^{i_2} Z_{3,+}^{i_3}\rd_{v_+}^{j} e^{is \varphi_+}f(v_+, r', \theta,\varphi_+)|^2 \,\vols\,dv_{+}  \leq \frac{C}{v_0^{q_r}} $$ and similarly for the southern hemisphere $\Sp^2_-$. A standard Sobolev inequality\footnote{See for example 8.8 Theorem in \cite{LiebLoss}. By choosing suitable coordinates for $\Sp^2_+$ the domain $(v_0, \infty) \times \Sp^2_+$ can be viewed as an open subset of $\R^3$ which satisfies a cone property that is uniform in $v_0$.} applied to $e^{\pm is \varphi_+} f$ thus gives $$\sup_{r' \in [r_0,r_1]} \sup_{(\theta, \varphi_+) \in \Sp^2\setminus \{ \theta = 0, \pi\}} |f(v_0, r',\theta, \varphi_+)| \leq \frac{C}{\sqrt{v_0^{q_r}}}$$ for $v_0 \geq 1$. We proceed similarly for $v_0 \leq -1$ and for $v_+ \in [-1,1], r \in [r_0, r_1]$ the field is uniformly bounded since it is regular. This shows \eqref{EqCorSobolevNoShift}.
\end{proof}

Note that the reason for why the constant $C>0$ in Corollary \ref{CorollaryNoShift} blows up when we let $r_1$ go to $r_+$ is because of the conversion of the $v_-$-weights to $v_+$-weights, which becomes worse and worse when $r_1 \to r_+$, cf.\ the proof of Corollary \ref{CorFinEnEstHorizons}. If we restrict to the region $v_+ \geq 1$ then the constant can be chosen uniformly up to $r = r_+$:
\begin{corollary} \label{CorUniformDecayNearHP}
Under the assumptions of Section \ref{SecAssumptions} and for given $r_- < r_0 < r_+$ there exists $C>0$ such that for $ f \in \{ \psi, \rd_r \psi, \rd_r^2 \psi\}$
\begin{equation}\label{EqSob22}
|f(v_+, r,\theta, \varphi_+)| \leq \frac{C}{v_+^\frac{q_r}{2}}
\end{equation}
holds for all $r \in [r_0, r_+], v_+ \geq 1, (\theta, \varphi) \in \Sp^2\setminus \{ \theta = 0, \pi\}$.
\end{corollary}
Indeed, the statement is only needed for $f= \rd_r^2 \psi$ (for Proposition \ref{PropLimitDis}).
\begin{proof}
In the same way as in the proof of \eqref{EqCorH2NoShift} in Corollary \ref{CorollaryNoShift} one obtains
$$\sup_{r' \in [r_0,r_+]} \sum_{0 \leq i_1 + i_2 + i_3 + j \leq 2} \int\limits_{\{r = r'\} \cap \{v_+ \geq 1\}} v_+^{q_r} |\widetilde Z_{1,+}^{i_1}\widetilde Z_{2,+}^{i_2}\widetilde Z_{3,+}^{i_3}\rd_{v_+}^{j}f|^2 \,\vols\,dv_{+}  \leq C $$
from Proposition \ref{PropRightHP} (and Proposition \ref{PropEnergyEstNoShift}) for $ f \in \{ \psi, \rd_r \psi, \rd_r^2 \psi\}$, but now with a constant which is uniform up to $r=r_+$.  As before one now proves \eqref{EqSob22} by Sobolev embedding.
\end{proof}

\subsection{Estimates near the Cauchy horizons} \label{SecEECauchy}

Recall that for the method of proof of Theorem \ref{Thm1} it is convenient to first establish the blow-up result \eqref{EqThm1} along the left Cauchy horizon and then to propagate it backwards. The estimates established in this section are used to show that 1) we can indeed extend $ \psi$ to the left Cauchy horizon (along with a convergence result); 2)  the $\chi$-weighted $L^2$-bound propagates all the way to the left Cauchy horizon; 3) the singularity can be propagated backwards from the left Cauchy horizon. All this is used in Section \ref{SecPfMainThm}.

\begin{proposition}\label{PropEnCauchyHorizons}
Under the assumptions of Section \ref{SecAssumptions} there exists an $r_{\mathrm{ered}} \in (r_-, \rred)$ and a constant $C>0$ such that the following holds
\begin{align}
\int\limits_{\mathclap{\{r_- < r \leq r_{\mathrm{ered}}\}}} \chi(v_+) \Big(|\rd_r \psi|^2 + |\rd_{v_+} \psi|^2 + \sum_i | \zt_{i,+} \psi|^2 + |\psi|^2\Big) \, \vols dv_+ dr &\leq C \label{PropCauchyEqILED} \\[4pt]
\sup_{r' \in [\ered, r_-)} \int\limits_{\{r = r'\}} \chi(v_+) \Big(|\Delta| |\rd_r \psi|^2 + |\rd_{v_+} \psi|^2 + \sum_i |\zt_{i,+} \psi|^2 + |\psi|^2\Big) \, \vols dv_+  &\leq C \;, \label{PropCauchyEqBoundaryTerms}
\end{align}
where the function $\chi$ is as in Corollary \ref{CorFinEnEstHorizons}.
\end{proposition}

\begin{proof}
We use that for $s=+2$ there is an \emph{effective} red-shift for the energy operating close to the left Cauchy horizon. The red-shift is effective in the sense that while it persists after one commutation of $\mathcal{T}_{[s]} \psi = 0$ with $\rd_r$, after two commutations with $\rd_r$ it turns into a blue-shift for the energy, which becomes stronger with subsequent commutations. We use $\chi_n(v_+)(1 + \lambda \Delta)(-\rd_r + \rd_{v_+} + \frac{a}{r_-^2 + a^2} \rd_{\varphi_+})$ as a multiplier. Note that, compared to the multiplier used in the proof of Proposition \ref{PropRightHP}, the additional contribution in $\rd_{\varphi_+}$ makes the vector field timelike near the Cauchy horizons. 

\underline{\bf{Step 1: The multiplier.}} We start out from the following multiplier identity, where $\lambda <0$ and  $\eta, \mu >0$ are constants to be chosen:
\begin{equation}\label{EqPropCauchyMultiplier}
\begin{split}
0 &= \Rea\Big(\mathcal{T}_{[s]} \psi \cdot \chi_n(v_+)(1 + \lambda \Delta)(-\rd_r + \rd_{v_+} + \frac{a}{r_-^2 + a^2} \rd_{\varphi_+})\overline{\psi}\Big) \\
&\qquad + \underbrace{\rd_r(\chi_n(v_+) \mu e^{\eta r}| \psi|^2) - \chi_n(v_+)\mu \eta e^{\eta r} | \psi|^2 - 2\chi_n(v_+) \mu e^{\eta r} \Rea(\overline{\psi}\rd_r \psi)}_{=0} \;.
\end{split}
\end{equation}
Here, the function $\chi_n : \R \to (0, \infty)$ results from locally smoothing out the corners of the function
\begin{equation*}
v_+ \mapsto \begin{cases} |v_+|^{q_l} \qquad &\textnormal{ for } v_+ \leq -(n)^{\nicefrac{1}{q_l}} \\
n &\textnormal{ for } -(n)^{\nicefrac{1}{q_l}} \leq v_+ \leq n^{\nicefrac{1}{q_r}} \\
v_+^{q_r} &\textnormal{ for } v_+ \geq n^{\nicefrac{1}{q_r}} \;.
\end{cases}
\end{equation*}
Given $\delta >0$ it is easy to see that one can choose $n\geq 1$ large enough such that $|\chi'_n(v_+)| \leq \delta \chi_n(v_+)$ holds for all $v_+ \in \R$. The parameter $n$ will be fixed in the next step.

After integration over the spheres, the right hand side of \eqref{EqPropCauchyMultiplier} consists of the sum of the following terms
\begin{enumerate}
\item the sum of all the terms on the right hand sides of \ref{ApR}, \ref{ApV}, and \ref{ApPhi} with $\chi(v_+) = \chi_n(v_+)$
\item the real parts of the terms
\begin{equation*}
\begin{split}
2\big(r(1-2s) &- isa \cos \theta\big) \rd_{v_+} \psi \cdot \chi_n(v_+)(1 + \lambda \Delta)(-\rd_r + \rd_{v_+} + \frac{a}{r_-^2 + a^2} \rd_{\varphi_+})\overline{\psi}\\[4pt]
&\dashuline{-\chi_n(v_+) (1+ \lambda \Delta) 2(r-M)(1-s)|\rd_r \psi|^2} \\[4pt]
&+ \chi_n(v_+)(1+\lambda \Delta)2(r-M)(1-s) \rd_r \psi (\overline{\rd_{v_+} \psi + \frac{a}{r_-^2 + a^2} \rd_{\varphi_+} \psi})\\[4pt]
&-2s\psi \cdot \chi_n(v_+)(1 + \lambda \Delta)(-\rd_r + \rd_{v_+} + \frac{a}{r_-^2 + a^2} \rd_{\varphi_+})\overline{\psi}
\end{split}
\end{equation*}
\item the underbraced terms in \eqref{EqPropCauchyMultiplier}.
\end{enumerate}
Again, our desired boundedness statement requires all the bulk terms to yield a negative definite contribution. We also recall here that $(\rd_r \Delta)(r_-) = 2(r_- - M) <0$.

\underline{\bf{Step 2: Estimating all bulk terms that are quadratic in derivatives of $\psi$.}}

We first consider  all those terms that are quadratic in $\rd_r \psi$. The leading order terms are the dashed term from 2.\ in Step 1 and the dashed term from \ref{ApR}. Their sum at $r=r_-$ equals
\begin{equation*}
\chi_n(v_+) 2(r_- - M) \big(\frac{1}{2} - (1-s)\big)|\rd_r \psi|^2 \;,
\end{equation*}
which is negative for $s=+2$. The other two bulk terms quadratic in $\rd_r \psi$ from \ref{ApR} and \ref{ApV} sum to
\begin{equation}\label{EqLowerOrderRTerms}
\chi'_n(v_+) (1 + \lambda \Delta)(r^2 + a^2 + \frac{1}{2} \Delta)|\rd_r\psi|^2 \;.
\end{equation}
We can now choose $n \gg 1$ large enough and $r_- < \ered$ close enough to $r_-$  ($\ered$ depending in particular on $\lambda$ at this point) such that \eqref{EqLowerOrderRTerms} is controlled by $-\frac{1}{2}$ times the sum of the dashed terms in $r_- < r \leq \ered$. 

We next consider all those terms quadratic in $\rd \psi$ that are leading order in $\lambda$; these are all the wavily underlined terms from \ref{ApR}, \ref{ApV}, and \ref{ApPhi}. They sum to
\begin{equation}\label{EqPfPropCauchyWavy}
\begin{split}
-\chi_n(v_+)& \lambda 2(r-M) \Big[\big(\frac{1}{2}a^2 \sin^2\theta + (r^2 + a^2)\big) |\rd_{v_+} \psi|^2 + \big(2a + \frac{a(r^2 + a^2)}{r_-^2 + a^2}\big) \Rea(\rd_{\varphi_+} \psi \overline{\rd_{v_+} \psi}) \\[3pt]
&\qquad \qquad \quad + \frac{1}{2} \sum_i |\zt_{i,+} \psi|^2 + \frac{a^2}{r_-^2 + a^2} |\rd_{\varphi_+} \psi|^2 \Big] \\[3pt]
&= -\chi_n(v_+) \lambda 2(r-M) \Big[\big(\frac{1}{2}a^2 \sin^2\theta + (r^2 + a^2)\big) |\rd_{v_+} \psi|^2 \\[3pt]
&\qquad \qquad \quad+ \big(2a + \frac{a(r^2 + a^2)}{r_-^2 + a^2}\big)  \sin \theta \Rea\big(\frac{1}{\sin\theta}(is \cos \theta \cdot \psi + \rd_{\varphi_+} \psi) \overline{\rd_{v_+} \psi}\big) \\[3pt]
&\qquad \qquad \quad + \frac{1}{2} \big( |\rd_\theta \psi|^2 + \frac{1}{\sin^2 \theta} |is \cos \theta \cdot \psi + \rd_{\varphi_+} \psi|^2 + \underbrace{s^2 |\psi|^2}\big) \\[3pt]
&\qquad \qquad \quad + \frac{a^2}{r_-^2 + a^2} \sin^2 \theta \big|\frac{1}{\sin \theta}(is \cos \theta \cdot \psi + \rd_{\varphi_+} \psi)\big|^2 \Big] \\[3pt]
&\quad + \underbrace{\chi_n(v_+) \lambda 2(r-M) \Big[ \big(2a + \frac{a(r^2 + a^2)}{r_-^2 + a^2}\big)\Rea(is \cos \theta \cdot \psi \overline{\rd_{v_+} \psi}) }\\[3pt]
&\qquad \qquad \quad + \underbrace{\frac{a^2}{r_-^2 + a^2}\big( s^2 \cos^2\theta |\psi|^2 + 2\Rea(is \cos \theta \cdot \psi \overline{\rd_{\varphi_+} \psi})\big)\Big]} \;,
\end{split}
\end{equation}
where we have used again Lemma \ref{LemStrongerBoundPhiDer} and we consider the underbraced terms again as error terms.
Considering the non-underbraced terms as a quadratic form in $\big(\rd_{v_+} \psi, \frac{1}{\sin \theta}(is \cos \theta \cdot \psi + \rd_{\varphi_+} \psi)\big)$, the corresponding matrix is $-\chi_n(v_+) \lambda 2(r-M) Q_1$ with
\begin{equation}\label{EqCauchyDetQ}
Q_1 = \begin{pmatrix}
\frac{1}{2}a^2 \sin^2\theta + (r^2 + a^2) & \big(a + \frac{a(r^2 + a^2)}{2(r_-^2 + a^2)} \big) \sin \theta \\
\big(a + \frac{a(r^2 + a^2)}{2(r_-^2 + a^2)} \big) \sin \theta  & \frac{1}{2} + \frac{a^2}{r_-^2 + a^2} \sin^2\theta
\end{pmatrix} \;.
\end{equation}
The determinant of $Q_1$ evaluated at $r=r_-$ is easily computed to be $\det Q_1(r_-) =\frac{1}{2(r_-^2 + a^2)} (r_-^2 + a^2 \cos^2 \theta)^2 >0$, and hence $Q_1$ is positive definite at $r=r_-$. Recalling that $\lambda <0$ and $2(r_- - M) <0$ it now follows that for $\ered > r_-$ close enough to $r_-$ there exist constants $c>0, C>0$ such that the following holds in $r_-< r\leq \ered$:
\begin{equation*}
\textnormal{\eqref{EqPfPropCauchyWavy}} \leq - c \chi_n(v_+) | \lambda|\big(|\rd_{v_+} \psi|^2 + \sum_i|\zt_{i,+}\psi|^2\big) + \chi_n(v_+) C |\lambda|\cdot |\psi|^2 \;.
\end{equation*}
Together with our earlier estimates for $\rd_r \psi$ this shows that the dashed terms, the other terms quadratic in $\rd_r \psi$, and the wavily underlined terms are bounded from  above by
\begin{equation*}
- c \chi_n(v_+)\Big[ |\rd_r \psi|^2 +  | \lambda|\big(|\rd_{v_+} \psi|^2 + \sum_i|\zt_{i,+}\psi|^2\big)\Big] + \chi_n(v_+) C |\lambda|\cdot |\psi|^2
\end{equation*}
in the region $r_-<r\leq \ered$.
We can now choose $\lambda <0$ large enough in absolute value such that the sum of all the non-underbraced terms on the right hand side of \eqref{EqPropCauchyMultiplier} that are not total derivatives are estimated from above by
\begin{equation}\label{EqPfPropCauchyBulkTermsEst}
- c \chi_n(v_+)\Big[ |\rd_r \psi|^2 +  |\rd_{v_+} \psi|^2 + \sum_i|\zt_{i,+}\psi|^2\Big] + \chi_n(v_+) C |\psi|^2 
\end{equation}
in a region $r_- < r\leq \ered$, where $c>0$, $C>0$ are (new) constants.

\underline{\bf{Step 3: Estimating boundary terms.}}
We now gather all the total derivatives appearing on the right hand side of \eqref{EqPropCauchyMultiplier}. They are $\rd_r(B)$ and $\rd_{v_+}(A)$ with
\begin{equation*}
\begin{split}
A&= \chi_n(v_+) (1 + \lambda \Delta) \Big[ \frac{1}{2} a^2 \sin^2 \theta |\rd_{v_+} \psi|^2 + \frac{a^3 \sin^2 \theta}{r_-^2 +a^2} \Rea(\rd_{v_+} \psi \overline{\rd_{\varphi_+} \psi}) - \frac{1}{2 \sin^2 \theta}\big|is \cos \theta \cdot \psi + \rd_{\varphi_+} \psi\big|^2 \\[3pt] 
&\qquad + \frac{a^2}{r_-^2 + a^2} |\rd_{\varphi_+} \psi|^2 - a^2 \sin^2 \theta \Rea(\rd_{v_+} \psi \overline{\rd_r \psi}) + \big(\frac{a(r^2 + a^2)}{r_-^2 + a^2} - 2a\big) \Rea(\rd_{r} \psi \overline{ \rd_{\varphi_+} \psi}) \\[3pt]
&\qquad - (r^2 + a^2 + \frac{1}{2} \Delta) |\rd_r \psi|^2 - \frac{1}{2} |\rd_\theta \psi|^2 + \frac{1}{2} s |\psi|^2 \Big]
\end{split}
\end{equation*}
and
\begin{equation*}
\begin{split}
B&=\chi_n(v_+)(1 + \lambda \Delta)\Big[\dotuline{\big(\frac{1}{2} a^2 \sin^2 \theta + r^2 + a^2\big) |\rd_{v_+} \psi|^2 + \big(2a + \frac{a(r^2 + a^2)}{r_-^2 + a^2}\big) \Rea(\rd_{\varphi_+} \psi \overline{\rd_{v_+} \psi})} \\[3pt]
&\qquad + \dotuline{\frac{a}{r_-^2 + a^2} \Delta \Rea(\rd_r \psi \overline{\rd_{\varphi_+}\psi}) + \frac{1}{2\sin^2\theta}\big|is \cos \theta \cdot \psi + \rd_{\varphi_+} \psi \big|^2 + \frac{a^2}{r_-^2 + a^2} |\rd_{\varphi_+} \psi|^2} \\[3pt]
&\qquad + \dotuline{\Delta \Rea(\rd_r \psi \overline{ \rd_{v_+} \psi}) - \frac{1}{2} \Delta |\rd_r \psi|^2} + \frac{1}{2}|\rd_\theta \psi|^2 - \frac{1}{2}s |\psi|^2 \Big] \\[3pt]
&\qquad \quad + \chi_n(v_+) \mu e^{\eta r}|\psi|^2 \;,
\end{split}
\end{equation*}
where we have used Lemma \ref{LemStrongerBoundPhiDer}. We begin by establishing the coercivity of $B$. We first only consider the dotted terms and in a procedure already familiar by now  we complete all individual $\rd_{\varphi_+} \psi$ terms into $\frac{1}{\sin \theta} (is \cos\theta \cdot \psi + \rd_{\varphi_+} \psi)$ at the expense of adding error terms.  We treat the arising expression (without the error terms) as a quadratic form in $(\rd_{v_+} \psi, \frac{1}{\sin \theta}(is \cos \theta \cdot \psi + \rd_{\varphi_+} \psi), \sqrt{-\Delta} \rd_r \psi)$ (note the weight in front of the $\rd_r$ derivative) the corresponding matrix of which is easily seen to be
\begin{equation*}
\begin{pmatrix}
\big(\frac{1}{2} a^2 \sin^2 \theta + r^2 + a^2\big) &  \big(a + \frac{a(r^2 + a^2)}{2(r_-^2 + a^2)}\big)\sin \theta & - \frac{1}{2} \sqrt{-\Delta} \\
\big(a + \frac{a(r^2 + a^2)}{2(r_-^2 + a^2)}\big)\sin \theta & \frac{1}{2} + \frac{a^2 \sin^2 \theta}{r_-^2 + a^2} & - \frac{1}{2} \frac{a \sqrt{-\Delta}}{r_-^2 + a^2} \sin \theta \\
- \frac{1}{2} \sqrt{-\Delta} & - \frac{1}{2} \frac{a \sqrt{-\Delta}}{r_-^2 + a^2} \sin \theta & \frac{1}{2}
\end{pmatrix} \;.
\end{equation*}
The positive definiteness of this matrix in a region $r_- < r \leq \ered$ follows easily from noting that the left-upper $2$-$2$ matrix has already been shown (below \eqref{EqCauchyDetQ}) to be positive definite in such a region while the other off-diagonal terms vanish at $r=r_-$. Choosing now $\mu(\eta)$ such that $\mu(\eta) e^{\eta r}$ is large enough we can control all the error terms to obtain
\begin{equation}\label{EqPropCauchyCoerB}
B \gtrsim \chi_n(v_+) \big(|\Delta|\cdot |\rd_r \psi|^2 + |\rd_{v_+} \psi|^2 + \sum_i|\zt_{i,+} \psi|^2 + |\psi|^2 \big)
\end{equation}
in a region $r_- < r \leq \ered$.

We next establish the coercivity of $B-A$. We find
\begin{equation*}
\begin{split}
B-A &= \chi_n(v_+) (1 + \lambda \Delta) \Big[(r^2 + a^2) |\rd_{v_+} \psi|^2 + \Big(2a +\frac{a(r^2 + a^2 \cos^2 \theta)}{r_-^2 + a^2}\Big) \Rea(\rd_{v_+} \psi \overline{\rd_{\varphi_+} \psi})  \\[3pt]
&\qquad + \frac{1}{\sin^2 \theta}\big|is \cos \theta \cdot \psi + \rd_{\varphi_+} \psi\big|^2 + (\Delta + a^2 \sin^2 \theta) \Rea(\rd_{v_+} \psi \overline{\rd_r \psi}) + \big( 2a - \frac{2Mar}{r_-^2 + a^2}\big) \Rea(\rd_r \psi \overline{\rd_{\varphi_+} \psi}) \\[3pt]
&\qquad +(r^2 + a^2) |\rd_r \psi|^2 + |\rd_\theta \psi|^2 - s|\psi|^2\Big] + \chi_n(v_+) \mu e^{\eta r} |\psi|^2 \;.
\end{split}
\end{equation*}
Again, completing the $\rd_{\varphi_+} \psi$ terms to $\frac{1}{\sin \theta}(is \cos \theta \cdot \psi + \rd_{\varphi_+} \psi)$ terms by introducing error terms and considering those terms that are quadratic in $(\rd_{v_+} \psi, \frac{1}{\sin \theta}(is \cos \theta \cdot \psi + \rd_{\varphi_+} \psi), \rd_r \psi)$ as a quadratic form (note that this time we do not include a weight in the $\rd_r$ derivative), we need to establish the positive definiteness of the matrix
\begin{equation*}
Q_2 = \begin{pmatrix}
r^2 + a^2 & \Big(a + \frac{a(r^2 + a^2 \cos^2 \theta)}{2(r_-^2 + a^2)}\Big) \sin \theta & \frac{1}{2}(\Delta + a^2 \sin^2\theta) \\
\Big(a + \frac{a(r^2 + a^2 \cos^2 \theta)}{2(r_-^2 + a^2)}\Big) \sin \theta & 1 & \Big(a - \frac{Mar}{r_-^2 + a^2}\Big) \sin \theta \\ 
\frac{1}{2}(\Delta + a^2 \sin^2\theta) & \Big(a - \frac{Mar}{r_-^2 + a^2}\Big) \sin \theta  & r^2 + a^2
\end{pmatrix} \;.
\end{equation*} 
The first main minor is clearly positive, the second main minor at $r=r_-$ is found to be $$\frac{(r_-^2 + a^2 \cos^2 \theta)^2(r_-^2 + a^2 \cos^2 \theta + 6Mr_-)}{4(r_-^2 + a^2)^2} >0\;,$$ and we compute $$\det Q_2 (r_-) = \frac{4Mr_-(r_-^2 + a^2 \cos^2 \theta)^2(r_-^2 + a^2 \cos^2 \theta + 2Mr_-)}{4(r_-^2 + a^2)}\;.$$ Again, choosing $\mu(\eta) e^{\eta r}$ large enough we control all the error terms and conclude that
\begin{equation}\label{EqPropCauchyCoerAB}
B-A \gtrsim \chi_n(v_+) \big( |\rd_r \psi|^2 + |\rd_{v_+} \psi|^2 + \sum_i|\zt_{i,+} \psi|^2 + |\psi|^2 \big)
\end{equation}
holds for $r_- < r \leq \ered$ for $\ered$ close enough to $r_-$.

Finally, we need to establish the coercivity of $B + \frac{|\Delta|}{r^2 + 2Mr + a^2} A$ for $r$ close enough to $r_-$. We compute
\begin{equation*}
\begin{split}
B - \frac{\Delta}{r^2 + 2Mr + a^2} A &= \chi_n(v_+)(1+ \lambda \Delta) \Big[ 
\big(\frac{1}{2} a^2 \sin^2 \theta + r^2 + a^2 + \mathcal{O}(|\Delta|)\big) |\rd_{v_+} \psi|^2 \\[3pt]
&\qquad + \big(2a + \frac{a(r^2 + a^2)}{r_-^2 + a^2} + \mathcal{O}(|\Delta|)\big) \Rea(\rd_{\varphi_+} \psi \overline{\rd_{v_+} \psi}) \\[3pt]
&\qquad + \Delta\Big(\frac{a}{r_-^2 + a^2} - \frac{1}{r^2 + 2Mr + a^2}\big(\frac{a(r^2 + a^2)}{r_-^2 + a^2} - 2a\big)\Big) \Rea(\rd_r \psi \overline{\rd_{\varphi_+}\psi}) \\[3pt]
&\qquad + \Big(\frac{1}{2}+ \mathcal{O}(|\Delta|)\Big) \frac{1}{\sin^2 \theta} \big|is \cos \theta \cdot \psi + \rd_{\varphi_+} \psi \big|^2  + \Big(\frac{a^2}{r_-^2 + a^2} + \mathcal{O}(|\Delta|)\Big)|\rd_{\varphi_+} \psi|^2 \\[3pt]
&\qquad + \Delta \big(1 + \frac{a^2 \sin^2 \theta}{r^2 + 2Mr + a^2}\big) \Rea(\rd_r \psi \overline{ \rd_{v_+} \psi}) + \frac{\Delta^2}{r^2 + 2Mr + a^2} |\rd_r \psi|^2 \\[3pt]
&\qquad + \Big(\frac{1}{2} + \mathcal{O}(|\Delta|)\Big)|\rd_\theta \psi|^2 - \big(\frac{1}{2}s + \mathcal{O}(|\Delta|) \big) |\psi|^2 \Big] \\[3pt]
&\qquad \quad + \chi_n(v_+) \mu e^{\eta r}|\psi|^2 \;,
\end{split}
\end{equation*}
Again, we complete all isolated $\rd_{\varphi_+} \psi$ terms into $\frac{1}{\sin \theta}( is \cos \theta + \rd_{\varphi_+} \psi)$ by adding error terms and treat the part of the expression that is quadratic in $\{\rd_{v_+} \psi, \frac{1}{\sin \theta} (is \cos \theta \cdot \psi + \rd_{\varphi_+} \psi), \Delta \rd_r \psi\}$ as a quadratic form (note the weight in front of the $\rd_r$ derivative). Its corresponding matrix at $r=r_-$, modulo the factor  $\chi_n(v_+)$, is easily seen to be
\begin{equation*}
Q_3 =\begin{pmatrix}
\frac{1}{2} a^2 \sin^2 \theta + r_-^2 + a^2 & \frac{3}{2} a \sin \theta & \frac{1}{2} + \frac{a^2 \sin^2 \theta}{4(r_-^2 + a^2)} \\
\frac{3}{2} a \sin \theta & \frac{1}{2} + \frac{a^2 \sin^2 \theta}{r_-^2 + a^2} & \frac{3a}{4(r_-^2 +a^2)} \sin \theta \\
\frac{1}{2} + \frac{a^2 \sin^2 \theta}{4(r_-^2 + a^2)}  & \frac{3a}{4(r_-^2 +a^2)} \sin \theta & \frac{1}{2(r_-^2 + a^2)}
\end{pmatrix} \;,
\end{equation*}
where we have used $2Mr_- = r_-^2 + a^2$.
The left upper $2\times 2$ matrix is already known to be positive definite. Moreover, we compute $$\det Q_3 = \frac{(2r_-^2 + a^2 + a^2 \cos^2 \theta)(r_-^2 + a^2 \cos^2 \theta)^2}{16(r_-^2 + a^2)^3} >0  \;.$$
Hence, $Q_3$ is positive definite and after choosing $\mu(\eta) e^{\eta r}$ large enough we obtain
\begin{equation}\label{EqPropCauchyCoerAB2}
B + \frac{|\Delta|}{r^2 + 2Mr + a^2} A \gtrsim \chi_n(v_+) \big( \Delta^2 |\rd_r \psi|^2 + |\rd_{v_+} \psi|^2 + \sum_i |\zt_{i,+} \psi|^2 + |\psi|^2 \big)
\end{equation}
in $r_- < r \leq \ered$ for $\ered$ close enough to $r_-$.

\underline{\bf{Step 4: Estimating the remaining bulk terms.}}
The last two terms in \eqref{EqPropCauchyMultiplier} are estimated by
\begin{equation*}
-\chi_n(v_+) \mu \eta e^{\eta r} |\psi|^2 - 2 \chi_n(v_+) \mu e^{\eta r} \Rea(\overline{\psi} \rd_r \psi) \leq -\frac{1}{2} \chi_n(v_+) \mu \eta e^{\eta r} |\psi|^2 + 2 \chi_n(v_+) \eta^{-1}\mu e^{\eta r} |\rd_r \psi|^2 \;.
\end{equation*}
Choosing now $\eta >0$ large enough and recalling \eqref{EqPfPropCauchyBulkTermsEst} we finally obtain from \eqref{EqPropCauchyMultiplier}
\begin{equation}\label{EqPropCauchyFinalBulk}
\rd_{v_+} (A) + \rd_r(B) \underset{\mathrm{a.i.}}{\gtrsim} \chi_n(v_+)\big(|\rd_r \psi|^2 + |\rd_{v_+} \psi|^2 + \sum_i |\zt_{i,+} \psi|^2 + |\psi|^2\big) 
\end{equation}
in the region $r_- < r \leq \ered$.

\underline{\bf{Step 5: Putting it all together.}} Let $r' \in (r_-, \ered)$. We integrate \eqref{EqPropCauchyFinalBulk} over the region $\{r' \leq r \leq \ered\} \cap \{f^- \leq t_0\} \cap \{f^+ \leq t_0\}$, with $t_0 \gg 1$, with respect to  $ dv_+ \wedge dr \wedge \vols = \frac{1}{\rho^2} \vol$. Moreover, using that on a level set of $f^+$ we have $dr = dv_+$ and on a level set of $f^-$ we have $dr = \frac{\Delta}{r^2 + 2Mr + a^2} dv_+$, we obtain
\begin{equation}\label{EqPropCauchyFinalEneEst}
\begin{split}
\int\limits_{\mathclap{\substack{\{r = r'\} \cap \{f^+ \leq t_0\} \\ \cap \{f^- \leq t_0\}}}} &B \; \vols dv_+ +  \int\limits_{\mathclap{\substack{\{f^- = t_0\}  \\ \cap \{r' \leq r \leq \ered\}}}} (B + \frac{|\Delta|}{r^2 + 2Mr + a^2} A ) \; \vols dv_+ +  \int\limits_{\mathclap{\substack{\{f^+ = t_0\} \\ \cap \{r' \leq r \leq \ered\}}}} (B-A) \; \vols dv_+ \\[3pt]
&\qquad +c \int\limits_{\mathclap{\substack{\{r' \leq r \leq \ered\} \cap \{f^- \leq t_0\} \\ \cap \{f^+ \leq t_0\}}}} \chi_n(v_+)\big(|\rd_r \psi|^2 + |\rd_{v_+} \psi|^2 + \sum_i |\zt_{i,+} \psi|^2 + |\psi|^2\big)  \; \vols dv_+ dr \\[3pt]
&\leq \int\limits_{\mathclap{\substack{\{r = \ered\} \cap \{f^+ \leq t_0\} \\ \cap \{f^- \leq t_0\}}}} B \; \vols dv_+ \;,
\end{split}
\end{equation}
where $c>0$.
Using \eqref{EqPropCauchyCoerB}, \eqref{EqPropCauchyCoerAB}, and \eqref{EqPropCauchyCoerAB2}, the trivial upper bounds on $B$ for the right hand side together with Proposition \ref{PropEnergyEstNoShift}, letting $t_0 \to \infty$ and $r' \to r_-$, we conclude the proof of the proposition.
\end{proof}

\subsubsection{Extension of $\psi$ to the Cauchy horizon $\CH_l$}

\begin{proposition} \label{PropExtCauchyL2}
Under the assumptions of Section \ref{SecAssumptions} the limit
\begin{equation*}
\lim_{r \to r_-} \psi(v_+, \theta, \varphi_+; r) =: \psi(v_+, \theta, \varphi_+;r_-)
\end{equation*}
exists in $L^2(\R \times \Sp^2)$ and satisfies
\begin{equation}\label{EqWeightedPsiCauchy}
\int\limits_{\R \times \Sp^2} \chi(v_+) |\psi(v_+, \theta, \varphi_+;r_-)|^2 \; \vols dv_+  < \infty\;,
\end{equation}
where the function $\chi(v_+)$ is as in Proposition \ref{PropEnCauchyHorizons}.
\end{proposition}

\begin{proof}
For $r_1, r_2 > r_-$ and for $\theta \neq 0, \pi$ by the fundamental theorem of calculus we have
\begin{equation*}
|\psi(v_+, \theta, \varphi_+;r_1) - \psi(v_+, \theta, \varphi_+;r_2)| \leq \int\limits_{[r_1,r_2]} |\rd_r \psi(v_+, \theta, \varphi_+;r)|\;dr \;.
\end{equation*}
Squaring and Cauchy-Schwarz gives
\begin{equation*}
|\psi(v_+, \theta, \varphi_+;r_1) - \psi(v_+, \theta, \varphi_+;r_2)|^2 \leq |r_1 - r_2| \cdot \int\limits_{[r_1, r_2]} |\rd_r \psi (v_+, \theta, \varphi_+; r)|^2 \; dr  \;.
\end{equation*}
Integrating with respect to $\chi(v_+) \vols dv_+$ gives 
\begin{equation}\label{EqConvergeCauchy}
\begin{split}
\int\limits_{\R \times \Sp^2} &|\psi(v_+, \theta, \varphi_+ ; r_1) - \psi(v_+, \theta, \varphi_+ ;r_2)|^2 \chi(v_+) \vols dv_+ \\
&\leq  |r_1 - r_2| \cdot \int\limits_{\R \times \Sp^2} \int\limits_{[r_1, r_2]} |\rd_r \psi (v_+, \theta, \varphi_+; r)|^2 \; \chi(v_+) dr  \vols dv_+  \;.
\end{split}
\end{equation}
Let $L^2_{\chi(v_+)}(\R \times \Sp^2)$ denote the $L^2$ space with respect to the measure $\chi(v_+) \vols dv_+$. By \eqref{PropCauchyEqBoundaryTerms}  we have $\psi(v_+, \theta, \varphi_+;r) \in L^2_{\chi(v_+)}(\R \times \Sp^2)$ for $r$ close enough to $r_-$ and by \eqref{PropCauchyEqILED} we have that the right hand side of \eqref{EqConvergeCauchy} is bounded by $|r_1 - r_2|\cdot C$.
This shows that $\psi(v_+, \theta, \varphi_+;r)$ is Cauchy in $L^2_{\chi(v_+)}(\R \times \Sp^2)$ for $r \to r_-$, from which both claims in the proposition follow.
\end{proof}

\subsubsection{Backwards propagation of the singularity}

\begin{proposition}\label{PropBackwards}
Under the assumptions of Section \ref{SecAssumptions}, and considering the hypersurface $\Sigma := \{f^- = v_0\}$ transversal to $\CH_r$ for some $v_0 \in \R$, there exists a constant $C>0$ such that we have for all $v' \gg 1$ large enough
\begin{equation*}
\Big| \int\limits_{\CH_l \cap \{v_+ \geq v'\}} |\psi(\cdot\; ; r_-)|^2 \; \vols dv_+ - \int\limits_{\Sigma \cap \{v_+ \geq v'\}} |\psi|^2 \;\vols dv_+ \Big| \leq C \cdot e^{\frac{1}{2} \kappa_- v'} \;,
\end{equation*} 
where $\psi(\cdot \; ;r_-)$ is the $L^2$-limit from Proposition \ref{PropExtCauchyL2}.\footnote{Also recall that $\kappa_- <0$.}
\end{proposition}

\begin{proof}
\underline{\textbf{Step 1:}} We recall that $f^-(v_+, r) = - v_+ + 2r^* - r + r_+$. Thus, on $\Sigma = \{f^- = v_0\}$ we have
\begin{equation}\label{EqSigma}
\begin{split}
v_+ &= 2r^* - r + r_+ - v_0\\
& = \frac{1}{\kappa_-} \log (r- r_-) + 2F_-(r) - r + r_+ - v_0 \;,
\end{split}
\end{equation}
where we have used \eqref{EqRStarCauchy} (recall that $F_-(r)$ extends regularly to $r_-$). The right hand side of \eqref{EqSigma} is clearly a strictly decreasing function in $r$ and thus the inverse function exists which we denote by $r_\Sigma$ so to obtain $r_\Sigma(v_+) = r$ on $\Sigma$. It is also immediate that we have $r_\Sigma(v_+) \to r_-$ for $v_+ \to \infty$.

Taking the exponential, we obtain from \eqref{EqSigma}
\begin{equation*}
e^{\kappa_- v_+} = (r-r_-) \cdot G(r)
\end{equation*}
 on $\Sigma$ with $\lim_{r \to r_-} G(r) >0$. Thus, for $v_+ \gg 1$ large enough we have
 \begin{equation}\label{Eqrvplus}
 r_\Sigma(v_+) - r_- \simeq e^{\kappa_- v_+} \;.
 \end{equation}

\underline{\textbf{Step 2:}}
Let now $r' > r_-$ be close to $r_-$. 
 \begin{figure}[h]
\centering
 \def\svgwidth{6cm}
    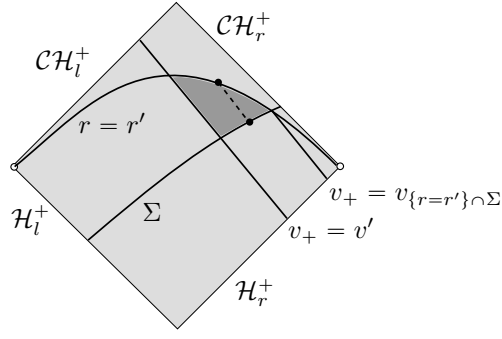
      \caption{The $L^2$-estimate} \label{FigBackwards}
\end{figure}
We now estimate, in a manner similar to the proof of Proposition \ref{PropExtCauchyL2}, as follows (see also Figure \ref{FigBackwards}):
\begin{equation*}
|\psi(v_+, r', \theta, \varphi_+) - \psi(v_+, r_\Sigma(v_+), \theta, \varphi_+)| \leq \int_{r'}^{r_\Sigma(v_+)} |\rd_r \psi(v_+, r, \theta, \varphi_+)|\; dr \;.
\end{equation*}
Squaring and Cauchy-Schwarz gives
\begin{equation*}
|\psi(v_+, r', \theta, \varphi_+) - \psi(v_+, r_\Sigma(v_+), \theta, \varphi_+)|^2 \leq |r_\Sigma(v_+) - r'| \cdot \int_{r'}^{r_\Sigma(v_+)} |\rd_r \psi(v_+, r, \theta, \varphi_+)|^2\; dr \;.
\end{equation*}
Let $v_{\{r = r'\} \cap \Sigma} = 2r^*(r') - r' + r_+ - v_0$ be the value of $v_+$ on $\Sigma$ where $r = r'$. For $v' < v_{\{r = r'\} \cap \Sigma}$ we  integrate to obtain
\begin{equation*}
\begin{split}
\int_{v'}^{v_{\{r = r'\} \cap \Sigma}} \int_{\Sp^2} &|\psi(v_+, r', \theta, \varphi_+) - \psi(v_+, r_\Sigma(v_+), \theta, \varphi_+)|^2 \; \vols dv_+ \\
&\leq |r_\Sigma(v') - r'| \cdot \int_{v'}^{v_{\{r = r'\} \cap \Sigma}} \int_{\Sp^2}  \int_{r'}^{r_\Sigma(v_+)} |\rd_r \psi(v_+, r, \theta, \varphi_+)|^2\; dr \vols dv_+ \;,
\end{split}
\end{equation*}
which gives
\begin{equation}\label{EqGHU}
\begin{split}
\Big| \Big(\int_{v'}^{v_{\{r = r'\} \cap \Sigma}} \int_{\Sp^2} &|\psi(v_+, r', \theta, \varphi_+) |^2 \; \vols dv_+\Big)^{\nicefrac{1}{2}} - \Big(\int_{v'}^{v_{\{r = r'\} \cap \Sigma}} \int_{\Sp^2}| \psi(v_+, r_\Sigma(v_+), \theta, \varphi_+)|^2 \; \vols dv_+ \Big)^{\nicefrac{1}{2}}\Big| \\
&\leq |r_\Sigma(v') - r'|^{\nicefrac{1}{2}} \cdot\Big( \int_{v'}^{v_{\{r = r'\} \cap \Sigma}} \int_{\Sp^2}  \int_{r'}^{r_\Sigma(v_+)} |\rd_r \psi(v_+, r, \theta, \varphi_+)|^2\; dr \vols dv_+\Big)^{\nicefrac{1}{2}} \;.
\end{split}
\end{equation}
We now let $r' \to r_-$ and note that this implies $v_{\{r = r'\} \cap \Sigma} \to \infty$. Moreover, we have
\begin{equation}\label{EqL2Conv}
\begin{split}
\big|\psi(v_+, \theta, \varphi_+; r') &\cdot \mathbbm{1}_{[v', v_{\{r = r'\} \cap \Sigma} ]} (v_+) - \psi(v_+, \theta, \varphi_+; r_-) \cdot\mathbbm{1}_{[v', \infty)}(v_+) \big| \\
&\leq \big|\mathbbm{1}_{[v', v_{\{r = r'\} \cap \Sigma} ]} (v_+) \cdot  \big[\psi(v_+, \theta, \varphi_+;r') - \psi(v_+, \theta, \varphi_+; r_-\big]\big| \\
&\quad + \big|\psi(v_+, \theta, \varphi_+ ; r_-) \cdot \big[\mathbbm{1}_{[v', v_{\{r = r'\} \cap \Sigma} ]} (v_+) - \mathbbm{1}_{[v', \infty)} (v_+)\big] \big| \;,
\end{split}
\end{equation}
where $\mathbbm{1}_A$ denotes the characteristic function of the set $A$. The first summand on the right hand side of \eqref{EqL2Conv} goes to zero in $L^2(\R \times \Sp^2)$ by Proposition \ref{PropExtCauchyL2} while the second goes to zero in $L^2(\R \times \Sp^2)$ by dominated convergence. We thus obtain from \eqref{EqGHU} after $r' \to r_-$ and for $v' \gg 1$ large enough
\begin{equation}\label{EqWERT}
\begin{split}
&\Big| \underbrace{\Big(\int_{v'}^\infty \int_{\Sp^2} |\psi(v_+, \theta, \varphi_+; r_-) |^2 \; \vols dv_+\Big)^{\nicefrac{1}{2}}}_{=:I} - \underbrace{\Big(\int_{v'}^\infty \int_{\Sp^2}| \psi(v_+, r_\Sigma(v_+), \theta, \varphi_+)|^2 \; \vols dv_+ \Big)^{\nicefrac{1}{2}}}_{=:II}\Big| \\
&\qquad \qquad \quad \leq |r_\Sigma(v') - r'|^{\nicefrac{1}{2}} \cdot\Big( \int_{v'}^\infty \int_{\Sp^2}  \int_{r_-}^{r_\Sigma(v_+)} |\rd_r \psi(v_+, r, \theta, \varphi_+)|^2\; dr \vols dv_+\Big)^{\nicefrac{1}{2}} \\
&\qquad \qquad \quad \leq C \cdot e^{\frac{1}{2}\kappa_- v'} \;,
\end{split}
\end{equation}
where we have used \eqref{Eqrvplus} and \eqref{PropCauchyEqILED} in the last step. Since $I$ is finite by Proposition \ref{PropExtCauchyL2}, $II$ is also finite. Multiplying \eqref{EqWERT} by $I + II \leq C$ concludes the proof.
\end{proof}

\begin{lemma}\label{LemPoly}
Let $f,g : [1, \infty) \to [0, \infty)$ be positive, integrable functions that satisfy for $v'$ sufficiently large 
\begin{equation*}
\Big|\int_{v'}^\infty f(v) \; dv - \int_{v'}^\infty g(v) \; dv \Big| \lesssim e^{-\kappa v'} \;,
\end{equation*}
with $\kappa >0$. For $p >0$ we then have
\begin{equation*}
\int_{1}^\infty v^p \cdot f(v) \; dv < \infty \qquad \textnormal{ if, and only if, } \qquad \int_{1}^\infty v^p g(v) \; dv < \infty \;.
\end{equation*}
\end{lemma}

\begin{proof}
Let us assume $\int_{1}^\infty v^p \cdot f(v) \; dv < \infty$. By assumption we have
\begin{equation*}
\Big|\int_{2^n}^\infty \big(f(v) - g(v)\big) \; dv \Big|\leq C e^{-\kappa \cdot 2^n} 
\end{equation*} 
for all $n \in \N$. It follows that
\begin{equation}\label{EqLargeN}
\Big|\int_{2^n}^{2^{n+1}} \big(f(v) - g(v)\big) \; dv\Big| \leq C \big( e^{-\kappa \cdot 2^n}  + e^{-\kappa \cdot 2^{n+1}}\big) \;. 
\end{equation}
We then compute, using \eqref{EqLargeN},
\begin{equation*}
\begin{split}
\int_1^\infty v^p \cdot g(v) \; dv &\leq \sum_{n=0}^\infty \int_{2^n}^{2^{n+1}} (2^{n+1})^p g(v) \; dv \\
&\leq \underbrace{\sum_{n=0}^\infty (2^{n+1})^p \cdot C\big( e^{-\kappa \cdot 2^n}  + e^{-\kappa \cdot 2^{n+1}}\big)}_{< \infty}  +  \sum_{n=0}^\infty \int_{2^n}^{2^{n+1}} (2^{n+1})^p f(v) \; dv \\
&\leq C + 2^p \sum_{n=0}^\infty \int_{2^n}^{2^{n+1}} (2^{n})^p f(v) \; dv \\
&\leq C + 2^p \int_{1}^\infty v^p \cdot f(v) \; dv \;.
\end{split}
\end{equation*}
\end{proof}

Applying the lemma with $f(v_+) = \int_{\Sp^2}|\psi(v_+,\theta, \varphi_+; r_-)|^2 \;\vols$ and $g(v_+) = \int_{\Sp^2} |\psi(v_+, r_\Sigma(v_+), \theta, \varphi_+)|^2 \; \vols$ gives the following
\begin{corollary} \label{CorPropBackwa}
In the setting of Proposition \ref{PropBackwards} we have
\begin{equation*}
 \int\limits_{\CH_l \cap \{v_+ \geq 1\}} v_+^p|\psi(\cdot\; ; r_-)|^2 \; \vols dv_+ < \infty \qquad \textnormal{ if, and only if, } \qquad \int\limits_{\Sigma \cap \{v_+ \geq 1\}} v_+^p|\psi|^2 \;\vols dv_+ < \infty \;,
\end{equation*}
where $p>0$. It follows in particular from Proposition \ref{PropExtCauchyL2} that 
\begin{equation}
\label{EqAddOn}\int\limits_{\Sigma \cap \{v_+ \geq 1\}} \chi(v_+)|\psi|^2 \;\vols dv_+ < \infty \;.
\end{equation}
\end{corollary}

\section{Teukolsky's separation of variables}
\label{SecTeukSep}

In this section we use the upper bounds derived on the Teukosky field in Corollary \ref{CorollaryNoShift} to establish the separation of variables. We begin by a discussion of the spin $2$-weighted spheroidal harmonics, then introduce the Teukolsky transform, and prove a non-trivial result regarding the relation of physical space $v_+$-weights and frequency domain $\omega$-derivatives. We then derive the radial Teukolsky equation belonging to \eqref{TeukolskyStar}.

\subsection{Spin $2$-weighted spheroidal harmonics} \label{SecSpin2Harmonics}

For $\omega \in \R$ and $f \in \SWs $ we define\footnote{This differs from the analogous operator defined in Section 6.2.1 in \cite{DafHolRod17} by an overall minus sign.}
\begin{equation}\label{DefSpinWeightedLaplacianOmega}
\mathring{\slashed\Delta}_{[s]}(\omega) f := \mathring{\slashed\Delta}_{[s]} f  + (a \omega)^2 \cos^2\theta \cdot f - 2sa\omega \cos \theta \cdot f \;.
\end{equation}
Clearly, $\mathring{\slashed\Delta}_{[s]}(\omega)$ maps $\SWs$ into $\SWs$.
\begin{proposition}\label{PropEigenfunctionsSWL}
The operator $\mathring{\slashed\Delta}_{[s]}(\omega) : L^2(\Sp^2) \supseteq \SWs \to \SWs \subseteq L^2(\Sp^2)$ has a complete and orthonormal (with respect to $L^2(\Sp^2)$) set of eigenfunctions $Y^{[s]}_{ml}(\omega) \in \SWs$ indexed by $m \in \Z$, $\N \ni l \geq \max (|m|, |s|)$ and eigenvalues $\lambda_{ml}^{[s]}(\omega) \in \R$ satisfying
\begin{equation}
\label{EqEVSWSH}
\mathring{\slashed\Delta}_{[s]}(\omega) Y^{[s]}_{ml}(\omega) = \lambda_{ml}^{[s]}(\omega) Y^{[s]}_{ml}(\omega)  \;.
\end{equation}
The eigenfunctions are known as \emph{spin $2$-weighted spheroidal harmonics} and are of the form $$Y^{[s]}_{ml}(\theta, \varphi ; \omega) = S^{[s]}_{ml}( \cos \theta ; \omega) e^{im \varphi} \;, $$
where the $S^{[s]}_{ml}(\cos \theta; \omega)$ form a complete and orthonormal (with respect to $L^2([-1,1], d \cos \theta)$) set of eigenfunctions of the operator
$$L^{[s]}_{m}(\omega) S := \frac{1}{\sin \theta} \rd_\theta (\sin \theta \rd_\theta S) - \frac{m^2}{\sin^2 \theta} S - 2sm \frac{\cos \theta}{\sin^2 \theta} S - (s^2 \frac{\cos^2 \theta}{\sin^2 \theta} - s) S + (a \omega)^2 \cos^2 \theta \cdot S - 2sa\omega \cos \theta \cdot S $$
with eigenvalues $\lambda^{[s]}_{ml}(\omega)$:\footnote{This is the same equation as (4.10) in \cite{Teuk73} with $A = - \lambda_{ml}^{[s]}(\omega)$.}
\begin{equation}
\label{EqEFXEq}
L^{[s]}_{m}(\omega) S^{[s]}_{ml}(\omega) = \lambda_{ml}^{[s]}(\omega) S^{[s]}_{ml}(\omega)  \;.
\end{equation}
The eigenvalues $\lambda^{[s]}_{ml}(\omega)$ depend analytically on $\omega$ and the eigenfunctions $S^{[s]}_{ml}(\cos \theta; \omega)$ are analytic in $\omega$ and in $x= \cos \theta$ away from $x= \pm 1$.   Near $x = \pm 1$ we have the the following asymptotic expansions for all $k \in \N$: Near $x = -1$ we have
$$\rd_\omega^k S_{ml}^{[s]}(x ; \omega) = (1 + x)^{\frac{1}{2}|m-s|} a_k(x; \omega)\;,$$
where $a_{k}(x; \omega)$ is analytic in both arguments near $x= -1$; and  near $x=+1$ 
$$\rd_\omega^k S_{ml}^{[s]}(x ; \omega)  = (x - 1)^{\frac{1}{2}|m+s|} b_k(x; \omega)$$ is valid with $b_k(x; \omega)$ being analytic in both arguments near $x = +1$.

Moreover, we have $\lambda_{ml}^{[s]}(\omega) - s = \lambda^{[-s]}_{ml}(\omega) + s$ and for $\omega = 0$ the eigenvalues are given by $\lambda_{ml}^{[s]}(0) = -(l -s)(l+s+1) = -l(l+1) + s(s+1)$.
\end{proposition}

\begin{proof}
The result is standard, see for example \cite{Teuk73} or \cite{DafHolRod17}, although we do not know a reference that includes a proof. We will thus give an outline of the proof here.

Making the separation of variables ansatz $Y_{m}(\theta, \varphi) = S_m(\theta) e^{im \varphi}$ we obtain $$\mathring{\slashed\Delta}_{[s]}(\omega) Y_m(\theta, \varphi) = \big(L_m^{[s]}(\omega) S_m(\theta)\big) \cdot e^{im \varphi} \;.$$ We will now find an orthonormal basis of eigenfunctions for $L_m^{[s]}(\omega)$ using Sturm-Liouville theory. The substitution $x = \cos \theta$ yields
\begin{equation}\label{EqXX}
L_m^{[s]}S_m = \frac{d}{dx} \big( (1-x^2) \frac{d}{dx} S_m\big)  - \frac{(m+sx)^2}{1 - x^2} S_m + \big( s + (a\omega)^2 x^2 - 2sa\omega x\big) S_m \;.
\end{equation}
We now go over to $L_m^{[s]} - \lambda$ for $\lambda \in \C$. The points $x = \pm 1$ are regular singular points of the second order differential operator and, moreover, it depends analytically on $\omega$ and $\lambda$, even for complex $\omega$. The Frobenius method, see for example \cite{Teschl}, shows that there is a fundamental system of solutions of $(L_m^{[s]} - \lambda)S_m = 0$, normalised at $x =-1$, of the form 
\begin{equation*}
\begin{split}
u_1(x; \lambda, \omega) &= (1+x)^{\frac{1}{2}|m - s|} h_1(x; \lambda, \omega) \\
u_2(x ; \lambda, \omega) &= (1+x)^{-\frac{1}{2}|m-s|} h_2(x; \lambda, \omega) + c \log(1 +x) u_1(x ; \lambda, \omega) \;,
\end{split}
\end{equation*}
where $h_1$ and $h_2$ are analytic in $[-1,1) \times \R \times \R$ and the constant $c$ might be zero unless $m =s$. Similarly, there is a fundamental system of solutions normalised at $x=+1$:
\begin{equation*}
\begin{split}
v_1(x; \lambda, \omega) &= (x-1)^{\frac{1}{2}|m+s|} g_1(x; \lambda, \omega) \\
v_2(x ; \lambda, \omega) &= (x-1)^{-\frac{1}{2}|m+s|} g_2(x; \lambda, \omega) + c \log (x-1) v_1(x ; \lambda, \omega) \;,
\end{split}
\end{equation*}
where $g_1$ and $g_2$ are analytic in $(-1,1] \times \R \times \R$ and the constant $c$ might be zero unless $m = -s$. Note that $u_1$ is regular at $x = -1$ while $v_1$ is regular at $x =+1$. For $\lambda = \lambda_0 >0$ large enough one can show that $u_1$ and $v_1$ are linearly independent. Using this pair of solutions one constructs the Green's function in the same way as for a regular Sturm-Liouville problem, c.f.\ \cite{Teschl}. The above asymptotics imply that the Green's function is in $L^2([-1,1] \times [-1,1])$, and thus the solution operator $K_{\lambda_0}$ is a symmetric and compact operator on $L^2([-1,1])$. It is easy to show that the kernel vanishes and thus, by the spectral theorem, there is an orthonormal basis of eigenfunctions $S^{[s]}_{ml}(x)$ of $K_{\lambda_0}$ with real eigenvalues $\mu^{[s]}_{ml}$. Using the asymptotics of $u_1$ and $v_1$ in the Green's function one shows that $S^{[s]}_{ml}$ are continuous at $x = \pm 1$. Moreover, they satisfy $(L_m^{[s]} - \lambda_0) (\mu^{[s]}_{ml}S_{ml}^{[s]} ) = S^{[s]}_{ml}$ and thus $$L_m^{[s]} S^{[s]}_{ml} - (\underbrace{\lambda_0 + \frac{1}{\mu^{[s]}_{ml}}}_{=: \lambda^{[s]}_{ml}}) S^{[s]}_{ml} = 0 \;.$$ 
It follows that $S^{[s]}_{ml} \sim u_1 \sim v_1$ -- and thus in particular that the eigenvalues are simple.

To show the analytic dependence of the eigenvalues on $\omega$ we notice that they are exactly the zeros of the modified Wronskian $W(u_1, v_1)(\lambda, \omega) := u_1(x; \lambda, \omega) (1-x^2) v_1'(x; \lambda, \omega) - (1-x^2) u_1'(x; \lambda, \omega) v_1(x ; \lambda, \omega)$. One now shows that the zeros of the Wronskian in $\lambda$ are simple, i.e., $\partial_\lambda W(u_1, v_1) (\lambda^{[s]}_{ml}, \omega) \neq 0$. The analytic implicit function theorem then yields that the eigenvalues $\lambda^{[s]}(\omega)$ depend analytically on $\omega$. It follows that $u_1(x; \lambda^{[s]}_{ml}(\omega), \omega)$ depends analytically on $\omega$. Normalising it in $L^2([-1,1])$ then gives $S^{[s]}_{ml}(x; \omega)$, which shows in particular the regularity claimed in the proposition. 

It is straightforward to show that $Y^{[s]}_{lm}(\omega)$ is an orthonormal basis of $L^2(\Sp^2)$. To show $Y^{[s]}_{lm}(\omega) \in \SWs$, we can use the asymptotics of $S^{[s]}_{ml}(\omega)$ given by the Frobenius solutions above and tediously verify the conditions in Proposition \ref{PropFirstCharSpinWeighted}. Alternatively, and more elegantly, we can multiply  \eqref{EqEVSWSH} by $\overline{Y^{[s]}_{ml}(\omega)}$, integrate over the sphere and check that the asymptotics of $S_{ml}^{[s]}(\omega)$ allow us to do one integration by parts to conclude that $Y^{[s]}_{ml}(\omega) \in H^1_{[s]}(\Sp^2)$. We now go over the corresponding trace-free and symmetric $2$-covariant tensor field $\alpha_{ml}(\omega)$ on $\Sp^2$ which is smooth except possibly at the poles of the sphere.
Using \eqref{EqContDepSpin} we now rewrite \eqref{EqEVSWSH} as a standard elliptic equation for $\alpha_{ml}(\omega)$. It now follows from standard elliptic regularity theory that $\alpha_{ml}(\omega)$ is smooth on all of the sphere -- showing the claim.

The relation $\lambda_{ml}^{[s]}(\omega) - s = \lambda^{[-s]}_{ml}(\omega) + s$ follows from the substitution $x \to - x$ in \eqref{EqXX}. Finally, we refer the reader to \cite{GoMaNe67} for the evaluation of the eigenvalues at $\omega = 0$.
\end{proof}

The following quantitative result on the $\omega$-dependence of the eigenfunctions $Y^{[s]}_{ml}(\omega)$ is needed for the proof of Proposition \ref{PropCharSlowDecay}.

\begin{proposition} \label{PropEstimatesDerivativesEigenfunctions}
By Proposition \ref{PropEigenfunctionsSWL} we know that $\partial_\omega^k S_{ml}^{[s]}( \omega) \in L^2([-1,1])$ for all $k \in \N$. We can thus expand in $L^2([-1,1])$ $$\rd_\omega^k S_{ml}^{[s]}(\omega) = \sum_{l' \geq \max(|m|, |s|)} D^{[s]}_{mll';k} (\omega) S^{[s]}_{ml'}(\omega) $$
with $D_{mll';k}^{[s]}(\omega) \in \R$.

There exists $\varepsilon >0$ such that for $|\omega| \leq \varepsilon$ we have
\begin{equation}
\label{EqPropLittleL2}
\sum_{l,l' \geq \max(|m|, |s|)} |D^{[s]}_{mll';k}(\omega)|^2 \leq C(k) < \infty \;,
\end{equation}
where the constant is independent of $m$ and $|\omega| \leq \varepsilon$.
\end{proposition}

Let us remark that \eqref{EqPropLittleL2} is equivalent to \begin{equation}
\label{EqEquivL2Bound}
\sum_{l \geq \mms} ||\rd_\omega^k \Sml(\omega)||_{L^2}^2 \leq C(k) < \infty \;.
\end{equation} 

\begin{proof}
Differentiating \eqref{EqEFXEq} in $\omega$ gives
\begin{equation}\label{EqDiffEq}
L_m^{[s]}(\omega) \rd_\omega S_{ml}^{[s]}(\omega) + 2a(a \omega x^2 - sx)\Sml(\omega) = \rd_\omega \lambda^{[s]}_{ml}(\omega) \cdot \Sml(\omega) + \lambda_{ml}^{[s]} (\omega)\cdot \rd_\omega \Sml(\omega) \;.
\end{equation}
Note that since $||\Sml(\omega)||_{L^2} = 1$ for all $\omega$, we have $\langle \Sml(\omega), \rd_\omega \Sml(\omega)\rangle_{L^2} = 0$. Multiplying \eqref{EqDiffEq} by $\Sml(\omega)$ and integrating gives
$$\langle L_m^{[s]}(\omega) \rd_\omega \Sml(\omega), \Sml(\omega) \rangle_{L^2} + \langle 2a(a \omega x^2 - sx) \Sml(\omega), \Sml(\omega)\rangle_{L^2} = \rd_\omega \lambda^{[s]}_{ml}(\omega) \;. $$
We now integrate by parts in the first term\footnote{Note that the arising boundary terms are $$\Big[(1 - x^2) \frac{d}{dx} \rd_\omega \Sml(\omega) \cdot \Sml(\omega)\Big]^1_{-1} - \Big[ \rd_\omega \Sml(\omega) \cdot (1- x^2) \frac{d}{dx} \Sml(\omega)\Big]^1_{-1}\;,$$ which vanishes given the asymptotics of $\rd_\omega^k \Sml(\omega)$ from Proposition \ref{PropEigenfunctionsSWL}.}
to obtain $$\langle L_m^{[s]}(\omega) \rd_\omega \Sml(\omega), \Sml(\omega) \rangle_{L^2} = \langle \rd_\omega \Sml(\omega), L_m^{[s]}(\omega) \Sml(\omega) \rangle_{L^2} = \langle\rd_\omega \Sml(\omega), \lambda_{ml}^{[s]}(\omega) \Sml(\omega) \rangle_{L^2} = 0 \;.$$
This finally leaves us with 
\begin{equation}
\label{EqOmegaDerLambda}
\rd_\omega \lambda^{[s]}_{ml}(\omega) = \langle 2a(a \omega x^2 - sx) \Sml(\omega), \Sml(\omega) \rangle_{L^2([-1,1])} \;.
\end{equation}
Multiplying \eqref{EqDiffEq} by $\Smlp(\omega)$, $l \neq l'$, and integrating over $[-1,1]$ in $x$ gives, after the integration by parts as before, 
\begin{equation*}
\langle\rd_\omega \Sml(\omega), \lambda^{[s]}_{ml'}(\omega) \Smlp(\omega) \rangle_{L^2} + \langle 2a(a \omega x^2 - sx) \Sml(\omega), \Smlp(\omega) \rangle_{L^2} = \lambda_{ml}^{[s]}(\omega) \langle \rd_\omega \Sml(\omega), \Smlp(\omega) \rangle_{L^2} \;.
\end{equation*}
We thus obtain for $l \neq l'$
\begin{equation}
\label{EqD1}
\begin{split}
D^{[s]}_{mll';1}(\omega) &= \langle \rd_\omega \Sml(\omega), \Smlp(\omega) \rangle_{L^2} \\
&= \frac{\langle 2a(a \omega x^2 - sx) \Sml(\omega), \Smlp(\omega) \rangle_{L^2}}{\lambda^{[s]}_{ml}(\omega) - \lambda^{[s]}_{ml'}(\omega)} \;.
\end{split}
\end{equation}
For $l = l'$ we have $D^{[s]}_{mll;1}(\omega) = 0$. To derive an expression for $D^{[s]}_{mll';k}$ we first note that by the regularity of $\rd_\omega^i\Sml(\omega)$ from Proposition \ref{PropEigenfunctionsSWL} we have that $\langle\rd_\omega \Sml(\omega), \Smlp(\omega)\rangle_{L^2}$ is smooth in $\omega$. We now compute
\begin{equation}
\label{EqExprHigherD}
\begin{split}
\rd_\omega^{(k-1)} D^{[s]}_{mll';1}(\omega) &= \rd_\omega^{(k-1)} \langle \rd_\omega \Sml(\omega), \Smlp\rangle_{L^2} \\
&=\sum_{n= 0}^{k-1} \binom{k-1}{n} \langle \rd_\omega^{(k-n)} \Sml(\omega), \rd_\omega^n \Smlp(\omega)\rangle_{L^2} \\
&=\langle\rd_\omega^k \Sml(\omega), \Smlp \rangle_{L^2} + \sum_{n = 1}^{k-1} \binom{k-1}{n} \langle\rd_\omega^{(k-n)}\Sml(\omega), \rd_\omega^n \Smlp(\omega)\rangle_{L^2} \\
&= D^{[s]}_{mll';k}(\omega) + \sum_{n=1}^{k-1} \binom{k-1}{n} \langle \sum_i D^{[s]}_{mli;k-n}(\omega) S^{[s]}_{mi}(\omega), \sum_j D^{[s]}_{ml'j;n} (\omega) S^{[s]}_{mj}(\omega) \rangle_{L^2} \\
&= D^{[s]}_{mll';k}(\omega) + \sum_{n = 1}^{k-1} \binom{k-1}{n} \underbrace{\sum_{i \geq \mms} D^{[s]}_{mli;k-n}(\omega) \cdot D^{[s]}_{ml'i;n} (\omega)}\;.
\end{split}
\end{equation}
We also need to estimate the eigenvalues for small $|\omega|$: for $|\omega| \leq 1$ it follows directly from \eqref{EqOmegaDerLambda} that
$$|\rd_\omega \lambda^{[s]}_{ml}(\omega)| \leq 2a(a + |s|) \;.$$
We now choose $1>\varepsilon>0$ such that for $|\omega| \leq \varepsilon$ we have
\begin{equation}
\label{EqEVSmallOmega}
|\lambda^{[s]}_{ml}(\omega) - \lambda^{[s]}_{ml}(0)|\leq \frac{1}{4} 
\end{equation}
uniformly in $m$ and $l$.

We now prove \eqref{EqPropLittleL2} by induction in $k$. We start with $k = 1$ and estimate \eqref{EqD1}. We have $|\langle 2a (a \omega x^2 - sx) \Sml(\omega), \Smlp(\omega)\rangle_{L^2}| \leq 2a(a \varepsilon +|s|)$. We now estimate the denominator using \eqref{EqEVSmallOmega} and $\lambda^{[s]}_{ml}(0) = -l(l+1) + s(s+1)$:
\begin{equation}
\label{EqSumLittleL2}
\begin{split}
\sum_{\substack{l, l' \geq \mms \\ l \neq l'}} \frac{1}{|\lambda^{[s]}_{ml}(\omega) - \lambda^{[s]}_{ml'}(\omega)|^2} &\leq \sum_{\substack{l, l' \geq \mms \\ l \neq l'}} \frac{1}{(|-l(l+1) + l'(l'+1)| - \frac{1}{2})^2} \\
&= \sum_{l \geq \mms } \sum_{\substack{k \in \Z \setminus\{0\} \\ k \geq -l + \mms}} \frac{1}{(|k(k+1 + 2l)| - \frac{1}{2})^2} \qquad \textnormal{with $l' = l+k$} \\
&\leq \sum_{l \geq \mms} \sum_{\substack{k \in \Z \setminus\{0\} \\ k \geq -l + \mms}} \frac{1}{|k|^2|k+ 2l|^2} \\
&\leq  \sum_{l \geq \mms} \sum_{\substack{k \in \Z \setminus\{0\} \\ k \geq -l + \mms}} \frac{1}{k^2 l^2} \\
&\leq \sum_{l \in \N} \frac{1}{l^2}\sum_{k \in \Z \setminus\{0\} } \frac{1}{k^2} \\
&= \frac{\pi^4}{18} \;.
\end{split}
\end{equation}
This proves the claim for $k = 1$. 

We now assume that \eqref{EqPropLittleL2} holds up to and including $k-1$. We first show that for $1 \leq j \leq k-1$ and $|\omega| \leq \varepsilon$ we have
\begin{equation}
\label{EqBoundDerLambda}
|\rd_\omega^j \lambda^{[s]}_{ml}(\omega)| \leq C(j) < \infty \;,
\end{equation}
where the constant is independent of $m,l$.
Let $f^{[s]}(x; \omega) := 2a(a \omega x^2 - sx)$. Thus $\rd_\omega \lambda^{[s]}_{ml}(\omega) = \langle f^{[s]}(\omega) \Sml(\omega), \Sml(\omega) \rangle_{L^2}$ and thus
\begin{equation} \label{EqAsBefore}
|\rd_\omega^j \lambda^{[s]}_{ml} (\omega)| = |\sum_{1 \leq i_1 + i_2 + i_3 \leq j-1} \binom{j-1}{i_1,i_2,i_3} \langle \rd_\omega^{i_1} f^{[s]}(\omega) \rd_\omega^{i_2} \Sml(\omega), \rd_\omega^{i_3} \Sml(\omega)\rangle_{L^2}| \;.
\end{equation}
Clearly, $\rd_\omega^{i_1} f^{[s]}(\omega)$ is bounded in $L^\infty_x([-1,1])$ by a constant only depending on $s$ for $|\omega| \leq \varepsilon$, and $\rd^i_\omega \Sml(\omega)$ is bounded in $L^2_x([-1,1])$ by a constant independent of $m,l$ by the induction hypothesis and \eqref{EqEquivL2Bound}. This shows \eqref{EqBoundDerLambda}.

We now use \eqref{EqExprHigherD} to show that $D^{[s]}_{mll';k}(\omega)$ is bounded in $\ell^2(l,l')$.
The induction hypothesis shows directly that the $\ell^2$-norm in $l,l'$ of the underbraced terms in \eqref{EqExprHigherD} is bounded. It thus remains to show that the $\ell^2$-norm of the left hand side of \eqref{EqExprHigherD} is bounded. Note that for $l = l'$ it vanishes identically. For $l \neq l'$ we compute using \eqref{EqD1}
\begin{equation*}
\begin{split}
|\rd_\omega^{(k-1)} D_{mll';1}^{[s]}(\omega)| = | \sum_{1 \leq i_1 + i_2 + i_3 + i_4 \leq k-1} \binom{k-1}{i_1,i_2,i_3,i_4} &\underbrace{\langle \rd_\omega^{i_1}f^{[s]}(\omega) \rd_\omega^{i_2} \Sml(\omega), \rd_\omega^{i_3} \Smlp (\omega) \rangle_{L^2}} \\
&\cdot \rd_\omega^{i_4}\Big( \frac{1}{\lambda^{[s]}_{ml}(\omega) - \lambda^{[s]}_{ml'}(\omega)}\Big)| \;.
\end{split}
\end{equation*}
The underbraced terms are bounded uniformly in $m,l,l'$ and $|\omega| \leq \varepsilon$ as in \eqref{EqAsBefore}. Using \eqref{EqBoundDerLambda}, we can bound $$|\rd_\omega^{i_4}\Big( \frac{1}{\lambda^{[s]}_{ml}(\omega) - \lambda^{[s]}_{ml'}(\omega)}\Big)| \leq \frac{C}{|\lambda^{[s]}_{ml}(\omega) - \lambda^{[s]}_{ml'}(\omega)|}\;,$$ where the constant is independent of $m,l,l'$. It now follows from \eqref{EqSumLittleL2} that the $\ell^2(l,l')$ norm of  the left hand side of \eqref{EqExprHigherD} is bounded. This concludes the proof.
\end{proof}

\subsection{Teukolsky's expansion} \label{SecTeukExp}

For $f(v_+, \theta, \varphi_+) \in L^1_{v_+} L^2_{\Sp^2}$ we define the Fourier transform $\widecheck{f}$ of $f$ by
\begin{equation}\label{EqDefFT}
\widecheck{f}(\theta, \varphi_+; \omega) := \frac{1}{\sqrt{2\pi}} \int\limits_\R f(v_+, \theta, \varphi_+) e^{i \omega v_+} \, dv_+ \;.
\end{equation}
It can be easily checked that this is a map $\widecheck{( \cdot)} : L^1_{v_+} L^2_{\Sp^2} \to C^0_\omega L^2_{\Sp^2}$. It gives rise in the standard way to an isometry $\widecheck{( \cdot)} : L^2_{v_+} L^2_{\Sp^2} \to L^2_\omega L^2_{\Sp^2}$ which we denote again in the same way.

For $g \in L^2_\omega L^2_{\Sp^2}$ we define the map $(\cdot)_{ml} : L^2_\omega L^2_{\Sp^2} \to L^2_\omega \ell^2_{m,l}$ by
\begin{equation}\label{EqDefDecompositionSphere}
g_{ml}(\omega) := \int_{\Sp^2} g(\theta, \varphi_+; \omega) \overline{Y^{[s]}_{ml}(\theta, \varphi_+; \omega)} \, \vols \;,
\end{equation}
which is also an isometry since for each $\omega \in \R$ the $Y^{[s]}_{ml}(\omega)$ form an orthonormal basis of $L^2(\Sp^2)$. The summation in $\ell^2_{m,l}$ is over $m \in \Z$ and $\N \ni l \geq \max(|m|, |s|)$.

For $f \in L^1_{v_+} L^2_{\Sp^2} \cap L^2_{v_+} L^2_{\Sp^2}$ the composite map $\widecheck{(\cdot)}_{ml} := ( \cdot)_{ml} \circ \widecheck{(\cdot)}$, which we call the \emph{Teukolsky transform}, is given by
\begin{equation}\label{EqDefCompositeMap}
\widecheck{f}_{ml} (\omega) = \frac{1}{\sqrt{2\pi}} \int_{\Sp^2} \int_\R f(v_+, \theta, \varphi_+) e^{i \omega v_+} \overline{Y^{[s]}_{ml}(\theta, \varphi_+; \omega)} dv_+ \vols \;.
\end{equation}
Note that by $$\int_\R \int_{\Sp^2} |f(v_+, \theta, \varphi_+)| \cdot | \overline{Y^{[s]}_{ml}(\theta, \varphi_+; \omega)}| \vols dv_+ \leq \int_\R \Big( \int_{\Sp^2} | f(v_+, \theta, \varphi_+)|^2 \, \vols \Big)^\frac{1}{2} dv_+ < \infty$$ the order of integration in \eqref{EqDefCompositeMap} does not matter.

The inverse map of \eqref{EqDefDecompositionSphere} is given by
\begin{equation*}
g(\theta, \varphi_+; \omega) = \sum_{m,l} g_{ml}(\omega) Y^{[s]}_{ml}(\theta, \varphi_+ ; \omega)
\end{equation*}
and the inverse map of $\widecheck{(\cdot)} : L^2_{v_+} L^2_{\Sp^2} \to L^2_\omega L^2_{\Sp^2}$ is given by
\begin{equation*}
f(v_+, \theta, \varphi_+) =  \frac{1}{\sqrt{2\pi}}\int_\R \widecheck{f}(\theta, \varphi_+; \omega)  e^{-i \omega v_+}\, d\omega \;,
\end{equation*}
where this can be taken literally for $\widecheck{f} \in L^1_\omega L^2_{\Sp^2} \cap L^2_\omega L^2_{\Sp^2}$ and serves as notation for $\widecheck{f} \in L^2_\omega L^2_{\Sp^2}$ in the standard way, which is then defined via approximation by functions in $L^1_\omega L^2_{\Sp^2} \cap  L^2_\omega L^2_{\Sp^2}$. In particular for $\widecheck{f}_{ml} \in L^1_\omega \ell^2_{m,l} \cap L^2_\omega \ell^2_{m,l}$ we have $\widecheck{f} \in L^1_\omega L^2_{\Sp^2} \cap L^2_\omega L^2_{\Sp^2}$ and thus we have literally
\begin{equation*}
f(v_+, \theta, \varphi_+) = \frac{1}{\sqrt{2 \pi}} \int_\R \sum_{m,l} \widecheck{f}_{ml}(\omega) Y^{[s]}_{ml}(\theta, \varphi_+ ; \omega) e^{-i \omega v_+} \, d \omega 
\end{equation*}
as a map $L^1_\omega \ell^2_{m,l} \cap L^2_\omega \ell^2_{m,l} \to L^2_{v_+} L^2_{\Sp^2}$.

For $f \in L^2_{v_+} L^2_{\Sp^2}$ we have the Plancherel relation
\begin{equation}\label{EqPlancherel}
\begin{split}
\int_\R \int_{\Sp^2} |f(v_+, \theta, \varphi_+)|^2 \, \vols dv_+ &= ||f||^2_{L^2_{v_+} L^2_{\Sp^2}} \\
&= ||\widecheck{f}||^2_{L^2_\omega L^2_{\Sp^2}} \\ 
&= ||\widecheck{f}_{ml}||^2_{L^2_\omega \ell^2_{ml}} = \int_\R \sum_{m,l} |\widecheck{f}_{ml}|^2 \, d\omega \;.
\end{split}
\end{equation}

We also have\footnote{Proof as in \cite{LiebLoss} 7.9 Theorem.}
\begin{equation}\label{EqRelationsFTWeightDerivative}
\begin{aligned}
\widecheck{\rd_{v_+} f} &= - i \omega \widecheck{f} \qquad &&\textnormal{ in } L^2_\omega L^2_{\Sp^2} \quad \textnormal{ if } f, \rd_{v_+} f \in L^2_{v_+} L^2_{\Sp^2} \\
\widecheck{ v_+ f} &= -i \rd_\omega \widecheck{f} \qquad &&\textnormal{ in } L^2_\omega L^2_{\Sp^2} \quad \textnormal{ if } f, v_+ f \in L^2_{v_+} L^2_{\Sp^2}
\end{aligned}
\end{equation}
and
\begin{equation*}
(\widecheck{\rd_{v_+} f})_{ml} = -i \omega \widecheck{f}_{ml} \qquad \textnormal{ in } L^2_\omega \ell^2_{m,l} \quad \textnormal{ if } f, \rd_{v_+} f \in L^2_{v_+} L^2_{\Sp^2}  \;.
\end{equation*}
Note, however, that in general $(\widecheck{v_+ f})_{ml} \neq -i \rd_\omega \widecheck{f}_{ml}$, since the orthonormal basis functions $Y^{[s]}_{ml}$ of $L^2(\Sp^2)$ in \eqref{EqDefDecompositionSphere} are $\omega$-dependent.

In the following we address this point and show that under suitable assumptions we can still infer limited decay of $f(v_+, \theta, \varphi_+)$ for $|v_+| \to \infty$ from limited regularity of $\widecheck{f}_{ml}$ in $\omega$.

\subsubsection{Slow decay in $v_+$ of $f$ in terms of limited regularity of $\widecheck{f}_{ml}$}

\begin{proposition}
\label{PropCharSlowDecay}
Let $\varepsilon >0$ be as in Proposition \ref{PropEstimatesDerivativesEigenfunctions}, let $\widecheck{f} \in L^2_{(-\varepsilon, \varepsilon)} L^2_{\Sp^2}$ and let $q_0 \in \N_0$. Then $\rd_\omega^q \widecheck{f} \in L^2_{(-\varepsilon, \varepsilon)} L^2_{\Sp^2}$, i.e., 
\begin{equation}
\label{EqHJW1}
\int_{(-\varepsilon, \varepsilon)} \int_{\Sp^2} | \rd_\omega^q \widecheck{f}(\omega, \theta, \varphi_+)|^2 \, \vols d\omega < \infty 
\end{equation} 
for all $0 \leq q \leq q_0$, $q \in \N_0$ if, and only if, $\rd_\omega^q (\widecheck{f}_{ml}) \in L^2_{(-\varepsilon, \varepsilon)} \ell^2_{ml}$, i.e., 
\begin{equation}
\label{EqHJW2}
\int_{(-\varepsilon, \varepsilon)} \sum_{m,l} | \rd_\omega^q ( \widecheck{f}_{ml})(\omega)|^2 \, d \omega < \infty
\end{equation} 
for all $0 \leq q \leq q_0$, $q \in \N$.
Here, all derivatives are weak derivatives.\footnote{For this paper only the `only if' direction, i.e., `$\implies$', is needed.}
\end{proposition}

\begin{proof}
Assume first that $\widecheck{f}$ has $q_0$ weak $\omega$-derivatives in $L^2_{(-\varepsilon, \varepsilon)}L^2_{\Sp^2}$. We then have for $0 \leq q \leq q_0$
\begin{equation}\label{EqRelateBothWeakDerivatives}
\begin{split}
(\rd_\omega^{q} \widecheck{f})_{ml} &= \int_{\Sp^2} \rd_\omega^{q} \widecheck{f}(\theta, \varphi_+; \omega) \cdot \overline{Y^{[s]}_{ml}(\theta, \varphi_+; \omega) } \, \vols \\
&= \int_{\Sp^2} \rd_\omega^{q} \big( \widecheck{f}(\theta, \varphi_+ ; \omega) \overline{Y^{[s]}_{ml}(\theta, \varphi_+ ; \omega)} \big) \, \vols - \sum_{q' = 1}^{q} \binom{q}{q'} \int_{\Sp^2} \rd_\omega^{q - q'} \widecheck{f}(\theta, \varphi_+ ; \omega) \cdot \overline{\rd_\omega^{q'} Y^{[s]}_{ml}(\theta, \varphi_+; \omega) } \, \vols \\
&= \rd_\omega^{q} \widecheck{f}_{ml}(\omega) - \sum_{q' = 1}^{q} \binom{q}{q'} \int_{\Sp^2} \rd_\omega^{q - q'} \widecheck{f}(\theta, \varphi_+ ; \omega) \cdot \sum_{l'} \overline{D^{[s]}_{mll';q'}(\omega) Y^{[s]}_{ml'}(\theta, \varphi_+;  \omega) } \, \vols \\
&= \rd_\omega^{q} \widecheck{f}_{ml}(\omega) - \sum_{q' = 1}^{q} \binom{q}{q'} \sum_{l'} \overline{D^{[s]}_{mll'; q'}(\omega)} (\rd_\omega^{q - q'} \widecheck{f})_{ml'}(\omega) \;,
\end{split}
\end{equation}
where, in the second equality we have used the smoothness of the $Y^{[s]}_{ml}$ in $\omega$ and the product rule $$\rd_\omega^{q}(a \cdot b) = \sum_{q' = 0}^{q} \binom{q}{q'} \rd_\omega^{q - q'} a \cdot \rd_\omega^{q'} b $$ which of course also holds for weak derivatives if $b$ is smooth, in the third equality we have used that we can pull out weak derivatives from under the integral\footnote{Let $g(\theta, \varphi_+; \omega), \rd_\omega g(\theta, \varphi_+; \omega) \in L^2_\omega L^1_{\Sp^2}$ and let $h(\omega) := \int_{\Sp^2} g(\theta, \varphi_+; \omega) \, \vols$. Then the weak derivative of $h$ is given by $\int_{\Sp^2} \rd_\omega g(\theta, \varphi_+ ; \omega) \, \vols$: for $\chi(\omega) \in C^\infty_0(\R)$ we compute
\begin{equation*}
\begin{split}
-\int_\R h(\omega) \rd_\omega \chi(\omega) \, d\omega &= -\int_\R \int_{\Sp^2} g(\theta, \varphi_+; \omega) \, \vols \, \rd_\omega \chi(\omega) \, d\omega \\
&= - \int_\R \int_{\Sp^2} g(\theta, \varphi_+ ; \omega) \rd_\omega \chi(\omega) \, \vols d \omega \\
&= \int_\R \int_{\Sp^2} \rd_\omega g(\theta, \varphi_+ ; \omega) \chi(\omega) \vols d \omega \;.
\end{split}
\end{equation*}} and the representation of $\rd_\omega^{q'} Y^{[s]}_{ml}$ from Proposition \eqref{PropEstimatesDerivativesEigenfunctions}, and in the fourth equation we just used that the limit in $l'$ is an $L^2(\Sp^2)$ limit, so we can pull it out of the integral.

Now by \eqref{EqHJW1} and Plancherel we have $$\infty > \int_{-\varepsilon}^\varepsilon | \rd_\omega^{q} \widecheck{f}(\theta, \varphi_+; \omega)|^2 \, \vols d\omega = \int_{-\varepsilon}^\varepsilon \sum_{m,l} |( \rd_\omega^{q} \widecheck{f})_{ml}(\omega)|^2 \, d \omega \;.$$
Thus \eqref{EqHJW2} follows if we show $$\int_{-\varepsilon}^\varepsilon \sum_{m,l} | \sum_{l'} D^{[s]}_{mll';q'} (\rd_\omega^{q - q'}\widecheck{f})_{ml'}(\omega)|^2 \, d \omega < \infty$$ for $0 \leq q' \leq q$.
By Cauchy-Schwarz, Proposition \ref{PropEstimatesDerivativesEigenfunctions}, Plancherel, and \eqref{EqHJW1} we compute
\begin{equation*}
\begin{split}
\int_{-\varepsilon}^\varepsilon \sum_{m,l} | \sum_{l'} D^{[s]}_{mll';q'} (\rd_\omega^{q - q'}\widecheck{f})_{ml'}(\omega)|^2 \, d \omega &\leq \int_{-\varepsilon}^\varepsilon \sum_{m,l} \Big(\sum_{l'} |D^{[s]}_{mll';q'}(\omega)|^2\Big) \cdot \Big( \sum_{l'} | (\rd_\omega^{q - q'} \widecheck{f})_{ml'}(\omega)|^2\Big) \, d \omega \\
&\leq \int_{-\varepsilon}^\varepsilon \sum_m C(q') \Big( \sum_{l'} | (\rd_\omega^{q - q'} \widecheck{f})_{ml'}(\omega)|^2\Big) \, d \omega  \\
&= C(q') \int_{-\varepsilon}^\varepsilon \int_{\Sp^2} | \rd_\omega^{q - q'} \widecheck{f}(\theta, \varphi_+; \omega) |^2 \, \vols d \omega \\ &< \infty \;.
\end{split}
\end{equation*}
To prove the reverse direction, we now assume that $\widecheck{f}_{ml}$ has $q_0$ weak $\omega$-derivatives satisfying \eqref{EqHJW2}. Let $0 \leq q \leq q_0$ and $\chi \in C^\infty_0\big((-\varepsilon, \varepsilon) \times \Sp^2\big)$. Then
\begin{equation*}
\begin{split}
\int\limits_{(-\varepsilon, \varepsilon)} \int_{\Sp^2} \widecheck{f}(\omega, &\theta, \varphi_+) \overline{ \rd_\omega^q \chi(\omega, \theta, \varphi_+)} \, \vols d\omega = \int_{(-\varepsilon, \varepsilon)} \langle \widecheck{f}(\omega), \rd_\omega^q \chi (\omega) \rangle_{L^2(\Sp^2)} \, d \omega \\
&= \int\limits_{(-\varepsilon, \varepsilon)} \underbrace{\sum_{l,m} \widecheck{f}_{ml}(\omega) \langle Y_{ml}^{[s]}(\omega), \rd_\omega^q \chi(\omega) \rangle_{L^2(\Sp^2)}}_{|-"-| \leq \Big( \sum_{m,l} | \widecheck{f}_{ml} (\omega)|^2 \Big)^{\nicefrac{1}{2}} \Big( \sum_{m', l'} |\rd_\omega^q \chi_{m'l'}(\omega)|^2 \Big)^{\nicefrac{1}{2}}} \, d \omega \\
&= \sum_{l,m} \int\limits_{(-\varepsilon, \varepsilon)} \widecheck{f}_{ml}(\omega) \sum_{j=0}^q (-1)^j \binom{q}{j} \rd_\omega^{q - j} \underbrace{\langle \rd_\omega^j Y^{[s]}_{ml}(\omega), \chi(\omega)\rangle_{L^2(\Sp^2)}}_{\in C^\infty_0\big((-\varepsilon, \varepsilon)\big)} \, d \omega \\
&= \sum_{j = 0}^q (-1)^q \binom{q}{j} \sum_{l,m} \int\limits_{(-\varepsilon, \varepsilon)} \rd_{\omega}^{q - j} \widecheck{f}_{ml}(\omega) \cdot \langle \rd_\omega^j Y_{ml}^{[s]}(\omega), \chi(\omega)\rangle_{L^2(\Sp^2)} \, d \omega \\
&= (-1)^q \int\limits_{(-\varepsilon, \varepsilon)} \int_{\Sp^2} \underbrace{\sum_{j = 0}^q \binom{q}{j} \sum_{l,m} \rd_\omega^{q - j} \widecheck{f}_{lm}(\omega) \sum_{l'} D_{mll'; j}(\omega) Y_{ml'}^{[s]}(\theta, \varphi_+; \omega) }_{= \rd_\omega^q \widecheck{f}(\omega, \theta, \varphi_+)} \cdot \overline{\chi(\omega, \theta, \varphi_+)} \, \vols d\omega \;,
\end{split}
\end{equation*}
where we introduced $\langle \cdot , \cdot \rangle_{L^2(\Sp^2)}$ for the standard Hermitian product on $L^2(\Sp^2)$ for brevity, used Plancherel in the second line, dominated convergence in the third line as well as the combinatorial formula $$\langle Y, \rd_\omega^q \chi \rangle = \sum_{j = 0}^q (-1)^j \binom{q}{j} \rd_\omega^{q - j} \langle \rd_\omega^j Y, \chi\rangle $$ which can be proved easily via induction; we used that $\widecheck{f}_{ml}$ admits $q_0$ weak $\omega$-derivatives in $L^2_{(-\varepsilon, \varepsilon)}$ in the fourth line and finally Proposition \ref{PropEstimatesDerivativesEigenfunctions}, \eqref{EqHJW2}, and dominated convergence again in the last line. Proposition \ref{PropEstimatesDerivativesEigenfunctions} and \eqref{EqHJW2} together now also show \eqref{EqHJW1}.
\end{proof}

\subsection{Application of Teukolsky's separation to the Teukolsky field $\psi$}

\begin{theorem}\label{ThmSeparationVariables}
Under the assumptions of Section \ref{SecAssumptions} and for every $r \in (r_-, r_+)$ the Teukolsky transform
\begin{equation}\label{EqThmSepVar2}
\widecheck{\psi}_{ml}(r; \omega) = \frac{1}{\sqrt{2\pi}} \int_\R \int_{\Sp^2} \psi(v_+, r ,\theta, \varphi_+) e^{i \omega v_+} \overline{Y^{[s]}_{ml}(\theta, \varphi_+; \omega) }\, dv_+ \vols
\end{equation}
of the Teukolsky field $\psi(v_+,r, \theta, \varphi_+)$ is well-defined and we have $\widecheck{\psi}_{ml}(r; \cdot) \in L^2_\omega \ell^2_{m,l}$. Moreover, for every $r \in (r_-, r_+)$, $m \in \Z$, and  $\N \ni l \geq \max\{|m|, |s|\}$ we have $\widecheck{\psi}_{ml}(r; \cdot) \in C^0_\omega(\R)$.

For fixed $\omega, m, l$ the Teukolsky transform $\widecheck{\psi}_{ml}(r; \omega)$ is twice continuously differentiable in $r \in (r_-, r_+)$ and we also have $\frac{d}{dr}\widecheck{\psi}_{ml}(r; \cdot), \frac{d^2}{dr^2} \widecheck{\psi}_{ml}(r; \cdot) \in C^0_\omega(\R)$ for every $r \in (r_-, r_+)$ and $m, l$.\footnote{We only need $\widecheck{\psi}_{ml}(r; \cdot), \frac{d}{dr}\widecheck{\psi}_{ml}(r; \cdot)  \in C^0_\omega(\R)$ (for Lemma \ref{LemRegAA}).} Moreover, the Teukolsky transform satisfies
\begin{equation}\label{EqThmSepVar3}
\begin{split}
\Delta \frac{d^2}{dr^2} \widecheck{\psi}_{ml}(r; \omega) &+ 2 \Big( - (r^2 + a^2) i \omega + i a m + (r-M) (1-s)\Big) \frac{d}{dr} \widecheck{\psi}_{ml}(r; \omega) \\
&+ \Big( \lambda_{ml}^{[s]}(\omega) - (a \omega)^2 + 2 \omega m a - 2i \omega r (1-2s) - 2s\Big) \widecheck{\psi}_{ml}(r; \omega) = 0
\end{split}
\end{equation}
for all $\omega \in\R$, $m \in \Z$, $ \N \ni l \geq \max\{|m|, |s|\}$. Since we have  $\widecheck{\psi}_{ml}(r; \cdot) \in L^2_\omega \ell^2_{m,l}$ the 
representation
\begin{equation} \label{EqThmSepVar1}
\psi(v_+,r, \theta, \varphi_+) = \frac{1}{\sqrt{2\pi}} \int_\R \sum_{m,l} \widecheck{\psi}_{ml}(r; \omega)  Y^{[s]}_{ml}(\theta, \varphi_+; \omega) e^{-i\omega v_+} \, d\omega
\end{equation}
is valid for every $r \in (r_-, r_+)$ in particular in $L^2_{v_+} L^2_{\Sp^2}$.
\end{theorem}

\begin{proof}
Corollary \ref{CorollaryNoShift} in particular states that for each $r \in (r_-,r_+)$ we have $\psi(v_+,r, \theta, \varphi_+)$ is in $L^2_{v_+}L^2_{\Sp^2}$. It now follows from Section \ref{SecTeukExp} that the Teukolsky transform is well-defined with $\widecheck{\psi}_{ml}(r; \cdot) \in L^2_\omega \ell^2_{m,l}$ and also that \eqref{EqThmSepVar1} holds. 

By \eqref{EqCorSobolevNoShift} for $f = \psi$, and since $q_r > 2$, we obtain that $\psi(v_+,r, \theta, \varphi_+) \in L^1(\R \times \Sp^2)$. Together with the boundedness of $Y^{[s]}_{ml}(\theta, \varphi; \omega)$ and its continuity in $\omega$ we obtain from \eqref{EqThmSepVar2}  that $\widecheck{\psi}_{ml}(r; \cdot) \in C^0_\omega(\R)$ for fixed $r,m,l$. By \eqref{EqCorSobolevNoShift} for $f = \rd_r\psi, \rd_r^2 \psi$ we also obtain that we can continuously differentiate in $r$ twice under the integral in \eqref{EqThmSepVar2} for fixed $\omega, m,l$ and also the continuous dependence of the derivatives on $\omega$ as before.

In order to derive \eqref{EqThmSepVar3} we recall the coordinate expression \eqref{TeukolskyStar} of $\T_{[s]} \psi = 0$ to see that
\begin{equation*}
\begin{split}
0 &= \int_\R \int_{\Sp^2} \Big[ a^2 \sin^2 \theta \,\partial_{v_+}^2\psi + 2a \,\partial_{v_+}\rd_{\varphi_+} \psi + 2(r^2 + a^2)\, \partial_{v_+}\partial_r \psi 
+2 a\, \rd_{\varphi_+}\partial_r \psi \\
&\qquad \quad + \Delta \,\partial_r^2 \psi 
 + 2\Big( r(1-2s) - isa\cos \theta\Big)\, \partial_{v_+} \psi  \\
 &\qquad \quad +2(r-M)(1-s) \,\partial_r \psi +  \mathring{\slashed{\Delta}}_{[s]} \psi - 2s \psi \Big]
e^{i \omega v_+} \underbrace{S^{[s]}_{ml}(\cos \theta; \omega) e^{-im \varphi_+}}_{= \overline{Y^{[s]}_{ml}(\theta, \varphi_+; \omega)}} \, dv_+ \vols \\
&= \int_\R \int_{\Sp^2} - (\omega a)^2 \sin^2\theta \, \psi + 2am \omega  \, \psi - 2 i \omega(r^2 +a^2) \, \rd_r \psi +2ima \rd_r \psi \\
&\qquad \quad + \Delta \rd_r^2 \psi -2i \omega \Big( r(1-2s) - isa \cos \theta\Big) \psi \\
&\qquad \quad + 2(r-M)(1-s) \rd_r \psi + \mathring{\slashed{\Delta}}_{[s]} \psi - 2s \psi \Big] e^{i \omega v_+} \overline{Y^{[s]}_{ml}(\theta, \varphi_+; \omega)} dv_+ \vols
\end{split}
\end{equation*}
holds for all $r \in (r_-, r_+)$ and all $\omega \in \R$, where we have used \eqref{EqCorH2NoShift}, which in particular implies\footnote{We use $$\int_\R \Big( \int_{\Sp^2} |\rd_{v_+}^a \rd_{\varphi_+}^b \psi(r')|^2 \vols \Big)^{\frac{1}{2}} dv_+ \leq \underbrace{\Big(\int_\R \frac{1}{\chi(v_+)} \, dv_+\Big)^\frac{1}{2}}_{< \infty} \Big( \int_{\R} \int_{\Sp^2} \chi(v_+) | \rd_{v_+}^a \rd_{\varphi_+}^b \psi(r') |^2 \vols dv_+\Big)^\frac{1}{2} \;.$$ } $\rd_{v_+}^a\rd_{\varphi_+}^{b_1} \rd_r^{b_2} \psi (r) \in L^1_{v_+} L^2_{\Sp^2}$ for $0 \leq a + b_1 + b_2 \leq 2$, $a,  b_1, b_2 \in \N$ which we use to do the integration by parts in $v_+$.
We assemble $\mathring{\slashed{\Delta}}_{[s]}(\omega)$ from \eqref{DefSpinWeightedLaplacianOmega} to find
\begin{equation*}
\begin{split}
0 &= \int_\R \int_{\Sp^2} \Big[ \Delta \rd_r^2 \psi + 2\Big( -(r^2 + a^2) i \omega + iam + (r-M)(1-s)\Big) \rd_r \psi \\
&\qquad \quad + \Big( -(a \omega)^2  + 2 \omega m a - 2i \omega r(1-2s) - 2s\Big) \psi  + \mathring{\slashed{\Delta}}_{[s]}(\omega) \psi \Big] e^{i \omega v_+} \overline{Y^{[s]}_{ml}(\theta, \varphi_+; \omega)} dv_+ \vols \;.
\end{split}
\end{equation*}
We now use that $\psi$ and $Y^{[s]}_{ml}$ are smooth spin $2$-weighted functions so that by Proposition \ref{PropIntPartsSphere} and \eqref{EqSWLV} we can do the integration by parts to bring $\mathring{\slashed{\Delta}}_{[s]}(\omega)$ over to obtain a term of the form  $\psi \cdot e^{i \omega v_+} \overline{\mathring{\slashed{\Delta}}_{[s]}(\omega) Y^{[s]}_{ml}(\theta, \varphi_+; \omega)} = \psi  \cdot e^{i \omega v_+} \lambda^{[s]}_{ml}(\omega) \overline{Y^{[s]}_{ml}(\theta, \varphi_+; \omega)} $, where we used that the eigenvalues $\lambda^{[s]}_{ml}(\omega)$ are real. Finally, by \eqref{EqCorSobolevNoShift} for $f = \rd_r \psi, \rd_r^2 \psi$ and the boundedness of $Y^{[s]}_{ml}(\theta, \varphi; \omega)$ dominated convergence allows us to pull the $r$-derivatives out of the integral to obtain \eqref{EqThmSepVar3}. 
\end{proof}

\section{Analysis of the Heun equation and transmission and reflection coefficients for $\omega = 0$} \label{SecAnaHeun}

This section analyses the radial Teukolsky equation \eqref{EqThmSepVar3}. We show that it is of the Heun-form and that the limit $\omega = 0$ is a hypergeometric equation. We introduce specific fundamental systems of solutions along with the corresponding transmission and reflection coefficients and investigate their regularity and their behaviour for $\omega \to 0$.

\subsection{The Heun equation} \label{SecHeun}

Setting $x := \frac{r - r_-}{r_+ - r_-}$ in \eqref{EqThmSepVar3} so that we have $x = 0$ for $r = r_-$ and $x = 1$ for $r = r_+$ the equation \eqref{EqThmSepVar3} transforms to the Heun equation
\begin{equation}\label{EqSepTeukEqX}
(1-x)x \frac{d^2}{dx^2} v(x) + \big( \alpha x^2 + \beta x + \gamma\big) \frac{d}{dx} v(x) + \big( \delta x + \varepsilon\big) v(x) = 0 \;,
\end{equation}
where we have just written $\widecheck{\psi}_{ml}(r(x); \omega) = v(x)$ for brevity and generality and where
\begin{equation} \label{EqGreek}
\begin{aligned}
\alpha &=2 i \omega (r_+ - r_-) \qquad \qquad && \delta = 2i \omega (1-2s)(r_+-r_-) \\
\beta &= 4ir_- \omega +2(s-1) && \varepsilon = - \lambda^{[s]}_{ml}(\omega) +(a \omega)^2 - 2 \omega ma +2s + 2i \omega(1 - 2s)r_- \\
\gamma &= \frac{4iMr_-}{r_+ - r_-} (\omega - \omega_- m) + 1 - s
\end{aligned}
\end{equation}
and $\omega_\pm := \frac{a}{2Mr_\pm}$.

Setting $y :=1 - x = \frac{r_+ - r}{r_+ - r_-}$ so that we have $y = 0$ for $r = r_+$ and $y= 1$ for $r = r_-$ the equation \eqref{EqSepTeukEqX} transforms to the Heun equation
\begin{equation}\label{EqSepTeukEqY}
(1-y)y \frac{d^2}{dy^2} v(y) + \big( \tilde{\alpha} y^2 + \tilde{\beta} y + \tilde{\gamma} \big) \frac{d}{dy} v(y) + \big( \tilde{\delta} y + \tilde{\varepsilon} \big) v(y) = 0 \;,
\end{equation}
where
\begin{align*}
\tilde{\alpha} &= - \alpha = -2 i \omega (r_+ - r_-) \\
\tilde{\beta} &= \beta + 2 \alpha = 4i \omega r_+ + 2(s-1) \\
\tilde{\gamma} &=  -(\beta + \gamma + \alpha) = -\frac{4iMr_+}{r_+ - r_-}(\omega - \omega_+ m) + 1-s \\
\tilde{\delta} &= - \delta = - 2i \omega (1-2s)(r_+-r_-) \\
\tilde{\varepsilon} &= \delta + \varepsilon= - \lambda^{[s]}_{ml}(\omega) +(a \omega)^2 - 2 \omega ma +2s + 2i \omega(1 - 2s)r_+ \;.
\end{align*}

\subsubsection{The hypergeometric equation arising as the limit $\omega = 0$ of the Heun equation} \label{SecHypergeom}

We compute the values of the Greek parameters $\alpha, \ldots, \varepsilon$ for $\omega = 0$, where we also use $\lambda_{ml}^{[s]}(0) = -l(l+1) + s(s+1)$ from Proposition \ref{PropEigenfunctionsSWL}:
\begin{equation}\label{EqParaOmega0}
\begin{aligned}
\alpha|_{\omega =0} &=0\qquad \qquad && \delta|_{\omega =0} = 0 \\
\beta|_{\omega =0} &= 2(s-1) && \varepsilon|_{\omega =0} = (l-s)(l+s+1) + 2s\\
\gamma|_{\omega =0} &= -\frac{2iam}{r_+ - r_-}+ 1 - s \;.
\end{aligned}
\end{equation}
A straightforward computation then shows that for $\omega = 0$ the Heun equation \eqref{EqSepTeukEqX} turns into the hypergeometric equation
\begin{equation}\label{EqHypergeometricX}
(1-x) x \frac{d^2}{dx^2} v(x) + \big(\underline{c} - (\underline{a} + \underline{b} +1)x \big) \frac{d}{dx} v(x) - \underline{a} \underline{b} \cdot v(x) = 0 
\end{equation}
with
\begin{equation}\label{EqParaHypergeom}
\underline{a} = l + 1 - s \qquad \quad \underline{b} = -s -l \qquad \quad \underline{c} = \gamma|_{\omega = 0} = -\frac{2iam}{r_+ - r_-}+ 1 - s  \;.
\end{equation}
Setting again $y = 1-x$, \eqref{EqHypergeometricX} transforms into
\begin{equation}\label{EqHypergeometricY}
(1-y) y \frac{d^2}{dy^2} v(y) + \big(\underline{\tilde{c}} - (\underline{a} + \underline{b} +1)y \big) \frac{d}{dy} v(y) - \underline{a} \underline{b} \cdot v(y) = 0 
\end{equation}
with $$\underline{\tilde{c}} = \underline{a} + \underline{b} +1 - \underline{c} = \frac{2iam}{r_+ - r_-} + 1-s \;. $$

\subsection{Fundamental systems of solutions and reflection and transmission coefficients}

We now recall the Frobenius method to determine the possible asymptotics  of solutions of the radial ODE \eqref{EqThmSepVar3} at the regular singular points $r = r_+$ and $r=r_-$ and to construct fundamental systems of solutions with these prescribed asymptotics. We only provide a sketch of the derivation, full details of this textbook material are found for example in Chapter 4 of \cite{Teschl}.

We begin with the discussion of the regular singular point $x=0$ in \eqref{EqSepTeukEqX}, which corresponds to the Cauchy horizon $r=r_-$. The asymptotics at the event horizon, which is $y=0$ for \eqref{EqSepTeukEqY}, then follow directly from this discussion by replacing the Greek parameters $\alpha, \ldots, \varepsilon$ by their tilded versions.

We make the ansatz $x^\sigma \sum_{j=0}^\infty d_j(\omega, m,l) x^j$ for a solution of \eqref{EqSepTeukEqX}. Entering this ansatz into \eqref{EqSepTeukEqX} and comparing powers of $x$ yields
\begin{equation}
\label{EqPowerSeriesSol}
d_{j+1}(\sigma + j + 1)(\sigma + j + \gamma) = d_j\big((\sigma + j)(\sigma + j -1 - \beta) - \varepsilon\big) + d_{j-1}\big(- \alpha (\sigma + j -1) - \delta\big) \;.
\end{equation} 
For $j = -1$ we obtain the indicial equation $\sigma(\sigma - 1 + \gamma) = 0$ which has the two solutions $\sigma = 0$ and $\sigma = 1 - \gamma$.

Consider first $\sigma=1 - \gamma$ and set $d_0 = 1$. It then follows from \eqref{EqPowerSeriesSol} with $\sigma = 1- \gamma$ that the coefficients are recursively determined by 
\begin{equation}\label{EqRecRel1}
d_{j+1} = \frac{1}{(2 + j - \gamma)(1+j)} \Big[d_j\big((1 - \gamma + j)(- \gamma + j - \beta) - \varepsilon\big) + d_{j-1}\big(- \alpha(j - \gamma) - \delta\big)\Big]
\end{equation}
Note that for $s= 2$ we have $\gamma = -1 + i \frac{4Mr_-}{r_+ - r_-}(\omega - \omega_-m)$ and thus the denominator in \eqref{EqRecRel1} is non-vanishing for all $j \in \N$ and for all $\omega \in \R$. It can be shown that this power series converges absolutely for $x \in [0,1)$. Also note that since the coefficients $\alpha, \ldots, \varepsilon$ depend analytically on $\omega$, so do all the $d_j(\omega, m,l)$. The radius of convergence of the power series of $d_j(\omega, m, l)$ in $\omega$ is uniformly lower bounded in $j$ (it essentially depends on the radius of convergence of the power series in $\omega$ of $\lambda^{[s]}_{ml}(\omega)$ and on $\frac{4Mr_-}{r_+ - r_-}$).  Since the convergence is uniform, we obtain that the arising power series\footnote{Note that the multiplying factor $x^{1 - \gamma}$ is not analytic at $x = 0$ if $\omega \neq \omega_-m$.} is also analytic in $\omega$ for all $\omega \in \R$. We label this solution by $B^{[s]}_{\CH_r, ml}(x ; \omega) := x^{1 - \gamma} \sum_{j = 0}^\infty d_j^{[s]}(\omega, m ,l) x^j$.

To construct a second linearly independent solution we make the other choice $\sigma = 0$, i.e., we are looking for a solution of the form $\sum_{j=0}^\infty c_j(\omega, m,l) x^j$, which we normalise by $c_0 =1$. From \eqref{EqPowerSeriesSol} with the $d's$ replaced by $c's$ we obtain the recursive relation
\begin{equation} \label{EqRecTeukXC}
c_{j+1} = \frac{1}{(j+1)(j + \gamma)} \Big[ c_j \big( j(j-1 - \beta) - \varepsilon\big) + c_{j-1} \big(- \alpha(j-1) - \delta\big) \Big] \;.
\end{equation}
Since $\gamma = \frac{4iMr_-}{r_+ - r_-} (\omega - \omega_-m) + 1-s$, the denominator vanishes for $j=1$ and $\omega = \omega_-m$, but for all other $\omega$ one can show as before that the power series converges absolutely for $x \in [0,1)$ and is analytic in $\omega \in \R \setminus \{\omega_- m\}$. We label this solution by $B^{[s]}_{\CH_l, ml}(x; \omega) := \sum_{j = 0}^\infty c_j^{[s]}(\omega, m, l) x^j$. From \eqref{EqRecTeukXC} we compute $c_1(\omega, m, l) = - \frac{\varepsilon}{\gamma}$ and $c_2(\omega, m, l) = \frac{1}{2(1 + \gamma)} \Big[ \frac{\varepsilon}{ \gamma} (\beta + \varepsilon) - \delta\Big]$ for later.

The Frobenius solutions of \eqref{EqSepTeukEqY} normalised at $y=0$ are obtained in the analogous way by replacing $\alpha, \ldots, \varepsilon$ by $\tilde{\alpha}, \ldots, \tilde{\varepsilon}$. We summarise this discussion in the following

\begin{proposition} \label{PropFrobSolutions}
\hfill
\begin{enumerate}
\item
For $ \omega \neq \omega_- m$ equation \eqref{EqSepTeukEqX} has a fundamental system of solutions  $\{B^{[s]}_{\CH_l, ml}(x; \omega), B^{[s]}_{\CH_r, ml}(x ;\omega)\}$ which are of the form
\begin{equation*}
\begin{aligned}
B^{[s]}_{\CH_l, ml}(x; \omega) &= \sum_{j = 0}^\infty c_j^{[s]}(\omega, m, l) x^j \\
B^{[s]}_{\CH_r, ml}(x ; \omega) &= x^{1 - \gamma} \sum_{j = 0}^\infty d_j^{[s]}(\omega, m ,l) x^j
\end{aligned}
\end{equation*}
and are normalised by $c_0^{[s]}(\omega, m,l) = 1$ and $d_0^{[s]}(\omega, m ,l) = 1$. The power series $\sum_{j = 0}^\infty d_j^{[s]}(\omega, m ,l) x^j$ is analytic in $[0,1) \times \R$ (and thus $B^{[s]}_{\CH_r, ml}(x; \omega)$ is in particular also a solution for $\omega = \omega_- m$) while the solution $B^{[s]}_{\CH_l, ml}(x ; \omega)$ is only analytic and defined on $[0,1) \times (\R \setminus \{\omega_- m\})$. The coefficients are determined recursively and we find in particular $c_1^{[s]}(\omega, m, l) = - \frac{\varepsilon}{\gamma}$ and $c_2^{[s]}(\omega, m, l) = \frac{1}{2(1 + \gamma)} \Big[ \frac{\varepsilon}{ \gamma} (\beta + \varepsilon) - \delta\Big]$.\footnote{Note that $\gamma(\omega) \to 1-s = -1$ for $\omega \to \omega_- m$. This shows that $B^{[s]}_{\CH_l, ml} (x;\omega)$ is in general not regular for $\omega \to \omega_- m$.} 

\item For $ \omega \neq \omega_+ m$ equation \eqref{EqSepTeukEqY} has a fundamental system of solutions $\{ A^{[s]}_{\Hp_r, ml}(y ;\omega) , A^{[s]}_{\Hp_l, ml}(y ;\omega) \}$ which are of the form
\begin{equation*}
\begin{aligned}
A^{[s]}_{\Hp_r, ml}(y ;\omega) &= \sum_{j= 0}^\infty a_j^{[s]}(\omega, m,l) y^j   \\
A^{[s]}_{\Hp_l, ml}(y; \omega) &= y^{1 - \tilde{\gamma}} \sum_{j= 0}^\infty b_j^{[s]}(\omega, m, l) y^j 
\end{aligned}
\end{equation*}
and are normalised by $a_0^{[s]}(\omega, m, l) = 1 $ and $b_0^{[s]}(\omega, m,l) = 1$. The power series $\sum_{j= 0}^\infty b_j^{[s]}(\omega, m, l) y^j $ is analytic in $[0,1) \times \R$ (and  thus $A^{[s]}_{\Hp_l, ml}(y; \omega)$ is in particular also a solution for $\omega = \omega_+ m$) while the solution $A^{[s]}_{\Hp_r, ml}(y ;\omega)$ is only analytic and defined on $[0,1) \times (\R \setminus \{\omega_+ m\})$.
The coefficients are determined recursively and we find in particular $a_1^{[s]}(\omega, m, l) = - \frac{\tilde{\varepsilon}}{\tilde{\gamma}} $ and $a_2^{[s]}(\omega, m, l ) = \frac{1}{2(1 + \tilde{\gamma})} \Big[ \frac{\tilde{\varepsilon}}{\tilde{\gamma}}(\tilde{\beta} + \tilde{\varepsilon}) - \tilde{\delta}\Big] $. \footnote{Note that $\tilde{\gamma}(\omega) \to 1-s = -1$ for $\omega \to \omega_+ m$. This shows that $A^{[s]}_{\Hp_r, ml} (y;\omega)$ is in general not regular for $\omega \to \omega_+ m$.} 
\end{enumerate}
\end{proposition}

Our reason for labelling the solutions with $\Hp_r, \Hp_l, \CH_l, \CH_r$ will become apparent in Sections \ref{SecDetA} and \ref{SecPfMainThm}. It follows that we can write for $\omega \neq \omega_+ m$ the Teukolsky transform $\widecheck{\psi}_{ml}$ from Theorem \ref{ThmSeparationVariables}, with the $r$-coordinate replaced by the $y$-coordinate, as
\begin{equation}\label{EqDefAA}
\widecheck{\psi}_{ml}(y; \omega) =: a_{\Hp_r, ml}(\omega) A_{\Hp_r, ml}(y; \omega) + a_{\Hp_l, ml}(\omega) A_{\Hp_l, ml}(y; \omega) \;,
\end{equation}
where $a_{\Hp_r, ml}, a_{\Hp_l, ml} : \R \setminus \{\omega_+m\} \to \C$ are functions which will be determined later in Section  \ref{SecDetA}. 
\begin{lemma}\label{LemRegAA}
Under the assumptions of Section \ref{SecAssumptions} we have $a_{\Hp_r, ml}, a_{\Hp_l, ml} \in C^0(\R \setminus\{ \omega_+m\}, \C)$.
\end{lemma}

\begin{proof}
Differentiating \eqref{EqDefAA} in $y$ we obtain for $\omega \neq \omega_+ m$
\begin{equation}\label{EqMatAA} \begin{pmatrix}
\widecheck{\psi}_{ml}(y; \omega) \\ \frac{d}{dy} \widecheck{\psi}_{ml}(y; \omega) 
\end{pmatrix} = \begin{pmatrix}
A_{\Hp_r, ml}(y; \omega) & A_{\Hp_l, ml}(y; \omega)   \\ \frac{d}{dy} A_{\Hp_r, ml}(y; \omega) & \frac{d}{dy} A_{\Hp_l, ml}(y; \omega) 
\end{pmatrix} \begin{pmatrix}
a_{\Hp_r, ml}(\omega)  \\  a_{\Hp_l, ml}(\omega)
\end{pmatrix} \;.\end{equation}
Fix $y \in (0,1)$. Since $\{ A_{\Hp_r, ml}(y; \omega) , A_{\Hp_l, ml}(y; \omega) \}$ are linearly independent, the matrix has an inverse which is also analytic in $\omega$ for $\omega \neq \omega_+m$. We can thus solve for $a_{\Hp_r, ml}(\omega),  a_{\Hp_l, ml}(\omega)$ and thus they inherit the regularity of the left hand side of \eqref{EqMatAA}, which is continuous by Theorem \ref{ThmSeparationVariables}.
\end{proof}

\subsubsection{Alternative representation of second Frobenius solution} \label{SecAltRepFrob}

Let us also recall a different way of constructing the second Frobenius solution $B^{[s]}_{\CH_l, ml}(x;\omega)$ which will be useful later on in Section \ref{SecM0}. This is the variation of constant ansatz, see for example Chapter 4 of \cite{Teschl} for full details.

To obtain a second linearly independent solution we make the variation of constants ansatz $$v(x) = e(x) \cdot B^{[s]}_{\CH_r, ml}(x; \omega)\;.$$ Entering this into equation \eqref{EqSepTeukEqX} gives
\begin{equation}\label{EqVarConstants}
x(1-x)\big[ e''(x) B^{[s]}_{\CH_r, ml}(x; \omega) + 2e'(x) \frac{d}{dx}B^{[s]}_{\CH_r, ml}(x; \omega) \big] + ( \alpha x^2 + \beta x + \gamma) e'(x) \cdot B^{[s]}_{\CH_r, ml}(x; \omega) = 0 \;.
\end{equation}
Here, the prime stands for $\frac{d}{dx}$. This is a first order equation for $e'(x)$. Again, making a power series ansatz one can show that \eqref{EqVarConstants} has a unique solution of the form 
\begin{equation}\label{EqAltFrobDer}
e'(x) = x^{\gamma - 2} \sum_{j=0}^\infty e_j x^j
\end{equation}
which we normalise by $e_0 = \gamma - 1$ and where the coefficients are determined recursively by an algebraic expression which involves $\alpha, \beta, \gamma, e_j, d_j$. In particular, each $e_j(\omega, m, l)$ is an analytic function of $\omega$ for all $\omega \in \R$.

As before let us now assume that $\omega \neq \omega_- m$, so the parameter $\gamma$ has an imaginary part. In particular $e'(x)$ does not have a term proportional to $\frac{1}{x}$. An integral of $e'(x)$ is thus given by
\begin{equation}\label{EqAltFrobDer2}
e(x) = x^{\gamma - 1} \underbrace{\sum_{j = 0}^\infty \frac{e_j}{\gamma - 1 + j} x^j }\;.
\end{equation}
The underbraced power series converges absolutely on $x \in [0,1)$ and is analytic in $\omega$ for $\omega \neq \omega_-m$.  Since we have chosen $e_0 = \gamma -1$ we see that the coefficient in the power series in front of $x^0$ is $1$. Thus, 
\begin{equation} \label{EqAltRep}
B^{[s]}_{\CH_l, ml} (x; \omega) := e(x) \cdot B^{[s]}_{\CH_r, ml}(x; \omega)
\end{equation} 
is a solution of \eqref{EqHypergeometricX} of the form
$$B^{[s]}_{\CH_l, ml} (x; \omega) =\sum_{j = 0}^\infty c_j(\omega, m, l) x^j$$ with $c_0(\omega, m, l) = 1$. The coefficients $c_j(\omega, m, l)$ can of course be computed from those of $e(x)$ and those of $B^{[s]}_{\CH_r, ml}(x; \omega) = x^{1 - \gamma} \sum_{j = 0}^\infty d_j(\omega, m ,l) x^j$, for example we have
\begin{equation*}
c_0 = 1 \qquad \qquad c_1 = \frac{e_1}{\gamma} + d_1 \qquad \qquad c_2 = \frac{e_2}{\gamma +1} + \frac{e_1}{\gamma} d_1 + d_2 \;.
\end{equation*}
On the other hand it follows from the asymptotics that the solution \eqref{EqAltRep} we have constructed here must agree with $B^{[s]}_{\CH_l, ml} (x; \omega)$ from Proposition \ref{PropFrobSolutions}, for which we have already obtained the explicit values of $c_1$ and $c_2$. We thus find
$$c_2 = \frac{1}{2(1+ \gamma)} \big[ \frac{\varepsilon}{\gamma}(\beta + \varepsilon) - \delta \big] = \frac{e_2}{\gamma + 1} + \frac{e_1}{\gamma} d_1 + d_2 \;.$$
Multiplying by $(1+\gamma)$ and setting $\omega = \omega_-m$ we obtain
\begin{equation}\label{EqE2OmegaM}
e_2(\omega = \omega_- m) = \frac{1}{2} \big[ \frac{\varepsilon}{\gamma}(\beta + \varepsilon) - \delta \big]  (\omega =\omega_-m ) \;.
\end{equation}

Similarly, we find an alternative expression of $A^{[s]}_{\Hp_r, ml}(y ;\omega)$ as
\begin{equation*}
A^{[s]}_{\Hp_r, ml}(y ;\omega) = \tilde{e}(y) \cdot A^{[s]}_{\Hp_l, ml}(y ;\omega) = \Big(\sum_{j = 0}^\infty \frac{\tilde{e}_j}{\tilde{\gamma} - 1 + j} y^j \Big) \cdot \Big(\sum_{k=0}^\infty b_k y^k \Big)\;,
\end{equation*}
where
$$\tilde{e}(y) = y^{\tilde{\gamma} - 1} \underbrace{\sum_{j = 0}^\infty \frac{\tilde{e}_j(\omega, m, l)}{\tilde{\gamma} - 1 + j} y^j} $$
with
\begin{equation}\label{Eq2OT}
\tilde{e}_2(\omega = \omega_+ m) = \frac{1}{2} \big[ \frac{\tilde{\varepsilon}}{\tilde{\gamma}}(\tilde{\beta} + \tilde{\varepsilon}) - \tilde{\delta} \big]  (\omega =\omega_+m ) \;.
\end{equation}
The underbraced power series converges absolutely on $y \in [0,1)$ and is analytic in $\omega$ for $\omega \neq \omega_+m$.

\subsubsection{Reflection and transmission coefficients}

Since $\{B^{[s]}_{\CH_l, ml},B^{[s]}_{\CH_r, ml}\}$ forms a fundamental system of solutions, we can express each of the two solutions $A^{[s]}_{\Hp_r, ml}, A^{[s]}_{\Hp_l, ml}$ as a linear combination thereof, i.e., we can write for each $\omega \in \R \setminus \{\omega_-m, \omega_+m\}$
\begin{equation}\label{EqDefTransReflCoeff}
\begin{split}
A^{[s]}_{\Hp_r, ml}(1-x; \omega) &= \mathfrak{T}^{[s]}_{\Hp_r, ml}(\omega) \cdot B^{[s]}_{\CH_l, ml}(x; \omega) + \mathfrak{R}^{[s]}_{\Hp_r, ml} (\omega) \cdot B^{[s]}_{\CH_r, ml} (x; \omega) \\
A^{[s]}_{\Hp_l, ml}(1-x; \omega) &= \mathfrak{T}^{[s]}_{\Hp_l, ml}(\omega) \cdot B^{[s]}_{\CH_r, ml}(x; \omega) + \mathfrak{R}^{[s]}_{\Hp_l, ml} (\omega) \cdot B^{[s]}_{\CH_l, ml} (x; \omega) 
\end{split}
\end{equation}
with $\mathfrak{T}^{[s]}_{\Hp_r, ml}, \mathfrak{R}^{[s]}_{\Hp_r, ml}, \mathfrak{T}^{[s]}_{\Hp_l, ml},  \mathfrak{R}^{[s]}_{\Hp_l, ml} : \R \setminus \{\omega_-m, \omega_+m\} \to \C$, where we call $\mathfrak{T}^{[s]}_{\Hp_r, ml}(\omega), \mathfrak{T}^{[s]}_{\Hp_l, ml}$  the \emph{transmission coefficients of the right and left even horizon}, respectively, and $\mathfrak{R}^{[s]}_{\Hp_r, ml}, \mathfrak{R}^{[s]}_{\Hp_l, ml}$ the \emph{reflection coefficients of the right and left event horizon}, respectively.

\subsubsection{The case $m \neq 0$} \label{SecMneq0}

\begin{proposition} \label{PropTRmNeq0}
Let $m \neq 0$. Then the transmission and reflection coefficients $\mathfrak{T}^{[s]}_{\Hp_r, ml}(\omega), \mathfrak{R}^{[s]}_{\Hp_r, ml}(\omega)$, $\mathfrak{T}^{[s]}_{\Hp_l, ml}(\omega),  \mathfrak{R}^{[s]}_{\Hp_l, ml}(\omega)$ are (defined and) analytic for $\omega \in (-|\omega_+|, |\omega_+|)$ and we have $\mathfrak{T}^{[s]}_{\Hp_r, ml}(0) \neq 0$.
\end{proposition}

\begin{proof}
Note that $0 < |\omega_+| < |\omega_-|$. Hence for $m \neq 0$ the fundamental solutions in \eqref{EqDefTransReflCoeff} are defined for $\omega \in (-|\omega_+|, |\omega_+|)$ and thus so are the transmission and reflection coefficients. 

Combining the first line of \eqref{EqDefTransReflCoeff} with its differentiated version in $x$ we obtain the vector equation
\begin{equation}\label{EqTRAnalytic} \begin{pmatrix}  A^{[s]}_{\Hp_r, ml}(1-x; \omega) \\ \frac{d}{dx} A^{[s]}_{\Hp_r, ml}(1-x; \omega) \end{pmatrix} = 
\begin{pmatrix}
B^{[s]}_{\CH_l, ml}(x; \omega)  & B^{[s]}_{\CH_r, ml} (x; \omega) \\
\frac{d}{dx} B^{[s]}_{\CH_l, ml}(x; \omega)  & \frac{d}{dx} B^{[s]}_{\CH_r, ml} (x; \omega)
\end{pmatrix} 
\begin{pmatrix}
\mathfrak{T}^{[s]}_{\Hp_r, ml}(\omega)  \\ \mathfrak{R}^{[s]}_{\Hp_r, ml} (\omega)
\end{pmatrix} \;. 
\end{equation}
Fix $x \in (0,1)$. Note that the matrix on the right hand side is invertible (since $B^{[s]}_{\CH_l, ml}(x; \omega) $ and $B^{[s]}_{\CH_r, ml} (x; \omega)$ are linearly independent) and analytic in $\omega$ for $\omega \in (-|\omega_+|, |\omega_+|)$. The left hand side is analytic for $\omega \in (-|\omega_+|, |\omega_+|)$ as well.   We can thus solve for $\mathfrak{T}^{[s]}_{\Hp_r, ml}(\omega)  $ and $\mathfrak{R}^{[s]}_{\Hp_r, ml} (\omega)$  and obtain that they are analytic in $\omega \in (-|\omega_+|, |\omega_+|)$. Similarly one obtains that the transmission and reflection coefficients of the left event horizon are analytic in $(-|\omega_+|, |\omega_+|)$.

To show $\mathfrak{T}^{[s]}_{\Hp_r, ml}(0) \neq 0$ we begin by noticing that $ A^{[s]}_{\Hp_r, ml}(y; 0)$ and $A^{[s]}_{\Hp_l, ml}(y; 0)$ solve \eqref{EqHypergeometricY} and $B^{[s]}_{\CH_l, ml} (x; 0)$ and $B^{[s]}_{\CH_r, ml} (x; 0)$  solve \eqref{EqHypergeometricX}. For the hypergeometric equation there are convenient closed expressions for the Frobenius solutions, which we recall in the following, see for example Chapter 8 of \cite{BealsWong}, but they can also be verified directly.

For $\underline{a} \in \C$ and $n \in \N_0$ we define $(\underline{a})_n := \underline{a}(\underline{a} +1) (\underline{a} + 2 ) \cdots (\underline{a} + n-1) = \frac{\Gamma(\underline{a} +n)}{\Gamma(\underline{a})}$, where $\Gamma$ is the Gamma function. Then, for $-\underline{c} \notin \N_0$
\begin{equation}
\label{EqFirstHyper}
F(\underline{a}, \underline{b}, \underline{c}; x) := \sum_{n=0}^\infty \frac{(\underline{a})_n (\underline{b})_n}{(\underline{c})_n n!} x^n 
\end{equation} 
is a solution of \eqref{EqHypergeometricX} with $
  F(\underline{a}, \underline{b}, \underline{c}; 0) = 1$. And for $\underline{c} -2 \notin \N_0$
\begin{equation}\label{EqSecSolHyper} 
 x^{1 - \underline{c}} F(\underline{a} + 1 - \underline{c}, \underline{b} + 1 - \underline{c}, 2 - \underline{c};x)
 \end{equation} 
 is also a solution of \eqref{EqHypergeometricX}. Clearly, for $\underline{c} \neq 1$ these two solutions are linearly independent. 
Recall from Section \ref{SecHypergeom} that $B^{[s]}_{\CH_l, ml} (x; 0)$ and $B^{[s]}_{\CH_r, ml} (x; 0)$ are solutions of the hypergeometric equation \eqref{EqHypergeometricX} with $\underline{c} = \gamma(\omega = 0) = -1 - \frac{2iam}{r_+ - r_-} $.  For $m \neq 0$ we thus obtain, by comparison of the asymptotics, that we must have
$$ B^{[s]}_{\CH_l, ml} (x; 0) = F(\underline{a}, \underline{b}, \underline{c}; x)  \qquad \textnormal{ and } \qquad B^{[s]}_{\CH_r, ml} (x; 0) = x^{1 - \underline{c}} F(\underline{a} + 1 - \underline{c}, \underline{b} + 1 - \underline{c}, 2 - \underline{c};x) $$
with $\underline{a}$, $\underline{b}$ and $\underline{c}$ as in Section \ref{SecHypergeom}. Similarly we obtain
$$A^{[s]}_{\Hp_r, ml}(y; 0) = F(\underline{a}, \underline{b}, \tilde{\underline{c}}; y) \qquad \textnormal{ and } \qquad A^{[s]}_{\Hp_l, ml} (y;0) = y^{1 - \tilde{\underline{c}}} F(\underline{a} + 1 - \tilde{\underline{c}}, \underline{b} + 1 - \tilde{\underline{c}}, 2 - \tilde{\underline{c}};y)$$ with $\underline{\tilde{c}}$ as in Section \ref{SecHypergeom}.

\emph{Now note that \eqref{EqFirstHyper} is a polynomial in $x$ if, and only if, $\underline{a}$ or $\underline{b}$ are negative integers.} Since we have $\underline{b} = -s -l = -2 -l$ and $\N \ni l \geq \max(|m|, |s|)$, $\underline{b}$ is a negative integer and thus $B^{[s]}_{\CH_l, ml} (x; 0)$ and $A^{[s]}_{\Hp_r, ml}(1-x; 0)$ are polynomials in $x$. Moreover, since $\underline{c}$ and $\tilde{\underline{c}}$ have non-vanishing imaginary parts it is straightforward to see that $B^{[s]}_{\CH_r, ml} (x; 0) $ and $A^{[s]}_{\Hp_l, ml} (x;0)$ are not polynomials in $x$. Entering with this information into \eqref{EqDefTransReflCoeff} gives directly that $\mathfrak{R}^{[s]}_{\Hp_r, ml} (0)$ has to vanish and thus  $\mathfrak{T}^{[s]}_{\Hp_r, ml}(0) \neq 0$.
\end{proof}

\begin{remark} \label{RemExplicitT}
Indeed, all the transmission and reflection coefficients at $\omega = 0$ for $m \neq 0$ can be  computed explicitly using the classical theory of linear relations of solutions of the hypergeometric ODE, see for instance Chapter 8 of \cite{BealsWong}. For example one obtains $\mathfrak{T}^{[s]}_{\Hp_r, ml}(0) = \frac{\Gamma(\underline{a} + \underline{b} + 1 - \underline{c}) \Gamma(1 - \underline{c})}{\Gamma(\underline{a} + 1 - \underline{c}) \Gamma(\underline{b} + 1 - \underline{c})}$. Setting $\xi(0) := - \frac{2iam}{r_+ - r_-}$ and plugging in the exact values of the parameters for $s=2$ from \eqref{EqParaHypergeom} we obtain $$\mathfrak{T}^{[s]}_{\Hp_r, ml}(0) = \frac{\Gamma\big(-1-\xi(0)\big) \Gamma\big(2 - \xi(0)\big)}{\Gamma\big(l +1 -  \xi(0)\big) \Gamma\big(-l -  \xi(0)\big)} = \frac{\big(-l -  \xi(0)\big) \cdot \big(-l+1 -\xi(0)\big) \cdot \ldots \cdot (1-  \xi(0)\big)}{\big(l -  \xi(0)\big) \cdot \big(l - 1 - \xi(0)\big) \cdot \ldots \cdot (-1 -\xi(0)\big)}$$ from which it also follows that $|\mathfrak{T}^{[s]}_{\Hp_r, ml}(0)|=1$.
\end{remark}

\subsubsection{The case $m=0$ via the Teukolsky-Starobinsky conservation law} \label{SecM0}

\begin{proposition}\label{PropTRM0}
The transmission coefficient $\mathfrak{T}^{[s]}_{\Hp_r, 0l}$ of the right event horizon and the reflection coefficient $\mathfrak{R}^{[s]}_{\Hp_l, 0l}$ of the left event horizon, as well as $\omega \cdot \mathfrak{R}^{[s]}_{\Hp_r, 0l}$ (all of which are a priori not defined at $\omega = 0$) extend analytically to $\omega \in \R$. Moreover, we have $\mathfrak{R}^{[s]}_{\Hp_l, 0l} (0) = 0$. 
\end{proposition}

\begin{proof}
We construct a set of fundamental solutions which is regular for all $\omega \in \R$. Recall from Proposition \ref{PropFrobSolutions} that $B^{[s]}_{\CH_r, ml}(x;\omega)$ is defined for all $\omega \in \R$. For $\omega \neq 0$ we now use the alternative representation of the second Frobenius solution $B^{[s]}_{\CH_l, 0l}(x; \omega)$ from Section \ref{SecAltRepFrob}  and define
\begin{equation*}
\begin{split}
U_{\CH,0l}(x; \omega) :&= B_{\CH_l,0l}(x ;\omega) - \frac{e_2(\omega,0,l)}{\gamma +1} B_{\CH_r,0l}(x;\omega) \\
&=x^{\gamma - 1} \Big( \sum_{j = 0}^\infty \frac{e_j(\omega, 0, l)}{\gamma - 1 + j} x^j \Big) B_{\CH_r,0l}(x;\omega) - \frac{e_2(\omega,0,l)}{\gamma +1} B_{\CH_r,0l}(x;\omega) \\
&=x^{\gamma - 1} \Big( \sum_{\substack{j = 0 \\ j \neq 2}}^\infty \frac{e_j(\omega, 0 , l)}{\gamma - 1 + j} x^j \Big) B_{\CH_r,0l}(x;\omega) + e_2(\omega,0,l)\cdot \underbrace{\frac{x^{\gamma +1} -1}{\gamma +1}}_{= \sum_{k=1}^\infty \frac{(\gamma + 1)^{k-1} (\log x)^k}{k!}} \cdot B_{\CH_r,0l}(x;\omega) \;.
\end{split}
\end{equation*}
We thus see that $U_{\CH,0l}(x; \omega)$ extends analytically as a solution\footnote{Note that the above construction corresponds to choosing as an integral of  \eqref{EqAltFrobDer} not \eqref{EqAltFrobDer2} but $$x^{\gamma - 1} \Big( \sum_{\substack{j = 0 \\ j \neq 2}}^\infty \frac{e_j}{\gamma - 1 + j} x^j \Big) + e_2 \frac{x^{\gamma +1} - 1}{\gamma + 1} \;,$$ which differs from \eqref{EqAltFrobDer2} by an $\omega$-dependent constant and makes it analytic for all $\omega \in \R$.} of \eqref{EqSepTeukEqX} to $\omega = 0$ for which we have
\begin{equation}\label{EqU1}
U_{\CH,0l}(x; 0) = x^{-2} \Big( \sum_{\substack{j = 0 \\ j \neq 2}}^\infty \frac{e_j(0, 0 , l)}{-2 + j} x^j \Big) B_{\CH_r,0l}(x;0) + e_2(0,0,l) (\log x )\cdot B_{\CH_r,0l}(x;0) \;.
\end{equation}
Hence, $\{ B_{\CH_r,0l}(x;\omega), U_{\CH,0l}(x; \omega)\}$ is a fundamental system of solutions of \eqref{EqSepTeukEqX} for $m = 0$ which is defined and analytic for all $\omega \in \R$. (The linear independence of the solutions is shown below.)

Similarly we set
$$U_{\Hp,0l}(y; \omega) := A_{\Hp_r,0l}(y ;\omega) - \frac{\tilde{e}_2(\omega,0,l)}{\tilde{\gamma} +1} A_{\Hp_l,0l}(y;\omega)$$
to obtain a fundamental system $\{A_{\Hp_l,0l}(y;\omega), U_{\Hp,0l}(y; \omega)\}$ of solutions of \eqref{EqSepTeukEqY} for $m =0$ which is defined and analytic for all $\omega \in \R$. Moreover, $U_{\Hp,0l}(y; 0)$ is of the form
\begin{equation}\label{EqU2}
U_{\Hp,0l}(y; 0) = y^{-2} \Big( \sum_{\substack{ j = 0 \\ j \neq 2}}^\infty \frac{\tilde{e}_j(0,0,l)}{-2 + j} y^j\Big) A_{\Hp_l,0l}(y;0) + \tilde{e}_2(0,0,l) (\log y) \cdot A_{\Hp_l,0l}(y;0) \;.
\end{equation}
Let us also note that it follows from \eqref{EqE2OmegaM}, \eqref{Eq2OT} and  Section \ref{SecHeun} that $$e_2(0,0,l) = \tilde{e}_2(0,0,l) = -\frac{1}{2} \Big(\big[ (l-2)(l+3) + 4\big] \big[ (l-2)(l+3) + 6\big] \Big) \neq 0 $$ for $l \geq 2$ and thus the solutions \eqref{EqU1} and \eqref{EqU2} do indeed have $\log$ terms and are linearly independent from $ B_{\CH_r,0l}(x;0)$ and $A_{\Hp_l,0l}(y ;0)$, respectively.

We can now expand for all $\omega \in \R$
\begin{equation}\label{EqDEFXY}
\begin{split}
U_{\Hp, 0l}(1-x ; \omega) &= X_{\Hp,0l}(\omega) U_{\CH, 0l}(x ;\omega) + Y_{\Hp, 0l} ( \omega) B_{\CH_r, 0l}(x ; \omega) \\
A_{\Hp_l, 0l}(1-x; \omega) &= X_{\Hp_l, 0l}(\omega) B_{\CH_r, 0l}(x; \omega) + Y_{\Hp_l, 0l}(\omega) U_{\CH, 0l}(x ; \omega)
\end{split}
\end{equation}
where $X_{\Hp,0l},  Y_{\Hp, 0l}, X_{\Hp_l, 0l}, Y_{\Hp_l, 0l}$ are complex valued functions. It follows as in \eqref{EqTRAnalytic} that they are analytic on all of $\R$.

We now show that we have $Y_{\Hp_l, 0l}(0) = 0$. Recall that for $m=0$ and $\omega =0$ the coefficients of the hypergeometric equations which $ \{ B_{\CH_r,0l}(x;0), U_{\CH,0l}(x; 0)\}$  and  $\{U_{\Hp, 0l}(y ; 0), A_{\Hp_l, 0l}(y; 0)\}$ are satisfying are $\underline{c} = \tilde{\underline{c}} = -1$, $\underline{a} = l -1$, and $\underline{b} = -2-l$. It thus follows that \eqref{EqSecSolHyper} still defines a solution to the hypergeometric equation which is clearly linearly independent to $U_{\CH,0l}(x; 0)$ ($U_{\Hp, 0l}(y ; 0)$), since the latter contains a non-vanishing $\log$-term. By comparison of the leading order coefficients we thus obtain
$$B_{\CH_r,0l}(x;0) = x^{1 - \underline{c}} F(\underline{a} +1 - \underline{c}, \underline{b} + 1 - \underline{c}, 2 - \underline{c}; x) \qquad \textnormal{ and } \qquad A_{\Hp_l, 0l}(y; 0) = y^{1 - \underline{\tilde{c}}} F(\underline{a} +1 - \underline{\tilde{c}}, \underline{b} + 1 - \underline{\tilde{c}}, 2 - \underline{\tilde{c}}; y) \;.$$
Note that $\underline{b} +1 - \underline{c} =  \underline{b} +1 - \underline{\tilde{c}} =-l \in -\N$ and thus, as we observed in the proof of Proposition \ref{PropTRmNeq0}, $B_{\CH_r,0l}(x;0)$ and $A_{\Hp_l, 0l}(y; 0) $ are polynomials. Since $U_{\CH,0l}(x; 0)$ is clearly not a polynomial because of the $\log$-term, it directly follows from \eqref{EqDEFXY} that we must have $Y_{\Hp_l, 0l}(0) = 0$.

Expanding \eqref{EqDEFXY} in terms of our original systems of fundamental solutions gives
\begin{equation*}
A_{\Hp_l, 0l}(1-x; \omega) = \underbrace{\Big(X_{\Hp_l, 0l}(\omega) - Y_{\Hp_l, 0l}(\omega) \frac{e_2(\omega, 0, l)}{\gamma + 1} \Big)}_{= \mathfrak{T}_{\Hp_l, 0l}(\omega)} B_{\CH_r, 0l}(x ; \omega) + \underbrace{Y_{\Hp_l, 0l}(\omega)}_{= \mathfrak{R}_{\Hp_l, 0l}(\omega)} B_{\CH_l, 0l}(x; \omega) \;,
\end{equation*}
which directly shows that $\mathfrak{R}_{\Hp_l, 0l}(\omega)$ is analytic on $\R$ and vanishes at $\omega = 0$, and
\begin{equation}\label{EqTRForM0}
\begin{split}
A_{\Hp_r, 0l}&(1-x; \omega) = U_{\Hp, 0l}(1-x; \omega) + \frac{\tilde{e}_2(\omega, 0, l)}{\tilde{\gamma } +1} A_{\Hp_l, 0l}(1-x; \omega) \\
&=X_{\Hp, 0l}(\omega) U_{\CH, 0l}(x ; \omega) + Y_{\Hp, 0l}(\omega) B_{\CH_r, 0l}(x; \omega)  \\
&\qquad + \frac{\tilde{e}_2(\omega, 0 , l)}{\tilde{\gamma} +1} \Big( X_{\Hp_l, 0l}(\omega) B_{\CH_r, 0l}(x; \omega) + Y_{\Hp_l, 0l}(\omega) U_{\CH, 0l}(x; \omega)\Big) \\
&= \Big(\underbrace{X_{\Hp, 0l}(\omega) + \frac{\tilde{e}_2(\omega, 0, l)}{\tilde{\gamma} +1} Y_{\Hp_l, 0l}(\omega)}_{= \mathfrak{T}_{\Hp_r, 0l}(\omega)}\Big) B_{\CH_l, 0l}(x; \omega) \\
&\quad + \Big(\underbrace{Y_{\Hp, 0l}(\omega) + \frac{\tilde{e}_2(\omega, 0, l)}{\tilde{\gamma} + 1} X_{\Hp_l, 0l}(\omega) - \frac{e_2(\omega, 0, l)}{\gamma +1} \big[X_{\Hp, 0l}(\omega) + \frac{\tilde{e}_2(\omega, 0, l)}{\tilde{\gamma} +1} Y_{\Hp_l, 0l}(\omega)\big]}_{= \mathfrak{R}_{\Hp_r, 0l}(\omega)}\Big) B_{\CH_r, 0l}(x; \omega) \;.
\end{split}
\end{equation}
Since we have shown that $Y_{\Hp_l, 0l}(0) =0$ it follows that $\mathfrak{T}_{\Hp_r, 0l}(\omega)$ is analytic on all of $\R$. Moreover, we have $\mathfrak{R}_{\Hp_r, 0l}(\omega) = Y_{\Hp, 0l}(\omega) + \frac{\tilde{e}_2(\omega, 0, l)}{\tilde{\gamma} + 1} X_{\Hp_l, 0l}(\omega) - \frac{e_2(\omega, 0, l)}{\gamma +1} \mathfrak{T}_{\Hp_r, 0l}(\omega)$, from which it follows that $\omega \cdot \mathfrak{R}_{\Hp_r, 0l}(\omega)$ extends analytically to $\omega  = 0$.
\end{proof}

Note that \eqref{EqTRForM0} directly shows that our previous approach for $m \neq 0$ of showing that $\mathfrak{T}_{\Hp_r, ml}(0) \neq 0$, namely by computing the transmission coefficient for the simpler hypergeometric equation, does not directly transfer to $m =0$, since here we actually need to know the value of $\rd_\omega Y_{\Hp_l, 0l}(0)$, which is a statement that goes beyond the hypergeometric equation. The omega derivative can be computed -- however, it seems easier to use the Teukolsky-Starobinsky conservation law instead which has been made use of recently and developed in much detail in \cite{Teix20}, \cite{ShlaTeix20}\footnote{The Teukolsky-Starobinsky conservation law also provides an alternative approach to showing that $\mathfrak{T}_{\Hp_r, ml}(0) \neq 0$ for $m \neq 0$, cf.\ Remark \ref{RemMneq00}}. 

\subsubsection*{The Teukolsky-Starobinsky conservation law}

What is needed of the Teukolsky-Starobinsky conservation law for this paper can be  developed quite quickly, which keeps the paper self-contained. To make contact with \cite{Teix20}, \cite{ShlaTeix20} we begin by noting that $\widecheck{\psi}_{ml}^{[s]}(r; \omega)$ satisfies \eqref{EqThmSepVar3} if, and only if, $R_{ml}^{[s]}(r; \omega) := e^{im\overline{r}} e^{-i \omega r^*} \frac{1}{ \Delta^s} \widecheck{\psi}_{ml}^{[s]}(r; \omega)$ satisfies
\begin{equation}\label{EqTeukStandard}
\begin{split}
\Delta^{-s} \frac{d}{dr} \big( \Delta^{s+1} \frac{d R_{ml}^{[s]}}{dr} \big)  + \Big( &\big[ ( r^2 + a^2)^2 \omega^2 - 4aMr \omega m + a^2 m^2 + 2ia(r-M) ms - 2iM(r^2 - a^2) \omega s\big] \cdot \frac{1}{\Delta} \\
&+ 2ir \omega s + \lambda_{ml}^{[s]}(\omega) - a^2 \omega^2 \Big) R_{ml}^{[s]} = 0 \;.
\end{split}
\end{equation}
Note that \eqref{EqTeukStandard} is the radial Teukolsky equation in its most common form as it also appears for example in (150) of \cite{DafHolRod17} where, however, their $\lambda_{ml}^{[s]}$ differs from ours here by a minus sign. A direct computation, see also \cite{DafHolRod17}, gives furthermore that $R_{ml}^{[s]}(r; \omega)$ satisfies \eqref{EqTeukStandard} if, and only if, $u_{ml}^{[s]}(r; \omega) := \Delta^{\nicefrac{s}{2}} (r^2 + a^2)^{\nicefrac{1}{2}} R_{ml}^{[s]}(r; \omega)$ satisfies
\begin{equation}
\label{EqRescaledTeuk}
\frac{d^2}{(dr^*)^2} u_{ml}^{[s]}(r; \omega) + V^{[s]}_{ml}(r; \omega) u_{ml}^{[s]}(r; \omega) = 0 
\end{equation}
with 
\begin{equation*}
\begin{split}
V^{[s]}_{ml}(r; \omega) = &\frac{\Delta}{(r^2 + a^2)^2} \Big( \frac{\big((r^2 + a^2) \omega - am\big)^2 - 2is(r-M)\big((r^2 + a^2) \omega - am\big)}{\Delta} + 4is \omega r + \lambda_{ml}^{[s]}(\omega) -s - a^2 \omega^2 + 2 a m \omega\Big) \\
&\qquad -\frac{s^2 (r-M)^2}{(r^2 + a^2)^2} + \frac{\Delta}{(r^2 + a^2)^3} \Big(-2(r-M) r - \Delta + \frac{3r^2 \Delta}{r^2 + a^2}\Big)  \;.
\end{split}
\end{equation*}
Note that one has $V^{[s]}_{ml}(r; \omega) = \overline{V^{[-s]}_{ml}(r; \omega)}$, for which we recall $\lambda_{ml}^{[s]}(\omega) - s = \lambda_{ml}^{[-s]}(\omega) + s$ from Proposition \ref{PropEigenfunctionsSWL}. It follows that if $u_{ml}^{[-2]}$ is a solution of \eqref{EqRescaledTeuk} with $s = -2$, then $\overline{u_{ml}^{[-2]}}$ is a solution of \eqref{EqRescaledTeuk} with $s = +2$. Unwinding the above relations we find that if $\widecheck{\psi}_{ml}^{[-2]}$ satisfies $\eqref{EqThmSepVar3}$ (or \eqref{EqSepTeukEqX}) with $s = -2$ then $\overline{(r^2 + a^2)^{\frac{1}{2}} e^{im \overline{r}} e^{-i \omega r^*} \Delta \widecheck{\psi}_{ml}^{[-2]}}$ satisfies \eqref{EqRescaledTeuk} with $s = +2$ and thus 
\begin{equation}\label{EqStoMinusS} 
\Delta^2 e^{-2im \overline{r}} e^{2 i \omega r^*} \overline{\widecheck{\psi}_{ml}^{[-2]}}
\end{equation}
satisfies $\eqref{EqThmSepVar3}$ (or \eqref{EqSepTeukEqX}) with $s = +2$. 

Moreover, we observe that since \eqref{EqRescaledTeuk} does not have any first order terms, the Wronskian $$W_{r^*}(u_{ml}^{[s]}, w_{ml}^{[s]} ) := \big(\frac{d}{dr^*} u_{ml}^{[s]} \big) w_{ml}^{[s]} - u_{ml}^{[s]}\big(\frac{d}{dr^*}w_{ml}^{[s]}  \big) $$ is conserved in $r$ for any two solutions $u_{ml}^{[s]}(r^*; \omega)$ and $w_{ml}^{[s]}(r^*; \omega)$ of \eqref{EqRescaledTeuk}. Hence, if $v_{1, ml}^{[+2]}(x; \omega)$ and $v_{2,ml}^{[+2]}(x; \omega)$ are two solutions of \eqref{EqSepTeukEqX} with $s = +2$ then
\begin{equation}
\label{EqWronskianX}
\begin{split}
\mathrm{const} &= W_{r^*}\big((r^2 + a^2)^{\frac{1}{2}} \frac{1}{\Delta} e^{im \overline{r}} e^{-i \omega r^*} v_{1, ml}^{[+2]}, (r^2 + a^2)^{\frac{1}{2}} \frac{1}{\Delta} e^{im \overline{r}} e^{-i \omega r^*} v_{2, ml}^{[+2]}\big) \\
&= (r^2 + a^2) \frac{1}{\Delta^2} e^{2im \overline{r}} e^{-2i \omega r^*} W_{r^*}(v_{1, ml}^{+2]}, v_{2,ml}^{[+2]} ) \\
&=\frac{1}{\Delta} e^{2im \overline{r}} e^{-2i \omega r^*} \frac{1}{r_+ - r_-} \Big( \underbrace{\frac{d}{dx} v_{1, ml}^{[+2]} \cdot v_{2,ml}^{[+2]} - v_{1,ml}^{[+2]} \frac{d}{dx} v_{2, ml}^{[+2]} }_{=:W_x(v_{1,ml}^{[+2]}, v_{2,ml}^{[+2]})}\Big) \;,
\end{split}
\end{equation}
where we have used $\frac{d}{dr^*} = \frac{\Delta}{r^2 + a^2} \frac{1}{r_+ - r_-} \frac{d}{dx}$.

The Teukolsky-Starobinsky identities allow us to produce a solution for the $s=-2$ equation from one of the $s=+2$ equation -- and vice versa. Here, only the first direction is needed which is straightforward to establish for the radial Teukolsky equation in the form \eqref{EqSepTeukEqX}. We claim that if $v^{[+2]}$ is a solution to \eqref{EqSepTeukEqX} with $s=+2$, then $\frac{d^4}{dx^4} v^{[+2]}$ is a solution to \eqref{EqSepTeukEqX} with $s=-2$. To prove this we first note that \eqref{EqGreek} gives the following relation of the parameters of the Heun equation for $s\pm 2$
\begin{equation*}
\begin{aligned}
\alpha^{[-2]} &= \alpha^{[+2]} \qquad \qquad \qquad \qquad &&\delta^{[-2]} = \delta^{[+2]} + 8 \alpha^{[+2]} \\
\beta^{[-2]} &= \beta^{[+2]} - 8  &&\varepsilon^{[-2]} = \varepsilon^{[+2]} + 4 \beta^{[+2]} - 12 \\
\gamma^{[-2]} &= \gamma^{[+2]} + 4
\end{aligned}
\end{equation*}
where we have used again $\lambda_{ml}^{[s]}(\omega) - s = \lambda_{ml}^{[-s]}(\omega) + s$. Taking $\frac{d^4}{dx^4}$ of \eqref{EqSepTeukEqX} now gives
\begin{equation*}
\begin{split}
0 &= \sum_{j = 0}^4 \binom{4}{j} \Big[ \frac{d^j}{dx^j} \big((1-x)x \big) \frac{d^{4-j}}{dx^{4-j}} \frac{d^2}{dx^2} v^{[+2]} + \frac{d^j}{dx^j} \big( \alpha^{[+2]} x^2 + \beta^{[+2]} x + \gamma^{[+2]} \big) \frac{d^{4-j}}{dx^{4-j}} \frac{d}{dx} v^{[+2]} \\
&\qquad \qquad \qquad  + \frac{d^j}{dx^j} \big(\delta^{[+2]} x + \varepsilon^{[+2]}\big) \frac{d^{4-j}}{dx^{4-j}} v^{[+2]} \Big]\\
&= (1-x) x \frac{d^6}{dx^6} v^{[+2]} + 4(-2x + 1) \frac{d^5}{dx^5} v^{[+2]} + 6(-2) \frac{d^4}{dx^4} v^{[+2]}  +(\alpha^{[+2]} x^2 + \beta^{[+2]} x + \gamma^{[+2]}) \frac{d^5}{dx^5} v^{[+2]} \\
&\quad + 4(2 \alpha^{[+2]} x + \beta^{[+2]}) \frac{d^4}{dx^4} v^{[+2]} + 6 \cdot 2 \alpha^{[+2]} \frac{d^3}{dx^3} v^{[+2]} +(\delta^{[+2]} x + \varepsilon^{[+2]}) \frac{d^4}{dx^4} v^{[+2]} + 4 \delta^{[+2]} \frac{d^3}{dx^3} v^{[+2]} \\
&= (1-x)x\frac{d^2}{dx^2} \frac{d^4}{dx^4}v^{[+2]} + \big( \alpha^{[-2]} x^2 + \beta^{[-2]} x + \gamma^{[-2]} \big) \frac{d}{dx} \frac{d^4}{dx^4}v^{[+2]} + \big(\delta^{[-2]} x + \varepsilon^{[-2]}\big) \frac{d^{4}}{dx^{4}} v^{[+2]} \;,
\end{split}
\end{equation*}
where we have used $12 \alpha^{[+2]} + 4 \delta^{[+2]} = 0$.

We now apply this to the Frobenius solutions for $\omega \neq \omega_+m, \omega_- m$. Recalling $\frac{d}{dy} = -  \frac{d}{dx}$ we have
\begin{equation}
\label{EqPhiDeltaP}
\begin{split}
\frac{d^4}{dy^4} A_{\Hp_r, ml}^{[+2]} (y; \omega) &= \sum_{j=0}^\infty (j+4)(j+3)(j+2)(j+1) a_{j+4}^{[+2]}(\omega, m, l) y^j \\
\frac{d^4}{dx^4} B_{\CH_l, ml}^{[+2]} (x; \omega) &= \sum_{j=0}^\infty (j+4)(j+3)(j+2)(j+1) c_{j+4}^{[+2]}(\omega, m, l) x^j \\
\frac{d^4}{dx^4} B_{\CH_r, ml}^{[+2]} (x; \omega) &=\sum_{j=0}^\infty (j+1 - \gamma^{[+2]})(j - \gamma^{[+2]}) (j - 1 - \gamma^{[+2]})(j-2 - \gamma^{[+2]}) d_j^{[+2]}(\omega, m, l) x^{j-3- \gamma^{[+2]}} \;.
\end{split}
\end{equation}
With the notation from Section \ref{SecRelTwoPhis} and \ref{SecDelta} we find near $r = r_+$
\begin{equation*}
\begin{split}
e^{-2im \overline{r}} e^{2i \omega r^*} &= e^{2i r^*(\omega - \omega_+ m)} e^{2im \phi_+(r)} \\
&= (r_+ - r)^{\frac{i}{\kappa_+}(\omega - \omega_+m)} e^{2i F_+(r) \cdot ( \omega - \omega_+ m)} e^{2im \phi_+(r)} \\
&= \underbrace{(r_+ - r_-)^{\frac{4iMr_+}{r_+ - r_-}(\omega - \omega_+ m)} e^{2iF_+(r) \cdot (\omega - \omega_+m)} e^{2im \phi_+(r)}}_{=: \mathcal{D}_+(r; m, \omega)} \cdot \underbrace{y^{\frac{4iMr_+}{r_+ - r_-}(\omega - \omega_+m)}}_{= y^{-(1 + \tilde{\gamma}^{[+2]})}}
\end{split}
\end{equation*}
and near $r=r_-$
\begin{equation*}
e^{-2im \overline{r}} e^{2i \omega r^*} = \underbrace{(r_+ - r_-)^{\frac{-4iMr_-}{r_+ - r_-}(\omega - \omega_- m)} e^{2iF_-(r) \cdot (\omega - \omega_-m)} e^{2im \phi_-(r)}}_{=: \mathcal{D}_-(r; m, \omega)} \cdot \underbrace{x^{\frac{-4iMr_-}{r_+ - r_-}(\omega - \omega_-m)}}_{= x^{-(1 + \gamma^{[+2]})}} \;.
\end{equation*}
Note that $\mathcal{D}_\pm$ is regular at $r= r_\pm$ and that we have $|\mathcal{D}_{\pm}| = 1$. We also recall that $\Delta = (r_+ - r_-)^2 (x-1)x = (r_+ - r_-)^2 (y-1)y$.

By \eqref{EqStoMinusS}
\begin{equation*}
\begin{split}
\Delta^2 e^{-2im \overline{r}} &e^{2 i \omega r^*} \overline{\frac{d^4}{dx^4} B_{\CH_l, ml}^{[+2]} } \\
&= (r_+ - r_-)^4 (x-1)^2x^2 \mathcal{D}_-\big(r(x); m, \omega\big) x^{-(1 + \gamma^{[+2]})} \cdot \sum_{j=0}^\infty \prod_{k=1}^4(j+k) \overline{c_{j+4}^{[+2]}(\omega, m, l)} x^j
\end{split}
\end{equation*}
is a solution of \eqref{EqSepTeukEqX} with $s = +2$. Comparing asymptotics we find
\begin{equation}\label{EqBL}
\Delta^2 e^{-2im \overline{r}} e^{2 i \omega r^*} \overline{\frac{d^4}{dx^4} B_{\CH_l, ml}^{[+2]} } = (r_+ - r_-)^4 \mathcal{D}_-\big(r_-; m, \omega\big) 4!\cdot  \overline{c_4^{[+2]}(\omega, m, l)} \cdot B^{[+2]}_{\CH_r, ml}(x; \omega) \;.
\end{equation}
Similarly we have
\begin{equation}
\label{EqAHRY}
\Delta^2 e^{-2im \overline{r}} e^{2 i \omega r^*} \overline{\frac{d^4}{dy^4} A_{\Hp_r, ml}^{[+2]} } = (r_+ - r_-)^4 \mathcal{D}_+\big(r_+; m, \omega\big) 4! \cdot \overline{a_4^{[+2]}(\omega, m, l)} \cdot A^{[+2]}_{\Hp_l, ml}(y; \omega) \;.
\end{equation}
Again by \eqref{EqStoMinusS}
\begin{equation*}
\begin{split}
\Delta^2 e^{-2im \overline{r}} &e^{2 i \omega r^*} \overline{\frac{d^4}{dx^4} B_{\CH_r, ml}^{[+2]} } \\
&= (r_+ - r_-)^4 (x-1)^2x^2 \mathcal{D}_-\big(r(x); m, \omega\big) x^{-(1 + \gamma^{[+2]})} \cdot \sum_{j=0}^\infty \overline{\prod_{k=-2}^1(j+k-\gamma^{[+2]})) d_{j}^{[+2]}(\omega, m, l) x^{j-3-\gamma^{[+2]}}}
\end{split}
\end{equation*}
is a solution of \eqref{EqSepTeukEqX} with $s = +2$. Noting that $\overline{\gamma^{[+2]}} = - \gamma^{[+2]} -2$ and comparing asymptotics we find
\begin{equation}\label{EqBR}
\Delta^2 e^{-2im \overline{r}} e^{2 i \omega r^*} \overline{\frac{d^4}{dx^4} B_{\CH_r, ml}^{[+2]} } = (r_+ - r_-)^4 \mathcal{D}_-(r_-; m, \omega) \overline{\prod_{k=-2}^1 (k - \gamma^{[+2]})} B^{[+2]}_{\CH_l, ml}(x; \omega) \;.
\end{equation}
We now apply \eqref{EqWronskianX} to $A_{\Hp_r, ml}^{[+2]}(x; \omega)$ and $\Delta^2 e^{-2im \overline{r}} e^{2 i \omega r^*} \overline{\frac{d^4}{dy^4} A_{\Hp_r, ml}^{[+2]} }$. 
Using again \eqref{EqPhiDeltaP} and \eqref{EqAHRY} we obtain
\begin{equation}
\label{EqFirstHalf}
\begin{split}
\mathrm{const} &= \frac{1}{\Delta} e^{2im \overline{r}} e^{-2i \omega r^*} \frac{1}{r_+ - r_-} W_x\big( A^{[+2]}_{\Hp_r, ml}, \Delta^2 e^{-2im \overline{r}} e^{2 i \omega r^*} \overline{\frac{d^4}{dy^4} A_{\Hp_r, ml}^{[+2]} }\big) \\
&= \frac{1}{(r_+ - r_-)^3} \frac{1}{(y-1)y} \overline{\mathcal{D}_+(r; m, \omega)} y^{(\tilde{\gamma}^{[+2]}+1)} (r_+-r_-)^4 \mathcal{D}_+(r_+;m, \omega) 4! \cdot \overline{a_4^{[+2]}(\omega, m, l)} W_x(A^{[+2]}_{\Hp_r, ml}, A^{[+2]}_{\Hp_l, ml}) \\
&\to -(r_+ - r_-) 4! \cdot \overline{a_4^{[+2]}(\omega, m, l)}(1 - \tilde{\gamma}^{[+2]}) \;,
\end{split}
\end{equation}
for $y \to 0$, where we have used
\begin{equation*}
\begin{split}
-y^{\tilde{\gamma}^{[+2]}} W_x(A^{[+2]}_{\Hp_r, ml}, A^{[+2]}_{\Hp_l, ml}) &= \underbrace{y^{\tilde{\gamma}^{[+2]}} \Big( \frac{d}{dy} A^{[+2]}_{\Hp_r, ml} \cdot A^{[+2]}_{\Hp_l, ml}}_{\to 0} - A^{[+2]}_{\Hp_r, ml} \underbrace{\frac{d}{dy} A^{[+2]}_{\Hp_l, ml}}_{\sim (1- \tilde{\gamma}^{[+2]}) y^{- \tilde{\gamma}^{[+2]}}} \Big) \\
&\to - (1 - \tilde{\gamma}^{[+2]})
\end{split}
\end{equation*}
for $y \to 0$.
We now evaluate \eqref{EqFirstHalf} for $x \to 0$. Note that it follows from differentiating $$A^{[+2]}_{\Hp_r, ml}(1-x; \omega) = \mathfrak{T}^{[+2]}_{\Hp_r, ml}(\omega)B^{[+2]}_{\CH_l, ml}(x; \omega) + \mathfrak{R}^{[+2]}_{\Hp_r}(\omega) B^{[+2]}_{\CH_r, ml}(x; \omega)$$  and from \eqref{EqBL}, \eqref{EqBR} that
\begin{equation*}
\begin{split}
\Delta^2 e^{-2im \overline{r}} e^{2 i \omega r^*} \overline{\frac{d^4}{dx^4} A_{\Hp_r, ml}^{[+2]} } = &\overline{\mathfrak{T}^{[+2]}_{\Hp_r, ml}(\omega)} (r_+ - r_-)^4 \mathcal{D}_-(r_-;m, \omega) 4! \cdot \overline{c_4^{[+2]}(\omega, m,l)} B^{[+2]}_{\CH_r, ml}(x; \omega) \\
&+ \overline{ \mathfrak{R}^{[+2]}_{\Hp_r}(\omega)} (r_+ - r_-)^4 \mathcal{D}_-(r_-;m, \omega) \overline{\prod_{k=-2}^1(k - \gamma^{[+2]})} B^{[+2]}_{\CH_l, ml}(x; \omega) \;.
\end{split}
\end{equation*}
Hence, the constant from \eqref{EqFirstHalf} is also given by
\begin{equation}
\label{EqSecHalf}
\begin{split}
\mathrm{const} &= \frac{1}{\Delta} e^{2im \overline{r}} e^{-2i \omega r^*} \frac{1}{r_+ - r_-} W_x\Big( \mathfrak{T}^{[+2]}_{\Hp_r, ml}(\omega)B^{[+2]}_{\CH_l, ml}+ \mathfrak{R}^{[+2]}_{\Hp_r}(\omega) B^{[+2]}_{\CH_r, ml}, \\
&\qquad \quad  (r_+ - r_-)^4 \mathcal{D}_-(r_-; m, \omega) \big[ \overline{\mathfrak{T}^{[+2]}_{\Hp_r, ml}(\omega)} 4! \cdot \overline{c_4^{[+2]}(\omega, m,l)} B^{[+2]}_{\CH_r, ml} + \overline{ \mathfrak{R}^{[+2]}_{\Hp_r}(\omega)} \overline{\prod_{k=-2}^1(k - \gamma^{[+2]})} B^{[+2]}_{\CH_l, ml}\big] \Big) \\
&= \frac{r_+ - r_-}{x-1} \overline{\mathcal{D}_-(r; m, \omega)} \mathcal{D}_-(r_-; m, \omega) x^{\gamma^{[+2]}} W_x \big(B^{[+2]}_{\CH_l, ml}, B^{[+2]}_{\CH_r, ml}\big) \\
&\quad \cdot \Big( 4! \cdot \overline{c_4^{[+2]}(\omega, m,l)}  |\mathfrak{T}^{[+2]}_{\Hp_r, ml}(\omega)|^2 - \overline{\prod_{k=-2}^1(k - \gamma^{[+2]})}  | \mathfrak{R}^{[+2]}_{\Hp_r}(\omega)|^2 \Big) \\
&\to (r_+ - r_-) (1 - \gamma^{[+2]}) \Big( 4! \cdot \overline{c_4^{[+2]}(\omega, m,l)}  |\mathfrak{T}^{[+2]}_{\Hp_r, ml}(\omega)|^2 - \overline{\prod_{k=-2}^1(k - \gamma^{[+2]})}  | \mathfrak{R}^{[+2]}_{\Hp_r}(\omega)|^2 \Big) \;,
\end{split}
\end{equation}
for $x \to 0$, where we have used
\begin{equation*}
\begin{split}
x^{\gamma^{[+2]}} W_x \big(B^{[+2]}_{\CH_l, ml}, B^{[+2]}_{\CH_r, ml}\big) &= \underbrace{x^{\gamma^{[+2]}}  \Big( \frac{d}{dx} B^{[+2]}_{\CH_l, ml} \cdot B^{[+2]}_{\CH_r, ml}}_{\to 0} -   B^{[+2]}_{\CH_l, ml} \cdot \underbrace{\frac{d}{dx} B^{[+2]}_{\CH_r, ml}}_{\sim (1- \gamma^{[+2]}) x^{-\gamma^{[+2]}}} \Big) \\
&\to -(1 - \gamma^{[+2]}) 
\end{split}
\end{equation*}
as $x \to 0$.
From \eqref{EqFirstHalf} and \eqref{EqSecHalf} we now obtain the conservation law
\begin{equation}
\label{EqPrelimConsLaw}
-4! \cdot \overline{a_4^{[+2]}(\omega, m, l)}(1 - \tilde{\gamma}^{[+2]}) = (1 - \gamma^{[+2]}) \cdot \Big( 4! \cdot \overline{c_4^{[+2]}(\omega, m,l)}  |\mathfrak{T}^{[+2]}_{\Hp_r, ml}(\omega)|^2 - \overline{\prod_{k=-2}^1(k - \gamma^{[+2]})}  | \mathfrak{R}^{[+2]}_{\Hp_r, ml}(\omega)|^2 \Big) 
\end{equation}
which is valid for $\omega \neq \omega_+m, \omega_-m$. From now on again we will drop the superscript $s=+2$.

We evaluate the coefficients next.
We set $\xi:= \frac{4iMr_-}{r_+ -r_-}(\omega - \omega_-m)$ and $\tilde{\xi} := - \frac{4iMr_+}{r_+ - r_-}(\omega - \omega_+m)$. Let us agree that $\xi(0)$ and $\tilde{\xi}(0)$ refer to $\xi(\omega = 0) = -\frac{2iam}{r_+ -r_-} = -\tilde{\xi}(\omega = 0)$. Then $\tilde{\gamma} = -1 + \tilde{\xi}$ and $\gamma = -1 + \xi$. We observe that
\begin{equation}\label{EqProd}
\overline{\prod_{k=-2}^1(k - \gamma)} = \overline{(-1 - \xi)(-\xi)(1-\xi)(2-\xi)}= \overline{\xi(2-\xi) |1+\xi|^2} = -\xi(2+\xi)|1+\xi|^2\;.
\end{equation}
The recursion relation \eqref{EqRecTeukXC} gives
\begin{equation}\label{EqRec22}
\begin{aligned}
c_2 &= \frac{1}{2(1 + \gamma)} \Big[ \frac{\varepsilon}{\gamma}(\beta + \varepsilon) - \delta\Big] \\
c_3 &= \frac{1}{3(2 + \gamma)} \Big[ c_2\big(2(1 - \beta) - \varepsilon\big) + c_1(-\alpha - \delta)\Big] \\
c_4 &= \frac{1}{4(3 + \gamma)} \Big[ c_3\big(3(2- \beta) - \varepsilon\big) + c_2(- 2 \alpha - \delta) \Big]\;,
\end{aligned}
\end{equation}
and similarly for the $a_i$, where all parameters are replaced by their tilded analogues.

\textbf{We now proceed by setting $m=0$. However, see Remark \ref{RemMneq00} for $m \neq 0$.} For $m=0$ \eqref{EqPrelimConsLaw} is valid for $\omega \neq 0$. We will show that if we multiply by $\omega$ then both sides extend analytically to $\omega = 0$.

From \eqref{EqRec22} and \eqref{EqParaOmega0} we obtain successively 
\begin{equation*}
\begin{aligned}
\lim_{\omega \to 0} (1 + \gamma) c_2 &= - \frac{1}{2} \big[ \varepsilon(0) \big(2 + \varepsilon(0)\big)\big] \\
\lim_{\omega \to 0} (1 + \gamma) c_3 &= \frac{1}{6} \big[ \varepsilon(0) \big(2 + \varepsilon(0)\big)^2\big] \\
\lim_{\omega \to 0} (1 + \gamma) c_4 &= -\frac{1}{48} \big(\varepsilon(0)\big)^2 \big(2 + \varepsilon(0)\big)^2\;,
\end{aligned}
\end{equation*}
where $\varepsilon(0) = (l-2)(l+3) +4 >0$. Thus
\begin{equation}\label{Eq111}
\begin{split}
\lim_{\omega \to 0} \omega \overline{c_4(\omega, 0, l)} &= \lim_{\omega \to 0} \frac{r_+ - r_-}{4iM r_-} \xi \overline{c_4(\omega, 0, l)} = - \lim_{\omega \to 0} \frac{r_+ - r_-}{4iMr_-} \overline{\xi c_4(\omega, 0, l)} = - \lim_{\omega \to 0} \frac{r_+ - r_-}{4iMr_-} \overline{(1 + \gamma) c_4(\omega, 0, l)} \\
&= \frac{r_+ - r_-}{4iMr_-} \frac{1}{48} \big(\varepsilon(0)\big)^2 \big(2 + \varepsilon(0)\big)^2
\end{split}
\end{equation}
and similarly
\begin{equation}\label{Eq222}
\lim_{\omega \to 0} \omega \overline{a_4(\omega, 0, l)} =  \lim_{\omega \to 0} - \frac{r_+ - r_-}{4iMr_+} \tilde{\xi} \overline{a_4(\omega, 0, l)} = - \frac{r_+ - r_-}{4iMr_+} \frac{1}{48} \big(\varepsilon(0)\big)^2 \big( 2 + \varepsilon(0)\big)^2 \;.
\end{equation}
Multiplying \eqref{EqPrelimConsLaw} by $\omega$ and using \eqref{EqProd} we get
\begin{equation}\label{Eqhjk}
-4! \cdot \omega \overline{a_4(\omega, 0, l)}(1 - \tilde{\gamma}) = (1 - \gamma) \cdot \Big( 4! \cdot \omega \overline{c_4(\omega, 0,l)}  |\mathfrak{T}_{\Hp_r, 0l}(\omega)|^2 + \frac{4iMr_-}{r_+ - r_-} (2 + \xi)|1+\xi|^2 |\omega \mathfrak{R}_{\Hp_r, 0l}(\omega)|^2 \Big) \;.
\end{equation}
By Proposition \ref{PropTRM0} $\mathfrak{T}_{\Hp_r, 0l}(\omega)$ and $\omega \mathfrak{R}_{\Hp_r, 0l}(\omega)$ extend analytically to $\omega =0$. By the above, $\omega \cdot c_4(\omega, 0, l)$ and $\omega \cdot a_4(\omega, 0, l)$ do as well. We may thus take the limit $\omega \to 0$ in \eqref{Eqhjk} to obtain
\begin{equation*}
\frac{r_+ - r_-}{4Mr_+} \big(\varepsilon(0)\big)^2 \big(2 + \varepsilon(0)\big)^2 = \frac{r_+ - r_-}{4Mr_-} \big(\varepsilon(0)\big)^2\big(2 + \varepsilon(0)\big)^2 |\mathfrak{T}_{\Hp_r, 0l}(0)|^2 - \frac{16 M r_-}{r_+ -r_-} | \lim_{\omega \to 0} \omega \mathfrak{R}_{\Hp_r, 0l}(\omega)|^2
\end{equation*}
where we used \eqref{Eq111} and \eqref{Eq222}. Brining the last term over to the left hand side this in particular implies the following

\begin{proposition}
We have $\mathfrak{T}^{[s]}_{\Hp_r, 0l}(0) \neq 0$.
\end{proposition}

We conclude this section with the following

\begin{remark} \label{RemMneq00}
For $m \neq 0$ the conservation law \eqref{EqPrelimConsLaw} may be evaluated directly at $\omega = 0$. A direct computation using \eqref{EqRec22} gives $\overline{c_4(0,m,l)} = \frac{[\varepsilon(0)]^2 [2+ \varepsilon(0)]^2}{24} \frac{1}{|1 + \xi|^2 (2 - \xi) \xi}$ and $\overline{a_4(0,m,l) } = \frac{[\varepsilon(0)]^2 [2+ \varepsilon(0)]^2}{24} \frac{1}{|1 + \tilde{\xi}|^2 (2 - \tilde{\xi}) \tilde{\xi}}$. Plugging those values into \eqref{EqPrelimConsLaw}, together with \eqref{EqProd}, gives
\begin{equation*}
\begin{split}
[\varepsilon(0)]^2 [2+ \varepsilon(0)]^2 &\frac{r_+-r_-}{2am \cdot |1+\tilde{\xi}(0)|^2} + \frac{2am}{r_+ - r_-} |2+\xi(0)|^2 |1+\xi(0)|^2 | \mathfrak{R}_{\Hp_r, ml}(0)|^2 \\
&= [\varepsilon(0)]^2 [2+ \varepsilon(0)]^2 \frac{r_+-r_-}{2am \cdot |1+\xi(0)|^2} |\mathfrak{T}_{\Hp_r, ml}(0)|^2 \;.
\end{split}
\end{equation*}
This would have been another way of showing that $\mathfrak{T}_{\Hp_r, ml}(0) \neq 0$ for $m \neq 0$. However, the approach taken in Section \ref{SecMneq0} is more direct. Note that if we use the additional information that $\mathfrak{R}_{\Hp_r, ml}(0) = 0$, which was shown in Section \ref{SecMneq0}, then we recover that $|\mathfrak{T}_{\Hp_r, ml}(0)| = 1$, which is of course compatible with Remark \ref{RemExplicitT}.
\end{remark}

\section{Determination of the coefficients $a_{\Hp_l, ml}(\omega)$ and $a_{\Hp_r, ml}(\omega)$ in terms of the initial data on $\Hp_l$ and $\Hp_r$}
\label{SecDetA}

We will replace in this section the $r$-coordinate by the $y$-coordinate for convenience. Recall that $y = \frac{r_+ - r}{r_+ - r_-}$.
Also recall from Theorem \ref{ThmSeparationVariables} and \eqref{EqDefAA} the representation 
\begin{equation}\label{EqRecallExp}
\psi(v_+,y, \theta, \varphi_+) = \frac{1}{\sqrt{2\pi}} \int_\R \sum_{m,l} \big[\underbrace{ a_{\Hp_r, ml}(\omega) A_{\Hp_r, ml}(y; \omega) + a_{\Hp_l, ml}(\omega) A_{\Hp_l, ml}(y; \omega) }_{=\widecheck{\psi}_{ml}(y; \omega)} \big]  Y^{[s]}_{ml}(\theta, \varphi_+; \omega) e^{-i\omega v_+} \, d\omega \;.
\end{equation}
Note that while we know that for example $A_{\Hp_r, ml}(y; \omega)$ has a pole at $\omega = \omega_+ m$, we know that the terms in the linear combination conspire so that the total underbraced term is more regular, in particular continuous  in $\omega$ for $y \in (0,1)$.
In this section we will relate the coefficients $a_{\Hp_r, ml}(\omega)$ and $a_{\Hp_l, ml}(\omega)$ to the initial data on $\Hp_r$ and $\Hp_l$  (at least in a neighbourhood of $\omega = 0$).

\subsection{Passing to the limit $r \to r_+$ in $(v_+,r, \theta, \varphi_+)$ coordinates}
\label{SecPassLimitHPR}

We begin with the following

\begin{lemma} \label{LemPsiFTC1}
Under the assumptions from Section \ref{SecAssumptions} we have $\widecheck{(\psi|_{\Hp_r})}_{ml}(\omega) \in C^0(\R, \C)$.
\end{lemma}

\begin{proof}
Recall that 
\begin{equation}
\label{EqDiffOmegaIntegral}
\widecheck{(\psi|_{\Hp_r})}_{ml}(\omega) = \frac{1}{\sqrt{2\pi}} \int_{\Sp^2} \int_{\R} \psi|_{\Hp_r}(v_+, \theta, \varphi_+) e^{i \omega v_+} S_{ml}^{[s]}(\cos \theta; \omega) e^{-im \varphi_+} \, dv_+ \vols \;.
\end{equation}
It follows from the exponential decay of $\psi|_{\Hp_r}$ in $v_+$ for $v_+ \to - \infty$ (by Assumption \ref{AssumptionReg}) together with \eqref{AssumpDecayRHp} and $q_r$ being in particular bigger than $1$ that
\begin{equation*}
\begin{split}
\int\limits_{\R}\int\limits_{\Sp^2} \big| &\psi|_{\Hp_r}(v_+, \theta, \varphi_+)\big| \, \vols  dv_+ \leq 2 \sqrt{\pi} \int\limits_{\R} \Big( \int\limits_{\Sp^2} \big| \psi|_{\Hp_r}(v_+, \theta, \varphi_+)\big|^2 \, \vols \Big)^{\nicefrac{1}{2}} \, dv_+ \\
&\leq 2 \sqrt{\pi} \Big( \int\limits_{\R} \frac{1}{(1+|v_+|)^{q_r}} \, dv_+ \Big)^{\nicefrac{1}{2}} \cdot \Big( \int\limits_{\R} \int\limits_{\Sp^2} (1+|v_+|)^{q_r} \big|\psi|_{\Hp_r}(v_+, \theta, \varphi_+)\big|^2 \, \vols dv_+ \Big)^{\nicefrac{1}{2}} < \infty \;.
\end{split}
\end{equation*}
Together with the boundedness of $S_{ml}^{[s]}(\cos \theta; \omega)$ and its continuous dependence on $\omega$ the result now follows from dominated convergence.
\end{proof}

\begin{proposition}\label{PropAHP}
Let $\psi$ satisfy the assumptions from Section \ref{SecAssumptions}. We then have $a_{\Hp_r, ml}(\omega) = (\widecheck{\psi|_{\Hp_r}})_{ml}(\omega)$ for all $\omega \neq \omega_+m$. In particular $a_{\Hp_r, ml}$ extends as a $C^0$ function to $\omega = \omega_+m$.
\end{proposition}

\begin{proof}
By \eqref{EqCorL2Limit} of Corollary \ref{CorLimitHP} we have $\int_{\R \times \Sp^2}|\psi(v_+, y, \theta, \varphi_+) - \psi(v_+, 0, \theta, \varphi_+)|^2 \, \vols dv_+ \to 0$ for $y \to 0$. By Plancherel \eqref{EqPlancherel} this gives
$$\int_{\R} \sum_{m,l} |(\widecheck{\psi})_{ml}(y; \omega) - (\widecheck{\psi|_{\Hp_r}})_{ml}(\omega)|^2 \, d \omega \to 0 \quad \textnormal{ for } y \to 0 \;. $$
Fix $m$ and $l$. It now follows that there is a sequence $y_n \to 0$ with $\widecheck{\psi}_{ml}(y_n; \omega) \to (\widecheck{\psi|_{\Hp_r}})_{ml}(\omega)$ for almost every $\omega \in \R$. For $\omega \neq \omega_+ m$ we have
$$\widecheck{\psi}_{ml}(y_n; \omega) = a_{\Hp_r, ml}(\omega) A_{\Hp_r, ml}(y_n; \omega) + a_{\Hp_l, ml}(\omega) A_{\Hp_l, ml}(y_n; \omega) \to a_{\Hp_r, ml}(\omega)$$ as $y_n \to 0$ by the normalisation of the Frobenius solutions -- and thus we have $a_{\Hp_r, ml}(\omega) = (\widecheck{\psi|_{\Hp_r}})_{ml}(\omega)$ for a.e.\ $\omega \in \R \setminus \{\omega_+m\}$. Since both functions are continuous on $\R \setminus \{\omega_+m\}$ by Lemmas \ref{LemRegAA} and \ref{LemPsiFTC1} they agree everywhere.
\end{proof}

\begin{proposition} \label{PropFourierAssump}
The assumptions \eqref{AssumpP0}, \eqref{AssumpP0SphHarm}, and \eqref{AssumpP0Better} from Section \ref{SecAssumptions}, together with the regularity Assumption \ref{AssumptionReg}, imply that $\rd_\omega^q(\widecheck{\psi|_{\Hp_r}})_{m_0l_0} \in L^2_\omega([-2,2])$ for any $0 \leq q < p_0$, $q \in \N_0$ and $\rd_\omega^{p_0} (\widecheck{\psi|_{\Hp_r}})_{m_0l_0}(\omega) \notin L^2_{\omega}(-\varepsilon, \varepsilon)$ for any $\varepsilon>0$.
\end{proposition}

\begin{proof}
We drop the $|_{\Hp_r}$ from $\psi|_{\Hp_r}$ here to ease the notation. We only consider $\psi$ restricted to the event horizon. 
It follows from  the regularity Assumption \ref{AssumptionReg}, which ensures exponential decay of $\psi$ for $v_+ \to -\infty$, that \eqref{AssumpP0}, \eqref{AssumpP0SphHarm}, \eqref{AssumpP0Better} imply
\begin{equation}\label{EqPhysSpaceHp}
\begin{split}
&\int_\R \int_{\Sp^2} \big|v_+^q \psi\big|^2 \, \vols dv_+ < \infty \\[3pt]
&\int_\R \int_{\Sp^2} \big|v_+^{p_0}\rd_{v_+} \psi\big|^2 \, \vols dv_+ < \infty  \\[3pt]
&\int_\R  \big|v_+^{p_0} \psi_{S(m_0l_0)}\big|^2 \,  dv_+ = \infty  
\end{split}
\end{equation}
for all $0 \leq q < p_0$, $q \in \N_0$. 
Recall from Section \ref{SecAssumptions} that $$\widecheck{\psi}_{S(ml)}(\omega) := \int_{\Sp^2} \widecheck{\psi}(\omega, \theta, \varphi_+) \overline{Y_{ml}^{[s]}(\theta, \varphi_+;0)} \, \vols$$ denotes the projection of the Fourier transform $\widecheck{\psi} \in L^2_\omega L^2_{\Sp^2}$ of $\psi \in L^2_{v_+}L^2_{\Sp^2}$ (see \eqref{EqDefFT}) onto the spin $2$-weighted \underline{spherical} harmonic $Y^{[s]}_{ml}(\theta, \varphi_+; 0)$. The map $(\cdot)_{S(ml)}$ is clearly an isometry $L^2_\omega L^2_{\Sp^2} \to L^2_\omega \ell^2_{ml}$.
By Plancherel \eqref{EqPhysSpaceHp} is thus equivalent to
\begin{align}
&\int_\R \sum_{m,l} \big|\rd_{\omega}^q \widecheck{\psi}_{S(ml)}\big|^2 \, d\omega < \infty  \label{EqFTH1} \\[3pt]
&\int_\R \sum_{m,l} \big|\rd_{\omega}^{p_0}(\omega \widecheck{\psi}_{S(ml)})\big|^2 \, d\omega < \infty  \label{EqFTH2} \\[3pt]
&\rd_\omega^{p_0} (\widecheck{\psi}_{S(m_0l_0)} )\notin L^2(\R)  \label{EqFTH3} 
\end{align}
for all $0 \leq q < p_0$, $q \in \N_0$. It follows from
$$\rd_{\omega}^{p_0} ( \omega \widecheck{\psi}_{S(ml)}) = \rd_{\omega}^{p_0-1} ( \widecheck{\psi}_{S(ml)} + \omega \rd_\omega \widecheck{\psi}_{S(ml)}) = p_0 \rd_\omega^{p_0-1} \widecheck{\psi}_{S(ml)} + \omega \rd_\omega^{p_0} \widecheck{\psi}_{S(ml)} $$
and \eqref{EqFTH1} and \eqref{EqFTH2} that
\begin{equation}\label{EqFTH4}
\int_\R  \sum_{m,l}\big|\omega \rd_{\omega}^{p_0} \widecheck{\psi}_{S(ml)}\big|^2 \, d\omega < \infty \;.
\end{equation}
Together with \eqref{EqFTH3} this gives in particular
\begin{equation}\label{EqNotL2SS}
\rd_\omega^{p_0} (\widecheck{\psi}_{S(m_0l_0)}) \notin L^2_{(-\varepsilon, \varepsilon)}  \; \textnormal{ for any } \; \varepsilon >0 \;.
\end{equation} 
We relate the projection onto the spin $2$-weighted spheroidal harmonics to that onto the spin $2$-weighted spherical harmonics in the next step.

We expand in $L^2([-1,1], d \cos \theta)$
$$S_{ml}^{[s]}(\cos \theta; \omega) = \underbrace{\int_{[-1,1]} S_{ml}^{[s]}(\cos \theta; \omega) S_{ml'}^{[s]}(\cos \theta; 0) \, d \cos \theta}_{=:E^{[s]}_{mll'}(\omega)} \;\cdot\, S^{[s]}_{ml'}(\cos \theta; 0) \;,$$
where $E_{mll'}^{[s]}(\omega) : \ell^2_{l'} \to \ell^2_{l}$ is a change of orthonormal basis map for every $\omega \in \R$ and it is also smooth in $\omega$. We have
\begin{equation}\label{EqBoundDE}
\sum_{l'} |\rd_\omega^q E_{mll'}^{[s]}(\omega)|^2 = \sum_{l'} \Big|\int_{[-1,1]} \rd_\omega^q S^{[s]}_{ml}(\cos \theta; \omega) S^{[s]}_{ml'}(\cos \theta; 0) \, d \cos \theta\Big|^2 = ||\rd_\omega^q S^{[s]}_{ml}(\omega)||_{L^2([-1,1])}^2 \leq C(\omega, m,l)
\end{equation}
where, for fixed $m,l$, the constant can be chosen uniform on compact subsets of $\omega$ by the smoothness of $S^{[s]}_{ml}(\omega)$ in $\omega$, see Proposition \ref{PropEigenfunctionsSWL}.
We have $\widecheck{\psi}_{m_0l_0}(\omega) = \sum_l \widecheck{\psi}_{S(m_0l)}(\omega) \cdot E^{[s]}_{m_0l_0l}(\omega)$ in $L^2_\omega(\R)$ and weak differentiation gives
\begin{equation}\label{EqDiffE}
\rd_\omega^{p_0} \widecheck{\psi}_{m_0l_0}(\omega) = \sum_{l } \sum_{{q'} = 0}^{p_0} \binom{p_0}{{q'}} \rd_\omega^{q'} E^{[s]}_{m_0l_0l}(\omega) \cdot \rd_\omega^{p_0 -{q'}}\widecheck{\psi}_{S(m_0l)}(\omega)  \;.
\end{equation}
Consider first all the terms  $\sum_l \rd_\omega^{q'} E^{[s]}_{m_0l_0l}(\omega) \cdot \rd_\omega^{p_0-{q'}} \widecheck{\psi}_{S(m_0l)}(\omega)$ for ${q'} \geq 1$. We estimate those on the compact subset $[-2,2] \subseteq \R$ using \eqref{EqFTH1} and \eqref{EqBoundDE} as follows:
\begin{equation*}
\int\limits_{[-2,2]} \Big| \sum_l \rd_\omega^{q'} E^{[s]}_{m_0l_0l}(\omega) \cdot \rd_\omega^{p_0 - {q'}} \widecheck{\psi}_{S(m_0l)}(\omega) \Big|^2 \, d \omega \leq \int\limits_{[-2,2]}  \underbrace{\Big(\sum_l \big| \rd_\omega^{q'} E^{[s]}_{m_0l_0l}(\omega)\big|^2 \Big)}_{\leq C(m_0,l_0)} \Big( \sum_l \big| \rd_\omega^{p_0 - {q'}} \widecheck{\psi}_{S(m_0l)}(\omega) \big|^2 \Big) \, d \omega \leq C \;.
\end{equation*}
Note in particular that if we replace $p_0$ in \eqref{EqDiffE} by $0 \leq q <p_0$ then all terms can be estimated in this way. This proves the first claim in the proposition. We go back to \eqref{EqDiffE} with $p_0$ and consider next those terms with ${q'}=0$ and $l \neq l_0$. We note that for $l \neq l_0$ we have $E^{[s]}_{m_0l_0l}(0) = 0$ and thus $E^{[s]}_{m_0l_0l}(\omega) = \int_0^\omega \rd_\omega E_{m_0l_0l}^{[s]}(\omega') \, d \omega'$. Using this we estimate  
\begin{equation}\label{Eq765}
\int\limits_{[-2,2]} \Big| \sum_{l \neq l_0} E^{[s]}_{m_0l_0l}(\omega) \cdot \rd_\omega^{p_0} \widecheck{\psi}_{S(m_0l)}(\omega) \Big|^2 \, d \omega \leq \int\limits_{[-2,2]} \Big( \sum_{l \neq l_0} \big| \frac{1}{\omega} \int_0^\omega \rd_\omega E^{[s]}_{m_0l_0l}(\omega') \, d\omega' \big|^2 \Big) \cdot \Big( \sum_{l \neq l_0} \big| \omega \rd_\omega^{p_0} \widecheck{\psi}_{S(m_0l)}(\omega)\big|^2 \Big) \, d \omega \;.
\end{equation}
We continue estimating the first factor on the right hand side for $\omega \in [-2,2]$ using \eqref{EqBoundDE}
\begin{equation*}
\sum_{l \neq l_0} \big| \frac{1}{\omega} \int_0^\omega \rd_\omega E^{[s]}_{m_0l_0l}(\omega') \, d\omega' \big|^2 \leq \frac{1}{|\omega|} \sum_{l \neq l_0} \int\limits_{[0, \omega]} |\rd_\omega E_{m_0l_0l}^{[s]}(w')|^2 \, d \omega' = \frac{1}{|\omega|} \int\limits_{[0, \omega]} \underbrace{\sum_{l \neq l_0}  |\rd_\omega E_{m_0l_0l}^{[s]}(w')|^2}_{\leq C(m_0,l_0)} \, d \omega'  \leq C \;.
\end{equation*}
Using this in \eqref{Eq765} together with \eqref{EqFTH4} gives $$\int\limits_{[-2,2]} \Big| \sum_{l \neq l_0} E^{[s]}_{m_0l_0l}(\omega) \cdot \rd_\omega^{p_0} \widecheck{\psi}_{S(m_0l)}(\omega) \Big|^2 \, d \omega  \leq C \;.$$
It remains the term with ${q'}=0$ and $l=l_0$ in \eqref{EqDiffE}, which is $E^{[s]}_{m_0l_0l_0}(\omega) \cdot \rd_\omega^{p_0} \widecheck{\psi}_{S(m_0l_0)}$. Note that we have $E^{[s]}_{m_0l_0l_0}(0)=1$ and thus we can find $\varepsilon'>0$ such that $E^{[s]}_{m_0l_0l_0}(\omega) \geq \frac{1}{2}$ for $|\omega|\leq \varepsilon'$. It thus follows from \eqref{EqNotL2SS} that this term is not in $L_{\omega}^2(-\varepsilon, \varepsilon)$ for any $\varepsilon>0$. Entering all this information into \eqref{EqDiffE} concludes the proof.
\end{proof}

\begin{corollary} \label{CorConHPR}
Under the assumption from Section \ref{SecAssumptions} we have $\rd_\omega^qa_{\Hp_r, m_0l_0} \in L^2_\omega([-2,2])$ for any $0 \leq q < p_0$, $q \in \N_0$ and $\rd_\omega^{p_0} a_{\Hp_r, m_0l_0}(\omega) \notin L^2_{\omega}(-\varepsilon, \varepsilon)$ for any $\varepsilon>0$.
\end{corollary}

\begin{proof}
This follows directly from Propositions \ref{PropAHP} and \ref{PropFourierAssump}.
\end{proof}

\subsection{Passing to the limit $r \to r_+$ in $(v_-,r, \theta, \varphi_-)$ coordinates}

The determination of $a_{\Hp_l, ml}(\omega)$ is more complicated. Note that $\psi$ vanishes on $\Hp_l$, so in order to take a non-vanishing limit we consider $(\rd_r|_+)^2 \psi = \frac{1}{(r_+ - r_-)^2} (\rd_y|_+)^2 \psi =: \frac{1}{(r_+ - r_-)^2} \rd_y^2 \psi$ instead of $\psi$. \textbf{Here, and throughout this section, we have made the convention that $\rd_y^2 \psi$ is always with respect to the $(v_+, y, \theta, \varphi_+)$-coordinate system, even if we otherwise use $(v_-, y, \theta, \varphi_-)$-coordinates.} This is simply to ease the amount of notation.
There are now two main differences to the limiting procedure of Section \ref{SecPassLimitHPR}. The first one is that $\rd_y^2 \psi$ does not vanish at the bottom bifurcation sphere, so one cannot hope to take an $L^2$-limit $y \to 0$ in $(v_-, y, \theta, \varphi_-)$-coordinates. We will instead take a limit in the sense of distributions. The second difference is that the branch $A_{\Hp_r, ml}(y; \omega)$ in \eqref{EqRecallExp} in general also gives a non-vanishing contribution under this limit (see Footnote \ref{FootnDelta}) -- by choosing the support of the test functions suitably though and, in the case of $m_0 =0$, also using that the reflection coefficient of the left event horizon vanishes at $\omega = 0$ (see Corollary \ref{CorConcAHL} and  Section \ref{SecPfMainThm}), we can circumvent this second difficulty.   We begin with introducing our test functions.

\begin{lemma} \label{LemXiFasterPolynom}
Let $\xi \in C^{\infty}_0(\R)$ and set $$\tau_{\xi, ml}(v, \theta) := \frac{1}{\sqrt{2\pi}} \int_\R \xi(\omega) S_{ml}^{[s]}(\cos \theta ; \omega) e^{i \omega v} \, d \omega \;.$$
Then $||v^j \rd_{v}^k \tau_{\xi, ml} ||_{L^\infty\big(\R \times (- \pi, \pi)\big)} \leq C(j,k) < \infty$ for all $j,k \in \N_0$, where $C(j,k)$ also depends on $m,l$ and $\xi$.
\end{lemma}

\begin{proof}
Differentiating under the integral we compute
\begin{equation*}
\begin{split}
(iv)^j (-i \rd_{v})^k \tau_{\xi, ml} &= \frac{1}{\sqrt{2\pi}} \int_\R \xi(\omega) S_{ml}^{[s]}(\cos \theta; \omega) \omega^k \rd_\omega^j e^{i \omega v} \, d \omega \\
&= \frac{1}{\sqrt{2\pi}} \int_\R (-1)^j \rd_\omega^j \big( \omega^k \xi(\omega) S_{ml}^{[s]}(\cos \theta ; \omega) \big) e^{i \omega v} \, d\omega \;.
\end{split}
\end{equation*}
Since $S_{ml}^{[s]}(\cos \theta ;\omega)$ and all its $\omega$-derivatives are continuous on $[-1,1] \times \R$ and since $\xi(\omega)$ is smooth and of compact support, the $L^\infty$ norm in $ \theta$ of the integrand is absolutely integrable.
\end{proof}

\begin{proposition} \label{PropLimitDis}
Let $\xi \in C^\infty_0(\R)$ and consider the assumptions in Section \ref{SecAssumptions}. Then as $y \to 0$
\begin{equation}\label{EqPropLimitDis}
\begin{split}
\int_{\R} \int_{\Sp^2} (\rd_y|_+)^2 &\psi (v_-, y, \theta, \varphi_-) \overline{\tau_{\xi, ml}(v_-, \theta) \cdot e^{im \varphi_-}} \, \vols dv_- \\ 
&\to \int_{\R} \int_{\Sp^2} (\rd_y|_+)^2 \psi (v_-, 0, \theta, \varphi_-) \overline{\tau_{\xi, ml}(v_-, \theta) \cdot e^{im \varphi_-}} \, \vols dv_-  \;.
\end{split}
\end{equation}
\end{proposition}

\begin{proof}
Let $\varepsilon >0$ be given. Then, using Lemma \ref{LemXiFasterPolynom} and for a $v_0 \leq 0 $ to be chosen later, we estimate 
\begin{equation*}
\begin{split}
\int_{\R} \int_{\Sp^2} &\Big|  (\rd_y|_+)^2 \psi (v_-, y, \theta, \varphi_-) -  (\rd_y|_+)^2 \psi (v_-, 0, \theta, \varphi_-) \Big| \cdot \big| \tau_{\xi, ml}(v_-, \theta) \cdot e^{im \varphi_-}\big| \, \vols dv_- \\
&\leq \Big(\int\limits_{v_0}^\infty \int\limits_{\Sp^2} \Big|  (\rd_y|_+)^2 \psi (v_-, y, \theta, \varphi_-) -  (\rd_y|_+)^2 \psi (v_-, 0, \theta, \varphi_-) \Big|^2 \, \vols dv_-\Big)^{\nicefrac{1}{2}}  \underbrace{\Big(\int\limits_{ v_0}^\infty \int\limits_{\Sp^2} \big| \tau_{\xi, ml}(v_-, \theta)\big|^2 \, \vols dv_-\Big)^{\nicefrac{1}{2}}}_{\leq ||\tau_{\xi,ml}||_{L^2(\R \times \Sp^2)}} \\
&\quad + \int\limits_{-\infty}^{v_0} \int\limits_{\Sp^2} C  \cdot|\tau_{\xi, ml}(v_-, \theta)| \, \vols dv_- \;,
\end{split}
\end{equation*}
where we have used Corollary \ref{CorUniformDecayNearHP} and the regularity assumption \ref{AssumptionReg} to infer that $\rd_y^2 \psi$ is uniformly bounded in $\{r_0 \leq r \leq r_+\} \cap \{v_- \leq v_0\}$ for some $r_- < r_0 < r_+$.
 By Lemma \ref{LemXiFasterPolynom} $\tau_{\xi,ml}$ is integrable, so we can choose $v_0 \ll -1$ such that the last term is less than $\frac{\varepsilon}{2}$. By \eqref{EqCorL2LocLimit} we have that for all $y$ close enough to $0$ the first summand on the right hand side is less than $\frac{\varepsilon}{2}$.
\end{proof}

We now compute both sides of \eqref{EqPropLimitDis}. We start with the right hand side and recall the convention $\rd_y = \rd_y|_+$.
We use $(V^-_{r_+}, \theta, \Phi_{r_+})$ coordinates on $\Hp_l$ and write
\begin{equation}\label{EqSeparatePsi}
\begin{split}
\rd_y^2 \psi|_{\Hp_l}(V^-_{r_+}, \theta, \Phi_{r_+}) &= \rd_y^2 \psi|_{\Hp_l}(0, \theta, \Phi_{r_+}) \cdot \mathbbm{1}_{(-\infty, 0)}(v_-) \\
&\qquad + \underbrace{\int_0^{V^-_{r_+}} \frac{\rd}{\rd V^-_{r_+}} \rd_y^2 \psi|_{\Hp_l} (V^-_{r_+}, \theta, \Phi_{r_+}) \, dV^-_{r_+} + \mathbbm{1}_{[0, \infty)}(v_-) \cdot \rd_y^2 \psi|_{\Hp_l}(0, \theta, \Phi_{r_+})}_{=:\Omega(V^-_{r_+}, \theta, \Phi_{r_+})} \;.
\end{split}
\end{equation}

We first consider the contribution of $\Omega$ to the right hand side of \eqref{EqPropLimitDis}. Note that we have $\Omega(v_-, \theta, \varphi_-) = \rd_y^2 \psi|_{\Hp_l}(v_-, \theta, \varphi_-)$ for $v_- \geq 0$ and also $|\Omega(V^-_{r_+}, \theta, \Phi_{r_+})| \leq C \cdot  V^-_{r_+}$ for $0 \leq V^-_{r_+} \leq 1$. Since we have $V^-_{r_+} = e^{\kappa_+ v_-}$ this gives us exponential decay in $v_-$ towards the bottom bifurcation sphere, i.e., $|\Omega(v_-, \theta, \varphi_-)| \leq C \cdot e^{\kappa_+ v_-}$ for $v_- \leq 0$. By assumption \eqref{AssumpLHpWeaker} and \eqref{EqRPsiNotWeight} we can thus use Fubini (or Plancherel) to obtain
\begin{equation*}
\begin{split}
\int_{\R_{v_-}} \int_{\Sp^2} \Omega(v_-, \theta, \varphi_-) & \overline{\tau_{\xi, ml}(v_-, \theta) \cdot e^{im \varphi_-}} \, \vols dv_-  \\
&=\int_{\R_{v_-}} \int_{\Sp^2} \Omega(v_-, \theta, \varphi_-)  \frac{1}{\sqrt{2\pi}} \int_{\R_\omega} \overline{\xi(\omega)} \Sml(\cos \theta; \omega) e^{-i \omega v_-} e^{-im \varphi_-} \,d \omega \vols dv_-\\
&= \int_{\R_\omega} \frac{1}{\sqrt{2\pi}} \int_{\R_{v_-}} \int_{\Sp^2} \Omega^-(v_-, \theta, \varphi_-) e^{i \omega v_-} \Sml(\cos \theta; \omega) e^{-im \varphi_-} \, dv_- \vols \cdot \overline{\xi(\omega)} \, d \omega \\
&= \int_{\R} (\widecheck{\Omega^-})_{ml} (\omega) \overline{\xi(\omega)} \, d \omega \;,
\end{split}
\end{equation*}
where we introduced $\Omega^-(v_-, \theta, \varphi_-) := \Omega(-v_-, \theta, \varphi_-)$ to  account for the different sign in the phase of the Fourier transform compared to our convention \eqref{EqDefCompositeMap}. 

By assumption \eqref{AssumpLHpWeaker} and \eqref{EqRPsiNotWeight} we have $\int_{\R} \int_{\Sp^2} (1 + |v_-|^{q_l}) | \Omega^-(v_-, \theta, \varphi_-)|^2 \, \vols dv_- < \infty$ and thus \linebreak $\int_{\R} \int_{\Sp^2} |\rd_\omega^q \widecheck{\Omega^-}(\omega, \theta, \varphi_-)|^2 \, \vols d\omega < \infty$ for all $0 \leq q \leq \frac{q_l}{2}$, $q \in \N_0$. Proposition \ref{PropCharSlowDecay} now gives
$$\int_{(-\varepsilon, \varepsilon)} \sum_{m,l} | \rd_\omega^q(\widecheck{\Omega^-})_{ml}|^2 \,d\omega < \infty$$
for all $0 \leq q \leq \frac{q_l}{2}$, $q \in \N_0$, where $\varepsilon >0$ is as in Proposition \ref{PropCharSlowDecay}. 
\footnote{We only need this statement for the single mode $m_0l_0$, for which we do not need to appeal to Proposition \ref{PropCharSlowDecay}.}

We now come to the contribution of the first term in \eqref{EqSeparatePsi} to \eqref{EqPropLimitDis}. We compute
\begin{equation}\label{EqDefPsiML}
\begin{split}
\int\limits_{\R_{v_-}} &\int\limits_{\Sp^2} \rd_y^2 \psi|_{\Hp_l}(0, \theta, \Phi_{r_+}) \cdot \mathbbm{1}_{(-\infty, 0)}(v_-) \underbrace{ \frac{1}{\sqrt{2\pi}} \int_{\R_\omega} \overline{\xi(\omega)} \Sml(\cos \theta; \omega) e^{-i \omega v_-} e^{-im \varphi_-} \, d \omega}_{= \overline{\tau_{\xi, ml}(v_-, \theta) \cdot e^{im \varphi_-}} } \, \sin \theta d \theta d \varphi_- dv_- \\
&=\int\limits_{-\infty}^0\int_{\Sp^2} \rd_y^2 \psi|_{\Hp_l}(0, \theta, \Phi_{r_+})\frac{1}{\sqrt{2\pi}} \int_{\R_\omega} \overline{\xi(\omega)} \Sml(\cos \theta; \omega) e^{-i \omega v_-} e^{-im \Phi_{r_+}} e^{im \omega_+ v_-} e^{-im \phi_+(r_+)} \, d \omega \, \sin \theta d \theta d \Phi_{r_+} dv_-  \\
&=e^{-im \phi_+(r_+)} \cdot \lim_{L \to \infty} \frac{1}{\sqrt{2\pi}} \int_{\R_\omega}  \overline{\xi(\omega)} \int\limits_{-L}^0 \underbrace{ \int\limits_{\Sp^2} \rd_y^2 \psi|_{\Hp_l}(0, \theta, \Phi_{r_+}) \Sml(\cos \theta; \omega) e^{-im \Phi_{r_+}}\, \vols}_{=: \Psi_{ml}(\omega)}  \cdot e^{-i \omega v_-} e^{im \omega_+ v_-} \, dv_- d\omega \\
&= e^{-im \phi_+(r_+)} \cdot \lim_{L \to \infty} \frac{1}{\sqrt{2\pi}} \int_{\R_\omega}  \overline{\xi(\omega)} \Psi_{ml}(\omega) \int\limits_{-L}^0  e^{-i v_-(\omega - \omega_+m)} \, dv_- d\omega  \\
&= e^{-im \phi_+(r_+)} \cdot \lim_{L \to \infty} \frac{1}{\sqrt{2\pi}} \int_{\R_\omega}  \overline{\xi(\omega)} \Psi_{ml}(\omega)  \frac{i}{\omega - \omega_+ m} [ 1 - e^{iL(\omega - \omega_+ m)}] \, d \omega
\end{split}
\end{equation}
Let us note that $\Psi_{ml}(\omega)$ is clearly smooth in $\omega$. We now divide the domain of integration into $|\omega - \omega_+ m| \leq \delta$ and its complement for some $\delta >0$. We first compute 
\begin{equation*}
\begin{split}
\lim_{L \to  \infty} &\int_{|\omega - \omega_+ m| \leq \delta} \overline{\xi}(\omega) \Psi_{ml}(\omega)  \frac{i}{\omega - \omega_+ m} [ 1 - e^{iL(\omega - \omega_+ m)}] \, d \omega \\
&= \lim_{L \to  \infty} \int_{|\omega - \omega_+ m| \leq \delta} \overline{\xi}(\omega_+m) \Psi_{ml}(\omega_+m)  \frac{i}{\omega - \omega_+ m} [ 1 - e^{iL(\omega - \omega_+ m)}] \, d \omega  \\
&\qquad + \underbrace{\lim_{L \to  \infty} \int_{|\omega - \omega_+ m| \leq \delta} \underbrace{\frac{ (\overline{\xi} \Psi_{ml})(\omega) - (\overline{\xi}\Psi_{ml})(\omega_+m)}{\omega - \omega_+m}}_{= \mathcal{O}(1)} [ 1 - e^{iL(\omega - \omega_+ m)}] \, d \omega  }_{= \mathcal{O}(\delta)} \\
&= \lim_{L \to  \infty} \int_{|\tilde{\omega} | \leq \delta}  \overline{\xi}(\omega_+m) \Psi_{ml}(\omega_+m) \frac{i}{\tilde{\omega}} [ 1 - \cos (L \tilde{\omega}) - i \sin (L \tilde{\omega})] \, d \tilde{\omega} + \mathcal{O}(\delta) \qquad \qquad  \textnormal{ with $\tilde{\omega} = \omega - \omega_+m$} \\
&= \lim_{L \to  \infty} 2\int_0^\delta  \overline{\xi}(\omega_+m) \Psi_{ml}(\omega_+m)  \frac{\sin(L \tilde{\omega})}{\tilde{\omega}} \, d \tilde{\omega} + \mathcal{O}(\delta)  \\
&= 2 \overline{\xi}(\omega_+m) \Psi_{ml}(\omega_+m)  \int_0^\infty \frac{\sin(\hat{\omega})}{\hat{\omega}} \, d \hat{\omega} + \mathcal{O}(\delta) \qquad \qquad  \textnormal{ with $\hat{\omega} = L \tilde{\omega}$} \\
&= \pi \overline{\xi}(\omega_+m) \Psi_{ml}(\omega_+m)  + \mathcal{O}(\delta) \;.
\end{split}
\end{equation*}
For the domain $|\omega - \omega_+m| > \delta$ we compute using Riemann-Lebesgue
\begin{equation}\label{EqDefMathringPsi}
\begin{split}
\lim_{L \to \infty}& \int_{|\omega - \omega_+ m| > \delta} \overline{\xi(\omega)} \Psi_{ml}(\omega)  \frac{i}{\omega - \omega_+ m} [ 1 - e^{iL(\omega - \omega_+ m)}] \, d \omega \\
&= \int_{|\omega - \omega_+ m| > \delta} \overline{\xi(\omega)} \Psi_{ml}(\omega)  \frac{i}{\omega - \omega_+ m} \, d \omega \\
&=  \int_{|\omega - \omega_+ m| > \delta} \overline{\xi(\omega)} \Psi_{ml}(\omega_+m)  \frac{i}{\omega - \omega_+ m} \, d \omega + \int_{|\omega - \omega_+ m| > \delta} \overline{\xi(\omega)}   i\underbrace{\frac{\Psi_{ml}(\omega) -\Psi_{ml}(\omega_+m)}{\omega - \omega_+ m}}_{=:\mathring{\Psi}_{ml}(\omega)} \, d \omega \;.
\end{split}
\end{equation}
Clearly $\mathring{\Psi}_{ml}(\omega)$ is smooth in $\omega$. Combining everything and letting $\delta$ go to zero we obtain
\begin{equation*}
\begin{split}
\int\limits_{\R_{v_-}} &\int\limits_{\Sp^2} \rd_y^2 \psi|_{\Hp_l}(0, \theta, \Phi_{r_+}) \cdot \mathbbm{1}_{(-\infty, 0)}(v_-) \cdot \overline{\tau_{\xi, ml}(v_-, \theta) \cdot e^{im \varphi_-}}  \, \vols dv_- \\
&= e^{-im \phi_+(r_+)} \Big( \sqrt{\frac{\pi}{2}} \overline{\xi(\omega_+m}) \Psi_{ml}(\omega_+m) + \frac{1}{\sqrt{2\pi}} \big[\lim_{\delta \to 0} \int\limits_{|\omega - \omega_+m| > \delta} \overline{\xi(\omega)} \frac{i \Psi_{ml}(\omega_+m)}{\omega - \omega_+m} \, d \omega + \int_{\R} \overline{\xi(\omega)} i \mathring{\Psi}_{ml}(\omega) \, d \omega \big] \Big)
\end{split}
\end{equation*}

We now claim that we have 
\begin{equation}
\label{EqClaimPsi}
\int_{-\varepsilon}^\varepsilon \sum_{m,l} |\rd_\omega^q \mathring{\Psi}_{ml}|^2 \, d \omega < \infty
\end{equation} for all $q \in \N_0$ for some $\varepsilon >0$. 
To see this, we first recall that $\Psi_{ml}(\omega) = \int_{\Sp^2} \rd_y^2 \psi|_{\Hp_l}(0, \theta, \Phi_{r_+}) \Sml(\cos \theta ; \omega) e^{-im \Phi_{r_+}} \, \vols$ and thus, using the notation from Proposition \ref{PropEstimatesDerivativesEigenfunctions},
\begin{equation*}
\begin{split}
\rd_\omega^q \Psi_{ml}(\omega) &= \int_{\Sp^2} \rd_y^2 \psi|_{\Hp_l}(0, \theta, \Phi_{r_+}) \rd_\omega^q \Sml(\cos \theta ; \omega) e^{-im \Phi_{r_+}} \, \vols \\
&= \int_{\Sp^2} \rd_y^2 \psi|_{\Hp_l}(0, \theta, \Phi_{r_+}) \sum_{l'} D^{[s]}_{mll'; q}(\omega) S^{[s]}_{ml'}(\cos \theta; \omega) e^{im \Phi_{r_+}}  \, \vols \\
&= \sum_{l'} D^{[s]}_{mll'; q}(\omega) \Psi_{ml'}(\omega) \;.
\end{split}
\end{equation*}
Hence, we obtain
\begin{equation}
\label{EqEstPsiSphere}
|\rd_\omega^q \Psi_{ml}(\omega)| \leq \Big(\sum_{l'} |D_{mll';q}^{[s]}(\omega)|^2 \Big)^{\nicefrac{1}{2}} \Big(\sum_{l'} |\Psi_{ml'}(\omega)|^2\Big)^{\nicefrac{1}{2}} \;.
\end{equation}
The claim \eqref{EqClaimPsi} with the sum restricted to $m \neq 0$ then follows directly: if necessary we choose $\varepsilon >0$ from Proposition \ref{PropEstimatesDerivativesEigenfunctions} even smaller than $|\omega_+|$ and then differentiate  $\mathring{\Psi}_{ml}(\omega) = \frac{\Psi_{ml}(\omega) - \Psi_{ml}(\omega_+m)}{\omega - \omega_+ m}$ in the region $(-\varepsilon, \varepsilon)$, which is disjoint from $\omega = \omega_+m$, and apply \eqref{EqEstPsiSphere} and Proposition \ref{PropEstimatesDerivativesEigenfunctions}.

To see that the contribution from $m=0$ to the sum in \eqref{EqClaimPsi} is also finite we observe that $\mathring{\Psi}_{0l}(\omega) = \frac{\int_0^\omega \rd_\omega \Psi_{0l}(\omega') \, d \omega'}{\omega} = \int_0^1 \rd_\omega \Psi_{0l}(\tau \omega) \, d \tau$ with $\omega' = \tau \omega$ and thus $$\rd_\omega^q \mathring{\Psi}_{0l}(\omega) = \int_0^1 \rd_\omega^{q+1} \Psi_{0l}(\tau \omega) \tau^q \, d \tau \;.$$
Using \eqref{EqEstPsiSphere} we continue to estimate
\begin{equation*}
\begin{split}
|\rd_\omega^q \mathring{\Psi}_{0l}(\omega)| &\leq \int_0^1 \Big( \sum_{l'}|D_{0ll';q+1}^{[s]}(\tau \omega)|^2\Big)^{\nicefrac{1}{2}} \Big(\sum_{l'} |\Psi_{0l'}(\tau \omega)|^2 \Big)^{\nicefrac{1}{2}} \, d \tau \\
&\leq \Big( \int_0^1 \sum_{l'}|D_{0ll';q+1}^{[s]}(\tau \omega)|^2 \, d \tau \Big)^{\nicefrac{1}{2}} \Big( \int_0^1 \Big(\sum_{l'} |\Psi_{0l'}(\tau \omega)|^2  \, d \tau\Big)^{\nicefrac{1}{2}} 
\end{split}
\end{equation*}
and for $\omega \in (-\varepsilon, \varepsilon)$
\begin{equation}\label{EqBeforeOmegaInt}
\begin{split}
\sum_{l}|\rd_\omega^q \mathring{\Psi}_{0l}(\omega)|^2 &\leq \int_0^1 \sum_{l,l'}  |D_{0ll';q+1}^{[s]}(\tau \omega)|^2 \, d \tau \cdot \int_0^1 \sum_{l'} |\Psi_{0l'}(\tau \omega)|^2  \, d \tau \\
&\leq C(q+1) \cdot \int_0^1 \underbrace{||\int_{\Sp^1} \rd_y^2 \psi|_{\Hp_l}(0, \theta, \Phi_{r_+}) \, d \Phi_{r_+} ||^2_{L^2_{\cos \theta}(- \pi, \pi)}} \, d \tau
\end{split}
\end{equation}
Note that the underbraced term is independent of $\tau \omega$. Thus the integration in $\tau$ is trivial and we can also trivially integrate \eqref{EqBeforeOmegaInt} in $\omega$ over $(-\varepsilon, \varepsilon)$. This finally proves the claim \eqref{EqClaimPsi}.

We summarise what we have shown in the following
\begin{proposition} \label{PropAHPL}
Under the assumptions from Section \ref{SecAssumptions} and  for every $\xi \in C^\infty_0(\R, \C)$ we have
\begin{equation*}
\begin{split}
\int_{\R} &\int_{\Sp^2} \rd_y^2 \psi|_{\Hp_l} (v_-, \theta, \varphi_-) \overline{\tau_{\xi, ml}(v_-, \theta) \cdot e^{im \varphi_-}} \, \vols dv_- \\
&= e^{-im \phi_+(r_+)} \Big( \sqrt{\frac{\pi}{2}} \overline{\xi(\omega_+m}) \Psi_{ml}(\omega_+m) + \frac{1}{\sqrt{2\pi}} \big[\lim_{\delta \to 0} \int\limits_{|\omega - \omega_+m| > \delta} \overline{\xi(\omega)} \frac{i \Psi_{ml}(\omega_+m)}{\omega - \omega_+m} \, d \omega + \int_{\R} \overline{\xi(\omega)} i \mathring{\Psi}_{ml}(\omega) \, d \omega \big] \Big)\\
&\qquad + \int_{\R} (\widecheck{\Omega^-})_{ml} (\omega) \overline{\xi(\omega)} \, d \omega   \;,
\end{split}
\end{equation*}
where $\Omega^-(v_-, \theta, \varphi_-) = \Omega(-v_-, \theta, \varphi_-)$ and $\Omega$ is defined in  \eqref{EqSeparatePsi}, $\Psi_{ml}$ is defined in \eqref{EqDefPsiML}, $\mathring{\Psi}_{ml}$ is defined in \eqref{EqDefMathringPsi}, and there exists an $\varepsilon >0$ such that $\rd_\omega^q (\widecheck{\Omega^-})_{ml}(\omega), \rd_\omega^q\mathring{\Psi}_{ml}(\omega) \in L^2_{(-\varepsilon, \varepsilon)} \ell^2_{lm}$ for all $0 \leq q \leq \frac{q_l}{2}$, $q \in \N_0$, with $q_l$ as in Section \ref{SecAssumptions}.
\end{proposition}
Let us remark that we only need the statement of this proposition for test functions $\xi$ which are supported away from $\omega = \omega_+m$. This would slightly shorten the proof -- the delta distribution term would be absent. Moreover, we only need the statement for the mode $m_0l_0$.
We next evaluate the left hand side of \eqref{EqPropLimitDis}.

\begin{proposition}\label{PropLimitHPL}
Under the assumptions of Section \ref{SecAssumptions} and for  $\xi \in C^\infty_0(\R \setminus \{ \omega_+m\}, \C)$ we have
\begin{equation*}
\begin{split}
\lim_{y \to 0}\int_{\R} \int_{\Sp^2} \rd_y^2 &\psi (v_-, y, \theta, \varphi_-) \overline{\tau_{\xi, ml}(v_-, \theta) \cdot e^{im \varphi_-}} \, \vols dv_- \\
&= \int_\R a_{\Hp_l, ml} (\omega) \Big(\frac{e^{- 2 \kappa_+ F_+(r_+)}}{r_+ - r_-}\Big)^{i \frac{\omega - \omega_+ m}{\kappa_+}} (\tilde{\gamma} - 1)\tilde{\gamma} e^{- 2 im \phi_+(r_+)} \overline{\xi(\omega)} \, d \omega \;,
\end{split}
\end{equation*}
where $F_+$ and $\phi_+$ are as in Sections \ref{SecRelTwoPhis} and \ref{SecDelta}.
\end{proposition}

It is important for the validity of the proposition as stated that one chooses  the support  of $\xi$ away from $\omega_+ m$.\footnote{\label{FootnDelta} One can evaluate the limit also for $\xi$ which are supported on $\omega_+m$; one then picks up a delta distribution term at $\omega_+ m$. With additional work it can be shown that it exactly agrees with the delta distribution term appearing in Proposition \ref{PropAHPL}, i.e., $a_{\Hp_l, ml}(\omega)$ does not contain a delta distribution, but only poles. This, however, is not needed for the method of proof of the main theorem chosen in this paper.}
\begin{proof}
Recall that we have $\varphi_- = \varphi_+ - 2 \overline{r}$ and $v_- = 2r^* - v_+$ and $\overline{r} = \omega_+ r^* - \phi_+(r)$. We thus obtain
\begin{equation}\label{EqTwoSummands1}
\begin{split}
\int_{\R} &\int_{\Sp^2} \rd_y^2 \psi (v_-, y, \theta, \varphi_-) \underbrace{\overline{\tau_{\xi, ml}(v_-, \theta) \cdot e^{im \varphi_-}}}_{= \frac{1}{\sqrt{2\pi}} \int_{\R_\omega} \overline{\xi(\omega)} \Sml(\cos \theta; \omega) e^{-i \omega v_-} e^{-im \varphi_-} \, d \omega} \, \vols dv_-  \\
&= \int_{\R_{v_+}} \int_{\Sp^2} \rd_y^2 \psi (v_+, y, \theta, \varphi_+) \, \frac{1}{\sqrt{2\pi}} \int_{\R_\omega} \overline{\xi(\omega)} \Sml(\cos \theta; \omega) e^{i \omega v_+} e^{-im \varphi_+} e^{-2ir^*(\omega - m \omega_+) }e^{-2im \phi_+(r)} \, d \omega \vols dv_+ \\
&= \int_{\R} (\widecheck{\rd_y^2 \psi})_{ml}(y; \omega) \overline{\xi(\omega)} e^{-2ir^*(\omega - \omega_+m)} e^{-2im \phi_+(r)} \, d \omega \\
&= \int_{\R} \rd_y^2 \widecheck{\psi}_{ml}(y; \omega) \overline{\xi(\omega)} e^{-2ir^*(\omega - \omega_+m)} e^{-2im \phi_+(r)} \, d \omega  \\
&= \int_{\R} [ a_{\Hp_r, ml}(\omega) A_{\Hp_r, ml}''(y ; \omega) + a_{\Hp_l, ml} (\omega) A_{\Hp_l, ml}'' (y; \omega) ]\overline{\xi(\omega)} e^{-2ir^*(\omega - \omega_+m)} e^{-2im \phi_+(r)} \, d \omega  \;,
\end{split}
\end{equation}
where we have used the same kind of reasoning as in the proof of Theorem \ref{ThmSeparationVariables}, $'$ denotes $\frac{d}{dy}$, and we consider $r^*$ and $r$ as functions of $y$.  Recall from Lemma \ref{LemRegAA} that $a_{\Hp_r, ml}, a_{\Hp_l, ml} \in C^0(\R \setminus\{ \omega_+m\}, \C)$. Since we have chosen the support of $\xi$ to be disjoint from $\omega_+m$ it is immediate that we can evaluate the integrals of the two summands separately. We begin with the first one.

We have $A_{\Hp_r, ml}''(y; \omega) = \sum_{j = 0}^\infty j(j-1) a_j(\omega, m, l) y^{j -2} = 2a_2(\omega, m, l) + \mathcal{O}(y)$. Note that the $\mathcal{O}(y)$ is uniform in $\omega$ on the support of $\xi$; and $a_2(\omega, m,l)$ is also uniformly bounded on $\mathrm{supp}(\xi)$. Moreover we have $r^*(y) \to - \infty$ for $y \to 0$. We thus obtain from Riemann-Lebesgue and direct estimation
\begin{equation*}
\begin{split}
\lim_{y \to 0}\int_{\R} &a_{\Hp_r, ml}(\omega) A_{\Hp_r, ml}''(y ; \omega) \overline{\xi(\omega)} e^{-2ir^*(\omega - \omega_+m)} e^{-2im \phi_+(r)} \, d \omega  \\
&= \lim_{y \to 0}\int_{\R} a_{\Hp_r, ml}(\omega) \big[2a_2(\omega, m, l) + \mathcal{O}(y)\big] \overline{\xi(\omega)} e^{-2ir^*(\omega - \omega_+m)} e^{-2im \phi_+(r)} \, d \omega  \\
&=0 \;.
\end{split}
\end{equation*}
In order to evaluate the second summand in \eqref{EqTwoSummands1} we first note that for $s=2$
\begin{equation*}
A_{\Hp_l, ml}''(y; \omega) = \sum_{j=0}^\infty (j+1 - \tilde{\gamma})(j - \tilde{\gamma}) b_j(\omega, m, l) y^{j-1-\tilde{\gamma}} = \underbrace{y^{\frac{4iMr_+}{r_+ - r_-}(\omega - \omega_+m)}}_{= y^{\frac{i (\omega - \omega_+m)}{\kappa_+}}} \sum_{j = 0}^\infty \underbrace{(j+1 - \tilde{\gamma})(j - \tilde{\gamma}) b_j(\omega, m, l)}_{=: \tilde{b}_j(\omega,m,l)} y^j \;.
\end{equation*}
Also recall that $y = \frac{r_+ - r}{r_+ - r_-}$ and, from Section \ref{SecDelta}, $r_+ - r= e^{2 \kappa_+ r^*} e^{- 2 \kappa_+ F_+(r)}$. This gives $$y^{i \frac{ \omega - \omega_+ m}{\kappa_+}} = \Big(\frac{e^{-2 \kappa_+ F_+(r)}}{r_+ - r_-}\Big)^{i \frac{\omega - \omega_+m}{\kappa_+}} \cdot e^{2ir^*(\omega - \omega_+m)} \;.$$
Thus
\begin{equation*}
\begin{split}
\lim_{y \to 0}\int_{\R}  &a_{\Hp_l, ml} (\omega) A_{\Hp_l, ml}'' (y; \omega) \overline{\xi(\omega)} e^{-2ir^*(\omega - \omega_+m)} e^{-2im \phi_+(r)} \, d \omega \\
&=\lim_{y \to 0}\int_{\R} a_{\Hp_l, ml} (\omega) \Big(\frac{e^{-2 \kappa_+ F_+(r)}}{r_+ - r_-}\Big)^{i \frac{\omega - \omega_+m}{\kappa_+}} \sum_{j = 0}^\infty \tilde{b}_j(\omega, m ,l) y^j \cdot \overline{\xi(\omega)} e^{-2im \phi_+(r)} \, d \omega \\
&= \int_\R a_{\Hp_l, ml} (\omega) \Big(\frac{e^{- 2 \kappa_+ F_+(r_+)}}{r_+ - r_-}\Big)^{i \frac{\omega - \omega_+ m}{\kappa_+}} \underbrace{\tilde{b}_0(\omega,m,l)}_{=(1- \tilde{\gamma} )(-\tilde{\gamma})} e^{- 2 im \phi_+(r_+)} \overline{\xi(\omega)} \, d \omega \;.
\end{split}
\end{equation*}
\end{proof}

It now follows from Propositions \ref{PropLimitDis}, \ref{PropAHPL}, and \ref{PropLimitHPL} that for $\xi \in C^\infty_0(\R \setminus \{ \omega_+m\}, \C)$ we have
\begin{equation*}
\begin{split}
 \int_\R a_{\Hp_l, ml} (\omega) &\Big(\frac{e^{- 2 \kappa_+ F_+(r_+)}}{r_+ - r_-}\Big)^{i \frac{\omega - \omega_+ m}{\kappa_+}} (\tilde{\gamma} - 1)\tilde{\gamma} e^{- 2 im \phi_+(r_+)} \overline{\xi(\omega)} \, d \omega \\
 &= \int_\R \Big((\widecheck{\Omega^-})_{ml}(\omega) + e^{-im\phi_+(r_+)} \big[ \frac{1}{\sqrt{2\pi}} \frac{i \Psi_{ml}(\omega_+m)}{\omega - \omega_+m} + i \mathring{\Psi}_{ml}(\omega) \big] \Big) \overline{\xi(\omega)} \, d \omega \;.
\end{split}
\end{equation*}
Note that  $(\widecheck{\Omega^-})_{ml}(\omega)$ is continuous in $\omega$ (this follows as in the proof of Lemma \ref{LemPsiFTC1}). Thus, all terms multiplying $\overline{\xi(\omega)}$ on each side are (at least) continuous in $\omega$ away from $\omega_+m$. We thus conclude that for $\omega \neq \omega_+m$
\begin{equation}
\label{EqFinalExprAHPL}
\begin{split}
a_{\Hp_l, ml} (\omega) = (r_+ - r_-)^{i \frac{\omega - \omega_+m}{\kappa_+}} &e^{2i F_+(r_+)(\omega - \omega_+m)} e^{2im \phi_+(r_+)} \frac{1}{\tilde{\gamma}(\tilde{\gamma} -1)} \\
&\cdot  \Big((\widecheck{\Omega^-})_{ml}(\omega) + e^{-im\phi_+(r_+)} \big[ \frac{1}{\sqrt{2\pi}} \frac{i \Psi_{ml}(\omega_+m)}{\omega - \omega_+m} + i \mathring{\Psi}_{ml}(\omega) \big] \Big)  \;.
\end{split}
\end{equation}

\begin{corollary} \label{CorConcAHL}
Under the assumptions from Section \ref{SecAssumptions} there exists $\varepsilon_0 >0$ such that
\begin{enumerate}
\item for $m \neq 0$ we have $\rd_\omega^q a_{\Hp_l, ml} \in L^2_{(-\varepsilon_0, \varepsilon_0)} \ell^2_{\substack{m,l \\ m \neq 0}}$ for all $0 \leq q \leq \frac{q_l}{2}$, $q \in \N_0$.
\item for $m = 0$ we have that $\omega \cdot a_{\Hp_l, 0l}$ extends continuously to $\omega =0$ and, moreover, we have $\rd_\omega^q (\omega a_{\Hp_l, 0l}) \in L^2_{(-\varepsilon_0, \varepsilon_0)} \ell^2_{l} $ for all $0 \leq q \leq \frac{q_l}{2}$, $q \in \N_0$.
\end{enumerate}
\end{corollary}

\begin{proof}
This is immediate from \eqref{EqFinalExprAHPL} and Proposition \ref{PropAHPL}.
\end{proof}

\section{Proof of the main theorems}
\label{SecPfMainThm}

\subsection{Proof of Theorem \ref{Thm1}}

We are now in a position to prove the main theorem.
\begin{proof}[Proof of Theorem \ref{Thm1}:]
Recall from Proposition \ref{PropExtCauchyL2} that the $L^2(\R \times \Sp^2)$-limit of $\psi(v_+,r, \theta, \varphi_+)$ for $r \to r_-$ exists and that we labelled it suggestively by $\psi(v_+, r_-, \theta, \varphi_+)$. Taking the Teukolsky transform we have
$$\psi(v_+, r_-, \theta, \varphi_+) = \frac{1}{\sqrt{2 \pi}} \int_{\R} \sum_{m,l} \widecheck{\psi}_{ml}(r_-; \omega) Y^{[s]}_{ml}(\theta, \varphi_+; \omega) e^{-i \omega v_+} \, d \omega \;,$$
where $\widecheck{\psi}_{ml}(r_-; \omega)$ is continuous in $\omega$ by \eqref{EqWeightedPsiCauchy}, cf.\ the proof of Lemma \ref{LemPsiFTC1}.\footnote{Note that $\widecheck{\psi}_{ml}(r_-; \omega)$ is a priori not related in any way to $\widecheck{\psi}_{ml}(r; \omega)$ for $r \in (r_-,r_+)$. The choice of terminology is justified by hindsight. However, it should not confuse the reader into believing that there is nothing to show in the following.} By Proposition \ref{PropExtCauchyL2} and Plancherel we have for $r \to r_-$
$$0 \leftarrow || \psi(v_+, r, \theta, \varphi_+) - \psi(v_+, r_-, \theta, \varphi_+)||^2_{L^2(\R \times \Sp^2)} = \int_{\R} \sum_{m,l} || \widecheck{\psi}_{ml}(r ; \omega) - \widecheck{\psi}_{ml}(r_-; \omega)|^2 \, d\omega \;. $$
Fix $m$ and $l$. Then there exists a sequence $r_n \to r_-$ such that $\widecheck{\psi}_{ml}(r_n; \omega) \to \widecheck{\psi}_{ml}(r_-; \omega)$ for almost all $\omega \in \R$.

It now follows from \eqref{EqDefAA} and \eqref{EqDefTransReflCoeff} that for $\omega \neq \omega_+m, \omega_-m$
\begin{equation*}
\begin{split}
\widecheck{\psi}_{ml}(x; \omega) &= \big[\mathfrak{R}_{\Hp_l, ml}(\omega) a_{\Hp_l, ml}(\omega) + \mathfrak{T}_{\Hp_r, ml}(\omega) a_{\Hp_r, ml}(\omega)\big] B_{\CH_l, ml}(x ; \omega) \\
&\quad + \big[\mathfrak{T}_{\Hp_l, ml}(\omega) a_{\Hp_l, ml}(\omega) + \mathfrak{R}_{\Hp_r, ml}(\omega) a_{\Hp_r, ml}(\omega) \big] B_{\CH_r, ml}(x; \omega) \;.
\end{split}
\end{equation*}
The asymptotics of the Frobenius solutions from Proposition \ref{PropFrobSolutions} imply  $\lim_{n \to \infty} \widecheck{\psi}_{ml}(x(r_n); \omega) = \mathfrak{R}_{\Hp_l, ml}(\omega) a_{\Hp_l, ml}(\omega) + \mathfrak{T}_{\Hp_r, ml}(\omega) a_{\Hp_r, ml}(\omega)$ for $\omega \neq \omega_+m, \omega_-m$ and thus we obtain
\begin{equation}\label{EqPsiCheckCH}
\widecheck{\psi}_{ml}(r_-; \omega) = \mathfrak{R}_{\Hp_l, ml}(\omega) a_{\Hp_l, ml}(\omega) + \mathfrak{T}_{\Hp_r, ml}(\omega) a_{\Hp_r, ml}(\omega) 
\end{equation}
for almost every $\omega \in \R \setminus\{ \omega_+m, \omega_-m\} $.
We claim that there is an $\varepsilon >0$ such that the right hand side is continuous for $\omega \in (-\varepsilon, \varepsilon)$. For $m \neq 0$ this follows directly from Lemma \ref{LemRegAA} and Proposition \ref{PropTRmNeq0}. For $m = 0$ Propositions \ref{PropAHP} and \ref{PropTRM0} imply that $\mathfrak{T}_{\Hp_r, 0l}(\omega) a_{\Hp_r, 0l}(\omega) $ is continuous and, moreover, Proposition \ref{PropTRM0} implies that $\mathfrak{R}_{\Hp_l, 0l}(\omega)$ is of the form $\mathfrak{R}_{\Hp_l, 0l}(\omega) =: \omega \cdot \widehat{\mathfrak{R}_{\Hp_l, 0l}}(\omega)$ with $\widehat{\mathfrak{R}_{\Hp_l, 0l}}(\omega)$ analytic for all $\omega \in \R$. Hence, Corollary \ref{CorConcAHL} implies that also $\mathfrak{R}_{\Hp_l, 0l}(\omega) a_{\Hp_l, 0l}(\omega) = \widehat{\mathfrak{R}_{\Hp_l, 0l}}(\omega) (\omega \cdot a_{\Hp_l, 0l}(\omega))$ is continuous in $\omega$. We thus obtain \eqref{EqPsiCheckCH} for all $\omega \in (-\varepsilon, \varepsilon)$.

We consider \underline{\textbf{the case} $m_0 \neq 0$} first and compute
\begin{equation}\label{EqLeib1}
\begin{split}
\rd_\omega^{p_0} \widecheck{\psi}_{m_0l_0} (r_-; \omega) &= \rd_\omega^{p_0} \big( \mathfrak{R}_{\Hp_l, m_0l_0}(\omega) a_{\Hp_l, m_0l_0}(\omega) + \mathfrak{T}_{\Hp_r, m_0l_0}(\omega) a_{\Hp_r, m_0l_0}(\omega)\big) \\
&= \sum_{q=0}^{p_0} \binom{p_0}{q} \big( \underbrace{\rd_\omega^{q} \mathfrak{R}_{\Hp_l, m_0l_0}(\omega)}_{|\cdot| \leq C} \cdot \rd_\omega^{p_0- q} a_{\Hp_l, m_0l_0}(\omega) + \underbrace{\rd_\omega^{q} \mathfrak{T}_{\Hp_r, m_0l_0}(\omega)}_{|\cdot | \leq C} \cdot \rd_\omega^{p_0 - q} a_{\Hp_r, m_0l_0}(\omega) \big) \;.
\end{split}
\end{equation}
The derivatives of the transmission and reflection coefficients are uniformly bounded on $(-\varepsilon, \varepsilon)$ by Proposition \ref{PropTRmNeq0}. It thus follows from Corollary \ref{CorConcAHL} and Corollary \ref{CorConHPR} that all terms on the right hand side of \eqref{EqLeib1}, with the exception of $\mathfrak{T}_{\Hp_r, m_0l_0}(\omega) \cdot \rd_\omega^{p_0} a_{\Hp_r, m_0l_0}(\omega)$, are in $L^2_\omega(\big(-\varepsilon_0,\varepsilon_0)\big)$ for some $\varepsilon_0 >0$. On the other hand  $\mathfrak{T}_{\Hp_r, m_0l_0}(\omega) $ is strictly bounded away from $0$ in a small neighbourhood of $\omega = 0$ by Proposition \ref{PropTRmNeq0}. It thus follows from Corollary \ref{CorConHPR} that 
$\mathfrak{T}_{\Hp_r, m_0l_0}(\omega) \cdot \rd_\omega^{p_0} a_{\Hp_r, m_0l_0}(\omega) \notin L^2_\omega\big((-\varepsilon, \varepsilon)\big)$ for any $\varepsilon >0$. Hence we obtain that $\rd_\omega^{p_0} \widecheck{\psi}_{m_0l_0} (r_-; \omega) \notin L^2_{\omega}\big((-\varepsilon, \varepsilon)\big)$ for any $\varepsilon >0$.

We proceed with \underline{\textbf{the case} $m_0 = 0$} and compute
\begin{equation}\label{EqLeib2}
\begin{split}
\rd_\omega^{p_0} \widecheck{\psi}_{0l} (r_-; \omega) &= \rd_\omega^{p_0} \big( \widehat{\mathfrak{R}_{\Hp_l, 0l}}(\omega)\big[ \omega \cdot a_{\Hp_l, 0l}(\omega) \big] + \mathfrak{T}_{\Hp_r, 0l}(\omega) a_{\Hp_r, 0l}(\omega)\big) \\
&= \sum_{q=0}^{p_0} \binom{p_0}{q} \big( \underbrace{\rd_\omega^{q} \widehat{\mathfrak{R}_{\Hp_l, 0l}}(\omega)}_{|\cdot| \leq C} \cdot \rd_\omega^{p_0- q}\big[ \omega \cdot  a_{\Hp_l, 0l}(\omega)\big] + \underbrace{\rd_\omega^{q} \mathfrak{T}_{\Hp_r, 0l}(\omega)}_{|\cdot | \leq C} \cdot \rd_\omega^{p_0 - q} a_{\Hp_r, 0l}(\omega) \big) \;.
\end{split}
\end{equation}
It follows again from Corollary \ref{CorConcAHL} and Corollary \ref{CorConHPR} that all terms on the right hand side of \eqref{EqLeib2}, with the exception of $\mathfrak{T}_{\Hp_r, 0l_0}(\omega) \cdot \rd_\omega^{p_0} a_{\Hp_r, 0l_0}(\omega)$, are in $L^2_\omega(\big(-\varepsilon_0,\varepsilon_0)\big)$ for some $\varepsilon_0 >0$. On the other hand  $\mathfrak{T}_{\Hp_r, 0l_0}(\omega) $ is strictly bounded away from $0$ in a small neighbourhood of $\omega = 0$ this time by Proposition \ref{PropTRM0}. It thus follows from Corollary \ref{CorConHPR} that 
$\mathfrak{T}_{\Hp_r, 0l_0}(\omega) \cdot \rd_\omega^{p_0} a_{\Hp_r, 0l_0}(\omega) \notin L^2_\omega\big((-\varepsilon, \varepsilon)\big)$ for any $\varepsilon >0$. Hence we obtain that $\rd_\omega^{p_0} \widecheck{\psi}_{0l_0} (r_-; \omega) \notin L^2_{\omega}\big((-\varepsilon, \varepsilon)\big)$ for any $\varepsilon >0$.

Taking the two cases together we have shown that 
\begin{equation}
\label{EqNotInL2BlowUp}
\rd_\omega^{p_0} \widecheck{\psi}_{m_0l_0} (r_-; \omega) \notin L^2_{\omega}\big((-\varepsilon, \varepsilon)\big) \; \textnormal{ for any } \;\varepsilon >0\;.
\end{equation}
We claim that this implies
\begin{equation}
\label{EqNearlyDone}
\int_{\R} \int_{\Sp^2} |v_+^{p_0} \psi(v_+, r_-, \theta, \varphi_+)|^2 \, \vols dv_+ = \infty \;.\end{equation}
If \eqref{EqNearlyDone} was finite then together with $\int_{\R}\int_{\Sp^2} | \psi(v_+, r_-, \theta, \varphi_+)|^2 \, \vols dv_+ < \infty$ from \eqref{EqWeightedPsiCauchy} we would have $\int_\R \int_{\Sp^2} |\rd_\omega^q \widecheck{\psi}( \omega, r_-, \theta, \varphi_+)|^2 \, \vols d \omega < \infty$ for all $0 \leq q \leq p_0$. Proposition \ref{PropCharSlowDecay} then gives $\int_{(-\varepsilon, \varepsilon)}\sum_{m,l} |\rd_\omega^q \widecheck{\psi}_{ml}( r_-; \omega)|^2 \, d \omega <\infty$ for all $0 \leq q \leq p_0$, where $\varepsilon>0$ is as in Proposition \ref{PropCharSlowDecay}. This clearly contradicts \eqref{EqNotInL2BlowUp} and we thus infer \eqref{EqNearlyDone}.

On the other hand by \eqref{EqWeightedPsiCauchy} we know that $$\int\limits_{-\infty}^1 \int_{\Sp^2} |v_+^{p_0} \psi(v_+, r_-, \theta, \varphi_+)|^2 \, \vols dv_+ < \infty$$ since $\frac{q_l}{2} \geq p_0$ and thus we must have 
$$\int\limits_{1}^\infty \int_{\Sp^2} |v_+^{p_0} \psi(v_+, r_-, \theta, \varphi_+)|^2 \, \vols dv_+ = \infty \;.$$
The theorem now follows from Corollary \ref{CorPropBackwa}.
\end{proof}

\subsection{Extension theorem and proof of Theorem \ref{Thm2}}
\label{SecExtension}

\begin{theorem} \label{ThmExt}
Let $v_0, v_1 \in \R$ and let $\psi \in \mathscr{I}^\infty_{[2]}(\Mun \cap \{f^+ \geq v_0\} \cap \{f^- \leq v_1\})$ be a solution of the Teukolsky equation $\mathcal{T}_{[2]} \psi = 0$ in $\Mun \cap \{f^+ \geq v_0\} \cap \{f^- \leq v_1\}$ satisfying the assumptions \eqref{AssumpP0}, \eqref{AssumpP0SphHarm}, \eqref{AssumpP0Better}, \eqref{AssumpDecayRHp} along the right event horizon\footnote{With $\Hp_r \cap \{v_+ \geq 1\}$ replaced by $\Hp \cap \{v_+ \geq v_0\}$ when appropriate.}. Then there exists a  $\chi \in \mathscr{I}^\infty_{[2]}(\Mun)$ which extends $\psi$ (i.e.\ $\chi|_{\Mun \cap \{f^+ \geq v_0\} \cap \{f^- \leq v_1\}} = \psi$) which moreover satisfies the Assumptions \ref{AssumptionReg}  and condition \eqref{AssumpLHp}.
\end{theorem}

\begin{proof}
The idea of the proof is outlined in Figure \ref{FigExt}. Also recall that the level sets of the functions $f^-$ and $f^+$ are spacelike hypersurfaces.

\begin{figure}[h]
\centering
 \def\svgwidth{6cm}
    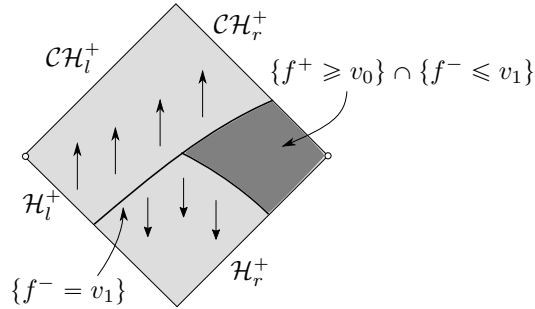
      \caption{Extending the Teukolsky field globally} \label{FigExt}
\end{figure}
We consider the Teukolsky equation \eqref{TeukEqRegKruskal} which is regular in $\underline{\mathcal{M}}$. Recall that $\tilde{\chi}$ satisfies \eqref{TeukEqRegKruskal} if, and only if, $\chi = (V^+_{r_+})^2 \tilde{\chi}$ satisfies \eqref{TeukolskyStar}. We now consider the induced initial data of $\tilde{\psi} := \frac{1}{(V^+_{r_+})^2} \psi$ on $\{f^+ = v_0\} \cap \{f^- \leq v_1\}$ and, moreover, extend the induced initial data on $\{f^- = v_1\} \cap \{f^+ \geq v_0\}$ smoothly to $\{f^- = v_1\} \cap \{f^+ \leq v_0\}$. Since $\tilde{\psi}$ is a solution of \eqref{TeukEqRegKruskal} in $\{f^+ \geq v_0\} \cap \{f^- \leq v_1\}$, it is clear that this choice of initial data satisfies the appropriate corner condition. In the appendix \ref{SecAppendixIVPTeuk} it is shown that the initial value problem for \eqref{TeukEqRegKruskal} is well-posed. We can thus solve backwards to obtain a smooth solution $\tilde{\chi}$ of \eqref{TeukEqRegKruskal} in the region $\{f^- \leq v_1\} \cap \{f^+ \leq v_0\}$ that attains the prescribed initial data.

In the second step we consider the initial value problem for \eqref{TeukEqRegKruskal} with compactly supported initial data on $\Hp_l \cap \{f^- \geq v_1\}$, which is a smooth extension of the induced initial data of $\tilde{\chi}$ on $\Hp_l \cap \{f^- \leq v_1\}$, and, moreover, with the induced initial data of $\tilde{\chi}$ on $\{f^- = v_1\}$. Again, the corner condition is satisfied and we obtain a solution in the region $\{f^- \geq v_1\}$. Patching these three solutions together proves the extension theorem.
\end{proof}

\begin{proof}[Proof of Theorem \ref{Thm2}:]
This is immediate from Theorem \ref{ThmExt} and Theorem \ref{Thm1}.
\end{proof}

\appendix 

\section{Kruskal-like coordinate transformations} \label{CoordTransformation} 

We express the $\{v_+, \varphi_+, r, \theta\}$ coordinate vector fields in terms of the $\{V^+_{r_+}, V^-_{r_+}, \theta, \Phi_{r_+}\}$ coordinate vector fields.

\begin{equation} \label{KruskalCoordTrafo}
\begin{split}
\partial_{v_+} &= \kappa_+V^+_{r_+} \partial_{V^+_{r_+}} - \kappa_+ V^-_{r_+}\partial_{V^-_{r_+}} - \frac{a}{r_+^2 + a^2} \partial_{\Phi_{r_+}} \\
\partial_{\varphi_+} &= \partial_{\Phi_{r_+}} \\
\partial_r &=2\kappa_+ \frac{r^2 + a^2}{\Delta} V^-_{r_+} \partial_{V^-_{r_+}} + \frac{a}{\Delta}\frac{r^2 - r_+^2}{r_+^2 + a^2} \partial_{\Phi_{r_+}} \\
\partial_\theta &=\partial_\theta \;.
\end{split}
\end{equation}
Note that the vector field $\partial_{v_+}$ does not vanish at the bottom bifurcation sphere $\Sp^2_b$.

Using \eqref{KruskalCoordTrafo} we thus compute
\begin{equation} \label{Eqe4BSphere}
e_4 = 2(r^2 + a^2) \kappa_+ V^+_{r_+} \partial_{V^+_{r_+}} + a\frac{r^2 - r_+^2}{r_+^2 + a^2} \partial_{\Phi_{r_+}} \;.
\end{equation}
It now follows from \eqref{r-r_+} that $\frac{1}{V^+_{r_+}} e_4$ is a regular and non-degenerate vector field at $\Hp_r \cup \Hp_l \cup \Sp^2_b$. Note, however, that compared to\footnote{Here $\mathcal{O}(1)$ is with respect to $r\to r_+$ and $c>0$.} $\hat{e}_4 = -\frac{1}{\Delta} e_4 = (\frac{c}{V^+_{r_+}V^-_{r_+}} + \mathcal{O}(1)) e_4$, which blows up at $\Sp^2_b$, $\frac{1}{V^+_{r_+}} e_4$ grows exponentially in $v_-$ for $v_- \to +\infty$.

\subsection{The regular Teukolsky equation in $\uline{\mathcal{M}}$ }

The above suggests that if $\psi$ satisfies $\mathcal{T}_s \psi = 0$, then the quantity $\tilde{\psi}_s := \frac{1}{(V^+_{r_+})^s} \psi_s$ should satisfy a regular equation in $\underline{\mathcal{M}}$, thus in particular near the bottom bifurcation sphere. In order to show this claim, we start by rewriting \eqref{TeukolskyStar} in $(v_+, r, \theta, \varphi_+)$ coordinates as
\begin{equation}\label{EqTeukStarBox}
\begin{split}
\frac{1}{\rho^2} \mathcal{T}_{[s]} \psi_s &= \Box_g \psi_s - \frac{2s}{\rho^2}(r-M) \rd_r \psi_s + \frac{2si}{\rho^2} \frac{\cos \theta}{\sin^2 \theta} \rd_{\varphi_+} \psi_s \\
&\quad - \frac{2s}{\rho^2}(2r + ia \cos \theta) \rd_{v_+} \psi_s - \frac{1}{\rho^2}\big( s + s^2 \frac{\cos^2 \theta}{\sin^2\theta}\big) \psi_s = 0 \;.
\end{split}
\end{equation}
A straightforward computation shows that $\psi_s \in \mathscr{I}^\infty_{[s]}(\mathcal{M})$ satisfies \eqref{EqTeukStarBox} if, and only if, $\tilde{\psi}_s$ satisfies
\begin{equation}\label{EqTeukTildeStar}
\begin{split}
\tilde{\mathcal{T}}_{[s]} \tilde{\psi}_s :&= \Box_g \tilde{\psi}_s + \frac{2s}{\rho^2}( \kappa_+ a^2 \sin^2\theta - 2r - ia\cos \theta) \rd_{v_+} \tilde{\psi}_s + \frac{2s}{\rho^2}\big( \kappa_+(r^2 + a^2) - (r-M)\big) \rd_r \tilde{\psi}_s \\
&\quad + \frac{2s}{\rho^2}\big(\kappa_+ a + i \frac{\cos\theta}{\sin^2\theta}\big) \rd_{\varphi_+} \tilde{\psi}_s + \frac{s \kappa_+}{\rho^2}(s \kappa_+ a^2 \sin^2\theta + 2r) \tilde{\psi}_s - \frac{2s^2 \kappa_+}{\rho^2}(2r + ia \cos\theta) \tilde{\psi}_s \\
&\quad - \frac{1}{\rho^2}(s + s^2 \frac{\cos^2\theta}{\sin^2\theta})\tilde{\psi}_s =0\;.
\end{split}
\end{equation}
Rewriting \eqref{EqTeukTildeStar} in terms of $\{V^+_{r_+}, V^-_{r_+}, \theta, \Phi_{r_+}\}$ coordinates (in the following we will drop the $r_+$, i.e., we will only write $\{V^+, V^-, \theta, \Phi\}$) gives
\begin{equation}\label{TeukEqRegKruskal}
\begin{split}
0 = \tilde{\mathcal{T}}_{[s]} \tilde{\psi}_s &= \Box_g \tilde{\psi}_s + \frac{2si}{\rho^2} \frac{\cos \theta}{\sin^2 \theta} \rd_{\Phi} \tilde{\psi}_s - \frac{1}{\rho^2}\big(s^2 \frac{\cos^2 \theta}{\sin^2 \theta} -s\big) \tilde{\psi}_s \\
&\qquad + \tilde{X}_s^{V^+} \rd_{V^+} \tilde{\psi}_s + \tilde{X}_s^{V^-} \rd_{V^-} \tilde{\psi}_s + \tilde{X}_s^{\Phi} \rd_\Phi \tilde{\psi}_s + \tilde{f}_s \tilde{\psi}_s \;,
\end{split}
\end{equation}
where
\begin{equation*}
\begin{split}
\tilde{X}_s^{V^+} &= \frac{2s}{\rho^2}(\kappa_+ a^2 \sin^2 \theta - 2r - ia \cos \theta)\kappa_+ V^+ \\
\tilde{X}_s^{V^-} &=\frac{2s}{\rho^2} \kappa_+ V^-\Big( - \kappa_+ a^2 \sin^2 \theta + 2r + ia \cos \theta + \frac{2(r^2 + a^2)}{\Delta}\big(\dashuline{\kappa_+(r^2 + a^2) - (r-M)}\big)\Big) \\
\tilde{X}_s^{\Phi} &= \frac{2s}{\rho^2}\Big(- \frac{a}{r_+^2 + a^2} ( \kappa_+ a^2 \sin^2\theta - 2r - ia \cos \theta) + \frac{a}{\Delta} \dashuline{\frac{r^2 - r_+^2}{r_+^2 + a^2}}\big[\dashuline{\kappa_+(r^2 + a^2) - (r-M)}\big] + \kappa_+ a\Big) \\
\tilde{f}_s &= \frac{s \kappa_+}{\rho^2} \Big( s(-4r -2ia\cos \theta + \kappa_+ a^2 \sin^2\theta) + 2r - \frac{2}{\kappa_+}\Big) \;.
\end{split}
\end{equation*}
Note that the dashed terms are $\mathcal{O}(r_+-r)$. To see this we recall that $r_+ - r_- = 2(r_+ -M)$ and compute
\begin{equation*}
 \kappa_+(r^2 + a^2) - (r-M) = \frac{r_+ - r_-}{2(r_+^2 + a^2)}(r^2 + a^2) - (r-M) = \frac{(r_+ - M)}{r_+^2 + a^2}(r^2 + a^2) - (r-M) = \mathcal{O}(r_+-r) \;.
 \end{equation*}
Hence, we have $\tilde{X}_s^{V^+} , \tilde{X}_s^{V^-} , \tilde{X}_s^{\Phi}, \tilde{f}_s \in C^\infty(\underline{\mathcal{M}})$.

\subsection{The initial value problem for the Teukolsky equation \eqref{TeukEqRegKruskal}}
\label{SecAppendixIVPTeuk}

In this section we show that the initial value problem for the Teukolsky equation \eqref{TeukEqRegKruskal} is well-posed by reducing it to an initial value problem for a tensorial wave equation. In the following we restrict to $s = +2$ and drop the subscript $s$ from $\tilde{\psi}_s$. For $\tilde{\psi} \in \mathscr{I}^\infty_{[2]}(\underline{\mathcal{M}}) $ there exists, by Remark \ref{RemSpinWSpacetimeS2Tensor}, a unique $\alpha \in \Gamma^\infty\big(S^2T^*\mathcal{M}\big)$ with $\alpha(m,m) = \tilde{\psi}$ that is trace-free with respect to $\slashed{g}_{\Sp^2} = d\theta^2 + \sin^2\theta \, d\varphi^2$ and that is an $\Sp^2$ tensor\footnote{I.e.\ we have $\alpha(\rd_{V^+}, \cdot) = \alpha(\rd_{V^-}, \cdot) =0$.}. We rewrite \eqref{TeukEqRegKruskal} as
\begin{equation}\label{EqTeukBoxToG}
\begin{split}
0 &= \tilde{\mathcal{T}}_{[s]} \big(\alpha(m,m)) \\
&= g^{\mu \nu} \mathcal{L}_{\rd_\mu} \mathcal{L}_{\rd_\nu} \big(\alpha(m,m)\big)  + (\Box_g x^{\mu}) \mathcal{L}_{\rd_\mu} \big(\alpha(m,m)\big)  + \frac{2si}{\rho^2} \frac{\cos \theta}{\sin^2 \theta} \mathcal{L}_{\rd_{\Phi}}\big(\alpha(m,m)\big)\\ 
&\qquad - \frac{1}{\rho^2}\big(s^2 \frac{\cos^2 \theta}{\sin^2 \theta} -s\big) \big(\alpha(m,m)\big)   + \tilde{X}_s^{V^+} \mathcal{L}_{\rd_{V^+}} \big(\alpha(m,m)\big)  + \tilde{X}_s^{V^-} \mathcal{L}_{\rd_{V^-}} \big(\alpha(m,m)\big)\\
&\qquad + \tilde{X}_s^{\Phi} \mathcal{L}_{\rd_\Phi} \big(\alpha(m,m)\big) + \tilde{f}_s \big(\alpha(m,m)\big) \;.
\end{split}
\end{equation}
The differentials of the Kruskal-like coordinate system are
\begin{equation*}
\begin{aligned}
dV^+ &= \kappa_+ V^+(dt + \frac{r^2 + a^2}{\Delta}dr) \qquad \qquad &&dV^- = \kappa_+ V^-(\frac{r^2 + a^2}{\Delta} dr - dt) \\
d\theta &= d\theta \qquad \qquad &&\;\;\, d\Phi = d\varphi - \frac{a}{r_+^2 + a^2} dt \;,
\end{aligned}
\end{equation*}
which, together with \eqref{gInverse}, easily yields that the components of the inverse metric in Kruskal-like coordinates satisfy
\begin{equation}\label{EqRegKruskalMetric}
\begin{split}
g^{\theta \theta} &= \frac{1}{\rho^2} \qquad  \textnormal{and} \qquad  g^{\theta \mu} = 0 \textnormal{ for } \mu \neq \theta \\
g^{\Phi \Phi} &= \frac{1}{\rho^2 \sin^2 \theta} + g^{\Phi \Phi}_{\mathrm{rem}} \qquad \textnormal{with }\; \; g^{\Phi \Phi}_{\mathrm{rem}} \in C^{\infty}(\Mun) \\
g^{\Phi V^+} &, g^{\Phi V^-}  \in C^{\infty}(\Mun)\quad  (\textnormal{ i.e., they do not have poles in $\theta$}).
\end{split}
\end{equation}
Moreover, we note that $\Box_g \Phi = 0$ and $\Box_g \theta = \frac{\cos \theta}{\rho^2 \sin \theta}$ away from the axis $\theta = 0, \pi$.  Also using \eqref{EqSpinWeightedLap} we obtain from \eqref{EqTeukBoxToG}
\begin{equation}\label{EqRewriteTeukReg}
\begin{split}
0 &= \tilde{\mathcal{T}}_{[s]} \big(\alpha(m,m)) \\
&= \sum_{\substack{(\mu, \eta) \notin \\ \{(\theta, \theta), (\Phi, \Phi)\}}} g^{\mu \nu} \mathcal{L}_{\rd_\mu} \mathcal{L}_{\rd_\nu} \big(\alpha(m,m)\big)  + \sum_{\mu \neq \theta, \Phi} (\Box_g x^{\mu}) \mathcal{L}_{\rd_\mu} \big(\alpha(m,m)\big)  + \frac{1}{\rho^2} \swl\big(\alpha(m,m)\big) \\ 
&\qquad +g^{\Phi \Phi}_{\mathrm{rem}} \mathcal{L}_{\rd_\Phi} \mathcal{L}_{\rd_\Phi} \big(\alpha(m,m)\big)   + \tilde{X}_s^{V^+} \mathcal{L}_{\rd_{V^+}} \big(\alpha(m,m)\big)  + \tilde{X}_s^{V^-} \mathcal{L}_{\rd_{V^-}} \big(\alpha(m,m)\big)\\
&\qquad + \tilde{X}_s^{\Phi} \mathcal{L}_{\rd_\Phi} \big(\alpha(m,m)\big) + \tilde{f}_s \big(\alpha(m,m)\big) \;.
\end{split}
\end{equation}
Using now  \eqref{EqSWLV}, \eqref{LieZTilde}, $\mathcal{L}_{\rd_{V^+}} m = \mathcal{L}_{\rd_{V^-}} m = \mathcal{L}_{\rd_{\Phi}} m =0$, and again Remark \ref{RemSpinWSpacetimeS2Tensor} yields that $\tilde{\psi} = \alpha(m,m)$ satisfies \eqref{EqRewriteTeukReg} if, and only if, $\alpha$ satisfies
\begin{equation}\label{TeukTensorEq}
\begin{split}
0 &=  \sum_{\substack{(\mu, \eta) \notin \\ \{(\theta, \theta), (\Phi, \Phi)\}}} g^{\mu \nu} \mathcal{L}_{\rd_\mu} \mathcal{L}_{\rd_\nu} \alpha + \sum_{\mu \neq \theta, \Phi} (\Box_g x^{\mu}) \mathcal{L}_{\rd_\mu} \alpha + \frac{1}{\rho^2} (\mathcal{L}_{Z_{1, r_+}}^2 + \mathcal{L}_{Z_{2, r_+}}^2 + \mathcal{L}_{Z_{3, r_+} }^2+ s + s^2) \alpha\\ 
&\qquad +g^{\Phi \Phi}_{\mathrm{rem}} \mathcal{L}_{\rd_\Phi} \mathcal{L}_{\rd_\Phi} \alpha   + \tilde{X}_s^{V^+} \mathcal{L}_{\rd_{V^+}} \alpha + \tilde{X}_s^{V^-} \mathcal{L}_{\rd_{V^-}} \alpha + \tilde{X}_s^{\Phi} \mathcal{L}_{\rd_\Phi} \alpha + \tilde{f}_s \alpha \;.
\end{split}
\end{equation}
Here, the vector fields $Z_{i, r_+}$ are the vector fields from \eqref{DefEqVectorFieldZ} with $\varphi$ replaced by $\Phi_{r_+}$. Note that firstly the equation \eqref{TeukTensorEq} extends regularly to the axis $\theta \in \{0, \pi\}$ and secondly it also extends regularly to all of $\Mun$ by virtue of 
$\tilde{X}_s^{V^+} , \tilde{X}_s^{V^-} , \tilde{X}_s^{\Phi}, \tilde{f}_s \in C^\infty(\underline{\mathcal{M}})$ and \eqref{EqRegKruskalMetric}.

It is now easy to see that \eqref{TeukTensorEq} is a tensorial wave equation with principal symbol $g^{-1}$ and thus the initial value problem is well-posed\footnote{For example one can reduce it to the initial value problem for a scalar wave equation as follows: Choose a frame field $(f_1, f_2, f_3, f_4)$ for $T\Mun$ that is smooth away from $\theta = 0$ and another one, $(\hat{f}_1, \hat{f}_2, \hat{f}_3, \hat{f}_4)$, that is smooth away from $\theta = \pi$. Equation \eqref{TeukTensorEq} yields now induced scalar equations for the components of $\alpha$ with respect to the frame  $(f_1, f_2, f_3, f_4)$, which, by putting the principal symbol back together, are manifestly  wave equations with principal symbol $g^{-1}$ on $\Mun \setminus \{\theta = 0\}$ -- and analogously for the components of $\alpha$ with respect to the hatted frame field. Given geometric initial data for \eqref{TeukTensorEq} one can now solve for the components of $\alpha$ with respect to $(f_1, f_2, f_3, f_4)$, and also with respect to the hatted frame field, in their corresponding domains of dependence (recall that $\theta = 0$, $\theta = \pi$ is removed from $\Mun$, respectively). By virtue of \eqref{TeukTensorEq} being a geometric equation, each set of  solutions transforms to solutions of the other set under the change of frame -- whenever they are both defined. By uniqueness of the initial value problem, the untransformed and the transformed sets have to agree and we can now patch the two sets of solutions together to obtain a local solution $\alpha$ of \eqref{TeukTensorEq}. We then iterate this procedure. Hence, the main point of this part of the appendix was to show that the Teukolsky equation \eqref{TeukEqRegKruskal} is the scalarisation of a regular geometric equation -- which is not surprising at all given its derivation...}. Taking the trace of \eqref{TeukTensorEq} with respect to $\slashed{g}_{\Sp^2}$ shows that if $\alpha \in \Gamma^\infty(S^2T^*\Mun)$ satisfies \eqref{TeukTensorEq}, then the trace of $\alpha$ satisfies a homogeneous wave equation. Similarly, inserting $\rd_{V^+}$ or $\rd_{V^-}$ into one of the components of \eqref{TeukTensorEq} (and using that the Lie-bracket of coordinate vector fields vanishes) shows that  $\alpha(\rd_{V^+}, \cdot)$ and $\alpha(\rd_{V^-}, \cdot)$ satisfy homogeneous wave equations. The same holds for the antisymmetric part of $\alpha$. Thus, symmetric and trace-free $\Sp^2$ initial data (cf.\ Remark \ref{RemSpinWSpacetimeS2Tensor}) for \eqref{TeukTensorEq} gives rise to a symmetric and trace-free $\Sp^2$ solution. Finally, we recall that by Remark \ref{RemSpinWSpacetimeS2Tensor} initial data for $\tilde{\psi}$ for equation \eqref{TeukEqRegKruskal} uniquely determines geometric symmetric and trace-free $\Sp^2$ initial data for $\alpha$ for \eqref{TeukTensorEq}. This establishes well-posedness for the Teukolsky equation \eqref{TeukEqRegKruskal} in $\Mun$.

\newpage
\section{Commutator computations for \eqref{TeukolskyEquationHat}}
\label{AppendixCommTeukHat}

The second order terms of $ \hat{\mathcal{T}}_{[s]}$ are 
\begin{equation*}
 a^2 \sin^2 \theta \,\partial_{v_-}^2\hat{\psi} - 2a \,\partial_{v_-}\partial_{\varphi_-} \hat{\psi} + 2(r^2 + a^2)\, \partial_{v_-}\partial_r \hat{\psi} -2 a\, \partial_{\varphi_-}\partial_r \hat{\psi} 
 + \Delta \,\partial_r^2 \hat{\psi} + \mathring{\slashed{\Delta}}_{[s]} \hat{\psi} \;.
\end{equation*}
We use $v_-^q \big(-(1+ \lambda \Delta)\partial_r + (1 + \lambda \Delta)\rd_{v_-}\big) \overline{\hp}$ as a multiplier and compute the commutator expressions in the following individually for the $\rd_r$ component and the $\rd_{v_-}$ component of the multiplier, term by term. We will use the notation $\underset{a.i.}{=}$ to denote equality after integration over the spheres with respect to $\vols$.

\subsection{The multiplier $-v_-^q(1 + \lambda \Delta)\overline{\partial_r \hat{\psi}}$} 
\label{ApHatR}
\begin{align*}
-v_-^q(1+\lambda \Delta) a^2 \sin^2 \theta \, \Rea(\partial_{v_-}^2\hp \overline{\rd_r \hp}) &= -\rd_{v_-} \Big(a^2 \sin^2 \theta v_-^q (1+ \lambda \Delta) \Rea(\partial_{v_-} \hp \overline{\rd_r \hp})\Big) \\[4pt]
&\quad + qv_-^{q-1}a^2 \sin^2 \theta (1 + \lambda \Delta) \Rea(\rd_{v_-}\hp \overline{\rd_r \hp})\\[4pt]
&\quad + \rd_r\Big(\frac{1}{2}a^2 \sin^2 \theta v_-^q(1+ \lambda\Delta)|\rd_{v_-}\hp|^2\Big)
- \uwave{\frac{1}{2} a^2 \sin^2\theta v_-^q \lambda \partial_r \Delta |\rd_{v_-}\hp|^2} \\[1cm]
v_-^q(1 + \lambda \Delta) 2a \Rea( \rd_{v_-}\rd_{\varphi_-} \hp  \overline{\rd_r \hp}) &\eai\rd_{v_-} \Big(a v_-^q (1 + \lambda \Delta) \Rea(\rd_{\varphi_-} \hp \overline{\rd_r \hp})\Big) - aqv_-^{q-1}(1 + \lambda \Delta)\Rea(\rd_{\varphi_-}\hp \overline{\rd_r\hp}) \\[4pt]
&\quad -\rd_r\Big(av_-^q(1+\lambda \Delta)\Rea(\rd_{\varphi_-}\hp \overline{\rd_{v_-}\hp})\Big) + \uwave{av_-^q\lambda(\rd_r\Delta)\Rea(\rd_{\varphi_-}\hp \overline{\rd_{v_-} \hp})} \\[1cm]
-v_-^q(1+ \lambda \Delta)2(r^2 + a^2) \Rea(\rd_{v_-}\rd_r\hp \overline{\rd_r\hp}) &= - \rd_{v_-}\Big(v_-^q(1+\lambda \Delta)(r^2 + a^2)|\rd_r\hp|^2\Big) + qv_-^{q-1}(r^2 + a^2)(1+ \lambda \Delta)|\rd_r\hp|^2 \\[1cm]
v_-^q(1+ \lambda \Delta)2a\Rea(\rd_r\rd_{\varphi_-}\hp \overline{\rd_r \hp}) &\eai 0 \\[1cm]
-v_-^q(1+ \lambda \Delta)\Delta \Rea(\rd_r^2\hp \overline{\rd_r\hp} )&= - \rd_r\Big(\frac{1}{2}v_-^q(1+ \lambda \Delta)\Delta|\rd_r\hp|^2\Big) + \dashuline{\frac{1}{2}v_-^q \rd_r\Delta(1+2\lambda \Delta)|\rd_r\hp|^2}\\[1cm]
-v_-^q(1+ \lambda \Delta)\Rea(\mathring{\slashed{\Delta}}_{[s]} \hat{\psi} \overline{\rd_r\hp}) &\eai -\rd_r\Big(\frac{1}{2} v_-^q(1+ \lambda \Delta)(s+s^2)|\hp|^2\Big) +\dotuline{\frac{1}{2}v_-^q\lambda\rd_r\Delta(s+s^2)|\hp|^2} \\
&\quad + \rd_r\Big(\frac{1}{2}v_-^q(1+ \lambda \Delta)\sum_i|\widetilde{Z}_{i,-}\hp|^2\Big) 
- \uwave{\frac{1}{2} v_-^q\lambda \rd_r \Delta \sum_i|\widetilde{Z}_{i,-}\hp|^2} 
\end{align*}

\newpage
\subsection{The multiplier $v_-^q(1 + \lambda \Delta)\overline{\partial_{v_-} \hat{\psi}}$} \label{ApHatV}
\begin{align*}
 v_-^q(1+ \lambda \Delta)a^2 \sin^2 \theta \, \Rea(\rd_{v_-}^2\hp \overline{\rd_{v_-}\hp})&= \rd_{v_-}\Big(\frac{1}{2} v_-^q(1+ \lambda \Delta)a^2 \sin^2\theta \, |\rd_{v_-}\hp|^2\Big) - \frac{1}{2}qv_-^{q-1}(1 + \lambda \Delta)a^2 \sin^2 \theta\, |\rd_{v_-}\hp|^2 \\[1cm]
 - v_-^q(1+ \lambda \Delta)2a\Rea(\rd_{v_-}\rd_{\varphi_-}\hp \overline{\rd_{v_-}\hp}) &\eai 0 \\[1cm]
  v_-^q(1+ \lambda \Delta)2(r^2 + a^2)\Rea(\rd_{v_-}\rd_r\hp \overline{\rd_{v_-}\hp}) &= \rd_r\Big( v_-^q(1+ \lambda \Delta)(r^2 + a^2)|\rd_{v_-}\hp|^2\Big) \\[4pt]
  &\quad - \uwave{v_-^q\big((r^2 + a^2)\lambda \partial_r \Delta + 2r(1 + \lambda \Delta)\big)|\rd_{v_-}\hp|^2}\\[1cm]
  - v_-^q(1+ \lambda \Delta)2a \Rea(\rd_r\rd_{\varphi_-}\hp\overline{\rd_{v_-}\hp}) &\eai   \rd_{v_-} \Big(a v_-^q (1 + \lambda \Delta) \Rea(\rd_{\varphi_-} \hp \overline{\rd_r \hp})\Big) - aqv_-^{q-1}(1 + \lambda \Delta)\Rea(\rd_{\varphi_-}\hp \overline{\rd_r\hp}) \\[4pt]
&\quad -\rd_r\Big(av_-^q(1+\lambda \Delta)\Rea(\rd_{\varphi_-}\hp \overline{\rd_{v_-}\hp})\Big) + \uwave{av_-^q\lambda(\rd_r\Delta)\Rea(\rd_{\varphi_-}\hp \overline{\rd_{v_-} \hp})} \\[1cm]
v_-^q(1+ \lambda \Delta)\Delta\Rea(\rd_r^2\hp \overline{\rd_{v_-} \hp}) &= \rd_r\Big( v_-^q(1+ \lambda \Delta) \Delta \Rea(\rd_r\hp \overline{\rd_{v_-}\hp})\Big) - v_-^q\rd_r\Delta(1+2\lambda \Delta)\Rea(\rd_r\hp \overline{\rd_{v_-}\hp}) \\[4pt]
&-\rd_{v_-}\Big(\frac{1}{2} v_-^q(1+ \lambda \Delta)\Delta|\rd_r\hp|^2\Big) + \frac{1}{2}qv_-^{q-1}(1 + \lambda \Delta)\Delta |\rd_r \hp|^2 \\[1cm]
 v_-^q(1+ \lambda \Delta)\Rea(\mathring{\slashed{\Delta}}_{[s]} \hat{\psi} \overline{\rd_{v_-}\hp}) &\eai \rd_{v_-}\Big(\frac{1}{2}  v_-^q(1+ \lambda \Delta)(s+s^2)|\hp|^2\Big) - \frac{1}{2}qv_-^{q-1}(1 + \lambda \Delta)(s+ s^2)|\hp|^2 \\[4pt]
 &-\rd_{v_-}\Big(\frac{1}{2}  v_-^q(1+ \lambda \Delta)\sum_i |\widetilde{Z}_{i,-} \hp|^2\Big) + \frac{1}{2}qv_-^{q-1}(1 + \lambda \Delta)\sum_i |\widetilde{Z}_{i,-} \hp|^2
\end{align*}

\section{Commutator computations for \eqref{TeukolskyStar}} \label{AppendixTeukStar}

The second order terms of $ \mathcal{T}_{[s]}$  in $\{v_+, r, \theta, \varphi_+\}$ coordinates are 
\begin{equation*}
 a^2 \sin^2 \theta \,\partial_{v_+}^2\psi + 2a \,\partial_{v_+}\partial_{\varphi_+} \psi + 2(r^2 + a^2)\, \partial_{v_+}\partial_r \psi +2 a\, \partial_{\varphi_+}\partial_r \psi 
 + \Delta \,\partial_r^2 \psi + \mathring{\slashed{\Delta}}_{[s]} \psi \;.
\end{equation*}
We use $\chi(v_+) \big(-(1+ \lambda \Delta)\partial_r + (1 + \lambda \Delta)\rd_{v_+}\big) \overline{\psi}$ as a multiplier and compute the commutator expressions in the following individually for the $\rd_r$ component and the $\rd_{v_+}$ component of the multiplier, term by term. Note that due to formal similarity all these expressions can be easily inferred from the computations in Appendix \ref{AppendixCommTeukHat} (or vice versa). They are listed here nevertheless for the convenience of the reader.
Again we use the notation $\underset{a.i.}{=}$ to denote equality after integration over the spheres with respect to $\vols$.

We also use $\chi(v_+)(1 + \lambda \Delta)(-\rd_r + \rd_{v_+} + \frac{a}{r_-^2 + a^2} \rd_{\varphi_+})\overline{\psi}$ as a multiplier. The $\rd_{\varphi_+} \overline{\psi}$ component is also computed here separately.

\newpage
\subsection{The multiplier $-\chi(v_+)(1 + \lambda \Delta)\overline{\partial_r \psi}$} 
\label{ApR}
\begin{align*}
-\chi(v_+)(1+\lambda \Delta) a^2 \sin^2 \theta \, \Rea(\partial_{v_+}^2\psi \overline{\rd_r \psi}) &= -\rd_{v_+} \Big(a^2 \sin^2 \theta \chi(v_+) (1+ \lambda \Delta) \Rea(\partial_{v_+} \psi \overline{\rd_r \psi})\Big)  \\[4pt]
&\quad+ \chi'(v_+)a^2 \sin^2 \theta (1 + \lambda \Delta) \Rea(\rd_{v_+}\psi \overline{\rd_r \psi})\\[4pt]
&\quad+ \rd_r\Big(\frac{1}{2}a^2 \sin^2 \theta \chi(v_+)(1+ \lambda\Delta)|\rd_{v_+}\psi|^2\Big) \\[4pt]
&\quad - \uwave{\frac{1}{2} a^2 \sin^2\theta \chi(v_+) \lambda \partial_r \Delta |\rd_{v_+}\psi|^2} \\[1cm]
-\chi(v_+)(1 + \lambda \Delta) 2a \Rea( \rd_{v_+}\rd_{\varphi_+} \psi  \overline{\rd_r \psi}) &\eai-\rd_{v_+} \Big(a \chi(v_+)(1 + \lambda \Delta) \Rea(\rd_{\varphi_+} \psi \overline{\rd_r \psi})\Big)\\[4pt] 
&\quad + a\chi'(v_+)(1 + \lambda \Delta)\Rea(\rd_{\varphi_+}\psi \overline{\rd_r\psi}) \\[4pt]
&\quad +\rd_r\Big(a\chi(v_+)(1+\lambda \Delta)\Rea(\rd_{\varphi_+}\psi \overline{\rd_{v_+}\psi})\Big)\\ 
&\quad - \uwave{a\chi(v_+)\lambda(\rd_r\Delta)\Rea(\rd_{\varphi_+}\psi \overline{\rd_{v_+} \psi})} \\[1cm]
-\chi(v_+)(1+ \lambda \Delta)2(r^2 + a^2) \Rea(\rd_{v_+}\rd_r\psi \overline{\rd_r\psi}) &= - \rd_{v_+}\Big(\chi(v_+)(1+\lambda \Delta)(r^2 + a^2)|\rd_r\psi|^2\Big) \\[4pt] 
&\quad + \chi'(v_+)(r^2 + a^2)(1+ \lambda \Delta)|\rd_r\psi|^2 \\[1cm]
-\chi(v_+)(1+ \lambda \Delta)2a\Rea(\rd_r\rd_{\varphi_+}\psi \overline{\rd_r \psi}) &\eai 0 \\[1cm]
-\chi(v_+)(1+ \lambda \Delta)\Delta\Rea(\rd_r^2\psi \overline{\rd_r\psi}) &= - \rd_r\Big(\frac{1}{2}\chi(v_+)(1+ \lambda \Delta)\Delta|\rd_r\psi|^2\Big) + \dashuline{\frac{1}{2}\chi(v_+) \rd_r\Delta(1+2\lambda \Delta)|\rd_r\psi|^2}\\[1cm]
-\chi(v_+)(1+ \lambda \Delta)\Rea(\mathring{\slashed{\Delta}}_{[s]} \psi \overline{\rd_r\psi}) &\eai -\rd_r\Big(\frac{1}{2} \chi(v_+)(1+ \lambda \Delta)(s+s^2)|\psi|^2\Big) +\dotuline{\frac{1}{2}\chi(v_+)\lambda\rd_r\Delta(s+s^2)|\psi|^2} \\
&\quad + \rd_r\Big(\frac{1}{2}\chi(v_+)(1+ \lambda \Delta)\sum_i|\widetilde{Z}_{i,+}\psi|^2\Big) 
- \uwave{\frac{1}{2} \chi(v_+)\lambda \rd_r \Delta \sum_i|\widetilde{Z}_{i,+}\psi|^2} 
\end{align*}

\newpage
\subsection{The multiplier $\chi(v_+)(1 + \lambda \Delta)\overline{\partial_{v_+} \psi}$} \label{ApV}
\begin{align*}
 \chi(v_+)(1+ \lambda \Delta)a^2 \sin^2 \theta \, \Rea(\rd_{v_+}^2\psi \overline{\rd_{v_+}\psi})&= \rd_{v_+}\Big(\frac{1}{2} \chi(v_+)(1+ \lambda \Delta)a^2 \sin^2\theta \, |\rd_{v_+}\psi|^2\Big) \\[4pt]
 &\quad - \frac{1}{2}\chi'(v_+)(1 + \lambda \Delta)a^2 \sin^2 \theta\, |\rd_{v_+}\psi|^2 \\[1cm]
  \chi(v_+)(1+ \lambda \Delta)2a\Rea(\rd_{v_+}\rd_{\varphi_+}\psi \overline{\rd_{v_+}\psi}) &\eai 0 \\[1cm]
  \chi(v_+)(1+ \lambda \Delta)2(r^2 + a^2)\Rea(\rd_{v_+}\rd_r\psi \overline{\rd_{v_+}\psi}) &= \rd_r\Big( \chi(v_+)(1+ \lambda \Delta)(r^2 + a^2)|\rd_{v_+}\psi|^2\Big) \\[4pt]
  &\quad - \uwave{\chi(v_+)\big((r^2 + a^2)\lambda \partial_r \Delta + 2r(1 + \lambda \Delta)\big)|\rd_{v_+}\psi|^2}\\[1cm]
   \chi(v_+)(1+ \lambda \Delta)2a \Rea(\rd_r\rd_{\varphi_+}\psi\overline{\rd_{v_+}\psi}) &\eai-\rd_{v_+} \Big(a \chi(v_+)(1 + \lambda \Delta) \Rea(\rd_{\varphi_+} \psi \overline{\rd_r \psi})\Big)\\[4pt] 
&\quad + a\chi'(v_+)(1 + \lambda \Delta)\Rea(\rd_{\varphi_+}\psi \overline{\rd_r\psi}) \\[4pt]
&\quad +\rd_r\Big(a\chi(v_+)(1+\lambda \Delta)\Rea(\rd_{\varphi_+}\psi \overline{\rd_{v_+}\psi})\Big)\\ 
&\quad - \uwave{a\chi(v_+)\lambda(\rd_r\Delta)\Rea(\rd_{\varphi_+}\psi \overline{\rd_{v_+} \psi})} \\[1cm]
\chi(v_+)(1+ \lambda \Delta)\Delta\Rea(\rd_r^2\psi \overline{\rd_{v_+} \psi}) &= \rd_r\Big( \chi(v_+)(1+ \lambda \Delta) \Delta \Rea(\rd_r\psi \overline{\rd_{v_+}\psi})\Big)\\[4pt] 
&\quad - \chi(v_+)\rd_r\Delta(1+2\lambda \Delta)\Rea(\rd_r\psi \overline{\rd_{v_+}\psi}) \\[4pt]
&\quad -\rd_{v_+}\Big(\frac{1}{2} \chi(v_+)(1+ \lambda \Delta)\Delta|\rd_r\psi|^2\Big) + \frac{1}{2}\chi'(v_+)(1 + \lambda \Delta)\Delta |\rd_r \psi|^2 \\[1cm]
 \chi(v_+)(1+ \lambda \Delta)\Rea(\mathring{\slashed{\Delta}}_{[s]} \psi \overline{\rd_{v_+}\psi}) &\eai \rd_{v_+}\Big(\frac{1}{2}  \chi(v_+)(1+ \lambda \Delta)(s+s^2)|\psi|^2\Big) - \frac{1}{2}\chi'(v_+)(1 + \lambda \Delta)(s+ s^2)|\psi|^2 \\[4pt]
 &\quad -\rd_{v_+}\Big(\frac{1}{2}  \chi(v_+)(1+ \lambda \Delta)\sum_i |\widetilde{Z}_{i,+} \psi|^2\Big) + \frac{1}{2}\chi'(v_+)(1 + \lambda \Delta)\sum_i |\widetilde{Z}_{i,+} \psi|^2
\end{align*}

\newpage
\subsection{The multiplier $\chi(v_+)(1 + \lambda \Delta)\frac{a}{r_-^2 + a^2}\overline{\partial_{\varphi_+} \psi}$} \label{ApPhi}
\begin{align*}
\chi(v_+)(1 + \lambda \Delta) \frac{a}{r_-^2 + a^2} a^2 \sin^2 \theta \Rea(\rd_{v_+}^2 \psi \overline{\rd_{\varphi_+} \psi}) &\eai \rd_{v_+} \big( \chi(v_+)(1 + \lambda \Delta) \frac{a^3 \sin^2 \theta}{r_-^2 + a^2} \Rea(\rd_{v_+} \psi \overline{\rd_{\varphi_+} \psi}) \\[4pt]
&\quad  - \chi'(v_+)(1 + \lambda \Delta) \frac{a^3 \sin^2 \theta}{r_-^2 + a^2} \Rea(\rd_{v_+} \psi \overline{\rd_{\varphi_+} \psi}) \\[1cm]
\chi (v_+)(1 + \lambda \Delta) \frac{a}{r_-^2 + a^2} 2a \Rea(\rd_{v_+} \rd_{\varphi_+} \psi \overline{\rd_{\varphi_+} \psi}) &= \rd_{v_+}\big( \chi(v_+)(1 + \lambda \Delta) \frac{a^2}{r_-^2 + a^2} | \rd_{\varphi_+} \psi|^2\big) \\[4pt]
&\quad - \chi'(v_+)(1 + \lambda \Delta) \frac{a^2}{r_-^2 + a^2} |\rd_{\varphi_+} \psi|^2 \\[1cm]
\chi(v_+)(1 + \lambda \Delta) \frac{a}{r_-^2 + a^2} 2 (r^2 + a^2) \Rea(\rd_{v_+} \rd_r \psi \overline{\rd_{\varphi_+} \psi}) &\eai \rd_{v_+} \big( \chi(v_+)(1+ \lambda \Delta) \frac{a(r^2 + a^2)}{r_-^2 + a^2} \Rea( \rd_r \psi \overline{\rd_{\varphi_+} \psi}) \big) \\[4pt]
&\quad  -\chi'(v_+)(1 + \lambda \Delta) \frac{a(r^2 + a^2)}{r_-^2 + a^2} \Rea(\rd_r \psi \overline{\rd_{\varphi_+} \psi}) \\[4pt]
&\quad + \rd_r\big( \chi(v_+)(1+ \lambda \Delta) \frac{a(r^2 +a^2)}{r_-^2 +a^2} \Rea(\rd_{\varphi_+} \psi \overline{\rd_{v_+} \psi}) \big)\\[4pt]
&\quad -\uwave{ \chi(v_+) \lambda (\rd_r \Delta) \frac{a(r^2 + a^2)}{r_-^2 + a^2} \Rea(\rd_{\varphi_+} \psi \overline{ \rd_{v_+} \psi})} \\[4pt]
&\quad - \chi(v_+)(1 + \lambda \Delta) \frac{2ar}{r_-^2 + a^2} \Rea( \rd_{\varphi_+} \psi \overline{\rd_{v_+} \psi})\\[1cm]
\chi(v_+)(1 + \lambda \Delta) \frac{a}{r_-^2 + a^2} 2a \Rea(\rd_{\varphi_+} \rd_r \psi \overline{\rd_{\varphi_+} \psi}) &= \rd_r\big(\chi(v_+)(1+ \lambda \Delta) \frac{a^2}{r_-^2 + a^2} |\rd_{\varphi_+} \psi|^2\big) \\[4pt]
&\quad - \uwave{\chi(v_+) \lambda (\rd_r \Delta) \frac{a^2}{r_-^2 + a^2} |\rd_{\varphi_+} \psi|^2} \\[1cm]
\chi(v_+)(1+ \lambda \Delta)\frac{a}{r_-^2 + a^2} \Delta \Rea(\rd_r^2 \psi \overline{ \rd_{\varphi_+} \psi}) &\eai \rd_r\big(\chi(v_+) \frac{a}{r_-^2 + a^2} (1+ \lambda \Delta) \Delta \Rea(\rd_r \psi \overline{ \rd_{\varphi_+} \psi})\big) \\[4pt]
&\quad -\chi(v_+) \frac{a}{r_-^2 + a^2} \lambda (\rd_r \Delta) \Delta \Rea(\rd_r \psi \overline{\rd_{\varphi_+} \psi}) \\[4pt]
&\quad -\chi(v_+) \frac{a}{r_-^2 + a^2}(1+ \lambda \Delta)(\rd_r \Delta)\Rea(\rd_r \psi \overline{\rd_{\varphi_+} \psi}) \\[1cm]
\chi(v_+) (1+ \lambda \Delta)\frac{a}{r_-^2 + a^2}\Rea(\mathring{\slashed{\Delta}}_{[s]} \psi \overline{ \rd_{\varphi_+} \psi}) &\eai 0
\end{align*}

\newpage

\section{Commutator computations for \eqref{TeukolskyEquation3}} \label{AppendixBL}

The second order terms of $\mathcal{T}_{[s]}$ in Boyer-Lindquist coordinates are
\begin{equation*}
-\Big[\frac{(r^2 + a^2)^2}{\Delta} - a^2 \sin^2 \theta\Big] \partial_t^2\psi - \frac{4Mar}{ \Delta} \partial_t \partial_\varphi \psi - \frac{a^2}{ \Delta}\partial^2_\varphi \psi + \Delta \partial_r^2 \psi  +\mathring{\slashed\Delta}_{[s]} \psi \;.
\end{equation*}
We use $-\chi(t) e^{\lambda r} \overline{\rd_r \psi}$ as a multiplier and compute the commutator expressions again term by term.

\subsection{The multiplier $-\chi(t) e^{\lambda r} \overline{\rd_r \psi}$}\label{ApBoyerR}

\begin{align*}
\chi(t) e^{\lambda r} \Big[\frac{(r^2 + a^2)^2}{\Delta} - a^2 \sin^2 \theta\Big] \Rea(\partial_t^2\psi \overline{\rd_r \psi}) &= \rd_t\Big( \chi(t) e^{\lambda r}\Big[\frac{(r^2 + a^2)^2}{\Delta} - a^2 \sin^2 \theta\Big] \Rea(\partial_t\psi \overline{\rd_r \psi})\Big) \\[4pt]
&\quad -\chi'(t) e^{\lambda r}\Big[\frac{(r^2 + a^2)^2}{\Delta} - a^2 \sin^2 \theta\Big] \Rea(\partial_t\psi \overline{\rd_r \psi}) \\[4pt]
&\quad -\frac{1}{2} \rd_r\Big(\chi(t) e^{\lambda r}\Big[\frac{(r^2 + a^2)^2}{\Delta} - a^2 \sin^2 \theta\Big] |\partial_t\psi|^2\Big) \\[4pt]
&\quad+\uwave{\frac{1}{2} \chi(t) e^{\lambda r}\Big[\lambda\Big(\frac{(r^2 + a^2)^2}{\Delta} - a^2 \sin^2 \theta\Big)} + \rd_r \Big(\frac{(r^2 + a^2)^2}{\Delta}\Big)\Big] \uwave{|\rd_t \psi|^2 }\\[1cm]
\chi(t)e^{\lambda r} \frac{4Mar}{\Delta} \Rea(\rd_t\rd_\varphi \psi \overline{\rd_r \psi}) &\eai \frac{1}{2} \rd_t\big(\chi(t) e^{\lambda r} \frac{4Mar}{\Delta} \Rea(\rd_\varphi \psi \overline{\rd_r \psi})\big) \\[4pt]
&\quad -\frac{1}{2} \chi'(t) e^{\lambda r} \frac{4Mar}{\Delta} \Rea(\rd_\varphi \psi \overline{\rd_r \psi}) \\[4pt]
&\quad -\frac{1}{2} \rd_r \big(\chi(t)e^{\lambda r} \frac{4Mar}{\Delta} \Rea(\rd_\varphi \psi \overline{\rd_t \psi})\big)\\[4pt]
&\quad +\uwave{\frac{1}{2} \chi(t)e^{\lambda r} \big(\lambda \frac{4Mar}{\Delta} } +\rd_r\Big(\frac{4Mar}{\Delta}\Big)\big)\uwave{\Rea(\rd_\varphi \psi \overline{\rd_t \psi})}\\[1cm]
\chi(t)e^{\lambda r} \frac{a^2}{\Delta} \Rea(\rd_{\varphi}^2 \psi \overline{\rd_r \psi}) &\eai -\rd_r \big( \frac{1}{2} \chi(t) e^{\lambda r} \frac{a^2}{\Delta}|\rd_\varphi\psi|^2\big) + \uwave{\frac{1}{2} \chi(t) e^{\lambda r}\Big[ \lambda \frac{a^2}{\Delta}} + \rd_r\Big(\frac{a^2}{\Delta}\Big)\Big]\uwave{|\rd_\varphi \psi|^2}\\[1cm]
-\chi(t) e^{\lambda r}{\Delta} \Rea(\rd_r^2 \psi \overline{\rd_r \psi}) &= - \rd_r\big(\frac{1}{2} \chi(t) e^{\lambda r} \Delta |\rd_r \psi|^2\big) + \uwave{\frac{1}{2}\chi(t)e^{\lambda r}(\lambda \Delta } + \rd_r \Delta)\uwave{|\rd_r\psi|^2}\\[1cm]
-\chi(t) e^{\lambda r} \Rea (\mathring{\slashed\Delta}_{[s]}\psi \overline{\rd_r \psi}) &\eai - \rd_r \big( \frac{1}{2} \chi(t) e^{\lambda r}(s + s^2) |\psi|^2 \big) + \frac{1}{2} \chi(t) (s+s^2)\lambda e^{\lambda r} |\psi|^2 \\[4pt]
&\quad + \rd_r \big(\frac{1}{2} \chi(t) e^{\lambda r} \sum_i |\zt_i \psi|^2\big) - \uwave{\frac{1}{2} \chi(t) \lambda e^{\lambda r} \sum_i |\zt_i \psi|^2}
\end{align*}

\bibliographystyle{acm}
\bibliography{Bibly}

\end{document}